\documentclass[aos]{imsart}

\RequirePackage{amsthm,amsmath,amsfonts,amssymb}
\usepackage{color}
\RequirePackage[authoryear]{natbib}
\usepackage{hyperref}
\usepackage{cleveref}
\usepackage{afterpage}
\usepackage{graphicx} 
\usepackage{caption}
\usepackage{subcaption}
\usepackage{float}
\startlocaldefs
\theoremstyle{plain}
\newtheorem{lemma}{Lemma}[section]
\newtheorem{theorem}{Theorem}
\newtheorem{mproposition}{Proposition}
\newtheorem{proposition}{Proposition}[section]
\theoremstyle{remark}
\newtheorem{assumption}{Assumption}
\newtheorem{definition}{Definition}
\newtheorem{remark}{Remark}[section]
\newcommand{\h}[2]{H_{i_{#1}i_{#2}}}
\newcommand{\D}[2]{D_{i_{#1}i_{#2}}}

\newcommand{\tda}[2]{\tilde{A}_{i_{#1}i_{#2}}}

\newcommand{\di}[1]{d_{i_{#1}}}
\newcommand{\fii}[2]{f_{i_{#1}i_{#2}}}
\newcommand{\gi}[1]{g_{i_{#1}}}
\newcommand{\giD}[1]{g_{D,i_{#1}}}

\newcommand{\yi}[1]{y_{i_{#1}}}
\newcommand{\yiD}[1]{y_{D,i_{#1}}}
\newcommand{\Ti}[1]{T_{i_{#1}}}

\newcommand{\sumi}[2]{\sum_{[i_{#1}\ldots i_{#2}]}}

\newcommand{\tdm}{{n}_1^\prime}

\newcommand{\E}{\mathbb{E}}

\newtheorem{corollary}{Corollary}
\newcommand{\Prob}{\mathbb{P}}
\renewcommand{\Pr}{\mathbb{P}}
\newcommand{\cov}{\operatorname{cov}}
\newcommand{\var}{\operatorname{var}}
\newcommand{\op}{o_{\mathbb{P}}}
\newcommand{\Op}{O_{\mathbb{P}}}
\newcommand{\sumtwo}{\sum_{[i_1,i_2]}}
\newcommand{\tr}{\operatorname{tr}}
\newcommand{\asyeq}{\stackrel{\cdot}{\sim}}
\newcommand{\htauadj}{\hat{\tau}_{\text{adj}}}
\newcommand{\htaudb}{\hat{\tau}_{\text{db}}}
\newcommand{\diag}{\operatorname{diag}}

\newcommand{\cre}{\textrm{cre}}
\newcommand{\cf}{\textrm{cf}}
\newcommand{\adj}{\textrm{adj}}
\newcommand{\hd}{\textrm{hd}}
\newcommand{\cb}{\textrm{cb}}
\newcommand{\db}{\textrm{db}}

\newcommand{\lin}{\textrm{lin}}
\newcommand{\sca}{\operatorname{Scale}}
\newcommand{\tra}{\operatorname{Trans}}
\newcommand{\ub}{\textrm{ub}}

\newcommand{\cI}{\mathcal{I}}
\newcommand{\bR}{\mathbb{R}}

\newcommand{\ii}{\mathrm{i}}
\newcommand{\e}{\varepsilon}
\newcommand{\dd}{\mathrm{d}}

\newcommand{\bv}{\mathbf v}
\newcommand{\bP}{\mathbf P}
\newcommand{\cW}{\mathcal W}
\newcommand{\mW}{\mathbf W}
\newcommand{\bW}{\boldsymbol{\mathcal W}}
\newcommand{\bcH}{\boldsymbol{\mathcal H}}
\newcommand{\bD}{\mathbf{D}}
\newcommand{\wt}{\tilde}

\newcommand{\OO}{O}
\newcommand{\bX}{\boldsymbol{X}}
\newcommand{\bXX}{\boldsymbol{V}}
\newcommand{\CMO}{\textrm{CMO}}
\newcommand{\CMA}{\textrm{CMA}}
\renewcommand{\cal}{\mathcal}

\def\bs{\boldsymbol}
\endlocaldefs
\def\unadj{\textrm{unadj}}

\def\rev{\color{black}}
\def\revone{\color{black}}

\usepackage[sort=use]{glossaries}

\newglossaryentry{Sa2}{name={\ensuremath{\bs{S}_{\bs{a}}^2}}, description={Finite population variance of $\bs{a}$}}
\newglossaryentry{Sab2}{name={\ensuremath{
\bs{S}_{\bs{a}, \bs{b}}}}, description={Finite population covariance between $\bs{a}$ and $\bs{b}$}}
\newglossaryentry{SAa2}{name={\ensuremath{\bs{S}^2_{\bs{A}, \bs{a}}}}, description={Scaled finite population variance of $\bs{a}$}}
\newglossaryentry{SAab2}{name={\ensuremath{
\bs{S}_{\bs{A}, \bs{a}, \bs{b}}}}, description={Scaled finite population covariance between $\bs{a}$ and $\bs{b}$}}
\newglossaryentry{bsQ}{name={\ensuremath{\bs{Q}}}, description={An $n\times n$ matrix used to define $\sigma^2_{\hd}$}}
\newglossaryentry{bsB}{name={\ensuremath{\bs{B}}}, description={An $n\times n$ matrix used to define $\sigma^2_{\hd}$}}
\newglossaryentry{alpha}{name={\ensuremath{\alpha}}, description={Covariate dimension to sample size ratio $p / n$}}
\newglossaryentry{R2}{name={\ensuremath{R^2}}, description={Canonical correlation between the covariate vector and the weighted potential outcomes}}

\newglossaryentry{diagA}{name={\ensuremath{\diag\{\bs{A}\}}}, description={The square matrix which keeps the diagonal elements of matrix $\bs{A}$}}
\newglossaryentry{diag-A}{name={\ensuremath{\diag^-\{\bs{A}\}}}, description={The square matrix which keeps the off-diagonal elements of matrix $\bs{A}$}}

\newglossaryentry{scre}{name={\ensuremath{\hat{\sigma}^2_\cre}}, description={The variance estimator of unadjusted estimator}}

\newglossaryentry{hsigmacre}{name={\ensuremath{\hat{\sigma}^2_\cre}}, description={An estimator of ${\sigma}_{\cre}^2$}}
\newglossaryentry{sumonek}{name={\ensuremath{\sum_{[i_{1}\ldots i_{k}]}}}, description= {Summation over all $(i_1,\ldots,i_k)$ with mutually distinct elements in $[n]$}}
\newglossaryentry{tldA}{name={\ensuremath{\bs{\tilde{A}}}}, description= {A centered matrix of matrix $\bs{A}$}}
\newglossaryentry{bsAtwo}{name={\ensuremath{\|\bs{A}\|_2}}, description= {The operator norm of $\bs{A}$}}
\newglossaryentry{tldU}{name={\ensuremath{\tilde{U}}}, description= {A centered random variable defined as $\tilde{U} := U - \E U$}}
\newglossaryentry{siz}{name={\ensuremath{s_i(z)}}, description= {$s_i(z) := H_{ii}(Y_i(z) - \bar{Y}(z)) - \frac{1}{n} \sum_{j=1}^n H_{jj} (Y_j(z) - \bar{Y}(z))$}}
\newglossaryentry{diz}{name={\ensuremath{d_i(z)}}, description= {The $i$-th entry of the vector $\bs{H}(Y_1(z)-\bar{Y}(z), \ldots, Y_n(z)-\bar{Y}(z))^\top$}}
\newglossaryentry{Az}{name={\ensuremath{\bs{A}(z)}}, description= {An $n\times n$ matrix defined as $\bs{A}(z) := \bs{H} \diag\left\{(Y_1(z) - \bar{Y}(z), \ldots, Y_n(z) - \bar{Y}(z))^\top\right\}$}}
\newglossaryentry{tldbsZ}{name={\ensuremath{\tilde{\bs{Z}}}}, description= {The vector of centered treatment assignments, i.e., $(\tilde{Z}_1, \ldots, \tilde{Z}_n)$}}
\newglossaryentry{deltaij}{name={\ensuremath{\delta_{ij}}}, description= {An indicator function that is equal to $1$ if and only if $i=j$}}
\newglossaryentry{diagy}{name={\ensuremath{\diag\{\bs{y}\}}}, description= {Diagonal square matrix having vector $\bs{y}$ as its diagonal elements}}
\newglossaryentry{bsT}{name={\ensuremath{\bs{T}}}, description= {The vector of the indicators of Bernoulli random sampling constructed by H\`{a}jek‘s Coupling of $\bs{Z}$}}

\makeglossaries

\begin{document}

\begin{frontmatter}
\title{Debiased regression adjustment in completely randomized experiments with moderately high-dimensional covariates}
\runtitle{Debiased regression adjustment}

\begin{aug}
\author[A]{\fnms{Xin}~\snm{Lu}\ead[label=e1]{lux20@mails.tsinghua.edu.cn}},
\author[B]{\fnms{Fan}~\snm{Yang}\ead[label=e2]{fyangmath@tsinghua.edu.cn}}
\and
\author[C]{\fnms{Yuhao}~\snm{Wang}\ead[label=e3]{yuhaow@tsinghua.edu.cn}}

\address[A]{Department of Statistics and Data Science, Tsinghua University\printead[presep={,\ }]{e1}}

\address[B]{Yau Mathematical Sciences Center, Tsinghua University\printead[presep={,\ }]{e2}}

\address[C]{Institute for Interdisciplinary Information Sciences, Tsinghua University\printead[presep={,\ }]{e3}}
\end{aug}

\begin{abstract}
Completely randomized experiment is the gold standard for causal inference. When the covariate information for each experimental candidate is available, one typical way is to include them in covariate adjustments for more accurate treatment effect estimation. In this paper, we investigate this problem under the randomization-based framework, i.e., that the covariates and potential outcomes of all experimental candidates are assumed as deterministic quantities and the randomness comes solely from the treatment assignment mechanism. Under this framework, to achieve asymptotically valid inference, existing estimators usually require either (i) 
that the dimension of covariates $p$ is much smaller than the sample size $n$;
or (ii) certain sparsity constraints on the linear representations of potential outcomes constructed via possibly high-dimensional covariates. In this paper, we consider the moderately high-dimensional regime where $p$ is allowed to be in the same order of magnitude as $n$. 
We develop a novel debiased estimator with a corresponding inference procedure and establish its asymptotic normality under mild assumptions. Our estimator is model-free and does not require any sparsity constraint on potential outcome's linear representations. We also discuss its asymptotic efficiency improvements over the unadjusted treatment effect estimator under different dimensionality constraints.
Numerical analysis confirms that compared to other regression adjustment based treatment effect estimators, our debiased estimator performs well in moderately high dimensions.
\end{abstract}



\end{frontmatter}


\section{Introduction} \label{sec:intro}


Since the seminal work of \cite{fisher1935design}, completely randomized experiment has been the gold standard for causal inference. By using only the randomization in treatment assignments as the reasoned basis, completely randomized experiments can provide a valid inference of treatment effects without any model or distributional assumptions on the experimental candidates, such as being i.i.d. sampled from some superpopulation or some other model assumption that may be unverifiable in practice. 
 Such inference is often called randomization-based or design-based inference, sometimes also called finite-population-based inference to emphasize its focus on just candidates in the experiment. Evaluating causal effects under this inferential framework has been an active area of research in the past few years  \citep[e.g.][]{Lin2013Agnostic,bloniarz2016lasso, Li9157,lei2021regression,wang2022rerandomization}, which is also the framework we focus on in this paper.

When the covariate information of each experimental candidate is available, it is often popular to use regression adjustment in the analysis stage to utilize the additional covariate information to improve estimation precision \citep{Lin2013Agnostic,bloniarz2016lasso,negi2021revisiting,lei2021regression}. \citet{Lin2013Agnostic} showed that in completely randomized experiments, regression adjustment can improve the asymptotic efficiency of average treatment effect estimation when the dimension of covariates $p$ is fixed as the sample size $n$ goes to infinity. However, in modern experiments, researchers can collect a large number of covariates. It is important to develop methodology and theory for high-dimensional settings where $p\rightarrow \infty$ as the sample size goes to infinity. 

In the high-dimensional regime where $p \gg n$, the LASSO-adjusted estimator \citep{bloniarz2016lasso} has been proposed to estimate causal effect with high accuracy. However, it is under the requirement that the potential outcomes can be well represented by a sparse linear function of the high-dimensional covariates, which can be unrealistic in practice.
In the lower dimensional regime, \cite{lei2021regression} improved Lin's method by debiasing and guaranteed improvement of estimation efficiency compared to the difference in means estimator when $p = O(n^{2/3} / (\log n)^{1/3})$, without any assumption on the sparsity of potential outcome's linear representations. {\revone Other research works that consider the diverging $p$ regime include \cite{wu2018loop,chiang2023regression,chang2024exact}, which will be discussed further in literature review (\Cref{sec:framework}).}

In this paper, we consider the \emph{moderately} high-dimensional regime where $p$ is allowed to be in the same order of magnitude as $n$. We develop a novel debiased estimator with a corresponding inference procedure and establish its asymptotic normality and inference validity. Our estimator guarantees improvement of estimation efficiency over the unadjusted estimator in the regime $p=o(n)$. In the higher dimensional regime where $p$ can be in the same order of magnitude as $n$, we prove that if the canonical correlation between potential outcomes and covariates is sufficiently large relative to $p/n$, we can still guarantee efficiency improvement compared to the unadjusted estimator. Noteworthy, our theory for asymptotic normality and inference validity is based on some standard regularity conditions on potential outcomes and their empirical regression residuals; beyond that, no assumption is required on the observed covariate features.



Before moving forward, it would be convenient to introduce some notations that will be used in the rest of this paper. Given $n$ $d$-dimensional samples $\bs{a}_1, \ldots, \bs{a}_n$ of a variable $\bs{a}$, we let $\bar{\bs{a}}$ denote its empirical average, and let $\gls*{Sa2} := \frac{1}{n - 1} \sum_{i=1}^n (\bs{a}_i - \bar{\bs{a}})(\bs{a}_i - \bar{\bs{a}})^\top$ be the empirical covariance matrix of $\bs{a}$. Given $n$ samples of two variables $\bs{a}$ and $\bs{b}$, we write $\gls*{Sab2} := \frac{1}{n - 1} \sum_{i=1}^n (\bs{a}_i - \bar{\bs{a}})(\bs{b}_i - \bar{\bs{b}})^\top$ as its empirical covariance matrix. Analogously, given a matrix $\bs{A} \in \bR^{n \times n}$, we define the \emph{scaled} variance of variable $\bs{a}$, and \emph{scaled} covariance of variables $\bs{a}$ and $\bs{b}$ as
\begin{align*}
\gls*{SAa2} & := \frac{1}{n - 1} \sum_{i=1}^n \sum_{j = 1}^n A_{ij} (\bs{a}_i - \bar{\bs{a}}) (\bs{a}_j - \bar{\bs{a}})^\top \\
 \gls*{SAab2}& := \frac{1}{n - 1} \sum_{i=1}^n \sum_{j = 1}^n A_{ij} (\bs{a}_i - \bar{\bs{a}}) (\bs{b}_j - \bar{\bs{b}})^\top.
\end{align*}
Apparently, $\bs{S}_{\bs{a}}^2 = \bs{S}^2_{\bs{I}, \bs{a}}$, where $\bs{I}$ is the identity matrix. 
Given a sequence of random variables $U_n$, we use $U_n \asyeq \mathcal{N}(0, 1)$ to denote that it converges in distribution to a standard normal distribution. We write $\bs{H} \in \bR^{n \times n}$ as the hat matrix where $H_{ij} := (n - 1)^{-1} (\bs{X}_i - \bar{\bs{X}})^\top \bs{S}_{\bs{X}}^{-2} (\bs{X}_j - \bar{\bs{X}})$. Let $\bs{1}_d$ and $\bs{0}_d$ be vectors of all ones and zeros with dimension $d$, respectively. When the dimension is clear from the context, we may omit the subscript $d$. Given any matrix $\bs{A} \in \bR^{n \times n}$, we write $\gls*{diagA}$ as a diagonal matrix with $(i,i)$-th entry equal to $A_{ii}$, and write $\gls*{diag-A}$ as a matrix with all diagonal entries equal to zero and off-diagonal entries equal to off-diagonal entries of $\bs{A}$.  Analogously, given any vector $\bs{y} \in \bR^n$,  with a slight abuse of notation, let $\gls*{diagy}$ be the diagonal square matrix having $\bs{y}$ as its diagonal elements.




\subsection{Framework, literature review and overview of contributions}\label{sec:framework}
We consider an experiment with $n$ experimental units and two arms $z \in \{0, 1\}$. We restrict the experiment to be a completely randomized experiment where the experimenter selects $n_1$ units uniformly at random to the treatment group, and the rest $n_0$ units to the control group. To describe causality, we adopt the potential outcome framework, where for each experimental candidate $i$, we assume there are two potential outcomes $Y_i(1)$ and $Y_i(0)$, where $Y_i(z)$ denotes the potential outcome of unit $i$ had unit $i$ been assigned to group $z$. Then the observed outcome $Y_i$ satisfies $Y_i = Z_i Y_i(1) + (1 - Z_i) Y_i(0)$, where the random variable $Z_i \in \{0, 1\}$ denotes the treatment arm assigned to unit $i$.

In this paper, we consider the randomization-based framework where all the potential outcomes $(Y_i(1), Y_i(0))$ are considered deterministic and the randomness comes only from the randomness in the treatment assignment mechanism. This regime has a long history in the study of randomized experiments~\citep{imbens2015causal}. Under this regime, our target of interest then becomes estimating the sample average treatment effect:
\[
\bar{\tau} :=\frac{1}{n} \sum_{i = 1}^n \tau_i \quad\textrm{where}\quad \tau_i = Y_i(1) - Y_i(0).
\]

According to the finite-population central limit theorem~\citep{hajek1960limiting}, one can prove that under some standard regularity conditions, as $n$ goes to infinity, the simple difference in mean estimator $\hat{\tau}_{\unadj} := \frac{1}{n_1} \sum_{i=1}^n Z_i Y_i - \frac{1}{n_0} \sum_{i=1}^n (1 - Z_i) Y_i$ is guaranteed to provide asymptotically normal and unbiased estimation of $\bar\tau$. Specifically, writing $r_z := n_z / n$ as the proportion of units in treatment arm $z$, we have
\[ \sqrt{n} (\hat{\tau}_{\unadj} - \bar\tau) / \sigma_\cre \asyeq \mathcal{N}(0, 1) \quad\textrm{where}\quad \sigma^2_\cre := r_1^{-1} S_{Y(1)}^2 + r_0^{-1} S_{Y(0)}^2 - S_\tau^2.
\]
When each experimental unit $i$ has a deterministic covariate information $\bX_i$ of dimension $p$ indicating its properties, such as age, education, and body weights, a typical choice is to use regression adjustment to incorporate these information for more efficient treatment effect estimation. Define $\tilde{\bs{\beta}}_z := \bs{S}_{\bs{X}}^{-2} \bs{S}_{\bs{X}, Y(z)}$,
and $\hat{\bs{\beta}}_z$ as an empirical estimate of $\tilde{\bs{\beta}}_z$ using samples in the treatment arm $z$:
\begin{equation}\label{eq:oldbeta}
	\hat{\bs{\beta}}_z := \bs{s}_{\bs{X}, z}^{-2} \bs{s}_{\bs{X}, Y(z)},
\end{equation}
where $\bs{s}_{\bs{X}, z}^2$ and $\bs{s}_{\bX, Y(z)}$ denotes empirical estimates of $\bs{S}_{\bs{X}}^2$ and $\bs{S}_{\bs{X}, Y(z)}$ using samples from treatment arm $z$:
\[
\bs{s}_{\bs{X}, z}^2 := \frac{1}{n_z - 1} \sum_{i: Z_i = z} (\bX_i - \bar{\bX}_z) (\bX_i - \bar{\bX}_z)^\top\;\&\; \bs{s}_{\bs{X}, Y(z)} := \frac{1}{n_z - 1} \sum_{i: Z_i = z} (\bX_i - \bar{\bX}_z) (Y_i - \bar{Y}_z),
\]
where $\bar{\bX}_z$ and $\bar{Y}_z$ are estimated means using the data in treatment arm $z$.
 \citet{Lin2013Agnostic} showed that in the regime where $p$ is assumed as a fixed constant, the regression-adjusted estimator:
\begin{equation}\label{eq:regadj}
\hat{\tau}_\lin := \frac{1}{n_1} \sum_{i=1}^n Z_i \{Y_i - \hat{\bs\beta}_1^\top (\bs{X}_i - \bar{\bs{X}})\} - \frac{1}{n_0} \sum_{i=1}^n (1 - Z_i) \{Y_i - \hat{\bs\beta}_0^\top (\bs{X}_i - \bar{\bs{X}})\}
\end{equation}
has the representation
\begin{equation}\label{eq:regrep}
\hat{\tau}_\lin -\bar\tau = \frac{1}{n} \sum_{i=1}^n Z_i (r_1^{-1} e_i(1) + r_0^{-1} e_i(0)) + \op(1 / \sqrt{n}),
\end{equation}
where 
\begin{equation}\label{eq:resid}
e_i(z) := Y_i(z) - \bar{Y}(z) - \tilde{\bs{\beta}}_z^\top (\bs{X}_i - \bar{\bs{X}})
\end{equation}
corresponds to the regression residual of $Y_i(z)$. 
Thus, from standard results in finite population central limit theorem (see e.g.~\citet{li2017general} and the references therein), it has the asymptotic distribution
\[
 \sqrt{n}(\hat{\tau}_\lin - \bar{\tau}) / \sigma_{\adj} \asyeq \mathcal{N}(0, 1) \quad\textrm{where}\quad \sigma^2_\adj := r_1^{-1} S_{e(1)}^2 + r_0^{-1} S_{e(0)}^2 - S_{\tau_e}^2,
\]
where $\tau_{e,i} := e_i(1) - e_i(0)$ corresponds to the individual difference of the residuals in two treatment arms.

As an extension, \citet{lei2021regression} showed that when $p = o(\sqrt{n})$ up to log factors, the above conclusion still holds under certain regularity conditions. Moreover, they proposed a debiased estimator that allows for the more relaxed regime $p = O(n^{2/3} / (\log n)^{1/3})$. {\revone Under similar regularity conditions, \cite{chiang2023regression} proposed a regression adjusted estimator using cross-fitting, allowing for $p=o(n^{3/4})$ up to log factors, and demonstrating superior empirical performance in bias control. \cite{chang2024exact} proposed an exactly unbiased estimator based on Lin's estimator, and proved that it achieves the same asymptotic representation as the right hand side of~\eqref{eq:regrep} when the maximum leverage score $\kappa := \max_{1 \le i \le n} H_{ii} = o(1)$, which requires $p = o(n)$. In a related vein, \cite{wu2018loop} proposed an unbiased ``leave-one-out potential-outcomes" estimator whose prediction functions support variable selection for high-dimensional covariates. 
However, this method currently lacks a theoretical guarantee.}

In this paper, we develop a new regression adjustment-based ATE estimator in the higher dimensional regime where $p$ is allowed to be in the same order of magnitude as $n$. Under this regime, a major challenge is that the inverse covariance matrix used in the construction of $\hat{\bs{\beta}}_z$ is based on the $\bs{X}_i$'s in treatment arm $z$, which is hard to analyze with large $p$.
Specifically, the $\bs{s}_{\bX, z}^{-2}$ in~\eqref{eq:oldbeta} is a complex nonlinear function of the treatment assignment $\bs{Z}$, adding complexity to the analysis of Lin's estimator when the dimension of covariates is diverging. In fact, even consistency requires $p = o(n / \log n)$~\citep{lei2021regression}.

In this paper, we instead consider a variant of regression adjustment estimator where the inverse covariance matrix is instead constructed using covariate information on the entire dataset, i.e. that we instead set\footnote{From here and below, we assume throughout that $p < n$ {\revone and that $\bs{S}^2_{\bX}$ is invertible}, so that the regression adjustment estimator is well defined.} 
	\begin{equation}\label{eq:newbeta}
		\hat{\bs{\beta}}_z := \bs{S}_{\bs{X}}^{-2} \bs{s}_{\bs{X}, Y(z)},
	\end{equation} 
where recall that as in~\eqref{eq:oldbeta}, $\bs{s}_{\bs{X}, Y(z)}$ denotes an empirical estimate of $\bs{S}_{\bs{X}, Y(z)}$ using samples from treatment arm $z$, and set $\hat{\tau}_\adj$ in exactly the same way as $\hat{\tau}_\lin$ in~\eqref{eq:regadj}, but using $\hat{\bs{\beta}}_z$ as defined in~\eqref{eq:newbeta}. In other words, compared to the definition at~\citet{Lin2013Agnostic,lei2021regression}, the main difference is that we use the $\bs{X}_i$'s of the entire sample to construct the inverse covariance matrix, not only those in treatment arm $z$. 
{\rev Consequently, the estimator becomes a quadratic function of $\bs{Z}$, allowing us to remove the main barrier of asymptotic analysis in both low and moderate dimensions.}
Such a variant has been discussed before by \citet{li2020rerandomization,wang2022rerandomization} in the lower dimensional regime. 
Building on this variant of regression adjustment-based estimator, we propose a new debiased estimator for average treatment effect estimation; we prove that in an asymptotic regime where $p = o(n)$, the debiased estimator is asymptotically normal with the same variance as in the fixed dimensional regime, namely $\sigma_\adj^2$. 
We also derive the asymptotic distribution of the debiased estimator in the higher-dimensional regime where $p$ can be in the same order of magnitude as $n$ and propose sufficient conditions so that the debiased estimator is asymptotically more efficient than the unadjusted estimator. As far as we are aware, both regimes have not been well investigated by existing randomization-based framework literature.

Finally, we would like to remark that besides the randomization-based framework we consider in this manuscript, another choice is the superpopulation framework, which assumes that the experimental units must be randomly sampled from some superpopulation. This framework is also popular in literature, some examples include~\citet{tsiatis2008covariate,wager2016high,negi2021revisiting}.

The rest of this paper is organized as follows. In \Cref{sec:reg}, we present our new estimator, and prove its asymptotic normality in the regime $p = o(n)$. In~\Cref{sec:highdim}, we prove its asymptotic convergence in the higher dimensional regime where $p$ is allowed to be in the same order of magnitude as $n$ and discuss conditions so that it is asymptotically more efficient than without regression adjustment. In~\Cref{sec:inference}, we present a new confidence interval construction method. In~\Cref{sec:assumption}, we discuss the regularity conditions involved in Sections~\ref{sec:reg}--\ref{sec:inference}. We further conduct a numerical analysis in~\Cref{sec:numerical}. {\rev Besides demonstrating the empirical performance of our proposed estimator, our empirical analysis also indicates that the use of full inverse covariance matrix in~\eqref{eq:newbeta} may not be just out of technical convenience, but is fundamental in improving our estimator.} We end with a concluding remark in~\Cref{sec:conclusion}.

\section{A debiased regression adjustment estimator}\label{sec:reg}

In this section, we discuss how our debiased estimator is constructed, and present its asymptotic property in the regime $p = o(n)$. As will be demonstrated in the Supplementary Material, under some regularity conditions that will be discussed further in this section, the following decomposition holds (see Supplementary Material for more details):
\begin{align}\label{eq:adjdecomp}
\hat{\tau}_\adj - \bar\tau = & - \frac{r_1 r_0}{n} \sum_{i = 1}^n H_{ii} \left(\frac{Y_i(1) - \bar{Y}(1)}{r_1^2} - \frac{Y_i(0) - \bar{Y}(0)}{r_0^2}\right) \\
& + \frac{1}{n} \sum_{i=1}^n (Z_i - r_1) c_i - \frac{1}{n} \sum_{i \neq j} (Z_i - r_1) (Z_j - r_1) \left(\frac{A_{ij}(1)}{r_1^2} -\frac{A_{ij}(0)}{r_0^2}\right) + \op(1 / \sqrt{n}), \notag
\end{align}
where $A_{ij}(z)$ is the $(i, j)$-th entry of the matrix
\[
\gls*{Az} := \bs{H} \diag\left\{(Y_1(z) - \bar{Y}(z), \ldots, Y_n(z) - \bar{Y}(z))^\top\right\};
\]
and $c_i$ is defined as
	\begin{align*}
		c_i & := r_0\frac{p}{n}\frac{Y_i(1)-\bar{Y}(1)}{r_1^2} +  r_1\frac{p}{n}\frac{Y_i(0)-\bar{Y}(0)}{r_0^2} -
		(r_0-r_1)\left(\frac{s_i(1)}{r_1^2}-\frac{s_i(0)}{r_0^2}\right)+ \frac{e_i(1)}{r_1}+\frac{e_i(0)}{r_0},\\
		& \qquad\textrm{with}\;\;\gls*{siz} := H_{ii}(Y_i(z) - \bar{Y}(z)) - \frac{1}{n} \sum_{j=1}^n H_{jj} (Y_j(z) - \bar{Y}(z)),\quad z = 0, 1.
	\end{align*}

Apparently, the second term {\rev on the right hand side of \eqref{eq:adjdecomp}} is of mean zero. For the third term, since the $Z_i$'s are just weakly dependent, one can also show that its mean is approximately zero. Therefore, the first term constitutes the bias. Based on this, we propose a new debiased ATE estimator via stripping the original $\hat{\tau}_\adj$ with an approximately unbiased estimator of the first term, which is constructed based on the observations in the two treatment arms:
\begin{equation}\label{eq:db}
\hat{\tau}_\db := \hat{\tau}_\adj + r_1 r_0\left(\frac{1}{n_1}\sum_{i: Z_i = 1} H_{ii} \frac{(Y_i - \bar{Y}_1)}{r_1^2} - \frac{1}{n_0}\sum_{i: Z_i = 0} H_{ii} \frac{(Y_i - \bar{Y}_0)}{r_0^2}\right),
\end{equation}
where recall that different from $\bar{Y}(z)$, $\bar{Y}_z$ is the estimated mean using the outcome data in treatment arm $z$.
Below we provide an asymptotic convergence guarantee of $\hat{\tau}_\db$ in the asymptotic regime $p = o(n)$. Our new guarantee is based on the following $4$ assumptions:

\begin{assumption}
\label{assumption:positive-limit-of-assigned-proportion}
    For $z=0,1$, $r_z$ tends to a limit in $(0,1)$.
\end{assumption}

\begin{assumption}
\label{assumption:2-moment-of-finite-population-for-y}
For $z=0,1$, $\sum_{i = 1}^n \left(Y_i(z)-\bar{Y}(z)\right)^2 = O(n).$
\end{assumption}
 

\begin{assumption}
    \label{assumption:lindeberg-type-condition-p-o(n)}
    Consider the $e_i(z)$ in~\eqref{eq:resid},
    as $n \to \infty$,
    \[
    \max_z \max_i |Y_i(z) - \bar{Y}(z)| / \sqrt{n} \to 0 \quad\&\quad \max_z \max_i |e_i(z)| / \sqrt{n} \to 0.
    \] 
\end{assumption}
\begin{assumption}
\label{assumption:Positive-variance-limit-p-o(n)}
   $\liminf_{n \to \infty} \sigma_\adj^2>0$.
\end{assumption}

Assumptions~\ref{assumption:positive-limit-of-assigned-proportion},~\ref{assumption:2-moment-of-finite-population-for-y} and~\ref{assumption:Positive-variance-limit-p-o(n)} are standard assumptions in randomization-based inference. \Cref{assumption:lindeberg-type-condition-p-o(n)} is a Lindeberg--Feller-type condition to guarantee that the representation in~\eqref{eq:regrep} has an approximately normal distribution in the large sample limit; similar condition has also appeared in previous regression adjustment literature; see e.g. \citet{lei2021regression} and the references therein. {\rev The second part of Assumption \ref{assumption:lindeberg-type-condition-p-o(n)} is implicitly related to covariates. In the fixed $p$ regime, it is usually considered as a mild assumption. In~\Cref{prop:justify-assumption-3}, we further provide some justifications of~\Cref{assumption:lindeberg-type-condition-p-o(n)} in the regime considered in this work. Specifically, we demonstrate that when $(Y_i(z), \bs{X}_i)$'s are i.i.d. generated from some superpopulation with an arbitrary nonlinear relationship between $Y_i(z)$ and $\bs{X}_i$, \Cref{assumption:lindeberg-type-condition-p-o(n)} can be satisfied with high probability by assuming further some moment conditions on $Y_i(z)$ and $\bs{X}_i$. Suppose instead that $(Y_i(z), \bX_i)$'s are from a parametric model where $Y_i(z) = \bs{\beta}_z^\top \bs{X}_i + \varepsilon_{z,i}$ where $\varepsilon_{z,i}$ is noise independent of $\bs{X}_i$. Then it directly follows from \citet[Proposition~F.2]{lei2021regression} that 
	we only need $\varepsilon_{z,i}$'s to have bounded $\delta$-th order moment for some $\delta > 2$ to satisfy the second part of \Cref{assumption:lindeberg-type-condition-p-o(n)}. In other words, no assumptions are required for $\bs{X}_i$; in fact, the $\bs{X}_i$'s can even be deterministic.}

With these assumptions, we are able to show that $\htaudb$ has the representation
\begin{align}\label{eq:dbdecomp}
\htaudb - \bar\tau = & \frac{1}{n} \sum_{i=1}^n Z_i (r_1^{-1} e_i(1) + r_0^{-1} e_i(0)) + \frac{1}{n} \sum_{i=1}^n Z_i (r_1^{-1} s_i(1)+r_0^{-1} s_i(0)) \\
& - \frac{1}{n} \sum_{i \neq j} (Z_i - r_1) (Z_j - r_1) \left(\frac{ A_{ij}(1)}{r_1^2} -  \frac{A_{ij}(0)}{r_0^2} \right) + \op(1 / \sqrt{n}). \notag
\end{align}

We now invoke the following condition on the ordered sequence of potential outcomes, which characterizes the tail of the population of potential outcomes:
\begin{assumption}
    \label{assumption:p-o(n)-case}
    As $n \to \infty$, $p / n \to 0$. Moreover, let $(Y_{(1)}(z) - \bar{Y}(z))^2 \geq \ldots \geq (Y_{(n)}(z) - \bar{Y}(z))^2$ be the ordered sequence of $\{(Y_i(z) - \bar{Y}(z))^2\}_{i=1}^n$. Then, for any $z\in \{0,1\}$,  we have that $\sum_{i=1}^p (Y_{(i)}(z) - \bar{Y}(z))^2 = o(n).$
\end{assumption}
{\rev When $p$ is a fixed constant, \Cref{assumption:p-o(n)-case} reduces to the first half of \Cref{assumption:lindeberg-type-condition-p-o(n)}. However, once we are faced with a diverging $p$, \Cref{assumption:p-o(n)-case} cannot be always satisfied given the assumptions metioned above. To elucidate, we now give a sequence of potential outcomes that satisfies \Cref{assumption:2-moment-of-finite-population-for-y} and the first half of \Cref{assumption:lindeberg-type-condition-p-o(n)}, but violates \Cref{assumption:p-o(n)-case}:  
	\[
	(Y_i(z))_{i=1}^{n} \equiv \Big(\frac{q_n(n-q_n)}{n^2}\Big)^{-1/2}(\bs{1}_{q_n}^\top,\bs{0}_{n-q_n}^\top)^\top,\quad 0 \le q_n<n/2,
	\]
	where {\rev $p=o(n)$}, $q_n$ is a sequence of integers to be determined. 
	Apparently, we have $\sum_i (Y_i(z)-\bar{Y}(z))^2 = n$, so that \Cref{assumption:2-moment-of-finite-population-for-y} is directly satisfied. Moreover, we have $\max_i |Y_i(z)-\bar{Y}(z)| = \{(n-q_n)/q_n\}^{1/2}$ and 
	\[
	\sum_{i=1}^p (Y_{(i)}(z)-\bar{Y}(z))^2 = \begin{cases}
		p\{(n-q_n)/q_n\} \quad p\leq q_n;\\
		q_n\{(n-q_n)/q_n\}+(p-q_n)\{q_n/(n-q_n)\} \quad p > q_n.
	\end{cases}  
	\]
This means that under the regime $p \to \infty$ as $n \to \infty$, by choosing $q_n$ such that $q_n \to \infty$ and $p / q_n \to \infty$, we have $\max_i |Y_i(z)-\bar{Y}(z)| = \{(n-q_n)/q_n\}^{1/2} = o(n^{1/2})$, and therefore the first half of  \Cref{assumption:lindeberg-type-condition-p-o(n)} holds. However, $n^{-1}\sum_{i=1}^p (Y_{(i)}(z)-\bar{Y}(z))^2 > (n-q_n)/n$, which tends to $1$ instead of $0$ (since $q_n = o(p) = o(n)$), and thus \Cref{assumption:p-o(n)-case} is not satisfied anymore.	 
}

Armed with the above five assumptions, we are able to show that the second and third terms in the decomposition~\eqref{eq:dbdecomp} are of order $\op(1 / \sqrt{n})$, so that $\htaudb$ has the same asymptotic representation (and therefore the same asymptotic distribution) as $\hat{\tau}_\adj$ in the fixed $p$ case. Specifically, we have the following result:

\begin{theorem}
\label{theorem:CLT-p=o(n)}
    Under Assumptions \ref{assumption:positive-limit-of-assigned-proportion}--\ref{assumption:p-o(n)-case},  we have
\begin{align*}
     n^{1/2}\left(\htaudb-\bar\tau\right)/\sigma_{\adj}
      \asyeq \mathcal{N}(0,1).
\end{align*}
\end{theorem}

 
 As discussed before, Assumptions \ref{assumption:positive-limit-of-assigned-proportion}--\ref{assumption:Positive-variance-limit-p-o(n)} are basic assumptions in previous literature; \Cref{assumption:p-o(n)-case} is a novel contribution from us, which requires no further assumption on the moment of covariate $\bs{X}_i$ or the leverage scores $H_{ii}$ to obtain asymptotically normal convergence.
 As we will show later in~\Cref{sec:assumption}, \Cref{assumption:p-o(n)-case} holds with high probability when $Y_i(z)$ are i.i.d. generated from a superpopulation with bounded second order moment. {\revone Finally, a simple inspection of the proof of \Cref{theorem:CLT-p=o(n)} reveals that, if we instead want to prove that $\htaudb - \bar{\tau}$ follows the representation on the right hand side of~\eqref{eq:regrep}, then only Assumptions~\ref{assumption:positive-limit-of-assigned-proportion}--\ref{assumption:2-moment-of-finite-population-for-y},~\ref{assumption:p-o(n)-case} and the first half of Assumption~\ref{assumption:lindeberg-type-condition-p-o(n)} are needed. In other words, no assumptions on covariate vectors $\bs{X}_i$'s are required, except that $\bs{S}^2_{\bX}$ must be invertible.}

In the next section, we further consider the regime where $p$ can be in the same order of magnitude as $n$. Under this regime, the second and third terms in the decomposition~\eqref{eq:dbdecomp} are not necessarily negligible anymore, so we need new analysis to understand their asymptotic convergence. 
The second term is easy to deal with. The main obstacle of the analysis is that the quadratic form of centered treatment indicators (i.e., the third term in~\eqref{eq:dbdecomp}) needs to be characterized by a new analytic tool, which is the central limit theorem of quadratic forms. We will discuss this further in the next section.

\section{Asymptotic convergence of debiased estimator with moderately high-dimensional covariates}\label{sec:highdim}

In this section, we consider the asymptotic convergence of our debiased estimator in the moderate high-dimensional regime where we allow $p$ to be in the same order of magnitude as $n$. As mentioned before, since in this regime, the second and third terms in the decomposition~\eqref{eq:dbdecomp} are not negligible, we need to derive their (joint) distribution in the large sample limit. With the help of standard results in combinatorial central limit theorem (CLT)~\citep{hajek1960limiting}, the second term is easy to derive. Whilst for the third term, due to the quadratic functions in the form $(Z_i - r_1) (Z_j - r_1)$, standard combinatorial CLT is not applicable to understand its convergence anymore.

In this paper, we will use the newly developed central limit theorem of the so-called \emph{homogeneous} sums from \citet{koike2022high} to characterize the third term of \eqref{eq:dbdecomp}. Specifically, \citet{koike2022high} studied the convergence of random variables satisfying the following form:
\begin{definition}\label{def:Wi}
Let $\boldsymbol{W}=\left(W_i\right)_{i=1}^n$ be a sequence of independent centered random variables with unit variance. A homogeneous sum is a random variable of the form
$$
 Q(f ; \boldsymbol{W})=\sum_{i_1, \ldots, i_q=1}^n f\left(i_1, \ldots, i_q\right) W_{i_1} \cdots W_{i_q},
$$
where $n, q \in \mathbb{N},[n]:=\{1, \ldots, n\}$ and $f:[n]^q \rightarrow \mathbb{R}$ is a symmetric function vanishing on diagonals, i.e., $f\left(i_1, \ldots, i_q\right)=0$ unless $i_1, \ldots, i_q$ are mutually different.
\end{definition}
This is an extension of the linear statistics studied by the standard CLT. Apparently, by setting $W_i \equiv (Z_i - r_1)/(r_1r_0)^{1/2}$, $q = 2$ and 
\[
f(i_1, i_2) \equiv r_1r_0 \left(r_1^{-2} A_{i_1 i_2}(1) + r_1^{-2} A_{i_2 i_1}(1) - r_0^{-2}  A_{i_1 i_2}(0) -  r_0^{-2} A_{i_2 i_1}(0)\right) / (2n),
\]
the third term in~\eqref{eq:dbdecomp} falls into this category, with the exception that in our problem, the $(Z_i - r_1)$'s are weakly dependent.

Below we give a brief literature review of this class of CLT. 
%
\cite{rotar1976limit} and \cite{rotar1979limit} studied the invariance principles of $Q(f;\bs{W})$ regarding the law of $W$. \cite{de1990central} established the univariate central limit theorem for $Q(f;\bs{W})$. \cite{koike2022high} extended it to the multivariate case and obtained the bound for the error of normal approximation. The special case of $q=1$ is the classical sum of independent random variables. The special case of $q=2$ has been extensively studied; see, for example, \cite{de1973asymptotic,de1987central,fox1987central}. Note that all of the results are for independent $W_i$'s.

Nevertheless, the results of \citet{koike2022high} are still not sufficient, since \citet{koike2022high} assumed that all the random variables are independent, whilst in our problem, the $Z_i$'s are weakly dependent due to simple random sampling. Mimicking the idea of H\`{a}jek's coupling \citep{hajek1960limiting} and its extension in~\citet{wang2022rerandomization}, we are able to propose a new combinatorial central limit theorem to characterize the joint distribution of the decomposition in~\eqref{eq:adjdecomp}, and furthermore, the asymptotic distribution of $\hat{\tau}_{\db} - \bar\tau$.

To formally describe the new convergence result, we first define 
\[
\sigma^2_{\hd, l} := r_1^{-1} S_{e(1)+s(1)}^2 + r_0^{-1} S_{e(0)+s(0)}^2 - S_{\tau_e + \tau_s}^2,
\]
where analogous to $\tau_{e, i}$, we write $\tau_{s, i} := s_i(1) - s_i(0)$. Moreover, we define
\[
\sigma^2_{\hd, q} := (r_1r_0)^2 S_{\bs{Q}, \;r_1^{-2} Y(1) - r_0^{-2} Y(0)}^2,
\]
where $\gls*{bsQ}$ is an $n \times n$ dimensional matrix such that $Q_{ij} := H_{ij}^2$ whenever $i \neq j$ and $Q_{ii} := H_{ii} - H_{ii}^2$. 
Apparently, $\sigma_{\hd, l}^2$ and $\sigma_{\hd, q}^2$ correspond to the variances contributed by the linear statistic and quadratic statistic in~\eqref{eq:dbdecomp}, respectively. We also write $\sigma^2_{\hd} := \sigma^2_{\hd, l} + \sigma^2_{\hd, q}$ as their total variance. We now invoke the following assumption regarding the asymptotics of our estimator variance.

\begin{assumption}
\label{assumption:Positive-variance-limit-alpha-positive}
 $\liminf_{n \to \infty} \sigma_{\hd,l}^2 > 0 $ or $\liminf_{n \to \infty} \sigma_{\hd,q}^2 > 0 $.
\end{assumption}

Armed with this assumption, we are able to show that our debiased estimator is asymptotically normal.

\begin{theorem}
\label{theorem:CLT-alpha-greater-than-0}
    If Assumptions \ref{assumption:positive-limit-of-assigned-proportion}--\ref{assumption:lindeberg-type-condition-p-o(n)}, \ref{assumption:Positive-variance-limit-alpha-positive} 
     hold, 
     we have as $n \to \infty$,
\begin{align*}
     n^{1/2}\left(\htaudb-\bar\tau\right)/\sigma_\hd \asyeq \mathcal{N}(0,1).
\end{align*}
\end{theorem}


Notice that in the above theorem, we do not require any assumption on the scaling of $p$; instead, we just require $p < n$ so that the debiased estimator is well-defined. Of course, as we will discuss later, we may still need $p / n$ to be asymptotically upper bounded by some constant in $(0, 1)$ to justify some assumptions. For more discussions, we refer the readers to~\Cref{sec:assumption}. 

Interestingly, without the constraint $p = o(n)$, we do not need \Cref{assumption:p-o(n)-case} anymore. This is because in the regime $p \asymp n$ we need to characterize more carefully the distribution of $\htaudb$ by considering the second and third terms in the decomposition~\eqref{eq:dbdecomp}. 
{\rev To our knowledge,  at the time we first circulated our manuscript, this was the first result which proves asymptotic normality in the moderately high-dimensional regime without any sparsity constraint on the potential outcome's linear representations. After we made our work publicly available, we have discovered an eariler work:~\citet{chang2021exact} {\revone (i.e., the arXiv version of~\citet{chang2024exact})}, which proposed a debiased estimator similar to ours. Indeed, as we prove in the Supplementary Material, under  Assumptions~\ref{assumption:positive-limit-of-assigned-proportion}--\ref{assumption:2-moment-of-finite-population-for-y} and the first half of \Cref{assumption:lindeberg-type-condition-p-o(n)}, the difference between our estimator and that of \citet{chang2021exact} (denoted as $\hat{\tau}_{\CMA}$) is asymptotically negligible: $\hat{\tau}_{\CMA} - \htaudb  = \Op(n^{-1})$. However, we would like to emphasize that though we are working on similar estimators, our theoretical focus is quite different. \cite{chang2021exact} proved that $\hat{\tau}_{\CMA}$ is exactly unbiased for $\bar\tau$, but all the asymptotic analysis of $\hat{\tau}_{\CMA}$ are under the fix $p$ regime. {\revone In their published version \citep{chang2024exact}, which became publicly available concurrently with our work, they further provided an asymptotic analysis allowing a diverging $p$ for an estimator that is slightly modified from $\hat{\tau}_{\CMA}$ (see also \Cref{sec:framework}).} 
	Beyond \cite{chang2024exact}, \cite{chiang2023regression} also proposed an estimator that has exactly unbiased property (denoted as $\hat{\tau}_{\CMO}$). Under again Assumptions~\ref{assumption:positive-limit-of-assigned-proportion}--\ref{assumption:lindeberg-type-condition-p-o(n)} and an additional assumption that $\max_z |\bar{Y}(z)| = O(1)$, we instead have
	\begin{align*}
		\hat{\tau}_{\CMO}-\htaudb = \Delta_{\CMO} + \op(n^{-1/2}),
	\end{align*}
	where  
	\[
	\Delta_{\CMO} := \frac{1}{n }\sum_{i=1}^n (Z_i-r_1)H_{ii} \Big(\frac{\bar{Y}(1)}{r_1}+\frac{\bar{Y}(0)}{r_0}\Big)  -\frac{1}{n}\sum_{i\ne j} (Z_i-r_1)(Z_j-r_1) H_{ij}\Big(\frac{\bar{Y}(1)}{r_1^2}-\frac{\bar{Y}(0)}{r_0^2}\Big)
	\]
	is a nuisance random component which depends on the population means. Moreover, as we prove in the Supplementary Material, $\Delta_{\CMO} = \op(n^{-1/2})$ when $p = o(n)$; and $\Delta_{\CMO} = \Op(n^{-1/2})$ when $p \asymp n$. In other words, our analysis shows that $\hat{\tau}_{\CMO}$ has the same asymptotic distribution as $\htaudb$ when $p = o(n)$, which we believe is of independent interest.
}

\subsection{Efficiency improvement for debiased regression adjustment with moderately high-dimensional covariates}
\label{sec:efficiency}

In \Cref{theorem:CLT-alpha-greater-than-0}, we present the asymptotic distribution of our new debiased estimator. In this section, we discuss conditions so that $\sigma_{\hd, l}^2 + \sigma_{\hd, q}^2$ is smaller than $\sigma_\cre^2$, i.e., that our debiased estimator is asymptotically more efficient than without doing regression adjustment at all. 

To shed light on how the covariate-dimension-to-sample-size-ratio $p / n$ influences the variance of our new estimator, we consider a class of pre-treatment covariates whose leverage scores concentrate around their mean $p / n$.  We formalize it into the following assumption.
\begin{assumption}\label{assumption:maximum-leverage-score-close-to-alpha}
Let $\gls*{alpha} := p / n$. As $n \to \infty$, we have that:
    \[
    \max_{1 \le i \le n} |H_{ii} - \alpha| \to 0 \quad\&\quad \max_{z \in \{0, 1\}} \max_{1 \le i \le n} |Y_i(z) - \bar{Y}(z)| / \sqrt{n} \to 0,
    \]
    or for some constant $\eta > 0$,
    \[
    \frac{1}{n}\sum_{i=1}^n (H_{ii} - \alpha)^2 \to 0 \quad\&\quad \max_{z \in \{0, 1\}}\frac{1}{n}\sum_{i=1}^n |Y_i(z) - \bar{Y}(z)|^{2 + \eta} {\rev = O(1)}.
    \]
\end{assumption}
In~\citet{lei2021regression}, the authors have also assumed similar assumptions. Specifically, they require $\max_i  H_{ii} = o(1)$, which is equivalent to our first constraint in the regime $\alpha \to 0$. To justify this assumption, they proved that this assumption holds with high probability when the covariates are randomly generated from a superpopulation with $(6+\delta)$-th moment. 
This proof, albeit being enough for the setting $p = o(n)$, cannot be used in the $p \asymp n$ regime. As we will discuss further in \Cref{sec:assumption}, in this paper, we provide a new proof to show that if the covariates are generated as i.i.d. realizations of a random variable with $(4+\delta)$-th moment, 
\Cref{assumption:maximum-leverage-score-close-to-alpha} holds with high probability.


Armed with the new assumption, we are able to propose some bounds on $\sigma_{\hd, l}$ and $\sigma_{\hd, q}$, depending on $\alpha$. By defining $\gls*{R2} := 1 - \frac{\sigma_\adj^2}{\sigma_\cre^2}$, which is equivalent to the {\rev squared canonical correlation} between $\bs{X}$ and $r_1^{-1} Y(1) + r_0^{-1} Y(0)$, we have the following result:

\begin{corollary}
\label{corollary:upper-lower-bound-for-sigma-alpha>0}
 Under Assumptions \ref{assumption:positive-limit-of-assigned-proportion}, \ref{assumption:2-moment-of-finite-population-for-y} and 
 \ref{assumption:maximum-leverage-score-close-to-alpha}, as $n \to \infty$, we have
 \begin{align*}
     & \sigma_{\hd,l}^2 = \left[(1+\alpha)^2-(1+2\alpha)R^2\right] \sigma_\cre^2 + o(1),\\
     & 0 \leq \sigma_{\hd,q}^2 \leq 2(r_1r_0)^2 \alpha(1-\alpha) S_{r_1^{-2} Y(1) - r_0^{-2} Y(0)}^2  + o(1).
 \end{align*}
 As a consequence, there is\footnote{{\rev Given a sequence of quantities $a_n$, we write $o(1) \le a_n$ if there exists a sequence $b_n \to 0$ such that $b_n \le a_n$.}}
\begin{align*}
    o(1)  \leq \sigma_\hd^2 - \left[(1+\alpha)^2-(1+2\alpha)R^2\right] \sigma_\cre^2 \leq 2(r_1r_0)^2 \alpha(1-\alpha) S_{r_1^{-2} Y(1) - r_0^{-2} Y(0)}^2 + o(1).
\end{align*}
\end{corollary}


Informally then, Corollary \ref{corollary:upper-lower-bound-for-sigma-alpha>0} shows that a necessary condition for the debiased estimator to give a variance smaller than that of $\hat{\tau}_\unadj$ is
 \begin{equation}\label{eq:necessarylb}
        (1+\alpha)^2-(1+2\alpha)R^2-1 < 0 \Leftrightarrow R^2 >\frac{\alpha^2+2\alpha}{1+2\alpha}.
    \end{equation}
At the same time, a sufficient condition for the debiased estimator to give a variance reduction is
\begin{equation}\label{eq:sufflb}
R^2 > \frac{\alpha^2+2\alpha}{1+2\alpha} + 2r_1r_0 \frac{\alpha(1-\alpha)}{1+2\alpha} \frac{S^2_{r_1^{-2} Y(1) - r_0^{-2} Y(0)}}{S^2_{r_1^{-1} Y(1) + r_0^{-1} Y(0)}}.
 \end{equation}

When $r_1 \equiv r_0 \equiv \frac{1}{2}$, then the above inequality can be written as
\[
R^2 > \underset{=: R_L^2}{\underbrace{\frac{\alpha^2+2\alpha}{1+2\alpha} +  \frac{\alpha(1-\alpha)}{1+2\alpha} \frac{S_\tau^2}{2 S_{\frac{Y(1) + Y(0)}{2}}^2}}}.
\]
We denote the right-hand side of the above inequality by $R_L^2$.  Informally, the magnitude of $R_L^2$ depends on two quantities: the first is the covariate-dimension-to-sample size ratio $\alpha$; the second is a scaled ratio between the variance of individual treatment effect $\tau_i = Y_i(1) - Y_i(0)$ and the variance of the average of two potential outcomes $\frac{Y_i(1) + Y_i(0)}{2}$, which we denote by $\gamma := \frac{S_\tau^2}{2 S_{\frac{Y(1) + Y(0)}{2}}^2}$.

\begin{figure}[t!]
    \centering
    \includegraphics[width=0.7\linewidth]{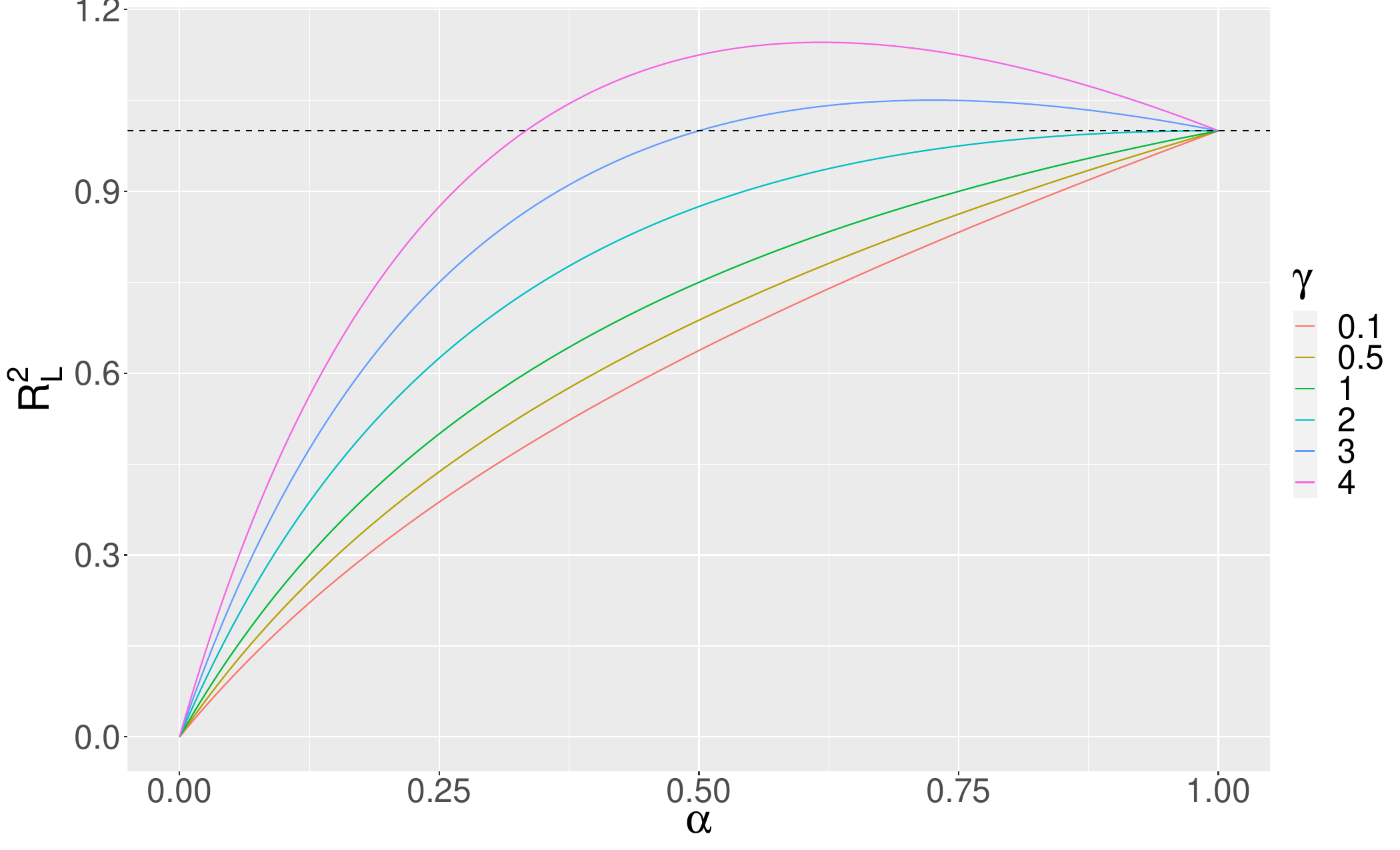}
    \caption{Curves of $R_L^2$ as a function of $\alpha := p / n$ with different magnitudes of $\gamma := \frac{S_\tau^2}{2 S_{\frac{Y(1) + Y(0)}{2}}^2}$. 
    The dashed line signifies $1$.}
    \label{fig:dependence}
\end{figure}

\Cref{fig:dependence} illustrates the dependency of $R_L^2$ on $\alpha$ and $\gamma$. Apparently, with a decreasing $\gamma$, $R_L^2$ decreases monotonically. This indicates that when the individual effects have smaller heterogeneity, less dependency between the potential outcomes and covariates is required for the debiased estimator to have an efficiency improvement compared to the unadjusted estimator. When $\gamma$ is small, $R_L^2$ increases monotonically as $\alpha$ goes from $0$ to $1$. When $\gamma$ is increased to above $2$, the trend then follows a different pattern. Regardless of the magnitude of $\gamma$, $R_L^2$ reduces to zero as $\alpha$ goes to zero. This is consistent with the theoretical findings in the low-dimensional setting. When $\alpha$ approaches $1$, both $R_L^2$ and the lower bound in~\eqref{eq:necessarylb} approaches $1$. This implies that we usually cannot achieve any efficiency improvement when $p$ is close to $n$. 
This also implies that adding more covariates to the regression will usually result in a phase transition from ``mostly harmless'' to ``harmful'' to the post-experiment analysis. In practice, we recommend practitioners to choose a moderate number of covariates so that $R^2$ is above $R_L^2$.

When $\alpha$ is small, say $0.1$, and we restrict $\gamma$ to be no larger than $2$, $R_L^2$ is at most $0.325$. We believe this already includes a large number of cases in practical applications; when $\alpha = 0.5$, one may require a relatively low $\gamma$ to keep $R_L^2$ away from $1$. 
Of course, since $R^2 > R_L^2$ is just a sufficient condition for our estimator to be more accurate than without using regression adjustment at all, in practice one may still observe an improved accuracy even when this is violated.




\section{Inference}\label{sec:inference}
Inference on $\htaudb$ relies on a valid estimation of $\sigma_\hd^2$, so that one can construct asymptotically valid Wald-type confidence intervals. We will derive the formula of the variance estimator in this section and show that this variance estimator is asymptotically valid with moderately high-dimensional covariates.

Our new inferential technique is constructed by a decomposition of $\sigma_\hd^2$, which is given below in~\eqref{eq:sigmadecomp}. In order to describe this decomposition, we define:
\begin{equation}\label{def:BB}
    \gls*{bsB} :=\bs{M}^\top \bs{M},\quad \textrm{where}\quad \bs{M}:=\left(\bs{I} - \frac{1}{n}\bs{1}\bs{1}^\top\right)  - \bs{H} + \left(\bs{I} - \frac{1}{n}\bs{1}\bs{1}^\top\right) \diag\{\bs{H}\}.
\end{equation}
As will be demonstrated in the Supplementary Material, with the above notation, we are able to rewrite $\sigma_\hd$ as
\begin{align}
\label{eq:rewrite-sigma-hd}
\sigma_\hd^2 = (r_1r_0) S_{\bs{B}, \;r_1^{-1} Y(1) + r_0^{-1} Y(0)}^2 + (r_1r_0)^2 S_{\bs{Q}, \;r_1^{-2} Y(1) - r_0^{-2} Y(0)}^2.
\end{align}



Now, using that for any \emph{symmetric} matrix $\bs{A}$ and any vectors $\bs{a}, \bs{b}$, 
\begin{align}
\label{eq:property-of-S2-1}
    S_{\bs{A}, \bs{a}, \bs{b}} = S_{\diag\{\bs{A}\}, \bs{a}, \bs{b}} + S_{\diag^-\{\bs{A}\}, \bs{a}, \bs{b}},
\end{align}
and that 
\begin{align}
\label{eq:property-of-S2-2}
    S^2_{\bs{A}, \bs{a} + \bs{b}} = S^2_{\bs{A}, \bs{a}} + S^2_{\bs{A}, \bs{b}} + 2 S_{\bs{A}, \bs{a}, \bs{b}},
\end{align}
(which we will clarify in the Supplementary Material), we further decompose $\sigma_{\hd}^2$ as

\begin{equation}\label{eq:sigmadecomp}
\begin{aligned}
 \sigma_{\hd}^2 &:= \underset{=: \cI_1}{\underbrace{(r_1r_0) \sum_{z \in \{0, 1\}} \left(S^2_{r_1r_0\diag\{\bs{Q}\}, \;r_z^{-2} Y(z)} + S^2_{r_z^{-2} \diag\{\bs{B}\}, \;Y(z)}\right)}} \\
& + \underset{=: \cI_2}{\underbrace{(r_1r_0) \sum_{z \in \{0, 1\}} \left(S^2_{r_1r_0\diag^-\{\bs{Q}\}, \;r_z^{-2} Y(z)} + S^2_{r_z^{-2} \diag^-\{\bs{B}\}, \;Y(z)}\right)}} \\
& + \underset{=: \cI_3}{\underbrace{2 S_{\diag\{\bs{B}\}, Y(1), Y(0)} - 2 S_{\diag\{\bs{Q}\}, Y(1), Y(0)}}} \\
&+ \underset{=: \cI_4}{\underbrace{2 S_{\diag^-\{\bs{B}\}, Y(1), Y(0)} - 2 S_{\diag^-\{\bs{Q}\}, Y(1), Y(0)}}}.
\end{aligned}
\end{equation}
Informally, $\cI_1$ and $\cI_2$ correspond to the variances of a single world, and $\cI_3$ and $\cI_4$ correspond to the covariance of counterfactual worlds. 


Armed with the above decomposition, we construct an estimation of $\sigma_\hd^2$ by estimating $\cI_1, \ldots, \cI_4$ separately. Since they represent variances from different sources, we need different estimation strategies for each term. We first consider $\cI_1$ and $\cI_2$. Since these quantities are quadratic functions of potential outcomes of a single arm, they can be consistently estimated using empirical observations from a single arm. Specifically, we can estimate $\cI_1$ via
\[
\hat{\cI}_1 := (r_1r_0) \sum_{z \in \{0, 1\}} \left(s^2_{r_1r_0\diag\{\bs{Q}\}, \;r_z^{-2} Y(z)} + s^2_{r_z^{-2} \diag\{\bs{B}\}, \;Y(z)}\right),
\]
where $s^2_{r_1r_0\diag\{\bs{Q}\}, \;r_z^{-2} Y(z)}$ and $s^2_{r_z^{-2} \diag\{\bs{B}\}, \;Y(z)}$ are empirical estimates of their oracle versions using samples from treatment arm $z$. For example, we write
\[
s^2_{r_1r_0\diag\{\bs{Q}\}, \;r_z^{-2} Y(z)} := \frac{1}{n_z} \sum_{i: Z_i = z} r_1 r_0 Q_{ii} (r_z^{-2} Y_i - r_z^{-2} \bar{Y}_z)^2.
\]
We now consider $\hat{\cI}_2$. Since it involves cross-sample products, we define $s^2_{r_1r_0\diag^-\{\bs{Q}\}, \;r_z^{-2} Y(z)}$ instead as
\[
s^2_{r_1r_0\diag^-\{\bs{Q}\}, \;r_z^{-2} Y(z)} := \frac{1}{r_z n_z} \sum_{i \ne j: Z_i, Z_j = z} r_1 r_0 Q_{ij} (r_z^{-2} Y_i - r_z^{-2} \bar{Y}_z) (r_z^{-2} Y_j - r_z^{-2} \bar{Y}_z)
\]
and similarly for $s^2_{r_z^{-2} \diag^{-}\{\bs{B}\}, \;Y(z)}$. 

Finally, we discuss the estimation of $\cI_3$ and $\cI_4$. Since $\cI_3$ corresponds to the covariances of potential outcomes from two worlds, it cannot be estimated consistently from observed data directly. Instead, it is only identifiable up to an upper bound. As we will show in the proof of \Cref{thm:inference}, ${\cI}_3$ can be decomposed as
\begin{align*}
    {\cI}_{3} = \sum_{z \in \{0, 1\}}& \left(S^2_{\diag\{\bs{B}\}, Y(z)} - S^2_{\diag\{\bs{Q}\}, Y(z)}- S^2_{\diag^-\{\bs{H}\},Y(z)}\right) + 2S_{\diag^-\{\bs{H}\},Y(1),Y(0)}  \\
    &-  S^2_{\diag\{\bs{H}\},Y(1)-Y(0)} - S^2_{e(1)-e(0)} + O(n^{-1}),
\end{align*}
where the last two terms ($S^2_{\diag\{\bs{H}\},Y(1)-Y(0)}$ and $S^2_{e(1)-e(0)}$) represent treatment effect variation and thus can not be estimated consistently. Fortunately, the last two terms are non-negative; this allows us to provide a consistent estimation of an upper bound of $\cI_3$ just with the first two terms in the above decomposition, which we denote by $\cI_{3, \ub}$. 
%
Noteworthy, besides the variance of a single world, the  ${\cI}_{3,\ub}$ involves a term representing the covariance of counterfactual worlds. We define its empirical estimate as  
\[
s_{\diag^-\{\bs{H}\},Y(1),Y(0)} := \frac{1}{nr_1r_0}\sum_{i \ne j: Z_i=1,Z_j=0} H_{ij} (Y_i-\bar{Y}_1)(Y_j-\bar{Y}_0)
\]
and similarly we can define $s_{\diag^-\{\bs{B}\},Y(1),Y(0)}$ and $s_{\diag^-\{\bs{Q}\},Y(1),Y(0)}$.
In light of the above, we obtain an empirical estimate of ${\cI}_{3,\ub}$ as
\[
\hat{\cI}_{3,\ub} = \sum_{z \in \{0, 1\}} \left(s^2_{\diag\{\bs{B}\}, Y(z)} - s^2_{\diag\{\bs{Q}\}, Y(z)}- s^2_{\diag^-\{\bs{H}\},Y(z)}\right) + 2s_{\diag^-\{\bs{H}\},Y(1),Y(0)}.
\]

For $\cI_4$, mimicking the estimates for $\cI_3$, we propose to estimate it via
\[
\hat{\cI}_4 = 2\left(s_{\diag^-\{\bs{B}\},Y(1),Y(0)}-s_{\diag^-\{\bs{Q}\},Y(1),Y(0)}\right).
\]
Putting together, we get the variance estimator $\hat{\sigma}_{\hd}^2 := \hat{\cI}_1+\hat{\cI}_2+\hat{\cI}_{3,\ub}+\hat{\cI}_4$. The following theorem characterizes the asymptotic convergence of this estimator. 
\begin{theorem}\label{thm:inference}
 {\revone If Assumptions \ref{assumption:positive-limit-of-assigned-proportion}--\ref{assumption:p-o(n)-case} hold, 
    we have, there exists a non-negative sequence $\epsilon_n = \op(1)$ such that,}
    \begin{align*}
         \hat{\sigma}^2_{\hd}= {\sigma}_{\adj}^2 +  S_{e(1)-e(0)}^2 + \op(1), \quad \& \quad  {\revone \frac{\hat{\sigma}^2_{\hd}}{\sigma^2_{\adj}} \geq 1-\epsilon_n.}
    \end{align*}
    {\revone Otherwise, if Assumptions \ref{assumption:positive-limit-of-assigned-proportion}--\ref{assumption:lindeberg-type-condition-p-o(n)}, and \ref{assumption:Positive-variance-limit-alpha-positive} hold,  we have,  there exists a non-negative sequence $\epsilon_n = \op(1)$ such that,}
    \begin{align*}
         \hat{\sigma}^2_{\hd}= {\sigma}_{\hd}^2 + S_{e(1)-e(0)}^2 +  S^2_{\diag\{\bs{H}\},Y(1)-Y(0)} + \op(1), \quad \& \quad  {\revone \frac{\hat{\sigma}^2_{\hd}}{\sigma^2_{\hd}} \geq 1-\epsilon_n.}
    \end{align*}
\end{theorem}
 {\revone \Cref{thm:inference} suggests the Wald-type confidence intervals using $\hat{\sigma}^2_{\hd}$ are conservative.} Since as mentioned earlier, ${\sigma}_{\hd}^2$ contains both the variances contributed by the linear statistic and the quadratic statistic, the construction of $\hat{\sigma}^2_{\hd}$ needs to estimate both variances to guarantee inference validity. We note that we are not the first to estimate the quadratic statistic variance of a covariate adjusted estimator for confidence interval construction. In the regime of $p = o(n^{3/4})$ up to log factors, \cite{chiang2023regression} used a bias-corrected version of a $\textnormal{HC}_3$-type variance estimator which captures the quadratic statistic variance of their cross-fitted estimator to improve the finite-sample performance of confidence interval construction.


Due to the unidentifiability of the counterfactual covariance, the estimated $\hat{\sigma}_\hd^2$ contains a variance inflation.
In the regime $p = o(n)$, our variance estimation has the same inflation as in the lower dimensional regime where $p= O(n^{2/3} / (\log n)^{1/3})$, see \cite{lei2021regression}. This variance inflation is always no greater than the usual inflation \emph{without} any covariate adjustment, namely $S^2_{\tau}$.
 Nevertheless, in the regime $p\asymp n$, the variance inflation $S_{e(1)-e(0)}^2 + S^2_{\diag\{\bs{H}\},Y(1)-Y(0)}$ is not always smaller than $S^2_{\tau}$, especially when there is strong co-linearity between $H_{ii}$ and $\tau_i$. We will demonstrate this in numerical analysis. On the other hand, as we will show in Section~\ref{sec:assumption}, when the data exhibit sufficient linearity and light tail, one can still expect $S_{e(1)-e(0)}^2 + S^2_{\diag\{\bs{H}\},Y(1)-Y(0)}<S_{\tau}^2$. 

 We now showcase an alternative variance estimator with variance inflation equal to $S^2_\tau$. 
 In this estimator, instead of estimating ${\cI}_{3,\ub}$, we focus on the following alternative upper bound of $\cI_3$:
\[
\cI_{3,\ub}^{\prime} := \sum_{z \in \{0, 1\}} \left(S^2_{\diag\{\bs{B}\}, Y(z)} - S^2_{\diag\{\bs{Q}\}, Y(z)}\right),
\]
which can be estimated via $$\hat{\cI}_{3,\ub}^\prime := \sum_{z \in \{0, 1\}} \left(s^2_{\diag\{\bs{B}\}, Y(z)} - s^2_{\diag\{\bs{Q}\}, Y(z)}\right).$$
Armed with $\hat{\cI}_{3,\ub}^\prime$, we define the alternative variance estimator as $\hat{\sigma}_\hd^\prime{}^2 := \hat{\cI}_1+\hat{\cI}_2+\hat{\cI}_{3,\ub}^\prime+\hat{\cI}_4$. The following corollary characterizes the asymptotic convergence of the alternative variance estimator. 
\begin{corollary}\label{cor:inference}
    {\revone If Assumptions \ref{assumption:positive-limit-of-assigned-proportion}--\ref{assumption:p-o(n)-case} hold,
    we have, there exists a non-negative sequence $\epsilon_n = \op(1)$,}
    \begin{align*}
        \hat{\sigma}_{\hd}^\prime{}^2= {\sigma}_{\adj}^2 +  S_{\tau}^2 + \op(1),\quad \& \quad  {\revone \frac{\hat{\sigma}_{\hd}^\prime{}^2}{\sigma^2_{\adj}} \geq 1-\epsilon_n.}
    \end{align*}
    {\revone Otherwise, if Assumptions \ref{assumption:positive-limit-of-assigned-proportion}--\ref{assumption:lindeberg-type-condition-p-o(n)}, and \ref{assumption:Positive-variance-limit-alpha-positive}  hold,  we have, there exists a non-negative sequence $\epsilon_n = \op(1)$,}
    \begin{align*}
         \hat{\sigma}_{\hd}^\prime{}^2= {\sigma}_{\hd}^2 + S_{\tau}^2 +  \op(1),\quad \& \quad  {\revone \frac{\hat{\sigma}_{\hd}^\prime{}^2}{\sigma^2_{\hd}} \geq 1-\epsilon_n.}
    \end{align*}
\end{corollary}
{\revone \Cref{cor:inference} shows the Wald-type confidence intervals using $\hat{\sigma}_{\hd}^\prime{}^2$ are conservative.} In practice, we recommend the use of $\min\{\hat{\sigma}_{\hd}^2,\hat{\sigma}_{\hd}^\prime{}^2\}$ for a shorter confidence interval. Then, no matter $S_{e(1)-e(0)}^2+ S^2_{\diag\{\bs{H}\},Y(1)-Y(0)} > S_{\tau}^2$ or not, the variance inflation is always no greater than $S_\tau^2$, i.e., the inflation without using any covariate adjustment.
Therefore, the confidence interval length from our inferential procedure is always asymptotically shorter than the unadjusted estimator whenever $\sigma^2_\hd < \sigma^2_\cre$.

Since the inferential procedure of \cite{Lin2013Agnostic} may behave poorly in practice when the covariate dimension is relatively large, existing literature recommended to use HC3-type standard error to boost finite sample performance which heavily penalizes dimension $p$ used by the analysis. However, the HC3-type standard error is typically conservative and has no theoretical guarantee in the moderately high-dimensional regime. Based on our theory, we provide an inference procedure that is valid under this regime and the estimated variance is ``tight'' in that the bias is the variance of unit-level treatment effect which can not be estimated from data. 

\section{Justification of assumptions}\label{sec:assumption}
In this section, we justify Assumptions~\ref{assumption:p-o(n)-case}--\ref{assumption:maximum-leverage-score-close-to-alpha}. The following proposition implies that \Cref{assumption:p-o(n)-case} holds with high probability if the potential outcomes are i.i.d. generated from a superpopulation with bounded variance. 
\begin{mproposition}
    \label{proposition:justify-sum-of-max-p-o(n)}
    Fix $z \in \{0, 1\}$. If $p=o(n)$ and $\{Y_{i}(z)\}_{i=1}^n$ are i.i.d.~random variables with $\var(Y_1(z)) < \infty$, then there exists a positive sequence $c_n\rightarrow 0$ such that
    \[
    \mathbb{P}\left(\frac{1}{n}\sum_{i=1}^p \left({Y}_{(i)}(z)-\bar{Y}(z)\right)^2 >c_n\right) \to 0.
    \]
\end{mproposition}

We now focus on the justification of~\Cref{assumption:maximum-leverage-score-close-to-alpha}. The assumptions on $Y_i(z)$'s are common in randomization-based literature; therefore, we only need to justify the assumptions on $\bs{H}$. Since $\max_{1 \le i \le n} |H_{ii} - \alpha| \to 0$ is a sufficient assumption for $\sum_{i = 1}^n (H_{ii} - \alpha)^2 / n \to 0$, we only need to show that the former condition holds with high probability under some superpopulation assumption on $\bs{X}_i$'s. In \Cref{proposition:4+eta-th-moment},  we show that when $\bs{X}_i$'s are i.i.d. realizations from some superpopulation with entry-wise bounded $(4 + \eta)$-th order moment up to some transformation, $\max_{1 \le i \le n} |H_{ii} - \alpha| \to 0$ holds with high probability.

\begin{mproposition}
    \label{proposition:4+eta-th-moment}
    Suppose that $\{\bs{X}_i\}_{i=1}^n$ are i.i.d. random vectors 
    generated from independent random variables as $\bX_i = \mathbf O \bXX_i,$ where $\mathbf O$ is a deterministic non-singular matrix and $\bXX_i$ have independent entries with mean 0, variance 1, and  $\max_j\E |V_{1j}|^{4+\eta} <C$ for some constants $\eta, C >0$; suppose also $\limsup_{n\to \infty}\alpha = \limsup_{n\to \infty}(p/n)<1$. 
    Then, for any constant $\delta$ satisfying that $0 < \delta < \frac{\eta}{8 + 2 \eta}$,
    \begin{equation}\label{eq:Hii}
        \mathbb{P}\left(\max_{i \in [n]}|H_{ii}-\alpha|> n^{- \delta}\right) \to 0.
    \end{equation}
\end{mproposition}

{\rev Under the same assumptions about $\bs{X}_i$ and additional assumptions regarding $Y_i(z)$, we can justify the second part of \Cref{assumption:lindeberg-type-condition-p-o(n)} by \Cref{prop:justify-assumption-3} without assuming a linear relationship between $\bs{X}_i$ and $Y_i(z)$.
	
	\begin{mproposition}
\label{prop:justify-assumption-3}
     Suppose that $\{\bs{X}_i,Y_i(z)\}_{i=1}^n$ are i.i.d. random vectors 
    generated from independent random variables as $\bX_i = \mathbf O \bXX_i,$ where $\mathbf O$ is a deterministic non-singular matrix and $\bXX_i$ have independent entries with mean 0, variance 1, and  $\max_j\E |V_{1j}|^{4+\eta} <C$ for some constants $\eta, C >0$. Then, the followings hold: 
    \begin{itemize}
    \item[(i)] if $p/n \rightarrow 0$, and $\E |Y_i(z)|^{2+\eta} < C$, then there exist constants $C_1,\delta > 0$ such that
    \begin{equation}\label{eq:ei-p/n-tends-to-0}
        \mathbb{P}\left(\max_{i\in [n]}\Big|\frac{e_{i}(z)}{\sqrt{n}}\Big|> C_1 \left(\frac{p}{n}\right)^{1/2}+C_1 n^{-\delta}\right) \to 0;
    \end{equation}
    \item[(ii)] if  $\limsup_{n\to \infty}(p/n)<1$ and $\E |Y_i(z)|^{4+\eta} < C$, then there exist constants $C_2,\delta > 0$ such that
     \begin{equation}\label{eq:ei-p/n-same-order}
        \mathbb{P}\left(\max_{i\in [n]}\Big|\frac{e_{i}(z)}{\sqrt{n}}\Big|> C_2n^{-\delta} \right) \to 0.
    \end{equation}
    \end{itemize}
\end{mproposition}}


We turn to~\Cref{assumption:Positive-variance-limit-alpha-positive}. In fact, it can be justified by simply applying \Cref{corollary:upper-lower-bound-for-sigma-alpha>0}. 
Specifically, under Assumptions \ref{assumption:positive-limit-of-assigned-proportion}, \ref{assumption:2-moment-of-finite-population-for-y} and 
 \ref{assumption:maximum-leverage-score-close-to-alpha}, we have $\liminf_{n \to \infty} \sigma_{\hd,l}^2 > 0$ when 
 either (i) \Cref{assumption:Positive-variance-limit-p-o(n)} holds or (ii) $\liminf_{n \to \infty} \sigma_{\cre}^2> 0$ and $\liminf_{n \to \infty} (p / n) > 0$. 
 These requirements are natural in randomization-based literature.

Finally, we investigate the relationship between $S_\tau^2$ and $S_{\diag\{\bs{H}\}, Y(1) - Y(0)}^2 + S_{e(1) - e(0)}^2$ under a superpopulation framework where $(Y_i(1), Y_i(0), \bs{X}_i, \varepsilon_i(1), \varepsilon_i(0))$ are i.i.d. generated from some  distribution. We assume a linear model where $\bs{X}_i, \varepsilon_i(1), \varepsilon_i(0)$ are independent and $\mu_z$, $z\in\{0,1\}$, are deterministic scalars:
\begin{align}
\label{DGP:linear-model}
    Y_i(1) = \mu_1 + \bs{X}_i^\top\bs{\beta}_1 + \varepsilon_i(1), \quad Y_i(0) = \mu_0+ \bs{X}_i^\top\bs{\beta}_0 + \varepsilon_i(0).
\end{align}
Proposition~\ref{proposition:compare-the-variance} shows that under this superpopulation framework, the confidence interval given by $\hat{\sigma}_{\hd}^2$ is asymptotically no larger than the confidence interval from $\hat{\sigma}_{\hd}^\prime{}^2$ with high probability.
\begin{mproposition}
    \label{proposition:compare-the-variance}
  Under model \eqref{DGP:linear-model}, for $z\in\{0,1\}$, we assume that $\E|\varepsilon_1(z)|^{4} < C$ and $\E|\bs{X}_1^\top\bs{\beta}_z|^{4}  < C$ for some constant $C>0$, and $\bs{X}_i$ satisfies the conditions in \Cref{proposition:4+eta-th-moment}. Then, there exists a positive sequence $c_n\rightarrow 0$ such that
    \[
    \mathbb{P}\left( S^2_{\tau}-S^2_{e(1)-e(0)}-S^2_{\diag\{\bs{H}\},Y(1)-Y(0)} >-c_n\right) \to 0.
    \]
\end{mproposition}

\section{Numerical analysis}\label{sec:numerical}
In this section, we conduct a numerical analysis to examine the finite sample performance of our debiased estimator and its corresponding inference procedure, together with several competitors.

\subsection{Experimental setup}
\paragraph*{Pre-treatment variable generation}
Let $\sca()$ be a standardization function:  for a finite population $\{a_i\}_{i=1}^n$ with $\bs{a} = (a_1,\ldots,a_n)$, $\sca(a_i):= (a_i-\bar{a})/\left(\sum_{i=1}^n (a_i-\bar{a})^2/n\right)^{1/2}$ and $\sca(\bs{a}) = (\sca(a_1),\ldots,\sca(a_n))$. Set $n=1000$ and $r_1=0.35$. We first generate a matrix $\bs{\mathcal{X}}\in\mathbb{R}^{n\times n}$ and two vectors, $\bs{\beta}\in \mathbb{R}^n$ and $\bs{\Delta}\in \mathbb{R}^n$, with i.i.d. entries from  $t$ distribution with $3$ degrees of freedom.  We keep $\bs{\mathcal{X}}$, $\bs{\beta}$  and $\bs{\Delta}$ fixed throughout the simulation. For each covariate-dimension-to-sample-size-ratio $\alpha$, let $\bs{X} = (\bs{X}_1,\ldots, \bs{X}_n)^\top$ be the first $p:=\alpha n$ columns of $\bs{\mathcal{X}}$.  We  generate the potential outcomes according to the following model
\begin{align*}
    Y_i(1) = \mu_1 + \sca(\bs{X}_i^\top\bs{\beta}_1) + \varepsilon_i(1)/\sqrt{\gamma}; \quad Y_i(0) = \mu_0+ \sca(\bs{X}_i^\top\bs{\beta}_0) + \varepsilon_i(0)/\sqrt{\gamma}.
\end{align*}
Here $\mu_z$ ($z\in\{0, 1\}$) are generated i.i.d. from  $t$ distribution with $3$ degrees of freedom. 
For any vector $\bs{a}$, let $\bs{a}_{[p]}$ be the subvector of the first $p$ elements; the coefficients $\bs{\beta}_1$ and $\bs{\beta}_0$ are generated by
\[
\bs{\beta}_{1} = \bs{\beta}_{[p]}+\delta\bs{\Delta}_{[p]}, \quad \bs{\beta}_0 = \bs{\beta}_{[p]} - \delta\bs{\Delta}_{[p]}.
\] 
The factor $\delta$ is introduced to control the heterogeneity of individual-level treatment effect.

\begin{figure}[t!]
    \centering
    \includegraphics[width = \linewidth]{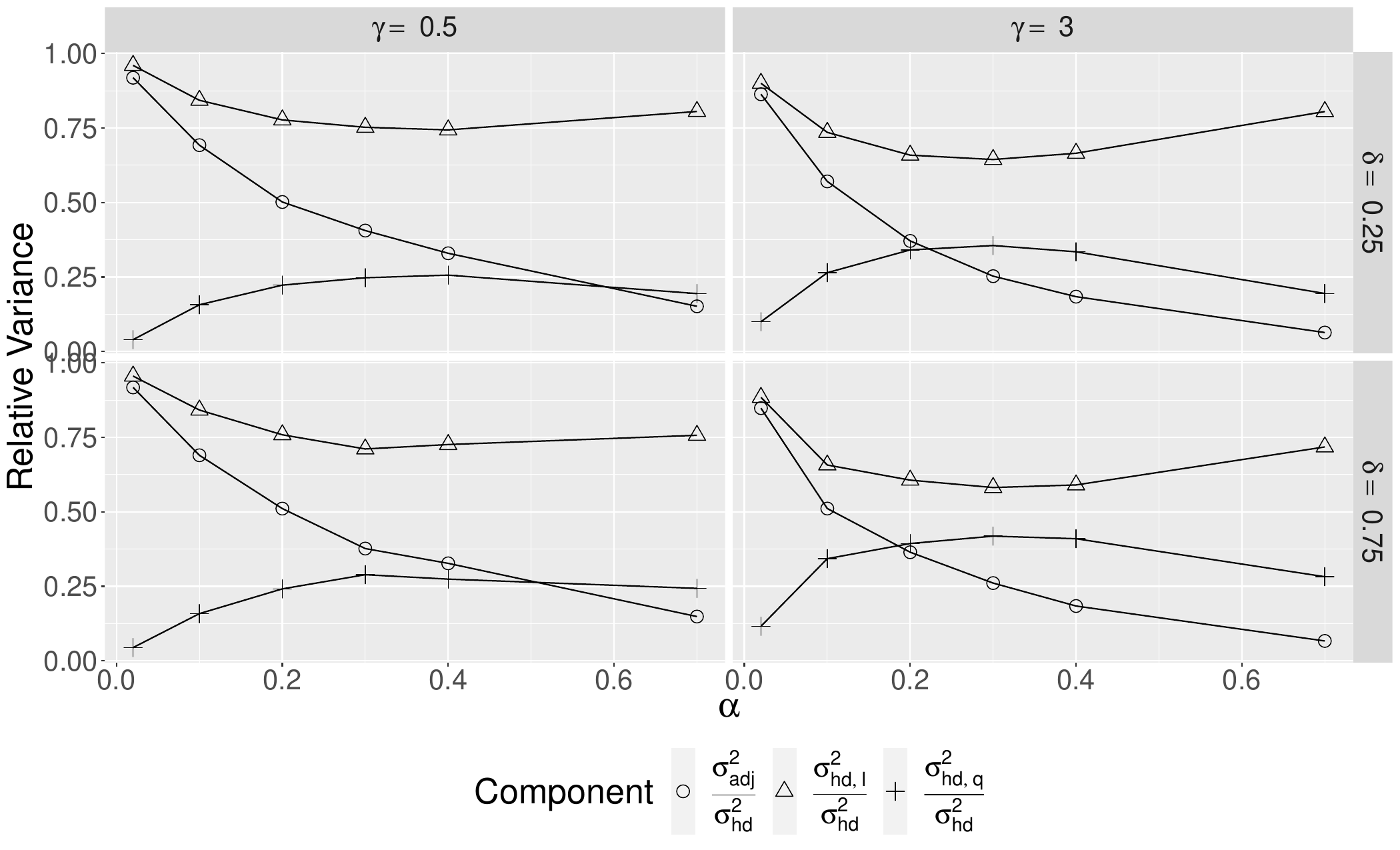}
    \caption{Relative size of variance components for different choices of $\gamma$, $\delta$ and $\alpha$ under the independent $t$ residual.}
    \label{fig:var-comp-q=0-resid-t}
\end{figure}

\begin{figure}[t!]
  \centering
\begin{subfigure}[b]{\linewidth}
  \includegraphics[width=\linewidth]{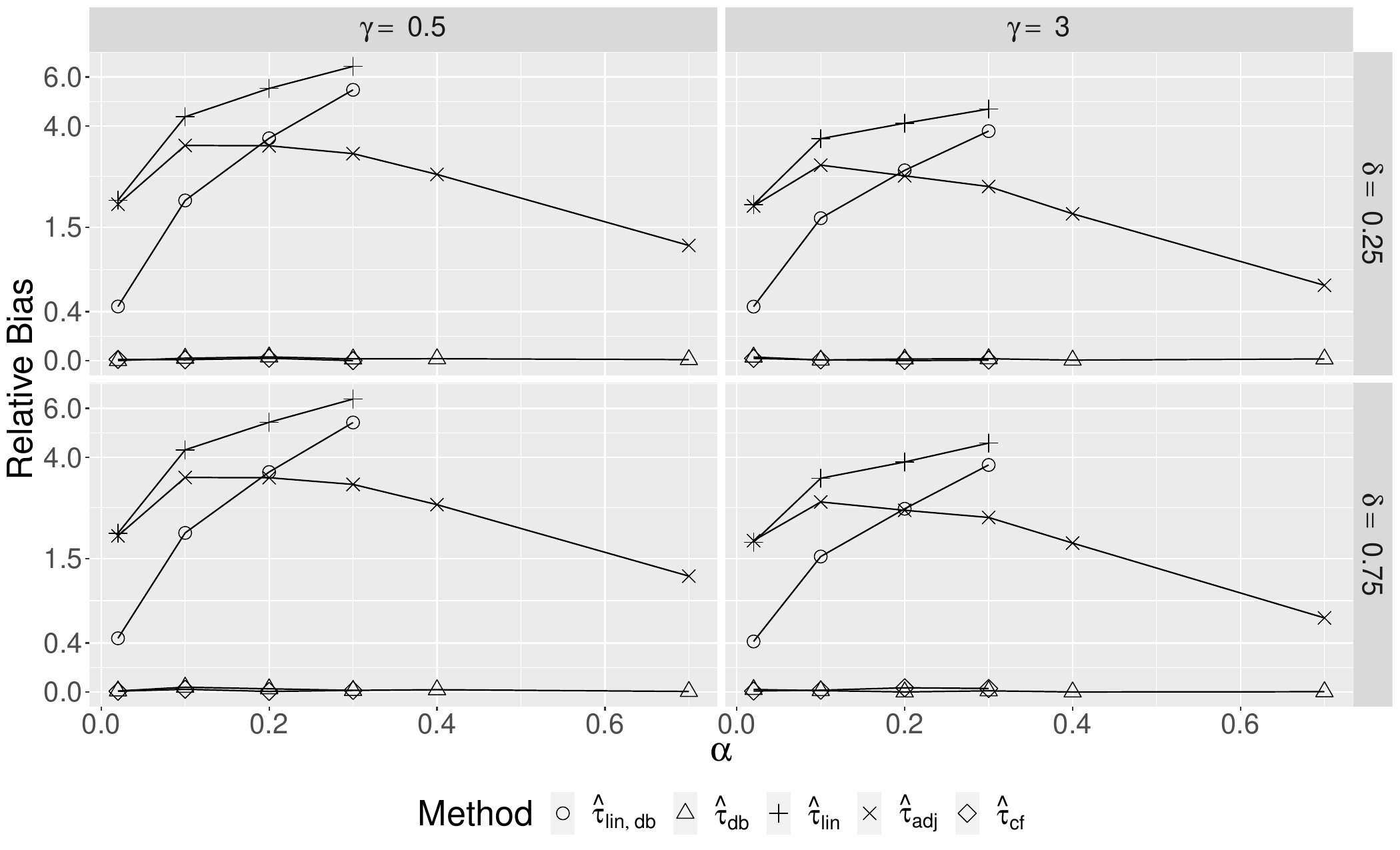}
  \caption{Worst case residual}     \label{fig:bias_worst}
\end{subfigure}
 \begin{subfigure}[b]{\linewidth}
  \centering
  \includegraphics[width=\linewidth]{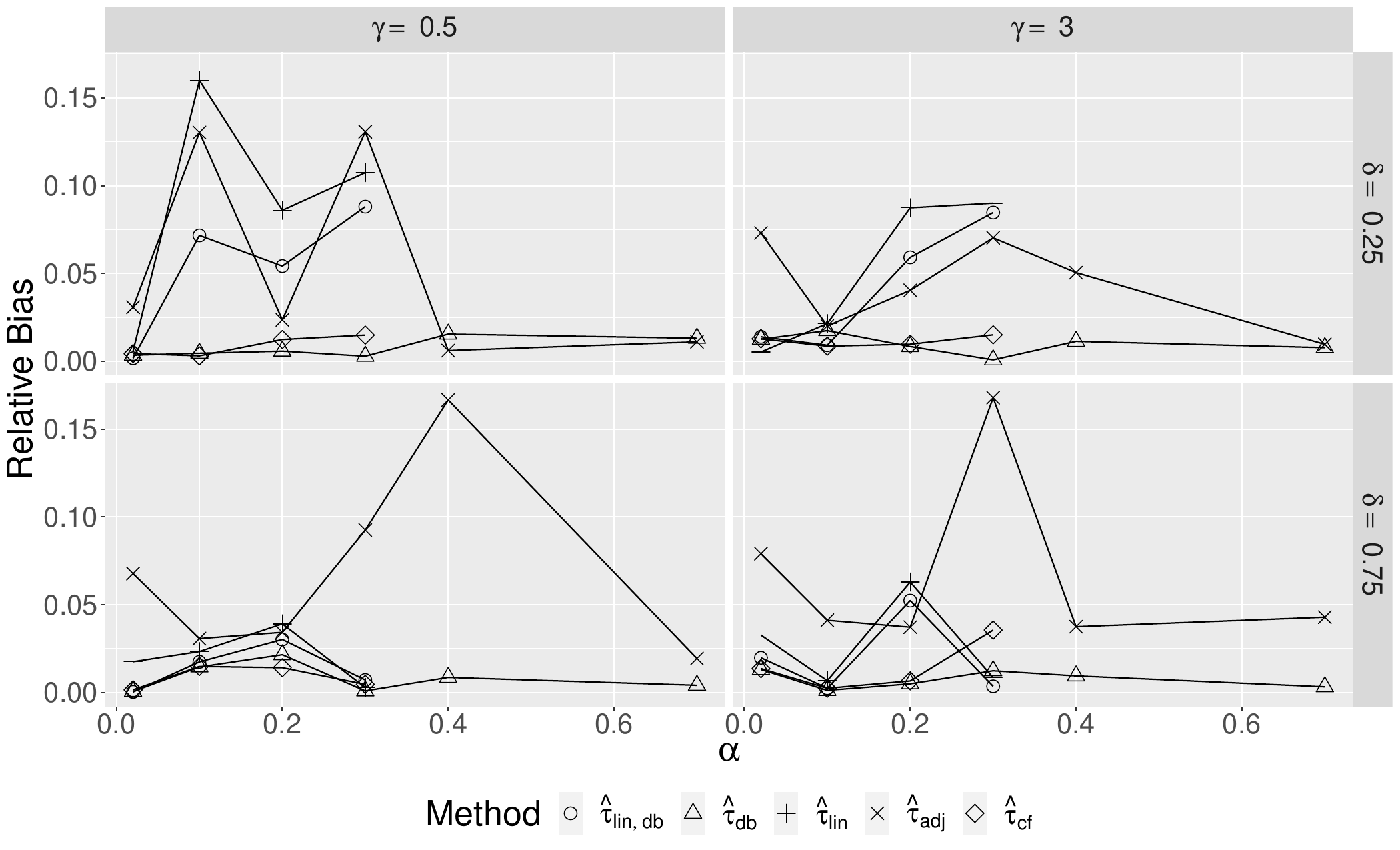}
 \caption{Independent $t$ residual} \label{fig:bias_t}
\end{subfigure}
\caption{Relative bias for different choices of $\gamma$, $\delta$ and $\alpha$ under the worst-case residual and independent $t$ residual. {\rev For (a), we use a transformation of $\log_{10}(1+x)$ for the y-axis to adapt the curve display.}
} \label{fig:bias}
\end{figure}

For the noise terms, $\gamma$ is the scaling factor for the magnitude of the noise.  In addition, we consider $2$ generating models of $\varepsilon_i(z)$:
\begin{itemize}
    \item Worst-case residual: let $\bs{\varepsilon}(z) := (\varepsilon_1(z),\ldots,\varepsilon_n(z))^\top$,
    \begin{align*}
        \bs{\varepsilon}(1) = \sca((\bs{I}-\bs{H})(H_{11},\ldots,H_{nn})^\top);\quad  \bs{\varepsilon}(0) = -2\sca((\bs{I}-\bs{H})(H_{11},\ldots,H_{nn})^\top).
    \end{align*}
    This residual is motivated by \cite{lei2021regression} and produces a large bias for regression-adjusted estimators in theory.
    \item $t$  residual: $\varepsilon_i(z)=\sca(\breve{\varepsilon}_i(z))$. $\breve{\varepsilon}_i(z)$ is generated i.i.d. from $t$ distribution with $3$ degrees of freedom. 
\end{itemize}


We view the simulation as a full factorial experiment and generate the data under all combinations of the following $4$ factors: $\delta = \{0.25,0.75\}$; $\gamma = \{0.5,3\}$ and the covariate-dimension-to-sample-size-ratio $\alpha = \{0.02,0.1,0.2,0.3,0.4,0.7\}$ and generating models of $\varepsilon_i(z)$. Throughout this section we fix $\bs{\mathcal{X}}$ to be generated from $t$ distribution with 3 degrees of freedom; in the Supplementary Material we further provide simulations with $\bs{\mathcal{{X}}}$ generated from Cauchy distribution. {\rev Note that the regression-adjusted estimator using the within-group covariance matrix, such as Lin's estimator, is only well defined when $p < \min\{n_1, n_0\}$. Therefore, the results for these estimators are not available for $\alpha = 0.4, 0.7$.}

\paragraph*{Repeated sampling evaluation}
  Once the pre-treatment variables $\{(\bX_i, Y_i(1), Y_i(0))\}_{i=1}^n$ are generated, we fix them and draw $10000$ random assignments. For evaluation criterion, We consider the empirical relative root mean squared error (relative RMSE) defined by the empirical RMSE of the estimators divided by the oracle standard errors of the unadjusted estimator $\hat{\tau}_{\unadj}$; and the empirical relative absolute bias {\revone (relative bias)} defined by the absolute value of the empirical bias divided by the asymptotic standard error of $\htaudb$, $\sigma_\hd/\sqrt{n}$. 
For inference procedures, we then compare, under a $0.05$ significance level, the empirical coverage probabilities and empirical averages of relative confidence interval length defined by the corresponding confidence interval length divided by the length of the confidence interval constructed \emph{without} covariate adjustment. 

\begin{figure}[h!]
\begin{subfigure}[b]{\linewidth}
  \centering
  \includegraphics[width=\linewidth]{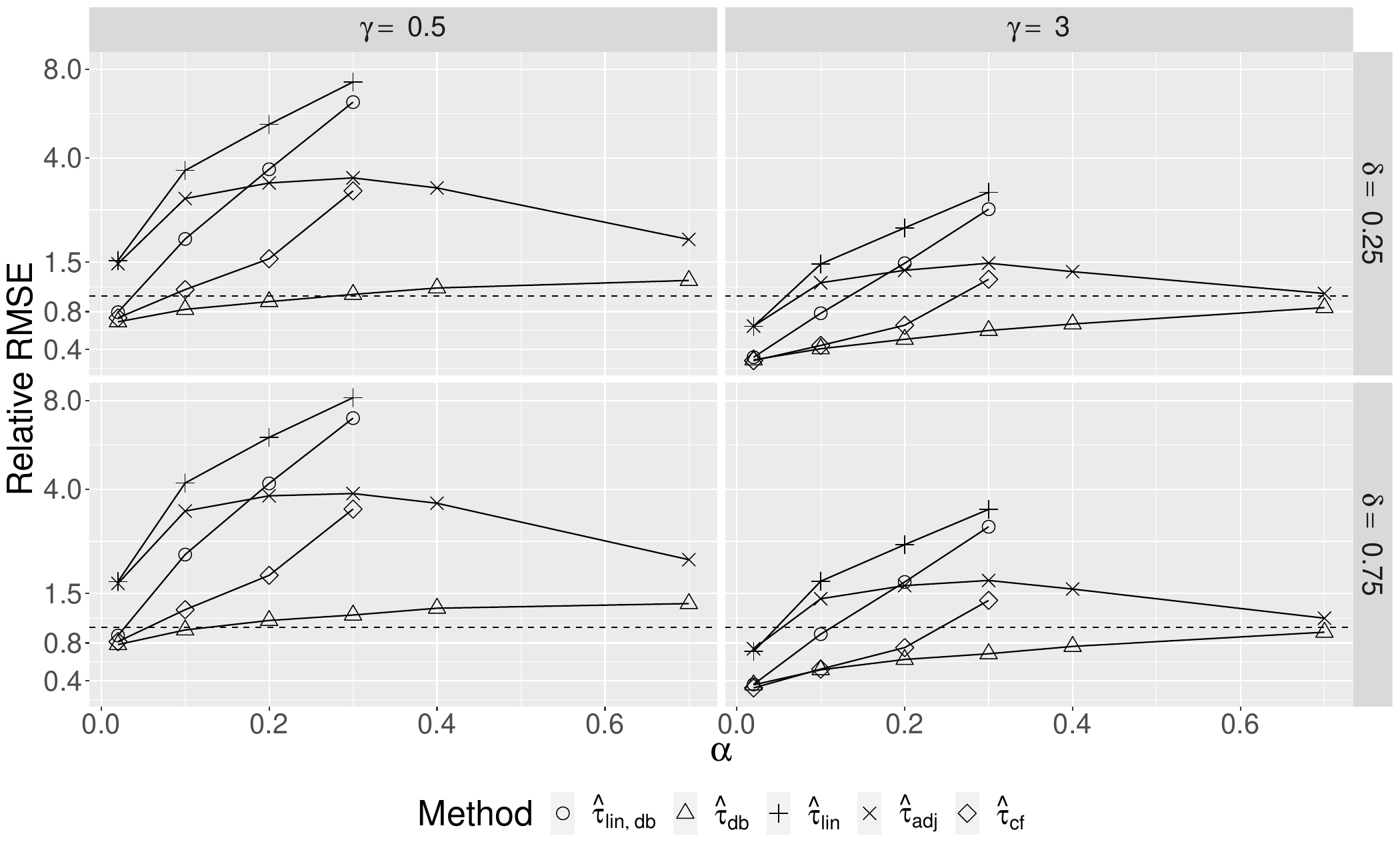}
  \caption{worst-case residual}     \label{fig:rmse_worst}
\end{subfigure}
\begin{subfigure}[b]{\linewidth}
  \centering
  \includegraphics[width=\linewidth]{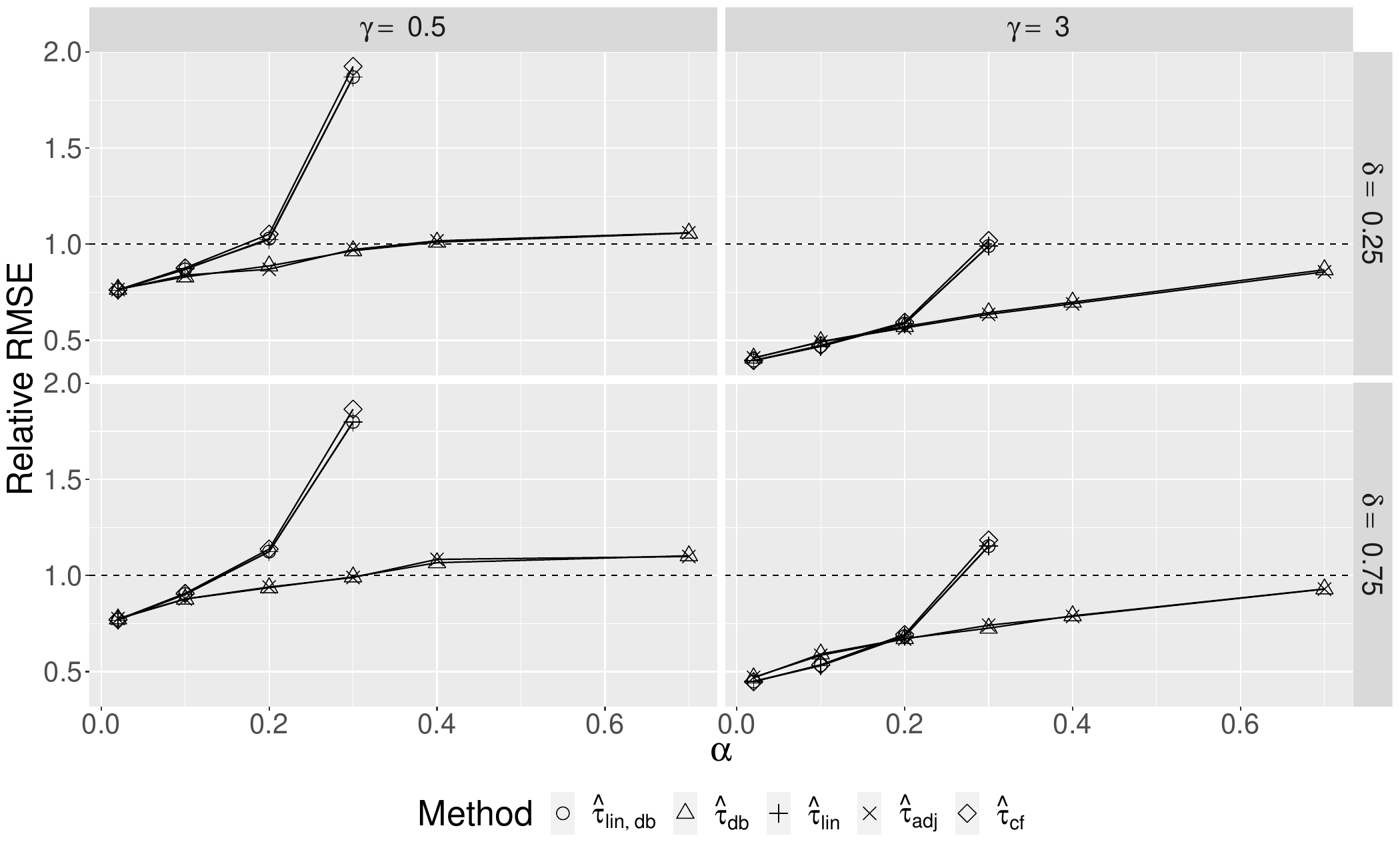}
  \caption{independent $t$ residual}     \label{fig:rmse_t}
\end{subfigure}
\caption{Relative RMSEs for different choices of $\gamma$, $\delta$ and $\alpha$ under the worst-case residual and independent $t$ residual.  The dashed lines signify $1$. {\rev For (a), we use a transformation of $\log_{10}(1+x)$ for the y-axis to adapt the curve display.}} \label{fig:rmse}
\end{figure}

\begin{figure}
 \begin{subfigure}[b]{1\linewidth}
  \centering
  \includegraphics[width=0.95\linewidth]{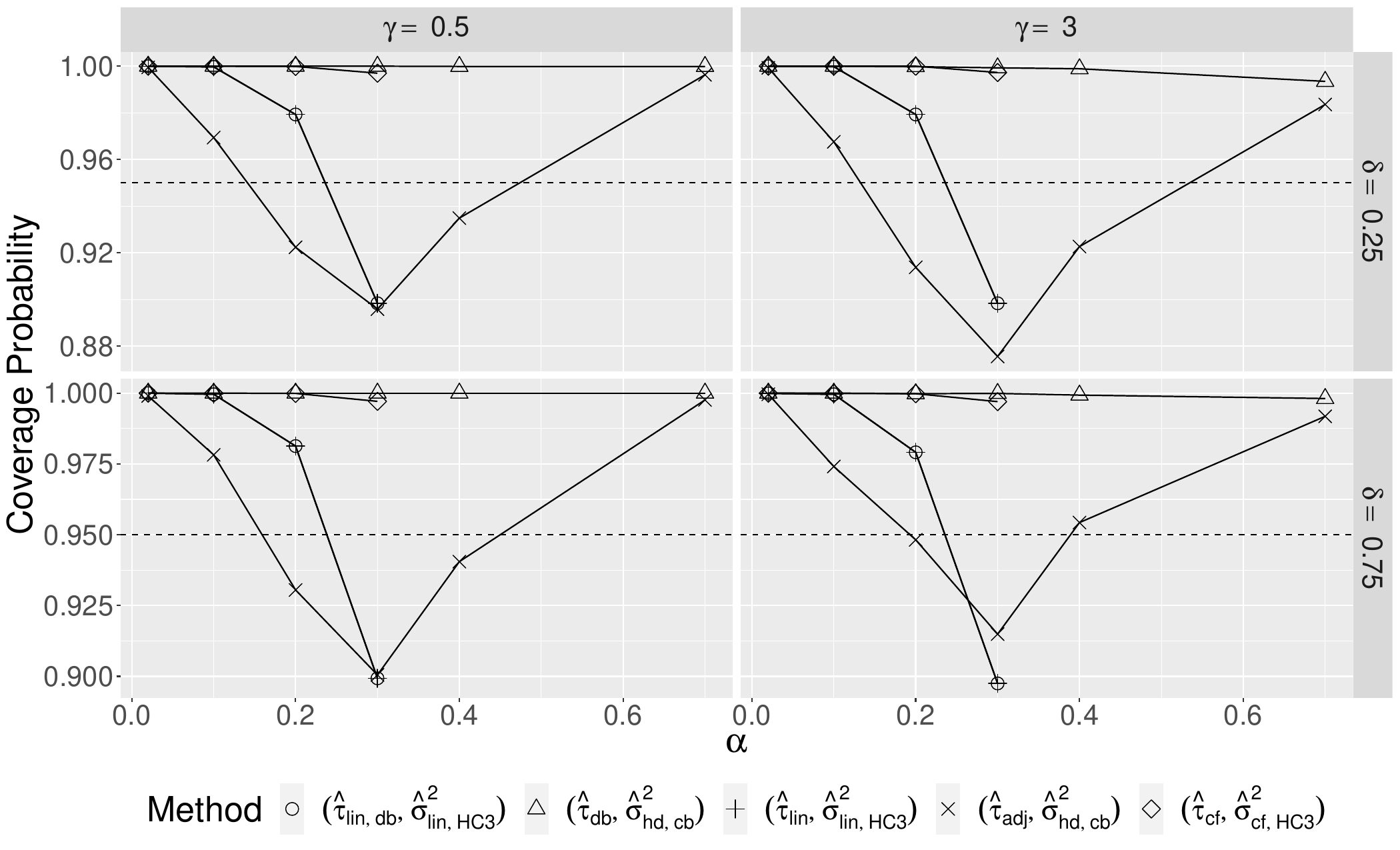}
 \caption{the worst-case residual} \label{fig:cp_worst_residual_q_0}
\end{subfigure}
\begin{subfigure}[b]{1\linewidth}
  \centering
  \includegraphics[width=0.95\linewidth]{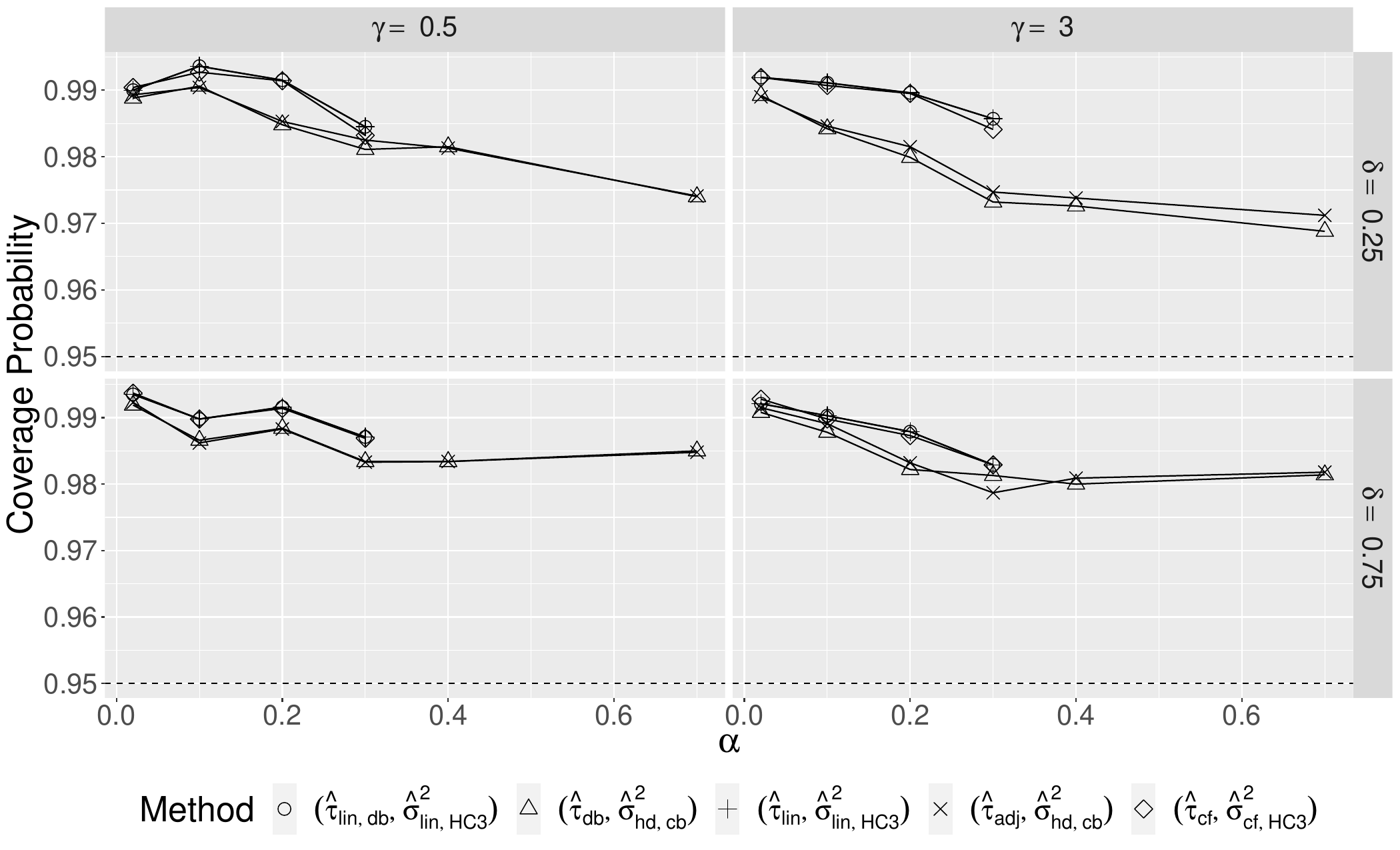}
  \caption{the independent $t$ residual}     \label{fig:cp_t_residual_q_0}
\end{subfigure}
\caption{Empirical coverage probabilities for different choices of $\gamma$, $\delta$ and $\alpha$ under the worst-case residual and independent $t$ residual.  The dashed lines signify $0.95$.} \label{fig:cp}
\end{figure}

\paragraph*{Methods for comparison}
For estimators, we consider our proposed high-dimensional regression estimator $\htaudb$ and its ``un-debiased`` version $\hat{\tau}_{\adj}$, i.e., the one without the debiasing step in~\eqref{eq:db}; Lin's regression estimator \citep{Lin2013Agnostic} $\hat{\tau}_{\textrm{lin}}$ as defined in~\eqref{eq:regadj},
and its debiased version \citep{lei2021regression}, $\hat{\tau}_{\textrm{lin},\textrm{db}}$; 
{\rev the regression adjusted estimator using cross-fitting proposed in \cite{chiang2023regression} denoted as $\hat{\tau}_{\cf}$}.

For the inference procedure, we consider $5$  Wald-type confidence intervals based on the $5$ point estimators and their corresponding variance estimators. In particular, for $\htaudb$ and $\hat{\tau}_{\adj}$, we pair them with our recommended variance estimator $\hat{\sigma}_{\hd,\textrm{cb}}^2:=\min\{\hat{\sigma}_{\hd}^2,\hat{\sigma}_{\hd}^\prime{}^2\}$ (``$\textrm{cb}$'' for combine); for $\hat{\tau}_{\lin}$ and $\hat{\tau}_{\lin,\db}$, we pair them with HC3 variance estimator proposed in~\citet{lei2021regression}, denoted as $\hat{\sigma}_{\lin,\textrm{HC3}}^2$; {\rev $\hat{\tau}_{\cf}$ is combined with the bias-corrected version of HC3-type variance estimator mentioned in \citet{chiang2023regression}, denoted as $\hat{\sigma}_{\cf,\textnormal{HC}3}^2$.}

\subsection{Results}

\paragraph*{Relative magnitude of variance components}
\label{sec:sim-var-comp}
Figure \ref{fig:var-comp-q=0-resid-t} shows relative magnitude of $\sigma_{\hd,l}^2$, $\sigma_{\hd,q}^2$, $\sigma_{\adj}^2$ divided by $\sigma_\hd^2$. With a relatively high dimension, the quadratic component $\sigma_{\hd,q}^2$ is non-negligible. 
Besides, as the dimension increases, $\sigma_{\adj}^2$ becomes an inaccurate approximation even for the linear component of variance, $\sigma^2_{\hd,l}$, not to mention $\sigma^2_{\hd}$.
\paragraph*{The effectiveness of debiasing}
\label{sec:sim-debias}
\Cref{fig:bias} shows the relative bias of different methods under the worst-case residual and independent $t$ residual, respectively. 
Apparently, {\rev the relative bias of both $\htaudb$ and $\hat{\tau}_{\cf}$ are negligible and significantly below $1$ in all cases, even under large $\alpha$ and worst-case residual.} At the same time, under the worst-case residual, the relative bias of $\hat{\tau}_{\adj}$ can be significantly above $1$. This suggests the necessity of debiasing. Moreover, under again the worst-case residual, not only $\hat{\tau}_\lin$ but also $\hat{\tau}_{\lin,\db}$ have explosive growth in bias as $\alpha$ grows.  $\hat{\tau}_{\lin,\db}$
has a bias smaller than $\hat{\tau}_\lin$; nevertheless its bias is still non-negligible under all the worst-case residual setups, except for $\alpha=0.02$. This is consistent with the theory of \cite{lei2021regression} requiring $p$ tending to infinity slow enough.

\paragraph*{Relative RMSE} 
\Cref{fig:rmse} shows
the relative RMSEs of different methods under the worst-case residual and independent $t$ residual. 
Under the worst-case residual, $\htaudb$ has the best performance among the $5$ methods. {\rev Although $\hat{\tau}_\cf$ exhibits negligible bias, its relative RMSE can increase significantly as $p$ increases.} The remaining methods have large RMSEs due to their non-negligible bias.  When signal to noise ratio is relatively high ($\gamma = 3$), $\htaudb$ can exploit the covariate information to keep relative RMSE smaller than $1$, i.e., to perform better than the unadjusted estimator, even with high dimensions. When the signal-to-noise ratio is low ($\gamma=0.5$) and the degree of heterogeneity is high $\delta = 0.75$, the efficiency improvement from our estimator is less compelling. This is consistent with our theory that the {\rev canonical correlation} needs to be relatively larger than $\alpha$ to secure efficiency improvement. 


Under the independent $t$ residual,  when $\alpha$ is small (say less than $0.1$), all methods tend to have similar RMSEs. As $\alpha$ gets larger, their RMSEs begin to diverge. Interestingly, both $\htaudb$ and $\hat{\tau}_{\adj}$ have the smallest relative RMSE, whilst the relative bias of $\htaudb$ is obviously smaller than $\hat{\tau}_{\adj}$ with independent $t$ residual. 

Finally, by checking the figures of both types of residual, we can conclude that when the signal is weak ($\gamma = 0.5$) and the dimension is high, no method can guarantee improvement. But even in the least favorable case, the relative RMSE of $\htaudb$ is just slightly above $1$. In other words, $\htaudb$ never does significant harm to RMSE. {\rev Comparing the two estimators $\hat{\tau}_{\adj}$ and $\hat{\tau}_{\lin}$, both of which do not use any debiasing technique and differ only in whether using the within-group inverse covariance matrix in the estimation of $\hat{\bs{\beta}}_z$, $\hat{\tau}_{\adj}$ significantly outperforms $\hat{\tau}_{\lin}$ in terms of RMSE as $p$ increases. This indicates that, even without debiasing, using the full population covariance alone can significantly improve the estimator. This is in striking contrast to the intuitions in the low dimensional regime, in which the application of the within-group inverse convariance matrix is usually more advocated since it accounts for the imbalance between treatment and control groups to improve precision~\citep{ding2019decomposing} (see also~\citet{cochran1977sampling} for related discussion on the ratio estimator in survey sampling).}


\paragraph*{Inference performance}
\Cref{fig:cp} shows the empirical coverage probabilities of different methods under the worst-case residual and independent $t$ residual. {\rev Only the combinations of $(\htaudb,\hat{\sigma}_{\hd,\cb}^2)$ gives a valid empirical coverage in all cases. $(\hat{\tau}_{\cf},\hat{\sigma}_{\cf,\textnormal{HC3}}^2)$ gives a valid empirical coverage when it is well defined, i,e, when $p < \min\{n_1, n_0\}$. } With worst-case residual, the other methods cannot guarantee a correct empirical coverage as $\alpha$ grows. 

\Cref{fig:ci} shows the relative confidence interval length produced by $\hat{\sigma}_{\hd,\cb}$, $\hat{\sigma}_{\lin,\textrm{HC3}}$ and {\rev $\hat{\sigma}_{\cf,\textnormal{HC3}}$. Overall, $\hat{\sigma}_{\hd,\cb}$ produce the shortest confidence intervals. The relative confidence interval lengths produced by $\hat{\sigma}_{\lin,\textrm{HC3}}$ and $\hat{\sigma}_{\cf,\textnormal{HC3}}$ are quite close.} The trend for the curve of $\hat{\sigma}_{\hd,\cb}$ is very similar to that of $\htaudb$ in \Cref{fig:rmse}. In particular, as long as the relative RMSE of $\htaudb$ is less than $1$, the relative confidence interval length of $\hat{\sigma}_{\hd,\cb}$ is also less than $1$. This echoes our discussion of \Cref{cor:inference}. In the least favorable case ($\gamma=0.5$, $\delta = 0.75$, $\alpha=0.7$, worst-case residual), the relative confidence interval length is about $1.1$, better than the relative RMSE of $\htaudb$ (about $1.5$). 

\begin{figure}[t!]
 \begin{subfigure}[b]{1\linewidth}
  \centering
  \includegraphics[width=0.95\linewidth]{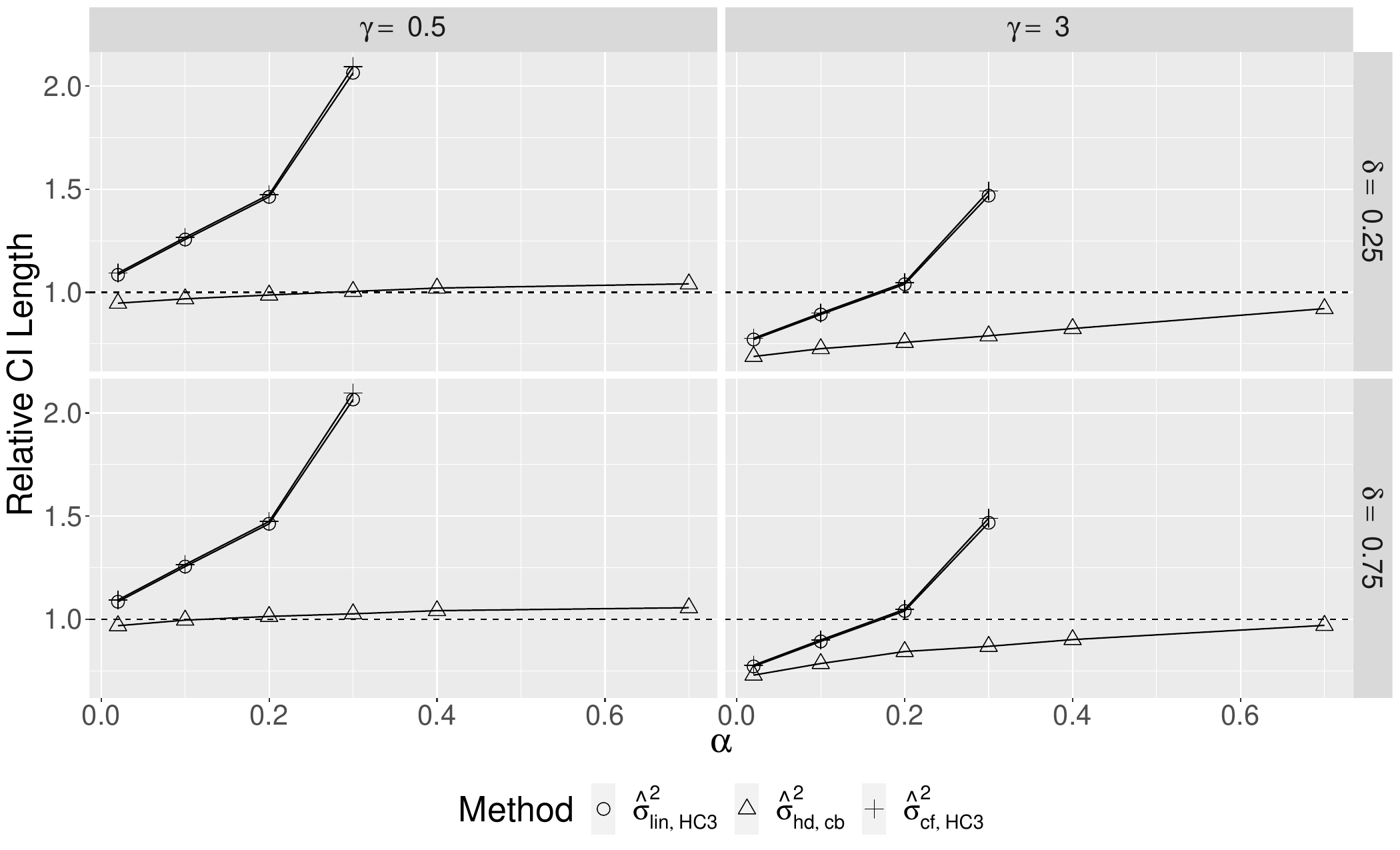}
 \caption{the worst-case residual} \label{fig:ci_worst_residual_q_0}
\end{subfigure}
\begin{subfigure}[b]{1\linewidth}
  \centering
  \includegraphics[width=0.95\linewidth]{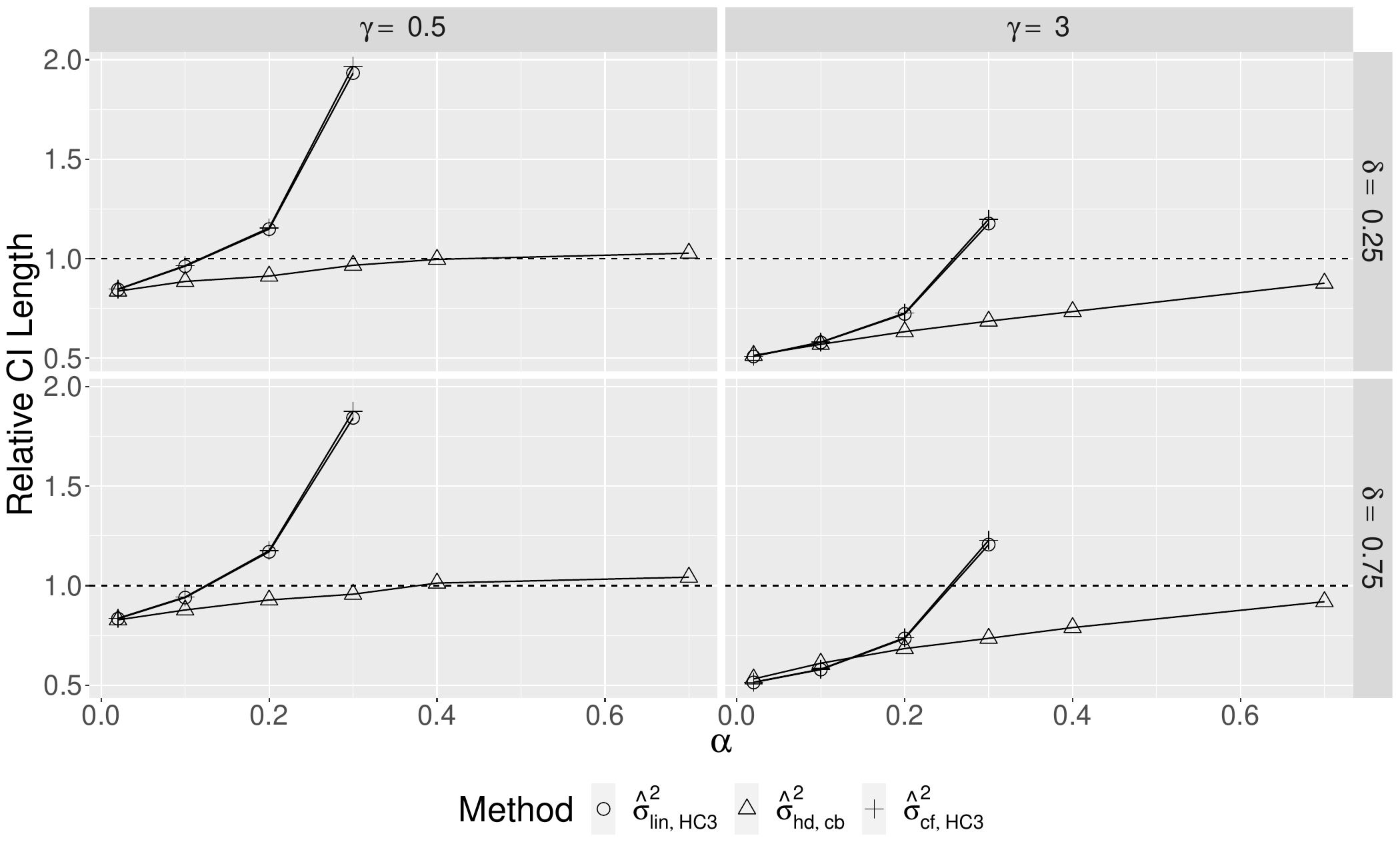}
  \caption{the independent $t$ residual}     \label{fig:ci_t_residual_q_0}
\end{subfigure}
\caption{Relative confidence interval length for different choices of $\gamma$, $\delta$ and $\alpha$ under the worst-case residual and independent $t$ residual.  The dashed lines signify $1$.} \label{fig:ci}
\end{figure}

\begin{figure}
 \begin{subfigure}[b]{0.48\linewidth}
  \centering
  \includegraphics[width=\linewidth]{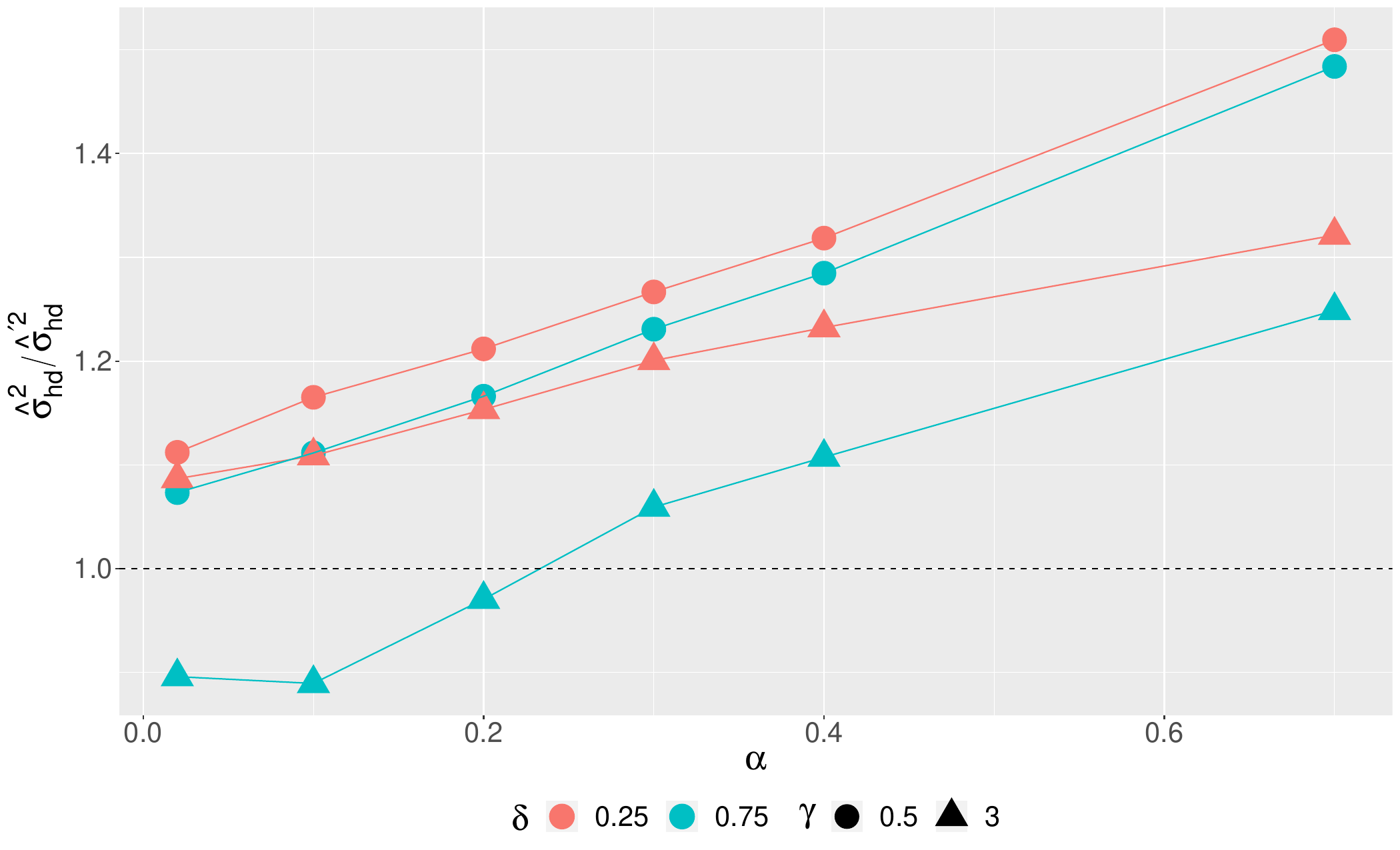}
 \caption{the worst-case residual} \label{fig:var_ratio_worst_residual_q_1}
\end{subfigure}
\begin{subfigure}[b]{0.48\linewidth}
  \centering
  \includegraphics[width=\linewidth]{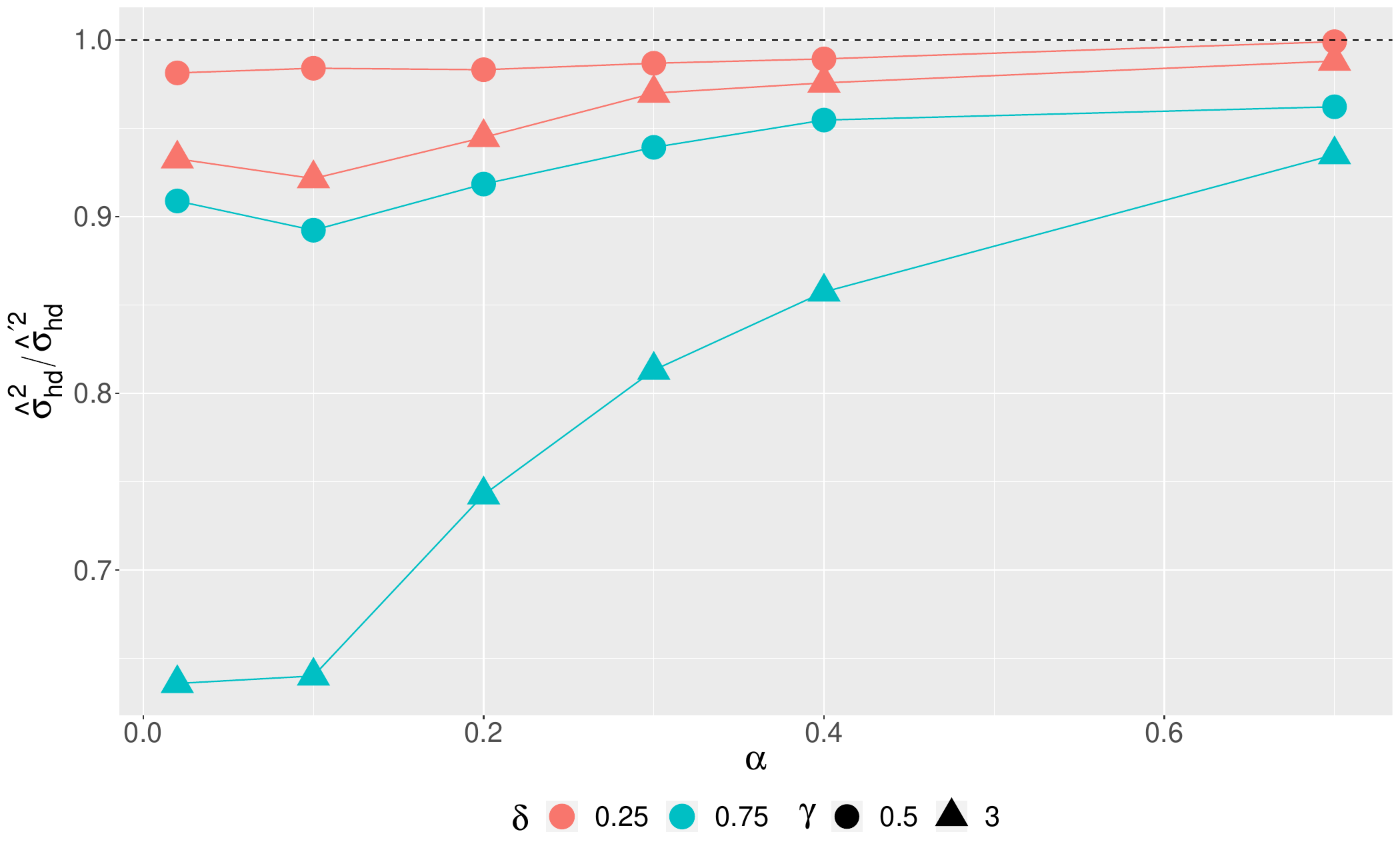}
  \caption{the independent $t$ residual}     \label{fig:var_ratio_t_residual_q_1}
\end{subfigure}
\caption{Ratios of $\hat{\sigma}^2_{\hd}$ to $\hat{\sigma}_{\hd}^{\prime}{}^2$ for different choices of $\gamma$, $\delta$ and $\alpha$ under the worst-case residual and independent $t$ residual.  The dashed lines signify $1$. } 
\end{figure}

\paragraph*{Usefulness of $\hat{\sigma}_{\hd,\textrm{cb}}$}
Figures ~\ref{fig:var_ratio_t_residual_q_1} and \ref{fig:var_ratio_worst_residual_q_1} demonstrate the ratio of $\hat{\sigma}^2_{\hd}$ to $\hat{\sigma}_{\hd}^{\prime}{}^2$.  Under independent $t$ residual, $\hat{\sigma}_{\hd}^2$ is smaller which echoes Proposition~\ref{proposition:compare-the-variance}. whilst under the worst-case residual $\hat{\sigma}_{\hd}^{\prime}{}^2$ is overall smaller. This supports our claim that $\hat{\sigma}_{\hd,\textrm{cb}}^2$ improves the estimation precision by taking the advantages of both $\hat{\sigma}_{\hd}^2$ and $\hat{\sigma}_{\hd}^{\prime}{}^2$.


\section{Conclusion}\label{sec:conclusion}
In practical applications, ignoring covariate dimension can result in catastrophic finite sample performance. Yet, at least under the context of finite-population inference, to the best of our knowledge, no theory explains this phenomenon. In this paper, we fill this gap by proposing a new debiased regression adjustment based average treatment effect estimator; and we study the conditions so that our estimator can have an advantage over the unadjusted estimator. In general, we require that the multiple correlation between covariates and potential outcomes increases with the covariate-dimension-to-sample-size ratio. Therefore, we recommend that practitioners use a moderate number of covariates that are predictive of the potential outcomes.

Our numerical analysis shows that compared to the other competitors, our estimator achieves the best performance in terms of estimation precision, bias reduction, inference reliability, and confidence interval length when the covariate-dimension-to-sample-size ratio is high; 
and an improved efficiency compared to the unadjusted estimator with a sufficiently large signal to noise ratio. It would be of interest to design new covariate adjustment based estimators that can bring improved accuracy even with low signal to noise ratio, which we leave for future work. 
Our additional numerical analysis in the Supplementary Material shows that our estimator is able to provide valid confidence interval even for heavy-tailed covariates, such as from Cauchy distribution.

Our theory builds upon a new central limit theorem of homogeneous sums~\citep{koike2022high}. It would also be interesting to use this new central limit theorem to study rerandomization \citep{morgan2012rerandomization,li2017general} in the moderately high-dimensional regime. In this paper, we mainly focus on completely randomized experiments. It would be interesting to extend our theory to more complex experiments such as stratified experiments \citep{liu2020regression}, and factorial experiments~\citep{liu2022randomization}. We study a high-dimensional extension of the OLS estimator and it would be interesting to consider the high-dimensional extension of the generalized linear estimator \citep{guo2023generalized}.


\begin{acks}[Acknowledgments]
The authors would like to thank the Editor, the Associate Editor and three anonymous referees for helpful comments that improved the paper. 

Fan Yang is also affiliated with Beijing Institute of Mathematical Sciences and Applications.
Yuhao Wang is also affiliated with Shanghai Qi Zhi Institute and Shanghai Artificial Intelligence Laboratory. Correspondence should be addressed to Yuhao Wang.
\end{acks}

\begin{funding}
The research of Yuhao Wang is supported by National Key R\&D Program (No. 2022YFA1008100), the grant of National Natural Science Foundation of China (NSFC) 12201341, and Shanghai Qi Zhi Institute Innovation Program SQZ202304.
Fan Yang is supported in part by the National Key R\&D Program of China (No. 2023YFA1010400). 
\end{funding}

\bibliographystyle{imsart-nameyear}
\bibliography{paper-ref}       

\newpage
\appendix

\begin{center}
	\bf \large
	\uppercase{Supplement to ``Debiased regression adjustment \\
 in completely randomized experiments \\ with moderately high-dimensional covariates''}
\end{center}

\bigskip

Appendix \ref{sec:useful-lemmas} provides some useful lemmas. It includes the technical details for the comments of \eqref{eq:property-of-S2-1} and \eqref{eq:property-of-S2-2}.

Appendix \ref{sec:decomp-htauadj} studies the decomposition of $\hat{\tau}_\adj$ and $\hat{\tau}_\db$. It includes the technical details for the comment of \eqref{eq:adjdecomp}.

Appendix \ref{sec:hajek-coupling} studies an extension of the H\`{a}jek's coupling.

Appendix \ref{sec:normality-p-o(n)} studies the asymptotic normality of $\hat{\tau}_\db$ in the regime $p=o(n)$. It includes the proof of \Cref{theorem:CLT-p=o(n)}.

Appendix \ref{sec:normality-p-O(n)} studies the asymptotic normality of $\hat{\tau}_\db$ in the moderately high-dimensional regime. It includes the proof of \Cref{theorem:CLT-alpha-greater-than-0}.

Appendix \ref{sec:inference-supp} studies the validity of the proposed inference procedure. It includes the proof of \Cref{thm:inference}, \Cref{cor:inference}, and the technical details for the comment of \eqref{eq:rewrite-sigma-hd} and \eqref{eq:sigmadecomp}.

{\rev Appendix \ref{sec:exact-unbiased} provides some comparisons with exactly unbiased estimators $\hat{\tau}_{\CMA}$ and $\hat{\tau}_{\CMO}$ introduced by~\citet{chang2021exact} and~\citet{chiang2023regression}, respectively.}

Appendix \ref{sec:justification-of-assumption} studies the justification of assumptions. It includes the proof of Propositions \ref{proposition:justify-sum-of-max-p-o(n)}--\ref{proposition:compare-the-variance}  and \Cref{corollary:upper-lower-bound-for-sigma-alpha>0}. 

Appendix \ref{sec:addsimu} contains additional numerical experiment results.

The last page of this supplement is a notation table containing all the major notations introduced in the main text as well as the Supplementary Material.
\vspace{4mm}

\noindent{\bf Notations and definitions.}
Define $[n]:= \{1,\ldots,n\}$.
 We use $\gls*{sumonek}$  to denote summation over all $(i_1,\ldots,i_k)$ with mutually distinct elements in $[n]$. So we may use $\sum_{[i,j]}$ and $\sum_{i\ne j}$ interchangeably. For any matrix $\bs{A}$, let $A_{ij}$ be its $(i,j)$th element. We use $\gls*{tldA}$ to denote a centered matrix with
\[
\tilde{A}_{ij} := \begin{cases}
    A_{ij} - \frac{\sum_{k\ne l} A_{kl}}{n(n-1)}, &i\ne j;\\
     A_{ii} - \frac{\sum_{k} A_{kk}}{n}, &i=j.
\end{cases}
\]
Let $\gls*{bsAtwo}$, $\tr(\bs{A})$ be the operator norm and trace of matrix $\bs{A}$, respectively.
For any random variable $U$, we define $\gls*{tldU} := U-\E U$ as a centered random variable. Let $\bs{Z}:= (Z_1,\ldots,Z_n)$. and $\gls*{tldbsZ} := (\tilde{Z}_1,\ldots,\tilde{Z}_n)$. Let $\gls*{deltaij}=1$ if $i=j$ and $\delta_{ij}=0$, otherwise.

Finally, let $\gls*{bsT} = (T_1,\ldots, T_n)\in \{0,1\}^n$ be the indicator of Bernoulli random sampling with each element i.i.d. generated from Bernoulli random variable with probability $r_1$. Moreover, we construct $\bs{T}$ so that the joint distribution of $\bs{T}$ and $\bs{Z}$ follows the so-called ``H\`{a}jek's Coupling'' which will be discussed further in Appendix~\ref{sec:hajek-coupling}. Let $\gls*{diz}$ be the $i$-th entry of the vector $\bs{H}(Y_1(z) - \bar{Y}(z), \ldots, Y_n(z) - \bar{Y}(z))^\top$.

\section{Some useful lemmas}
\label{sec:useful-lemmas}
We start by stating several useful lemmas and then proceed to prove the main results. \Cref{lem:S2-property} shows the technical details of the comments of \eqref{eq:property-of-S2-1} and \eqref{eq:property-of-S2-2}.
\begin{lemma}
\label{lem:S2-property}
For any symmetric matrix $\bs{A},\bs{B}\in \mathbb{R}^{n\times n}$, any constant $c,d$, and any population $\{\bs{a}_i\}_{i=1}^n$, $\{\bs{b}\}_{i=1}^n$, we have
\begin{align*}
        S_{\bs{A}+\bs{B}, \bs{a}, \bs{b}} = S_{\bs{A}, \bs{a}, \bs{b}} & + S_{\bs{B}, \bs{a}, \bs{b}},\quad S_{\bs{A}, c\bs{a}, d\bs{b}} = cd S_{\bs{A}, \bs{a}, \bs{b}} \\
        S^2_{\bs{A}, \bs{a} + \bs{b}} &= S^2_{\bs{A}, \bs{a}} + S^2_{\bs{A}, \bs{b}} + 2 S_{\bs{A}, \bs{a}, \bs{b}}.
\end{align*}
\end{lemma}
\begin{proof}[Proof of \Cref{lem:S2-property}]
    By definition,  we have
    \begin{align*}
         S_{\bs{A}+\bs{B}, \bs{a}, \bs{b}} =& \frac{1}{n-1}\sum_{i=1}^n \sum_{j=1}^n (A_{ij}+B_{ij})(\bs{a}_i-\bar{\bs{a}})(\bs{b}_j-\bar{\bs{b}})^\top\\
         =& \frac{1}{n-1}\sum_{i=1}^n \sum_{j=1}^n A_{ij}(\bs{a}_i-\bar{\bs{a}})(\bs{b}_j-\bar{\bs{b}})^\top+\frac{1}{n-1}\sum_{i=1}^n \sum_{j=1}^n B_{ij}(\bs{a}_i-\bar{\bs{a}})(\bs{b}_j-\bar{\bs{b}})^\top\\
         =&S_{\bs{A}, \bs{a}, \bs{b}} + S_{\bs{B}, \bs{a}, \bs{b}},
    \end{align*}
    and
    \begin{align*}
        S_{\bs{A}, c\bs{a}, d\bs{b}}& = \frac{1}{n-1}\sum_{i=1}^n \sum_{j=1}^n A_{ij}cd(\bs{a}_i-\bar{\bs{a}})(\bs{b}_j-\bar{\bs{b}})^\top 
        = cd S_{\bs{A}, \bs{a}, \bs{b}}.
    \end{align*}
Using that $A_{ij} = A_{ji}$, we have
\begin{align*}
     S^2_{\bs{A}, \bs{a} + \bs{b}}&=  \frac{1}{n-1}\sum_{i=1}^n \sum_{j=1}^n A_{ij}(\bs{a}_i+\bs{b}_i-\bar{\bs{a}}-\bar{\bs{b}})(\bs{a}_j+\bs{b}_j-\bar{\bs{a}}-\bar{\bs{b}})^\top\\
     & =\frac{1}{n-1}\sum_{i=1}^n \sum_{j=1}^n A_{ij}(\bs{a}_i-\bar{\bs{a}})(\bs{a}_j-\bar{\bs{a}})^\top + \frac{1}{n-1}\sum_{i=1}^n \sum_{j=1}^n A_{ij}(\bs{b}_i-\bar{\bs{b}})(\bs{b}_j-\bar{\bs{b}})^\top  \\
     &\quad + \frac{1}{n-1}\sum_{i=1}^n \sum_{j=1}^n A_{ij}(\bs{a}_i-\bar{\bs{a}})(\bs{b}_j-\bar{\bs{b}})^\top + \frac{1}{n-1}\sum_{i=1}^n \sum_{j=1}^n A_{ij}(\bs{b}_i-\bar{\bs{b}})(\bs{a}_j-\bar{\bs{a}})^\top\\
     &=\frac{1}{n-1}\sum_{i=1}^n \sum_{j=1}^n A_{ij}(\bs{a}_i-\bar{\bs{a}})(\bs{a}_j-\bar{\bs{a}})^\top + \frac{1}{n-1}\sum_{i=1}^n \sum_{j=1}^n A_{ij}(\bs{b}_i-\bar{\bs{b}})(\bs{b}_j-\bar{\bs{b}})^\top\\
     &\quad  + \frac{2}{n-1}\sum_{i=1}^n \sum_{j=1}^n A_{ij}(\bs{a}_i-\bar{\bs{a}})(\bs{b}_j-\bar{\bs{b}})^\top\\
     &=S^2_{\bs{A}, \bs{a}} + S^2_{\bs{A}, \bs{b}} + 2 S_{\bs{A}, \bs{a}, \bs{b}}.
\end{align*}
\end{proof}
\begin{lemma}
\label{lem:S2a-S2AAb}
    If $(a_1-\bar{a},\ldots,a_n-\bar{a})^\top = \bs{A}(b_1-\bar{b},\ldots,b_n-\bar{b})^\top$, 
    we have 
    \[
    S^2_{a} = S^2_{\bs{A}^\top \bs{A},b}.
    \]
\end{lemma}
\begin{proof}[Proof of \Cref{lem:S2a-S2AAb}]
    Let $\bs{a} := (a_1-\bar{a},\ldots,a_n-\bar{a})^\top$ and $\bs{b} := (b_1-\bar{b},\ldots,b_n-\bar{b})^\top$.
    Using $\bs{a} = \bs{A}\bs{b}$, we get that 
\begin{align*}
    S^2_{a} = \frac{1}{n-1}\sum_{i=1}^n (a_i-\bar{a})^2 =
\frac{1}{n-1}\bs{a}^\top \bs{a} = \frac{1}{n-1}\bs{b}^\top\bs{A}^\top\bs{A}\bs{b}  = S^2_{\bs{A}^\top \bs{A},b}
\end{align*}
\end{proof}
\begin{lemma}
\label{lem:variance-crt-2-forms}
For any populations $\{{a}_i\}_{i=1}^n$ and $\{{b_i}\}_{i=1}^n$, we have that
\[
(r_1r_0)S^2_{r_1^{-1}a+r_0^{-1}b} = r_1^{-1}S^2_{a}+ r_0^{-1}S^2_{b} - S^2_{a-b}.
\]
\end{lemma}
\begin{proof}[Proof of \Cref{lem:variance-crt-2-forms}]
    Using \Cref{lem:S2-property}, we obtain that
    \begin{align*}
        (r_1r_0)S^2_{r_1^{-1}a+r_0^{-1}b} &= \frac{r_0}{r_1}S^2_{a} + \frac{r_1}{r_0}S^2_{b} + 2S^2_{a,b}\\
        &= \big(\frac{1}{r_1}-1\big)S^2_{a} + \big(\frac{1}{r_0}-1\big)S^2_{b} + 2S^2_{a,b}\\
        &= \frac{1}{r_1}S^2_{a} + \frac{1}{r_0}S^2_{b} -S^2_{a}-S^2_{b} + 2S^2_{a,b}\\
        &=\frac{1}{r_1}S^2_{a} + \frac{1}{r_0}S^2_{b} - S^2_{a-b}.
    \end{align*}
\end{proof}

\begin{lemma}
  \label{lem:variance-of-swr-mean}
  Consider any finite population $\{y_i\}_{i\in[n]}$, the variance of its sample total is
  \begin{align*}
    \var\bigg(\sum_{i:Z_i=1} y_i\bigg) = \frac{n_1n_0}{n(n-1)}\sum_{i} ({y}_i-\bar{y})^2.
  \end{align*}
\end{lemma}
\begin{proof}[Proof of \Cref{lem:variance-of-swr-mean}]
    See Theorem 2.2 of \cite{cochran1977sampling}.
\end{proof}
\begin{lemma}
    \label{lem:order-sample-mean}
    If $\sum_i (y_i-\bar{y})^2 = O(n)$ and $r_1$ tends to a limit in $(0,1)$, then we have
    \[
    \sum_{i:Z_i=1} (y_i - \bar{y})/n_1 = \Op(n^{-1/2}).
    \]
\end{lemma}
\begin{proof}[Proof of Lemma~\ref{lem:order-sample-mean}]
    It follows from Lemma~\ref{lem:variance-of-swr-mean} and Chebyshev's inequality.
\end{proof}

\begin{lemma}
\label{lem:bound-of-the-trace}
 Let $\bs{A}_l$, $l=1, \ldots, q$, be  $n \times n$ deterministic matrices with $n \geqslant 2$. Let $\alpha_{l1}^2, \ldots, \alpha_{ln}^2$ be the eigenvalues of $\bs{A}_l \bs{A}_l^{\top}$ in descending order with $\alpha_{li}\geq 0$ for $l \in [q]$ and  $i\in [n]$. Then, we have that
  $$
  -\sum_{i=1}^n \alpha_{1 i} \cdots \alpha_{q i} \leqslant \operatorname{tr}\left(\bs{A}_1 \cdots \bs{A}_q\right) \leqslant \sum_{i=1}^n \alpha_{1 i} \cdots \alpha_{q i}.
  $$
\end{lemma}
\begin{proof}[Proof of Lemma~\ref{lem:bound-of-the-trace}]
    This Lemma follows directly from Theorem (second version) of \cite{kristof1970theorem} with $\Gamma_l = \bs{I}$, $l=1,\ldots,q$.
\end{proof}
We will use Lemma~\ref{lem:bound-of-the-trace} repeatedly. For example, to bound the following quadratic form of $y_i$'s, $\sum_{[i,j]} H_{ij}^4 y_i^2 y_j^2$. Let $\bs{y} = (y_1,\ldots,y_n)$. We rewrite the quadratic form as the trace of the product of several matrices
\begin{align*}
  \sumtwo \h{1}{2}^4 \yi{1}^{2} \yi{2}^{2} = \operatorname{tr}\left(\operatorname{diag}(\bs{y})^{2}\diag^-(\bs{Q})\operatorname{diag}(\bs{y})^{2}\diag^-(\bs{Q})\right).
\end{align*}
We can apply  Lemma~\ref{lem:bound-of-the-trace} with $A_1,A_2,A_3,A_4$ being $\operatorname{diag}(\bs{y})^{2},\diag^-(\bs{Q}),\operatorname{diag}(\bs{y})^{2},\diag^-(\bs{Q})$, respectively. Let $|y_{(1)}|\geq\ldots\geq |y_{(n)}|$ be the ordered sequence of $\{|y_i|\}_{i=1}^n$. Note that, for $i=1,\ldots,n$, 
\[
\alpha_{1i} =\alpha_{3i} = {\rev |y_{(i)}|^2}, \quad \alpha_{2i} = \alpha_{4i} < \|\diag^-(\bs{Q})\|_2.
\]
Therefore, there is 
\[
\left|\sumtwo \h{1}{2}^4 \yi{1}^{2} \yi{2}^{2}\right| < \|\diag^-(\bs{Q})\|_2^2 \sum_i {\rev y_{(i)}^4} = \|\diag^-(\bs{Q})\|_2^2 \sum_i {\rev y_{i}^4}. 
\]
\begin{lemma}
\label{lem:norm-of-two-matrix}
We have that
\begin{align*}
 \|\diag^-\{\bs{Q}\}\|_2 \leq 1,\quad \|\diag^-\{\bs{H}\}\|_2 \leq 2.
\end{align*}
\end{lemma}
\begin{proof}[Proof of \cref{lem:norm-of-two-matrix}]
 Using the Gershgorin circle theorem (see Theorem~0 of \cite{bell1965gershgorin}), we get that
 \begin{align*}
  \|\diag^-\{\bs{Q}\}\|_2 <\max_i \sum_{j\in [n]\backslash i} H_{ij}^2  = \max_i (H_{ii}-H_{ii}^2) <1.
 \end{align*}
 On the other hand, by the triangle inequality, we have 
 \begin{align*}
  \|\diag^-\{\bs{H}\}\|_2 \leq\|\bs{H}\|_2 + \|\operatorname{diag}\{\bs{H}\}\|_2 \leq 2.
 \end{align*}
\end{proof}
Recall the definition of $\bs{A}(z)$ in the main text. 
\begin{lemma}
\label{lem:equality-td-A}
Fix $z\in\{0,1\}$, we have that
\[
 \sum_{i\in [n]\backslash \{j\}}\tilde{A}_{ij}(z) = -s_j(z),\quad \sum_{j\in [n]\backslash \{i\}}\tilde{A}_{ij}(z)  = d_i(z)  - s_i(z), \quad \tilde{A}_{ii}(z) = s_i(z).
 \]
\end{lemma}
\begin{proof}[Proof of Lemma~\ref{lem:equality-td-A}]
By the fact $\sum_{i,j} A_{ij}(z)=0$, we see that
\begin{align}
\label{eq:eq9}
        \tilde{\bs{A}}(z)={\bs{A}}(z)- \frac{\sum_i A_{ii}(z)}{n-1}\left(\bs{I} - \frac{1}{n}\bs{1}\bs{1}^\top\right).
\end{align}
 Let $\bs{Y}(z) := (Y_1(z),\ldots,Y_n (z))^\top$ and $\bs{d}(z) := (d_1(z),\ldots,d_n(z))^\top$.  Therefore, we have
    \[
    \tilde{\bs{A}}(z)\bs{1} = \bs{A}(z)\bs{1} = \bs{H}\diag(\bs{Y}(z)-\bar{Y}(z)\bs{1})\bs{1} = \bs{H}(\bs{Y}(z)-\bar{Y}(z)\bs{1}) = \bs{d}(z),
    \]
    and 
    \[
        \tilde{\bs{A}}(z)^\top\bs{1} = \bs{A}(z)^\top\bs{1} = \diag(\bs{Y}(z)-\bar{Y}(z)\bs{1})\bs{H}\bs{1} = \bs{0},
    \]
    which implies that
    \begin{align}
    \label{eq:eq10}
          \sum_{j}\tilde{A}_{ij}(z)  = d_i(z), \quad   \sum_{i}\tilde{A}_{ij}(z) = 0.
    \end{align}
    By \eqref{eq:eq9}, we have
    \begin{equation}
    \begin{split}
        \label{eq:eq11}
            \tilde{A}_{ii}(z) &= A_{ii}(z) - \frac{\sum_i A_{ii}(z)}{n} \\
            &= H_{ii}(Y_i(z) - \bar{Y}(z)) - \frac{1}{n} \sum_{j=1}^n H_{jj} (Y_j(z) - \bar{Y}(z)) = s_i(z). 
    \end{split}        
    \end{equation}
    Combining \eqref{eq:eq10} and \eqref{eq:eq11}, we get that
    \[
     \sum_{i\in [n]\backslash \{j\}}\tilde{A}_{ij}(z) = -s_j(z),\quad \sum_{j\in [n]\backslash \{i\}}\tilde{A}_{ij}(z)  = d_i(z)  - s_i(z).
    \]
    This concludes the proof.
\end{proof}

\section{Decompostion of $\htauadj$}
\label{sec:decomp-htauadj}

In this section, we derive the decompositions of $\hat{\tau}_\adj$ and $\hat{\tau}_\db$, which correspond to Propositions \ref{proposition:asymptotic-equivalance-htauadj} and \ref{proposition:decompose-hat-taudb}. Before proving these results, we first state some useful lemmas.
\begin{lemma}
\label{lem:order-of-sumi-hii_yi}
Fix $z\in\{0,1\}$. Under Assumption~\ref{assumption:2-moment-of-finite-population-for-y}, we have
    \[
\frac{\sum_i A_{ii}(z)}{n} = \frac{\sum_i H_{ii}\left(Y_i(z)-\bar{Y}(z)\right)}{n} = O(1).
\]
\end{lemma}
\begin{proof}[Proof of Lemma~\ref{lem:order-of-sumi-hii_yi}]
    By Cauchy-Schwartz inequality, we have
    \begin{align*}
        &\frac{\sum_i H_{ii}\left(Y_i(z)-\bar{Y}(z)\right)}{n}\leq  \left(\frac{\sum_i H_{ii}^2}{n}\right)^{1/2}\left(\frac{\sum_i\left(Y_i(z)-\bar{Y}(z)\right)^2}{n}\right)^{1/2}.
    \end{align*}
Recall that for $i \in [n]$, $H_{ii} \le 1$, we have $\sum_i H_{ii}^2 \le \sum_i H_{ii} = p$. Therefore,
\[
\frac{\sum_i H_{ii}\left(Y_i(z)-\bar{Y}(z)\right)}{n} \leq \left(\frac{p}{n}\right)^{1/2}\left(\frac{\sum_i\left(Y_i(z)-\bar{Y}(z)\right)^2}{n}\right)^{1/2} = O(1),
\]
where in the last step we applied  Assumption~\ref{assumption:2-moment-of-finite-population-for-y}.
\end{proof}

\begin{lemma}
  \label{lem:ZHZ-order-A(z)-order}
Fix $z\in \{0,1\}$. Under Assumptions~\ref{assumption:positive-limit-of-assigned-proportion}--\ref{assumption:2-moment-of-finite-population-for-y} and the first half of \Cref{assumption:lindeberg-type-condition-p-o(n)}, we have that 
  \begin{align*}
     \bs{Z}^\top \tilde{\bs{A}}(z) \bs{Z} = \op(n).
  \end{align*}
 Moreover, we have
  \[
  \bs{Z}^\top \tilde{\bs{H}}\bs{Z}= \op(n).
  \]
\end{lemma}
\begin{proof}[Proof of Lemma~\ref{lem:ZHZ-order-A(z)-order}]
Since $\sum_{i,j} A_{ij}(z)=0$ and $\sum_{i,j} H_{ij} = 0$, we have 
$$\sum_{i \ne j} A_{ij}(z) = - \sum_{i} A_{ii}(z),\quad \sum_{i \ne j} H_{ij} = - \sum_{i} H_{ii},$$ which gives that
\[
\tilde{\bs{A}}(z)={\bs{A}}(z)- \frac{\sum_i A_{ii}(z)}{n-1}\left(\bs{I} - \frac{1}{n}\bs{1}\bs{1}^\top\right),\quad
\tilde{\bs{H}}=\bs{H}- \frac{\sum_i H_{ii}}{n-1}\left(\bs{I} - \frac{1}{n}\bs{1}\bs{1}^\top\right).
\]
In light of these equations, we now analyze $\bs{Z}^\top \bs{A}(z) \bs{Z}, \bs{Z}^\top \bs{H} \bs{Z}$ and $\bs{Z}^\top \left(\bs{I} - \bs{1}\bs{1}^\top/n\right) \bs{Z}$ one by one. 

We first consider $\bs{Z}^\top \bs{A}(z) \bs{Z}$. Observe that
    \[
    \bs{Z}^\top \bs{A}(z) \bs{Z} = \sum_{[i,j]} Z_i Z_j H_{ij} (Y_j(z)-\bar{Y}(z))+  \sum_{i} Z_i H_{ii} (Y_i(z)-\bar{Y}(z)).
    \]
    Applying Lemmas~\ref{lem:quadra-term-1} and \ref{lem:linear-term-1} with $y_i=1$, $g_i = Y_i(z)-\bar{Y}(z)$, $D_{ij}=H_{ij}$, and $a_i = H_{ii}$, we get
    \[
    \sum_{[i,j]} Z_i Z_j H_{ij} (Y_j(z)-\bar{Y}(z)) = r_1^2\sum_{[i,j]}H_{ij} (Y_j(z)-\bar{Y}(z))+\op(n),
    \]
    and
    \[
   \sum_{i} Z_i H_{ii} (Y_i(z)-\bar{Y}(z)) =r_1\sum_{i} H_{ii} (Y_i(z)-\bar{Y}(z))+\op(n).
    \]
    Therefore, we have 
    \[
    \bs{Z}^\top \bs{A}(z) \bs{Z} = \sum_{[i,j]} r_1^2 H_{ij} (Y_j(z)-\bar{Y}(z))+  \sum_{i} r_1 H_{ii} (Y_i(z)-\bar{Y}(z)) + \op(n).
    \]
    
    Applying similar analysis to $\bs{Z}^\top \bs{H} \bs{Z}$  and $\bs{Z}^\top (\bs{I}-\bs{1}\bs{1}^\top/n) \bs{Z}$, we get
    \begin{align*}
            &\bs{Z}^\top \bs{H} \bs{Z} = \sum_{[i,j]} r_1^2 H_{ij} +  \sum_{i} r_1 H_{ii}  + \op(n),\\
            &\bs{Z}^\top (\bs{I}-\bs{1}\bs{1}^\top/n) \bs{Z} = -\sum_{[i,j]} r_1^2  \frac{1}{n}  +  \sum_{i} r_1 \frac{n-1}{n} + \op(n).
    \end{align*}
    Putting together, we have
    \begin{align*}
        & \bs{Z}^\top \tilde{\bs{H}} \bs{Z} =\bs{Z}^\top \left\{{\bs{H}}- \frac{\sum_i H_{ii}}{n-1} (\bs{I}-\bs{1}\bs{1}^\top/n) \right\}\bs{Z} \\
        & =\sum_{[i,j]} r_1^2 \left(H_{ij}+\frac{\sum_k H_{kk}}{n(n-1)}\right) +  \sum_{i} r_1 \left(H_{ii}-\frac{\sum_j H_{jj}}{n}\right) + \op(n) + \op(n)\frac{\sum_i H_{ii}}{n-1}\\
        &= 0+ 0 + \op(n) + \op(n) \cdot p / (n - 1) = \op(n),
    \end{align*}
    where the last equality again uses $
    \sum_{i , j} H_{ij} = 0.
    $
    Similarly, we have
        \begin{align*}
          \bs{Z}^\top \tilde{\bs{A}}(z) \bs{Z} &=\bs{Z}^\top \left\{{\bs{A}}(z)- \frac{\sum_i A_{ii}(z)}{n-1} (\bs{I}-\bs{1}\bs{1}^\top/n) \right\}\bs{Z} \\
        & =\sum_{[i,j]} r_1^2 \left(A_{ij}(z)+\frac{\sum_k A_{kk}(z)}{n(n-1)}\right) +  \sum_{i} r_1 \left(A_{ii}(z)-\frac{\sum_j A_{jj}(z)}{n}\right) \\
        &\quad + \op(n) + \op(n)\frac{\sum_i A_{ii}(z)}{n-1}.
    \end{align*}
    Apparently, the second term on the right-hand side of the above decomposition is equal to zero. For the first term, using that
    \[
    \sum_i \sum_j A_{ij}(z) = \sum_j  \sum_iH_{ij}(Y_j(z)-\bar{Y}(z)) = 0,
    \]
    we see that the first term is equal to zero as well. For the last term, applying Lemma~\ref{lem:order-of-sumi-hii_yi} yields $\frac{\sum_i A_{ii}(z)}{n-1} = O(1)$. Putting together, we have $\bs{Z}^\top \tilde{\bs{A}}(z) \bs{Z} = \op(n)$, which concludes the proof.
\end{proof}

The following proposition gives the detailed formulation and proof of \eqref{eq:adjdecomp}.
\begin{proposition}
\label{proposition:asymptotic-equivalance-htauadj}
If Assumptions~\ref{assumption:positive-limit-of-assigned-proportion}--\ref{assumption:2-moment-of-finite-population-for-y} and the first half of \Cref{assumption:lindeberg-type-condition-p-o(n)} hold, then we have 
  \begin{align*}
     &\hat{\tau}_{\adj}-\bar{\tau} +  \frac{r_1 r_0}{n} \sum_{i = 1}^n H_{ii} \left(\frac{Y_i(1) - \bar{Y}(1)}{r_1^2} - \frac{Y_i(0) - \bar{Y}(0)}{r_0^2}\right)\\
     &= \frac {1}{n} \sum_i ({Z}_i-r_1) c_i- \frac {1}{n} \sum_{[i,j]} ({Z}_i-r_1)({Z}_j-r_1) \left(\frac{{A}_{ij}(1)}{r_1^2} -\frac{{A}_{ij}(0)}{r_0^2} \right)+\op(n^{-1/2}),
\end{align*}
where 
\begin{align*}
           c_i =   \alpha r_0\frac{Y_i(1)-\bar{Y}(1)}{r_1^2} + \alpha r_1\frac{Y_i(0)-\bar{Y}(0)}{r_0^2} -
        (r_0-r_1)\left(\frac{s_i(1)}{r_1^2}-\frac{s_i(0)}{r_0^2}\right)+ \frac{e_i(1)}{r_1}+\frac{e_i(0)}{r_0}.
\end{align*}
\end{proposition}
\begin{proof}[proof of Proposition~\ref{proposition:asymptotic-equivalance-htauadj}]
In the following proof, for ease of presentation, we write $\bs{X}_i  := \bs{X}_i-\bar{\bs{X}}$. 
We observe that
\[
\hat{\tau}_\adj = \frac{\sum_i Z_iY_i(1)}{n_1}-\frac{\sum_i (1-Z_i) Y_i(0)}{n_0} - (r_1\hat{\bs{\beta}}_0+r_0\hat{\bs{\beta}}_1)^\top \left(\sum_i Z_i\bs{X}_i \frac{n}{n_1n_0}\right).
\]
Expanding the third term in the expression, we get
\begin{align*}
    (r_1\hat{\bs{\beta}}_0+r_0\hat{\bs{\beta}}_1)^\top \left(\sum_i Z_i\bs{X}_i \frac{n}{n_1n_0}\right) &= \left(\sum_i Z_i\bs{X}_i^\top \frac{n}{n_1n_0}\right)\bs{S}_{\bs{X}}^{-2}\left(r_1\bs{s}_{\bs{X,Y(0)}}+r_0\bs{s}_{\bs{X,Y(1)}}\right).
\end{align*}
We define 
\[
M_1 = r_0\left(\sum_i Z_i\bs{X}_i^\top \frac{n}{n_1n_0}\right)\bs{S}_{\bs{X}}^{-2}\bs{s}_{\bs{X,Y(1)}},\quad M_2 = r_1\left(\sum_i Z_i\bs{X}_i^\top \frac{n}{n_1n_0}\right)\bs{S}_{\bs{X}}^{-2}\bs{s}_{\bs{X,Y(0)}}.
\]
We now analyze the two terms $M_1$ and $M_2$.

For $M_1$, we write it as
\begin{align*}
    M_1 =& r_0\left(\sum_i Z_i\bs{X}_i^\top \frac{n}{n_1n_0}\right)\bs{S}_{\bs{X}}^{-2}\bs{s}_{\bs{X,Y(1)}}\\
    =&r_0\left(\sum_i Z_i\bs{X}_i^\top \frac{n}{n_1n_0}\right)\bs{S}_{\bs{X}}^{-2}\left(\frac{1}{n_1-1}\sum_i Z_i \bs{X}_i(Y_i(1)-\bar{Y}_1)\right)\\
    =&\frac{n-1}{(n_1-1)n_1}\bs{Z}^\top \bs{A}(1) \bs{Z} +\frac{n-1}{(n_1-1)n_1}\bs{Z}^\top \bs{H} \bs{Z} (\bar{Y}(1)-\bar{Y}_1) =: M_{11}+M_{12}. 
\end{align*}
For $M_{11}$, by the definition of $\tilde{\bs{A}}(1)$ and the fact that
\[
\bs{Z}^\top (\bs{I}-\bs{1}\bs{1}^\top/n) \bs{Z} = \frac{n_1n_0}{n},
\]
we decompose it as
\begin{align*}
    M_{11} =& \frac{n-1}{(n_1-1)n_1}\bs{Z}^\top \bs{A}(1) \bs{Z} =\frac{n-1}{(n_1-1)n_1}\left(\bs{Z}^\top \tilde{\bs{A}}(1) \bs{Z}+\frac{\sum_i H_{ii}\left(Y_i(1)-\bar{Y}(1)\right)}{n-1} \frac{n_1n_0}{n} \right)\\
    =:&M_{111}+M_{112}.
\end{align*}
For $M_{111}$, we further expand $\bs{Z}^\top \tilde{\bs{A}}(1) \bs{Z}$.  Applying Lemma~\ref{lem:equality-td-A} and using the fact $\sum_{[i,j]}\tilde{A}_{ij}(z) = 0$ repeatedly, we obtain that
\begin{align*}
  &\bs{Z}^\top \tilde{\bs{A}}(1) \bs{Z} = \sum_{i} Z_i s_i(1) + \sum_{[i,j]} Z_i Z_j  \tilde{A}_{ij}(1) \\
  &= \sum_{i} Z_i s_i(1) +\sum_{[i,j]} (Z_i-r_1)(Z_j-r_1) \tilde{A}_{ij}(1) + \sum_{[i,j]} (Z_i-r_1)r_1 \tilde{A}_{ij}(1) \\
  & \qquad + \sum_{[i,j]} (Z_j-r_1)r_1 \tilde{A}_{ij}(1) + \sum_{[i,j]} r_1^2 \tilde{A}_{ij}(1)\\
  &=\sum_{i} (Z_i-r_1) s_i(1) + \sum_{[i,j]} (Z_i-r_1)(Z_j-r_1) \tilde{A}_{ij}(1) + \sum_{i}r_1 (Z_i-r_1)(d_i(1)  - s_i(1)) \\
  &\qquad + \sum_{j}r_1 (Z_j-r_1)(-s_j(1))\\
& =\sum_{i} (Z_i-r_1) (s_i(1)+r_1 d_i(1)-2r_1 s_i(1)) + \sum_{[i,j]} (Z_i-r_1)(Z_j-r_1) \tilde{A}_{ij}(1)\\
&= \sum_{i} (Z_i-r_1) (s_i(1)+r_1 d_i(1)-2r_1 s_i(1)) + \sum_{[i,j]} (Z_i-r_1)(Z_j-r_1) {A}_{ij}(1)+\Op(1),
\end{align*}
where in the last step, we used that 
\[
\sum_{[i,j]}({Z}_i-r_1)({Z}_j-r_1) = -\sum_{i}({Z}_i-r_1)^2= O(n),
\]
 and that by Lemma~\ref{lem:order-of-sumi-hii_yi}, 
\[
A_{ij}(1)-\tilde{A}_{ij}(1) = -\frac{\sum_i H_{ii}\left(Y_i(1)-\bar{Y}(1)\right)}{(n-1)n} = O(n^{-1}),\quad \forall i\ne j.
\]
Moreover, by Lemma~\ref{lem:ZHZ-order-A(z)-order}, $\bs{Z}^\top \tilde{\bs{A}}(1) \bs{Z} = \op(n).$
Thus, we obtain that 
\begin{align*}
    M_{111} &= \left(\frac{n}{n_1^2}+ O(n^{-2})\right)\bs{Z}^\top \tilde{\bs{A}}(1) \bs{Z} = \frac{n}{n_1^2}\bs{Z}^\top \tilde{\bs{A}}(1) \bs{Z} + \op(n^{-1})\\
    &=\frac{1}{nr_1^2}\sum_{i} (Z_i-r_1) (s_i(1)+r_1 d_i(1)-2r_1 s_i(1)) \\
    &\quad + \frac{1}{nr_1^2} \sum_{[i,j]} (Z_i-r_1)(Z_j-r_1) {A}_{ij}(1)+\Op(n^{-1}).
\end{align*}
On the other hand, for $M_{112}$, we use Lemma~\ref{lem:order-of-sumi-hii_yi} to get that
\begin{align*}
   M_{112} &= \left(\frac{n_0}{n_1}+O(n^{-1})\right)\frac{\sum_i H_{ii}\left(Y_i(1)-\bar{Y}(1)\right)}{n} \\
   &=  \frac{r_0}{r_1}\frac{\sum_i H_{ii}\left(Y_i(1)-\bar{Y}(1)\right)}{n}+O(n^{-1}). 
\end{align*}
For $M_{12}$, we see that
\begin{align}
\label{eq:decompose-Z-H-Z}
    \bs{Z}^\top \bs{H} \bs{Z} = \bs{Z}^\top \tilde{\bs{H}} \bs{Z} + \frac{\sum_i H_{ii}}{n-1}\frac{n_1n_0}{n} = \bs{Z}^\top \tilde{\bs{H}} \bs{Z} + \frac{\alpha n_1n_0}{n-1},
\end{align}
with which we can decompose $M_{12}$ as
 \begin{align}
M_{12}=\frac{n-1}{(n_1-1)n_1}\bs{Z}^\top \tilde{\bs{H}} \bs{Z} (\bar{Y}(1)-\bar{Y}_1) + \frac{\alpha n_0}{n_1-1}(\bar{Y}(1)-\bar{Y}_1)=:M_{121}+M_{122}. 
 \end{align}
By Lemmas~\ref{lem:order-sample-mean} and \ref{lem:ZHZ-order-A(z)-order}, we have that 
 \[
 M_{121} = O(n^{-1})\op(n)\Op(n^{-1/2 })=\op(n^{-1/2}).
 \]
 For $M_{122}$, we can derive that
 \[
 M_{122} = \left(\frac{\alpha r_0}{r_1}+O(n^{-1})\right)(\bar{Y}(1)-\bar{Y}_1) = \frac{\alpha r_0}{r_1}(\bar{Y}(1)-\bar{Y}_1)+\Op(n^{-3/2}).
 \]
Combining the above results, we obtain that 
\begin{align}
    M_{1} &= \frac{r_0}{r_1}\frac{\sum_i H_{ii}\left(Y_i(1)-\bar{Y}(1)\right)}{n} \nonumber \\
    &+ \frac{1}{nr_1^2}\sum_{i} (Z_i-r_1) \left[s_i(1)+r_1 d_i(1)-2r_1 s_i(1)-\alpha r_0(Y_i(1)-\bar{Y}(1))\right] \nonumber \\ 
& +\frac{1}{nr_1^2} \sum_{[i,j]} (Z_i-r_1)(Z_j-r_1) {A}_{ij}(1)+\op(n^{-1/2}). \label{eq:eq12}
\end{align}

Now, notice that
\begin{align*}
    &M_1 = r_0\left(\sum_{i:Z_i=1} \bs{X}_i^\top/n_1-\sum_{i:Z_i=0} \bs{X}_i^\top/n_0\right)\bs{S}_{\bs{X}}^{-2}\bs{s}_{\bs{X,Y(1)}},\\
    &M_2 = -r_1\left(\sum_{i:Z_i=0} \bs{X}_i^\top/n_0-\sum_{i:Z_i=1} \bs{X}_i^\top/n_1\right)\bs{S}_{\bs{X}}^{-2}\bs{s}_{\bs{X,Y(0)}}.
\end{align*}
So, similar arguments also apply to $M_2$. By symmetry, replacing $Z_i$ with $1-Z_i$, replacing the treatment-group-specific quantities with their control-group analogues in the formula of \eqref{eq:eq12}, and multiplying with a negative sign, we obtain that 
\begin{align*}
    M_{2} &= -\frac{r_1}{r_0}\frac{\sum_i H_{ii}\left(Y_i(0)-\bar{Y}(0)\right)}{n}\\
    &+ \frac{1}{nr_0^2}\sum_{i} (Z_i-r_1) \left[s_i(0)+r_0 d_i(0)-2r_0 s_i(0)-\alpha r_1(Y_i(0)-\bar{Y}(0))\right]\\
& -\frac{1}{nr_0^2} \sum_{[i,j]} (Z_i-r_1)(Z_j-r_1) {A}_{ij}(0)+\op(n^{-1/2}).
\end{align*}

Finally, the conclusion follows immediately from the equation
\begin{align*}
    \hat{\tau}_\adj-\bar{\tau} & = \frac{\sum_i Z_i(Y_i(1)-\bar{Y}(1))}{nr_1}-\frac{\sum_i (1-Z_i)( Y_i(0)-\bar{Y}(0))}{n r_0} - M_1-M_2 \\
    & = \frac{1}{n} \sum_{i=1}^n (Z_i - r_1) \left(\frac{Y_i(1) - \bar{Y}(1)}{r_1} + \frac{Y_i(0) - \bar{Y}(0)}{r_0}\right) - M_1 - M_2,
\end{align*}
and that  $e_i(z) = Y_i(z) - \bar{Y}(z) - d_i(z)$ for $z \in \{0, 1\}$.
\end{proof}

As a direct consequence of Proposition~\ref{proposition:asymptotic-equivalance-htauadj}, the bias term  is 
\[
  b = -\frac{r_1 r_0}{n} \sum_{i = 1}^n H_{ii} \left(\frac{Y_i(1) - \bar{Y}(1)}{r_1^2} - \frac{Y_i(0) - \bar{Y}(0)}{r_0^2}\right).
\]
Recall that we estimate the bias via (see also~\eqref{eq:db})
  \begin{align*}
    \hat{b}=:- r_1 r_0\left(\frac{1}{n_1}\sum_{i: Z_i = 1} H_{ii} \frac{(Y_i - \bar{Y}_1)}{r_1^2} - \frac{1}{n_0}\sum_{i: Z_i = 0} H_{ii} \frac{(Y_i - \bar{Y}_0)}{r_0^2}\right).
\end{align*}
We apply the following proposition to characterize $\hat{b}$:
\begin{proposition}
  \label{proposition:decompose-hat-b}
  If Assumptions~\ref{assumption:positive-limit-of-assigned-proportion}--\ref{assumption:2-moment-of-finite-population-for-y} hold, then we have that 
\begin{align*}
  \hat{b} =&~ b - \frac{1}{n}\sum_i ({Z}_i-r_1)\left\{r_0 \frac{s_i(1)}{r_1^2}-r_0\alpha \frac{Y_i(1)-\bar{Y}(1)}{r_1^2} + r_1 \frac{s_i(0)}{r_0^2}-r_1\alpha\frac{Y_i(0)-\bar{Y}(0)}{r_0^2}\right\} \\
  &~+ \op(n^{-1/2}).
\end{align*}
\end{proposition}
\begin{proof}[Proof of Proposition~\ref{proposition:decompose-hat-b}]
    We see that
    \begin{align*}
        \hat{b} = - r_1 r_0\left(\frac{1}{n_1}\sum_{i: Z_i = 1} (H_{ii}-\alpha )\frac{Y_i - \bar{Y}_1}{r_1^2} - \frac{1}{n_0}\sum_{i: Z_i = 0} (H_{ii} -\alpha)\frac{Y_i - \bar{Y}_0}{r_0^2}\right) =: M_1+M_2,
    \end{align*}
    where 
    \begin{align*}
        M_1 =  - r_1 r_0\left(\frac{1}{n_1}\sum_{i: Z_i = 1} (H_{ii}-\alpha )\frac{\bar{Y}(1) - \bar{Y}_1}{r_1^2} - \frac{1}{n_0}\sum_{i: Z_i = 0} (H_{ii} -\alpha)\frac{\bar{Y}(0) - \bar{Y}_0}{r_0^2}\right),\\
        M_2 = - r_1 r_0\left(\frac{1}{n_1}\sum_{i: Z_i = 1} (H_{ii}-\alpha )\frac{Y_i - \bar{Y}(1)}{r_1^2} - \frac{1}{n_0}\sum_{i: Z_i = 0} (H_{ii} -\alpha)\frac{Y_i - \bar{Y}(0)}{r_0^2}\right).
    \end{align*}
 For any $\{a_1, \ldots, a_n\}$ and $\{b_1, \ldots, b_n\}$ with empirical averages $\bar{a}$ and $\bar{b}$, there is 
    \begin{align*}
            \sum_{i:Z_i=1} a_i/n_1-\sum_{i:Z_i=0} b_i/n_0 =& \bar{a} - \bar{b} +\frac{1}{n}\sum_i Z_i \left( \frac{a_i - \bar{a}}{r_1}  + \frac{b_i - \bar{b}}{r_0} \right)\\
            =& \bar{a} - \bar{b} +\frac{1}{n}\sum_i (Z_i-r_1) \left( \frac{a_i - \bar{a}}{r_1}  + \frac{b_i - \bar{b}}{r_0} \right).
    \end{align*}
Applying the above equation  with $a_i$ and $b_i$ replaced by $(H_{ii}-\alpha )\frac{Y_i(1) - \bar{Y}(1)}{r_1^2}$ and $(H_{ii} -\alpha)\frac{Y_i(0) - \bar{Y}(0)}{r_0^2}$, respectively,
    we obtain that
    \[
    M_2 = b -\frac{1}{n}\sum_i ({Z}_i-r_1)\left\{r_0 \frac{s_i(1)}{r_1^2}-r_0\alpha \frac{Y_i(1)-\bar{Y}(1)}{r_1^2} + r_1 \frac{s_i(0)}{r_0^2}-r_1\alpha\frac{Y_i(0)-\bar{Y}(0)}{r_0^2}\right\}.
    \]
    It suffice to show that $M_1 = \op(n^{-1/2})$.
    Applying Lemma~\ref{lem:order-sample-mean} with $y_i = H_{ii}$ or $Y_i(z)$, we get $\sum_{i:Z_i=z} (H_{ii}-\alpha)/n_z = \Op(n^{-1/2})$ and  $\bar{Y}_z-\bar{Y}(z)= \Op(n^{-1/2})$, which implies that
    \[
    M_1 = \Op(n^{-1/2})\Op(n^{-1/2}) = \op(n^{-1/2}).
    \]
    This concludes the proof.
\end{proof}
Combining  Propositions~\ref{proposition:asymptotic-equivalance-htauadj} and \ref{proposition:decompose-hat-b}, it is straightforward to derive the following result.
\begin{proposition}
  \label{proposition:decompose-hat-taudb}
  If Assumptions~\ref{assumption:positive-limit-of-assigned-proportion}--\ref{assumption:2-moment-of-finite-population-for-y} and the first half of \Cref{assumption:lindeberg-type-condition-p-o(n)} hold, then we have that 
\begin{align*}
  \htaudb-\bar{\tau} &=  n^{-1}\sum_i ({Z}_i-r_1) \left\{\frac{e_i(1)}{r_1}+ \frac{e_i(0)}{r_0}+ \frac{s_i(1)}{r_1}+\frac{s_i(0)}{r_0}\right\}  \\
  &\qquad -   n^{-1}\sum_{[i,j]} ({Z}_i-r_1)({Z}_j-r_1) \left(\frac{{A}_{ij}(1)}{r_1^2}-\frac{{A}_{ij}(0)}{r_0^2}\right) + \op(n^{-1/2}).
\end{align*}
\end{proposition}
\section{H\`{a}jek's coupling}
\label{sec:hajek-coupling}
In this section, we study H\`{a}jek's coupling for sampling without replacement. We prove the second-order H\`{a}jek's coupling which is \Cref{proposition:second-order-hajek-coupling}. Then we use it to prove that $\hat{\tau}_\db$ is asymptotically equal to the summation of several homogeneous sums which is \Cref{prop:decompose-taudb-with-Ti}.

For ease of presentation, throughout~\Cref{proposition:first-order-hajek-coupling,proposition:second-order-hajek-coupling}, we consider a finite population $\{y_i\}_{i\in [n]}$ with $\sum_{i = 1}^n y_i = 0$.  Let $A_{ij}$ be the $(i,j)$-th element of $\bs{A} := \bs{H}\diag\left\{y_1,\ldots,y_n\right\}$. Let $d_i$ be the $i$th element of $\bs{H}(y_1,\ldots,y_n)^\top$. 
We can see that 
\[
\sum_{j} \tilde{A}_{ij} = d_i;\quad \sum_{i} \tilde{A}_{ij} = 0.
\] 
Recall that $\bs{Z} = (Z_1,\ldots,Z_n)\in \{0,1\}^n$ is the indicator of a completely randomized experiment with $\sum_i Z_i = n_1$ and $n_1/n=r_1$. Let $\bs{T} = (T_1,\ldots, T_n)\in \{0,1\}^n$ be the indicator of Bernoulli random sampling with each element i.i.d. generated from Bernoulli random variable with probability $r_1$. Let $n_1^\prime = \sum_i T_i$ and $\mathcal{T} = \{i:T_i=1\}$. We assume the following coupling between $\bs{T} $ and $\bs{Z}$:
\begin{itemize}
    \item If $n_1^\prime = n_1$, $\bs{Z} = \bs{T}$,
    \item If $n_1^\prime > n_1$, we  select a random sample $\mathcal{D}$ of size $n_1^\prime-n_1$ in $\mathcal{T}$ and define $Z_i = 0$ for $i\in\mathcal{D}$ and $Z_i = T_i$ for $i\in [n]\backslash \mathcal{D}$,
     \item If $n_1^\prime < n_1$, we  select a random sample $\mathcal{D}$ of size $n_1-n_1^\prime$ in $[n]\backslash\mathcal{T}$ and define $Z_i = 1$ for $i\in\mathcal{D}$ and $Z_i = T_i$ for $i\in [n]\backslash \mathcal{D}$.
\end{itemize}
\begin{proposition}[First-order H\`{a}jek's coupling]
  \label{proposition:first-order-hajek-coupling}
  If Assumption~\ref{assumption:positive-limit-of-assigned-proportion} holds and $\sum_i y_i^2 = O(n)$, then we have that
  \[
  n^{-1/2}\sum_i (Z_i-T_i)y_i = \op(1) .
  \]
  \end{proposition}
  \begin{proof}[Proof of Proposition~\ref{proposition:first-order-hajek-coupling}]
      The proposition follows from Lemma~A3 (iii) of \cite{wang2022rerandomization} with $u_i = y_i$.
  \end{proof}

  \begin{proposition}[Second-order H\`{a}jek's coupling]
    \label{proposition:second-order-hajek-coupling}
      Under Assumption~\ref{assumption:positive-limit-of-assigned-proportion} and $\sum_i {y}_i^2 = O(n)$, we have
  \[
  n^{-1/2}\sum_{[i,j]} (Z_iZ_j-T_iT_j) \tilde{A}_{ij} = \op(1).
  \]
  \end{proposition}
  \begin{proof}[Proof of Proposition~\ref{proposition:second-order-hajek-coupling}]
  Let $v= (v_1,\ldots,v_n)^\top$ be a uniform at random permutation of $\{1,\ldots,n\}$ and is independent from $\tdm$. Write $D := \sumtwo \tda{1}{2} (\Ti{1} \Ti{2}-Z_{i_1} Z_{i_2})$; apparently $\E D = 0$. We now bound $\E[D^2]$. First, from the coupling between $\bs{T}$ and $\bs{Z}$, by conditioning on $\tdm$, the random variable $D$ is equal in distribution to
  \[
  \sum_{i=n_1+1}^{\tdm}\sum_{j=1}^{n_1} \tilde{A}_{v_iv_j}+\sum_{i=1}^{n_1}\sum_{j=n_1+1}^{\tdm}\tilde{A}_{v_iv_j}+\sum_{i=n_1+1}^{\tdm} \sum_{j=n_1+1}^{\tdm} \tilde{A}_{v_iv_j}(1-\delta_{ij})
  \]
  if $\tdm > n_1$, 
  \[
   -\sum_{i=\tdm+1}^{n_1}\sum_{j=1}^{\tdm} \tilde{A}_{v_iv_j}-\sum_{i=1}^{\tdm}\sum_{j=\tdm+1}^{n_1}\tilde{A}_{v_iv_j}-\sum_{i=\tdm+1}^{n_1} \sum_{j=\tdm+1}^{n_1} \tilde{A}_{v_iv_j}(1-\delta_{ij})
  \]
  if $\tdm < n_1$, 
  and $0$ if $\tdm = n_1$.
  
  We first consider $D^2$ conditioning on some $\tdm>n_1$. Under this event, we can write $D = \sum_{(i,j) \in \mathcal{S}} \tilde{A}_{v_i v_j}$, where 
  \begin{align*}
      &\mathcal{S} := \{(i,j):~n_1+1 \leq i\leq n_1^\prime, ~1\leq j \leq n_1\}\cup \{(i,j):~ 1\leq i \leq n_1,~n_1+1 \leq j\leq n_1^\prime\}   \cup\\ &\qquad\qquad\qquad\qquad\qquad\qquad\{(i,j): n_1+1 \leq i,j\leq n_1^\prime,i\ne j\}. 
  \end{align*}
  Then, we have that
\begin{align*}
    & D^2 = \sum_{i\ne j,(i,j)\in \mathcal{S}} \tilde{A}_{v_iv_j}^2 + \sum_{i\ne j,(i,j)\in \mathcal{S}} \tilde{A}_{v_iv_j}\tilde{A}_{v_jv_i} + \sum_{i\ne j\ne k,(i,j),(k,j)\in \mathcal{S}} \tilde{A}_{v_iv_j}\tilde{A}_{v_kv_j} + \\
    &+ \sum_{i\ne j\ne k, (i,j),(j,k)\in \mathcal{S}} \tilde{A}_{v_iv_j}\tilde{A}_{v_j v_k} + \sum_{i \ne j\ne k, (i,j),(i,k)\in \mathcal{S}} \tilde{A}_{v_iv_j}\tilde{A}_{v_i v_k} +\sum_{i\ne j\ne k\ne l, (i,j),(k,l)\in \mathcal{S}} \tilde{A}_{v_iv_j}\tilde{A}_{v_k v_l}.
\end{align*}
For the first term, we have that for each index,
  \[
  \E \tilde{A}_{v_iv_j}^2 = \frac{\sumtwo\tda{1}{2}^2}{n(n-1)}.
  \]
  Similarly, we have 
\begin{align*}
    &\E \tilde{A}_{v_iv_j}\tilde{A}_{v_jv_i} = \frac{\sumtwo\tda{1}{2}\tda{2}{1}}{n(n-1)},\quad  \E\tilde{A}_{v_iv_j}\tilde{A}_{v_kv_j}=\frac{\sumi{1}{3}\tda{1}{2}\tda{3}{2}}{n(n-1)(n-2)}, \\ & \E\tilde{A}_{v_iv_j}\tilde{A}_{v_j v_k}=\frac{\sumi{1}{3}\tda{3}{2}\tda{2}{1}}{n(n-1)(n-2)},  \quad\E\tilde{A}_{v_iv_j}\tilde{A}_{v_i v_k}= \frac{\sumi{1}{3}\tda{2}{3}\tda{2}{1}}{n(n-1)(n-2)},\\
    &\E\tilde{A}_{v_iv_j}\tilde{A}_{v_k v_l}=\frac{\sumi{1}{4}\tda{1}{2}\tda{3}{4}}{n(n-1)(n-2)(n-3)}.
\end{align*}
To understand the order of the above terms, we introduce $M_1, \ldots, M_5$ as
\begin{align*}
  &M_1:=\sum_{i_1} \tda{1}{1}^2,\quad M_2:=\sumtwo \tda{1}{2}^2,\quad 
  M_3:=\bigg|\sumtwo \tda{1}{2}\tda{2}{1}\bigg|,\\
  &M_4:=\sum_{i_1} \di{1}^2, \quad M_5:=\bigg|\sum_{i_1}\di{1} \tda{1}{1}\bigg|,
\end{align*}
 Now, by repeatedly applying $
\sum_{j} \tilde{A}_{ij} = d_i$ and $ \sum_{i} \tilde{A}_{ij} = 0$, we obtain that
\begin{align*}
  \sumi{1}{3}\tda{1}{2}\tda{1}{3} &=\sumtwo \di{1}\tda{1}{2} -\sumtwo (\tda{1}{2}\tda{1}{1}+\tda{1}{2}\tda{1}{2}) \\
  &=  \sumtwo\di{1}\tda{1}{2} -\sumtwo \tda{1}{2}\tda{1}{1} -M_2 \\
  &= \sumtwo \di{1}\tda{1}{2} -\sum_{i_1} (\di{1}\tda{1}{1}- \tda{1}{1}\tda{1}{1}) - M_2 \\
  &= \sumtwo \di{1}\tda{1}{2} + O(M_1+M_2+M_5)\\
&= \sum_{i_1}(\di{1}^2-\di{1}\tda{1}{1})+O(M_1+M_2+M_5) = O(M_1+M_2+M_4+M_5);\\
 \sumi{1}{3}\tda{1}{2}\tda{2}{3} &= \sumtwo\tda{1}{2}\di{2} -\sumtwo (\tda{1}{2}\tda{2}{1}+\tda{1}{2}\tda{2}{2}) \\
 &= \sumtwo\tda{1}{2}\di{2} -M_3+ \sum_{i_2} \tda{2}{2}^2 \\
& =\sumtwo\tda{1}{2}\di{2} + O(M_1+M_3) \\
&= -\sum_{i_2}\tda{2}{2}\di{2} + O(M_1+M_3) = O(M_1+M_3+M_5);\\
\sumi{1}{3}\tda{1}{2}\tda{3}{2} &=  -\sumtwo (\tda{1}{2}\tda{1}{2}+\tda{1}{2}\tda{2}{2}) = O(M_1+M_2).
\end{align*}
Applying $
\sum_i d_i = 0$ and $\sum_i \tilde{A}_{ii} = 0,$, we obtain that 
\begin{align*}
    &\sumi{1}{4}\tda{1}{2}\tda{3}{4} = \sumi{1}{3}\tda{1}{2}\di{3}-\sumi{1}{3}(\tda{1}{2}\tda{3}{1}+\tda{1}{2}\tda{3}{2}+\tda{1}{2}\tda{3}{3})\\
    &=-\sumtwo(\tda{1}{2}\di{1}+\tda{1}{2}\di{2})-\sumi{1}{3}(\tda{1}{2}\tda{2}{3}+\tda{1}{2}\tda{3}{2}+\tda{1}{2}\tda{3}{3})\\
    &= -\sumtwo(\tda{1}{2}\di{1}+\tda{1}{2}\di{2})-\sumi{1}{3}\tda{1}{2}\tda{3}{3} -\sumi{1}{3}(\tda{1}{2}\tda{2}{3}+\tda{1}{2}\tda{3}{2})\\
    &=-\sumtwo(\tda{1}{2}\di{1}+\tda{1}{2}\di{2})+\sumtwo(\tda{1}{2}\tda{1}{1}+\tda{1}{2}\tda{2}{2}) \\
    &\quad -\sumi{1}{3}(\tda{1}{2}\tda{2}{3}+\tda{1}{2}\tda{3}{2}).
\end{align*}
Then, using the derivations in the analysis of $\sumi{1}{3}\tda{1}{2}\tda{1}{3}$, $\sumi{1}{3}\tda{1}{2}\tda{2}{3}$ and $\sumi{1}{3}\tda{1}{2}\tda{3}{2}$, we get that 
\[
\sumi{1}{4}\tda{1}{2}\tda{3}{4} = O(M_1 + \cdots + M_5).
\]
On the other hand, writing $\Delta := |\tdm - n_1|$, we have
\begin{align*}
    & |\mathcal{S}| \le 2 n \Delta, \quad |\{(i, j, k): i \neq j \neq k, (i,j), (k,j) \in \mathcal{S}\}| \le 2 n^2 \Delta,\\
    & |\{(i,j,k,l): i \neq j \neq k \neq l, (i,j), (k, l) \in \mathcal{S}\}| \le 4 n^2 \Delta^2 \le 4 n^3 \Delta.
\end{align*}
Putting together, we obtain that when $\tdm > n_1$, there exists a universal constant $C>0$ which does not depend on $\Delta$ such that
\[
\E[D^2 \mid \tdm] \le C \Delta n^{-1} \sum_{t = 1}^5 M_t.
\]
With similar arguments, we can obtain the same bound for $\E[D^2 \mid \tdm]$ when $\tdm < n_1$. Finally, with the law of total expectation, we obtain that 
\[
\E[D^2] \le C \E[\Delta] n^{-1} (M_1+M_2+M_3+M_4+M_5).
\]
Now, by Assumption~\ref{assumption:positive-limit-of-assigned-proportion}, we can bound $\E\Delta$ as
\[
\E\Delta \le (\E\Delta^2)^{1/2} = \left(nr_1r_0\right)^{1/2} = O(n^{1/2}).
\]
Combining these results, we get
\begin{align*}
  \E D^2 =O\left(n^{-1/2}(M_1+M_2+M_3+M_4+M_5)\right).
\end{align*}

It remains to bound $M_i$, $i=1,\ldots,5$. Since $\h{1}{1}\le 1$ and $\h{2}{2} = \sum_{i_1} \h{1}{2}^2$ (due to the fact $\bs{H} = \bs{H}^2$), we have that 
\begin{align*}
  M_1&= \sum_{i_1} \tda{1}{1}^2  \le  \sum_{i_1}\h{1}{1}^2\yi{1}^2\le  \sum_{i_1}\yi{1}^2  = O(n),\\
  M_2 &= \sumtwo \tda{1}{2}^2  \le  \sumtwo \h{1}{2}^2 \yi{2}^2 \le  \sum_{i_1}\h{1}{1}\yi{1}^2\le \sum_{i_1}\yi{1}^2=O(n).
\end{align*}
By Cauchy-Schwarz inequality, we have that $  M_3 \le M_2 =O(n)$. 
Finally, we have
\begin{align*}
   M_4 = \sum_{i} d_i^2 \le \sum_{i} y_i^2 = O(n),\quad M_5 \leq (M_4 M_1)^{1/2} = O(n).
\end{align*}
The above bounds give that $\E D^2 = O(n^{1/2})$.
Therefore, by Chebyshev's inequality, we have
   \[
   \sumtwo \tda{1}{2} (\Ti{1} \Ti{2}-Z_{i_1} Z_{i_2})= O_p(n^{1/4}) = \op(n^{1/2}).
   \]
The conclusion then follows.
  \end{proof}
  
 Equipped with Propositions~\ref{proposition:first-order-hajek-coupling} and \ref{proposition:second-order-hajek-coupling}, we can now approximate
  $\htaudb$ with a polynomial of $T_i$'s.

  \begin{proposition}
  \label{prop:decompose-taudb-with-Ti}
    If Assumptions~\ref{assumption:positive-limit-of-assigned-proportion}--\ref{assumption:2-moment-of-finite-population-for-y} and the first half of \Cref{assumption:lindeberg-type-condition-p-o(n)} hold, then we have that
    \begin{align*}
      \htaudb-\bar{\tau} &=  n^{-1}\sum_i ({T}_i-r_1) \left(\frac{e_i(1)}{r_1}+ \frac{e_i(0)}{r_0}+ \frac{s_i(1)}{r_1}+ \frac{s_i(0)}{r_0}\right)  \\
      & -  n^{-1}\sum_{[i,j]} \left({T}_i-r_1\right)\left({T}_j-r_1\right) \left(\frac{{A}_{ij}(1)}{r_1^2}-\frac{{A}_{ij}(0)}{r_0^2}\right) + \op(n^{-1/2}).
    \end{align*}
  \end{proposition}
\begin{proof}
For ease of presentation, we write $\tilde{Z}_i = Z_i -r_1$ and $\tilde{T}_i = T_i - r_1$.
By \Cref{proposition:decompose-hat-taudb}, it remains to show that
\begin{align*}
    M_1 - M_2&:=n^{-1}\sum_i (\tilde{T}_i - \tilde{Z}_i)\left(\frac{e_i(1)}{r_1}+ \frac{e_i(0)}{r_0}+ \frac{s_i(1)}{r_1}+ \frac{s_i(0)}{r_0}\right) \\
      &-  n^{-1}\sum_{[i,j]} (\tilde{T}_i\tilde{T}_j -\tilde{Z}_i\tilde{Z}_j)\left(\frac{{A}_{ij}(1)}{r_1^2}-\frac{{A}_{ij}(0)}{r_0^2}\right) = \op(n^{-1/2}).
\end{align*}

For the term $M_2$, using $\sum_{i,j } A_{ij}(z) = 0$, we obtain the decomposition
  \begin{align*}
    &\frac{1}{n}\sum_{[i,j]} (\tilde{T}_i\tilde{T}_j-\tilde{Z}_i\tilde{Z}_j) \left(\frac{{A}_{ij}(1)}{r_1^2}-\frac{{A}_{ij}(0)}{r_0^2}\right) \\
    & =
    \frac{1}{n}\sum_{[i,j]} (\tilde{T}_i\tilde{T}_j-\tilde{Z}_i\tilde{Z}_j) \left(-\frac{\sum_k A_{kk}(1)}{r_1^2 n(n-1)}+\frac{\sum_k A_{kk}(0)}{r_0^2 n(n-1)}\right) \\
    &\quad + \frac{1}{n}\sum_{[i,j]} (\tilde{T}_i\tilde{T}_j-\tilde{Z}_i\tilde{Z}_j) \left(\frac{\tilde{A}_{ij}(1)}{r_1^2}- \frac{\tilde{A}_{ij}(0)}{r_0^2} \right)   =: M_{21}+M_{22}.
  \end{align*}
  For ${M}_{21}$,
  as shown in the proof of Proposition~\ref{proposition:asymptotic-equivalance-htauadj}, we have 
  \[
  \sum_{[i,j]} \tilde{Z}_i\tilde{Z}_j = O(n).
  \]
  This, together with Lemma~\ref{lem:order-of-sumi-hii_yi}, yields that
  \[
  \sum_{[i,j]} \tilde{Z}_i\tilde{Z}_j\frac{\sum_k A_{kk}(z)}{n(n-1)}= O(1).
  \]
  Moreover, we see that for $i\ne j\ne k\ne l$,
  \[
  \E \tilde{T}_i^2\tilde{T}_j\tilde{T}_k = 0,\quad  \E \tilde{T}_i\tilde{T}_j\tilde{T}_k \tilde{T}_l = 0,
  \]
  which implies that
  \begin{align*}
\E\bigg(\sum_{[i,j]}\tilde{T}_i\tilde{T}_j\frac{\sum_k A_{kk}(z)}{n(n-1)} \bigg)^2 &= 2\sum_{[i,j]}\E(\tilde{T}_i \tilde{T}_j)^2 \left(\frac{\sum_k A_{kk}(z)}{n(n-1)}\right)^2 \\
&= 2(r_1r_0)^2 n(n-1)\left(\frac{\sum_k A_{kk}(z)}{n(n-1)}\right)^2 = O(1),   
\end{align*}
    where the last inequality follows from Lemma~\ref{lem:order-of-sumi-hii_yi}.
    Then, by Chebyshev's inequality, we have  
      \[
  \sum_{[i,j]}\tilde{T}_i\tilde{T}_j\frac{\sum_k A_{kk}(z)}{n(n-1)}  = \Op(1).
  \]
  Combining these results, we obtain $ M_{21} = \Op(n^{-1})$.

 We now focus on $M_{22}$. We first expand it as
 \begin{align*}
&M_{22} = \frac{1}{n}\sum_{[i,j]}  ({T}_i{T}_j-{Z}_i{Z}_j-r_1T_i+r_1 Z_i -r_1T_j+r_1Z_j) \left(\frac{\tilde{A}_{ij}(1)}{r_1^2}-\frac{\tilde{A}_{ij}(0)}{r_0^2} \right)\\
  &=\frac{1}{n}\sum_{[i,j]}  ({T}_i{T}_j-{Z}_i{Z}_j) \left(\frac{\tilde{A}_{ij}(1)}{r_1^2} -\frac{\tilde{A}_{ij}(0)}{r_0^2}\right) + \frac{1}{n}\sum_{[i,j]} r_1(Z_i-T_i)\left(\frac{\tilde{A}_{ij}(1)}{r_1^2} -\frac{\tilde{A}_{ij}(0)}{r_0^2}\right) \\
  &\quad +\frac{1}{n}\sum_{[i,j]} r_1(Z_j-T_j)\left(\frac{\tilde{A}_{ij}(1)}{r_1^2} -\frac{\tilde{A}_{ij}(0)}{r_0^2}\right).
 \end{align*}
 Then, by Lemma~\ref{lem:equality-td-A}, we see that 
 \begin{align*}
   M_{22} & = \frac{1}{n}\sum_{[i,j]}  ({T}_i{T}_j-{Z}_i{Z}_j) \left(\frac{\tilde{A}_{ij}(1)}{r_1^2} -\frac{\tilde{A}_{ij}(0)}{r_0^2}\right) \\
  &+ \frac{1}{n}\sum_i r_1(Z_i-T_i)\left(\frac{d_i(1)}{r_1^2}-\frac{s_i(1)}{r_1^2}-\frac{d_i(0)}{r_0^2}+\frac{s_i(0)}{r_0^2}\right)  \\
  &+ \frac{1}{n}\sum_i r_1(Z_i-T_i)\left(-\frac{s_i(1)}{r_1^2}+\frac{s_i(0)}{r_0^2}\right).
 \end{align*}
 Applying Proposition~\ref{proposition:second-order-hajek-coupling} with $y_{i} = Y_i(z)-\bar{Y}(z)$, we get
 \[
  \sum_{[i,j]} ({T}_i{T}_j-{Z}_i{Z}_j) \left(\frac{\tilde{A}_{ij}(1)}{r_1^2}-\frac{\tilde{A}_{ij}(0)}{r_0^2}\right) = \op(n^{1/2}).
  \]
Then, applying Proposition~\ref{proposition:first-order-hajek-coupling} with $y_i = s_i(z)$ and $d_i(z)$, we get 
\begin{align}
\label{eq:hajek-coupling}
    \sum_i ({Z}_i-T_i) s_i(z) = \op(n^{1/2}),\quad \sum_i ({Z}_i-T_i) d_i(z) = \op(n^{1/2}),
\end{align}
which implies that $M_{22} = \op(n^{-1/2})$. Together with the bound on $M_{21}$, it implies that 
$$M_2 = \op(n^{-1/2}).$$

It remains to bound $M_1$. By Proposition~\ref{proposition:first-order-hajek-coupling} with $y_i= e_i(z)$, we have
\[
\sum_i ({Z}_i-T_i) e_i(z) = \op(n^{1/2}),
\]
which, combined with \eqref{eq:hajek-coupling}, implies that $M_1 = \op(n^{-1/2})$. This concludes the proof.
\end{proof}
\section{Asymptotic normality of $\htaudb$ when $p=o(n)$}
\label{sec:normality-p-o(n)}
In this section, we study the asymptotic normality of $\hat{\tau}_\db$ when $p=o(n)$ and give the proof of \Cref{theorem:CLT-p=o(n)}. We first state some lemmas that will be used in the proof.

The following Lemma demonstrates that for a sequence of population $\{y_i\}_{i\in [n]}$ satisfying a bounded second moment and a Lindeberg--Feller--type condition, the linear type statistics is asymptotically normal.
\begin{lemma}
    \label{lem:asymptotic-normality-linear-term}
    Let $\sigma_y^2 =  r_1 r_0 \sum_i (y_i-\bar{y})^2/n$.  If $\liminf_{n \to \infty} \sigma_y >0$ and $\max_i (y_i-\bar{y})^2 = o(n)$, then we have
\begin{align*}
    \frac{\sum_i {T}_i (y_i-\bar{y})}{\sqrt{n}\sigma_y} \asyeq \mathcal{N}(0,1).
\end{align*}
\end{lemma}
\begin{proof}[Proof of \Cref{lem:asymptotic-normality-linear-term}]
    By the Theorem~1 of \cite{berry1941accuracy}, we have 
    \[
   d_K\left( \frac{\sum_i {T}_i (y_i-\bar{y})}{\sqrt{n}\sigma_y},\mathcal{N}(0,1)\right) < C \max_i \left|\frac{(y_i-\bar{y})}{\sqrt{n}\sigma_y}\right|,
    \]
    where $d_K$ denotes the Kolmogorov distance between two distributions.
    If $\liminf_{n \to \infty} \sigma_y >0$ and $\max_i (y_i-\bar{y})^2 = o(n)$, then 
    \[
    \max_i \left|\frac{y_i-\bar{y}}{\sqrt{n}\sigma_y}\right| = o(1).
    \]
    The conclusion follows.
\end{proof}

{\rev
	\begin{lemma}\label{lem:lpsum}
	Let $y_1 \ge y_2 \ge \cdots \ge y_n \ge 0$ be a sequence of real values in descending order. We have for any positive interger $p < n$,
	\[
	\sup_{\bv \in \mathcal{V}_p} \sum_{i = 1}^n v_i y_i = \sum_{i = 1}^p y_i,
	\]
	where 
	\[
	\mathcal{V}_p := \left\{\bv := (v_1, \ldots, v_n) \in \mathbb{R}^n: \quad\sum_{i=1}^n v_i \le p \;\&\; 0 \le v_i \le 1, \forall 1 \le i \le n\right\}.
	\]
\end{lemma}
\begin{proof}
	It is equivalent to proving that the vector $\bv^\ast$ with $v_i^\ast = 1$ if $i \le p$ and $0$ otherwise maximimizes the objective function. For any $\bv \in \mathcal{V}_p$, apparently we have $v_i \le v_i^\ast$ if $i \le p$, and $v_i \ge v_i^\ast$ othewise; moreover,
	\[
	\sum_{i = 1}^p (v_i^\ast - v_i) \ge \sum_{i=p + 1}^n (v_i - v_i^\ast) \ge 0.
	\]
	In light of the above, and the ordering property of $y_i$'s, we have
	\[
	\sum_{i = 1}^n v_i^\ast y_i - \sum_{i = 1}^n v_i y_i = \sum_{i = 1}^p (v_i^\ast - v_i) y_i - \sum_{i = p+1}^n (v_i - v_i^\ast) y_i \ge \sum_{i = 1}^p (v_i^\ast - v_i) y_p - \sum_{i = p+1}^n (v_i - v_i^\ast) y_p \ge 0,
	\]
	which proves the desired result.
\end{proof}
}

\begin{proof}[Proof of Theorem~\ref{theorem:CLT-p=o(n)}]
First, we see that
\begin{align*}
    \var\left(\frac1{\sqrt{n}}\sum_i {T}_i \left(\frac{e_i(1)}{r_1}+ \frac{e_i(0)}{r_0}\right)\right) =\frac{1}{n} \sum_i \left(\frac{e_i(1)}{r_1}+ \frac{e_i(0)}{r_0}\right)^2 = \frac{n-1}{n} S^2_{r_1^{-1}e(1)+ r_0^{-1}e(0)}.
\end{align*}
By \Cref{lem:variance-crt-2-forms}, we have 
\[
S^2_{r_1^{-1}e(1)+ r_0^{-1}e(0)} = r_1^{-1} S_{e(1)}^2 + r_0^{-1} S_{e(0)}^2 - S_{\tau_e}^2 = \sigma_\adj^2.
\]
Putting together, we get 
\[
\var\left(\frac1{\sqrt{n}}\sum_i {T}_i \left(\frac{e_i(1)}{r_1}+ \frac{e_i(0)}{r_0}\right) \right) = \sigma_\adj^2 + o(1).
\]
Then, by Lemma~\ref{lem:asymptotic-normality-linear-term}, we see that under Assumptions~\ref{assumption:positive-limit-of-assigned-proportion}-\ref{assumption:Positive-variance-limit-p-o(n)},   
\[
\frac{1}{\sqrt{n}\sigma_{\adj}}\sum_i {T}_i \left(\frac{e_i(1)}{r_1}+ \frac{e_i(0)}{r_0}\right)      \asyeq \mathcal{N}(0,1).
\]
Recall that $\tilde{T}_i = T_i-r_1$.  By Proposition~\ref{prop:decompose-taudb-with-Ti}, it remains to show that for $z\in \{0,1\}$
    \[
    n^{-1/2}\sum_i {T}_i s_i(z) = \op(1), \quad n^{-1/2}\sum_{[i,j]} \tilde{T}_i \tilde{T}_j{A}_{ij}(z) = \op(1).
    \]
 
 We have
 {\rev
 	\begin{equation}\label{eq:tssqbnd}
 		\begin{aligned}
 			\E\bigg(n^{-1/2}\sum_i {T}_i s_i(z)\bigg)^2 & =  \frac{r_1r_0}{n} \sum_i s_i (z)^2 \le \frac{r_1 r_0}{n} \sum_i H_{i i}^2\left(Y_i(z)-\bar{Y}(z)\right)^2 \\
 			& \le \frac{r_1r_0}{n} \sup_{\bv \in \mathcal{V}_p} \sum_i v_i (Y_i(z)-\bar{Y}(z))^2.
 		\end{aligned}
 	\end{equation}
where recall the definition of $\mathcal{V}_p$ in \Cref{lem:lpsum}.
Apparently, the first inequality follows from the inequality that $\sum_i(a_i-\bar{a})^2 \leq \sum_i a_i^2 $; the second inequality follows from $\sum_i H_{ii}^2 \leq \sum_i H_{ii} = p$ and $H_{ii} < 1$.
	Since all the $(Y_i(z)-\bar{Y}(z))^2$'s are nonnegative, it follows from~\Cref{lem:lpsum} that
	\[
	\E\bigg(n^{-1/2}\sum_i {T}_i s_i(z)\bigg)^2 \le \frac{r_1r_0}{n} \sum_{i=1}^p (Y_{(i)} - \bar{Y}(z))^2 = O(n^{-1}) o(n) = o(1),
	\]
	 where the penultimate equation is due to Assumption \ref{assumption:p-o(n)-case}.

From above, it follows that} 
    \[
     n^{-1/2}\sum_i {T}_i s_i(z) = \op(1).
    \]   
    Next, using that for $i\ne j\ne k\ne l$,
  \[
  \E \tilde{T}_i^2\tilde{T}_j\tilde{T}_k = 0,\quad  \E \tilde{T}_i\tilde{T}_j\tilde{T}_k \tilde{T}_l = 0,
  \] we can derive that
    \begin{align*}
       &\E\bigg(n^{-1/2}\sum_{[i,j]} \tilde{T}_i\tilde{T}_j {A}_{ij}(z)\bigg)^2 =n^{-1}\sum_{[i,j]}\E(\tilde{T}_i\tilde{T}_j)^2\left({A}_{ij}(z)^2+{A}_{ij}(z){A}_{ji}(z) \right)\\
       &= \frac{(r_1r_0)^2}{n} \sum_{[i,j]} \left({A}_{ij}(z)^2+{A}_{ij}(z){A}_{ji}(z) \right)\leq   \frac{2(r_1r_0)^2}{n} \sum_{[i,j]} {A}_{ij}(z)^2\\
       &
       = \frac{2(r_1r_0)^2}{n}  \sum_{[i,j]} H_{ij}^2\left({Y}_j(z)-\bar{Y}(z)\right)^2 \le \frac{2(r_1r_0)^2}{n}  \sum_{i} H_{ii}\left({Y}_i(z)-\bar{Y}(z)\right)^2
      \\ &
       \leq \frac{2(r_1r_0)^2}{n}  \sum_{i=1}^p  (Y_{(i)}(z)-\bar{Y}(z))^2 = o(1).
    \end{align*}
    Thus, by Chebyshev's inequality, we have
    \[
    n^{-1/2}\sum_i \tilde{T}_i \tilde{T}_j {A}_{ij}(z) = \op(1).\]
    Then, the conclusion follows.
\end{proof}
\section{The CLT of quadratic forms and the asymptotic normality of $\htaudb$ when $p \asymp n$}
\label{sec:normality-p-O(n)}
In this section, we study the asymptotic normality of $\hat{\tau}_\db$ when $p\asymp n$ and give the proof of \Cref{theorem:CLT-alpha-greater-than-0}. The main intermediate step is to show that the Kolmogorov distance between the normal distribution and the joint distribution of the linear and quadratic terms of $\hat{\tau}_\db$ is negligible (see \Cref{prop:negligible-Kolmogorov-distance}). 

For a symmetric function vanishing on diagonals, define the influence of the $i$-th variable of $f$ by
$$
\operatorname{Inf}_i\left(f\right) :=\sum_{i_2, \ldots, i_q} f\left(i, i_2, \ldots, i_q\right)^2.
$$
if $q\geq 2$ and $\operatorname{Inf}_i\left(f\right):= f\left(i\right)^2$ if $q=1$.
Denote
\begin{align*}
  \|f\|_{\ell_2} :=\bigg\{\sum_{i_1, \ldots, i_q} f\left(i_1, \ldots, i_q\right)^2\bigg\}^{1/2},\quad \mathcal{M}(f):=\max _{1 \leq i \leq n} \operatorname{Inf}_i(f).
\end{align*}

For a random variable $U$, define its fourth-order cumulant $\kappa_4(U) :=\E U^4-3 (\E U^2)^2.$ Set $W_i$ in \Cref{def:Wi} in the main context as $W_i = (T_i-r_1)/(r_1r_0)^{1/2}$, $i\in [n]$. With the decomposition in Proposition~\ref{prop:decompose-taudb-with-Ti}, we define $\bs{Q} := (Q_1,Q_2)^\top$ as
\begin{align*}
    &Q_1 :=\frac{\sqrt{r_1r_0}}{\sqrt n}\sum_i W_i \left(\frac{e_i(1)}{r_1}+ \frac{e_i(0)}{r_0}+ \frac{s_i(1)}{r_1}+ \frac{s_i(0)}{r_0}\right), \\
    &Q_2 :=\frac{r_1r_0}{\sqrt n}\sum_{[i,j]} W_i W_j \left(-\frac{{A}_{ij}(1)}{r_1^2}+\frac{{A}_{ij}(0)}{r_0^2}\right),
\end{align*}
and it is easy to see that $\sqrt{n}(\hat{\tau}_\db - \bar{\tau}) = Q_1+Q_2+\op(1)$. Moreover, we rewrite $Q_2$ into the form of a homogeneous sum
\[
Q_2=\frac{ r_1r_0}{2\sqrt{n}}\sum_{[i,j]} W_i W_j\left(-\frac{{A}_{ij}(1)}{r_1^2}-\frac{{A}_{ji}(1)}{r_1^2}+\frac{{A}_{ij}(0)}{r_0^2}+\frac{{A}_{ji}(0)}{r_0^2}\right).
\]
$Q_1$ and $Q_2$ define two homogeneous sums
\[
Q_1 = \sum_i f_1(i) W_i,\quad Q_2 = \sum_{[i,j]} f_2(i,j) W_i W_j,
\]
with 
\begin{align*}
    &f_1(i) = \frac{\sqrt{r_1r_0}}{\sqrt n}\left(\frac{e_i(1)}{r_1}+ \frac{e_i(0)}{r_0}+ \frac{s_i(1)}{r_1}+ \frac{s_i(0)}{r_0}\right),\\
    &f_2(i,j) = \frac{ r_1r_0}{2\sqrt{n}} \left(-\frac{{A}_{ij}(1)}{r_1^2}-\frac{{A}_{ji}(1)}{r_1^2}+\frac{{A}_{ij}(0)}{r_0^2}+\frac{{A}_{ji}(0)}{r_0^2}\right).
\end{align*}
The following Proposition gives the variances of the linear and quadratic terms.
\begin{proposition}
    \label{prop:var-of-l-q-term}
    We have
    \begin{align*}
        \var(Q_1) = \frac{n-1}{n} \sigma_{\hd,l}^2,\quad  \var(Q_2) = \frac{n-1}{n}\sigma_{\hd,q}^2, \quad {\rev \cov(Q_1,Q_2) = 0}.
    \end{align*}
\end{proposition}
\begin{proof}[Proof of Proposition~\ref{prop:var-of-l-q-term}]
Using $\E W_i = 0$, $\E W_i^2 = 1$, and the independence between $W_i$'s, we obtain that
\begin{align*}
    &\var(Q_1) = \sum_i f_1(i)^2 \var(W_1) = \sum_i f_1(i)^2 \\
    &=  \frac{r_1r_0}{n}\sum_i \left(\frac{e_i(1)}{r_1}+ \frac{e_i(0)}{r_0}+ \frac{s_i(1)}{r_1}+ \frac{s_i(0)}{r_0}\right)^2 \\
    &= \frac{n-1}{n} (r_1r_0) S^2_{r_1^{-1}e(1)+ r_1^{-1}s(1) + r_0^{-1} e(0)+ r_0^{-1}s(0)}  = \frac{n-1}{n} \sigma_{\hd,l}^2,
\end{align*}
where the last equation follows from Lemma~\ref{lem:variance-crt-2-forms} with $a_i = e_i(1)+s_i(1)$, $b_i = e_i(0)+s_i(0)$.

Using $\sum_{j:j\ne i} H_{ij}^2 = H_{ii}-H_{ii}^2$ and recalling the definition of $\bs{Q}$, we obtain that
\begin{align*}
        &\var(Q_2) = \sum_{[i,j]} \frac{(r_1r_0)^2}{n}\left(\frac{{A}_{ij}(1)}{r_1^2}-\frac{{A}_{ij}(0)}{r_0^2}\right)^2 \E(W_i^2 W_j^2)  \\
        &\qquad + \sum_{[i,j]} \frac{(r_1r_0)^2}{n}\left(\frac{{A}_{ij}(1)}{r_1^2}-\frac{{A}_{ij}(0)}{r_0^2}\right)\left(\frac{{A}_{ji}(1)}{r_1^2}-\frac{{A}_{ji}(0)}{r_0^2}\right) \E(W_i^2 W_j^2)\\
    &= \frac{(r_1r_0)^2}{n} \left\{\sum_{[i,j]}H_{ij}^2\left( \frac{{Y}_i(1)}{r_1^2}-\frac{{Y}_i(0)}{r_0^2}\right)\left(\frac{{Y}_j(1)}{r_1^2}-\frac{{Y}_j(0)}{r_0^2}\right)+ \sum_{[i,j]}H_{ij}^2\left( \frac{{Y}_i(1)}{r_1^2}-\frac{{Y}_i(0)}{r_0^2}\right)\right\}\\
&= \frac{(r_1r_0)^2}{n} \left\{\sum_{[i,j]}H_{ij}^2\left( \frac{{Y}_i(1)}{r_1^2}-\frac{{Y}_i(0)}{r_0^2}\right)\left(\frac{{Y}_j(1)}{r_1^2}-\frac{{Y}_j(0)}{r_0^2}\right)+\sum_i (H_{ii}-H_{ii}^2)\left( \frac{{Y}_i(1)}{r_1^2}-\frac{{Y}_i(0)}{r_0^2}\right)^2 \right\}   \\
&= \frac{(r_1r_0)^2}{n}  \sum_i \sum_j Q_{ij}\left( \frac{{Y}_i(1)}{r_1^2}-\frac{{Y}_i(0)}{r_0^2}\right)\left(\frac{{Y}_j(1)}{r_1^2}-\frac{{Y}_j(0)}{r_0^2}\right)\\
&=\frac{(r_1r_0)^2}{n}  \sum_i \sum_j Q_{ij}\left( \frac{{Y}_i(1)-\bar{Y}(1)}{r_1^2}-\frac{{Y}_i(0)-\bar{Y}(0)}{r_0^2}\right)\left(\frac{{Y}_j(1)-\bar{Y}(1)}{r_1^2}-\frac{{Y}_j(0)-\bar{Y}(0)}{r_0^2}\right)\quad\\
& =\frac{n-1}{n}\sigma_{\hd,q}^2,
\end{align*}
where the penultimate equation is due to $\bs{Q}^\top\bs{1}=\bs{Q}\bs{1}=\bs{0}$.

{\rev Finally, we have
	\begin{align*}
		\cov(Q_1, Q_2) = & \frac{r_1^{3/2} r_0^{3/2}}{n} \sum_{[i,j,k]} W_i W_j W_k \left(\frac{e_i(1)}{r_1}+ \frac{e_i(0)}{r_0}+ \frac{s_i(1)}{r_1}+ \frac{s_i(0)}{r_0}\right) \left(-\frac{{A}_{jk}(1)}{r_1^2}+\frac{{A}_{jk}(0)}{r_0^2}\right) \\
		& + \frac{r_1^{3/2} r_0^{3/2}}{n} \sum_{[i,j]} W_i^2 W_j  \left(\frac{e_i(1)}{r_1}+ \frac{e_i(0)}{r_0}+ \frac{s_i(1)}{r_1}+ \frac{s_i(0)}{r_0}\right) \left(-\frac{{A}_{ij}(1)}{r_1^2}+\frac{{A}_{ij}(0)}{r_0^2}\right) \\
		& + \frac{r_1^{3/2} r_0^{3/2}}{n} \sum_{[i,j]} W_i W_j^2   \left(\frac{e_j(1)}{r_1}+ \frac{e_j(0)}{r_0}+ \frac{s_j(1)}{r_1}+ \frac{s_i(0)}{r_0}\right) \left(-\frac{{A}_{ij}(1)}{r_1^2}+\frac{{A}_{ij}(0)}{r_0^2}\right).
	\end{align*}
Therefore, $\cov(Q_1,Q_2) = 0$ follows from all the $W_i$'s are independent and $\E[W_i] = 0$.}
\end{proof}
Let $\bs{G} := (G_1,G_2)^\top$ be a normal approximation of $(Q_1,Q_2)$, i.e.
\[
  \begin{pmatrix}
      G_1\\
      G_2
  \end{pmatrix}   \sim \mathcal{N}\left(0,\begin{pmatrix}
      \sigma_{\hd,l}^2& 0\\
      0 & \sigma_{\hd,q}^2
  \end{pmatrix}\right).
\]
The following proposition shows the order of $ \sigma^2_{\hd,l}$ and $ \sigma^2_{\hd,q}$.
\begin{proposition}
\label{prop:order-of-sigma-l-q}
If Assumptions~\ref{assumption:positive-limit-of-assigned-proportion} and \ref{assumption:2-moment-of-finite-population-for-y} hold, then we have that
    \[
    \sigma^2_{\hd,l} = O(1),\quad \sigma^2_{\hd,q} = O(1).
    \]
\end{proposition}
\begin{proof}[Proof of Proposition~\ref{prop:order-of-sigma-l-q}]
Recall that $\bs{B}$ is defined as
\[
\bs{B} =\bs{M}^\top \bs{M},\quad \textrm{where}\quad \bs{M}=\left(\bs{I} - \frac{1}{n}\bs{1}\bs{1}^\top\right)  - \bs{H} + \left(\bs{I} - \frac{1}{n}\bs{1}\bs{1}^\top\right) \diag\{\bs{H}\}.
\]
and $\sigma_{\hd,l}^2 $ and $\sigma_{\hd,q}^2$ are defined as
\begin{align*}
\sigma_{\hd,l}^2 = (r_1r_0) S_{\bs{B}, \;r_1^{-1} Y(1) + r_0^{-1} Y(0)}^2,\quad \sigma_{\hd,q}^2 = (r_1r_0)^2 S_{\bs{Q}, \;r_1^{-2} Y(1) - r_0^{-2} Y(0)}^2.
\end{align*}
By sub-additivity and sub-multiplicativity of the $l_2$-norm and the trivial bounds
\[
\left\|\bs{I} - \frac{1}{n}\bs{1}\bs{1}^\top\right\|_2 \leq 1, \quad \|\bs{H}\|_2 \leq 1,
\]
we see that $\|\bs{B}\|_2 = O(1)$, which, combined with Assumptions~\ref{assumption:positive-limit-of-assigned-proportion} and \ref{assumption:2-moment-of-finite-population-for-y}, yields that
$
\sigma_{\hd,l}^2 = O(1).
$
For $\sigma_{\hd,q}^2$, we notice that
\[
\|\diag(\bs{Q})\|_2 = \max_{i \in [n]} (H_{ii}-H_{ii}^2) \leq \frac{1}{4} = O(1).
\]
Therefore, by Lemma~\ref{lem:norm-of-two-matrix}, we have 
\[
\|\bs{Q}\|_2 \leq \|\diag(\bs{Q})\|_2+\|\diag^-(\bs{Q})\|_2 = O(1),
\]
which, combined with Assumption \ref{assumption:positive-limit-of-assigned-proportion} and \ref{assumption:2-moment-of-finite-population-for-y}, yields that $
\sigma_{\hd,q}^2 = O(1).
$
This concludes the proof.
\end{proof}
The following proposition shows that the Kolmogorov distance between $(G_1,G_2)$ and $(Q_1,Q_2)$ is negligible.
\begin{proposition}
\label{prop:negligible-Kolmogorov-distance}
Under Assumptions \ref{assumption:positive-limit-of-assigned-proportion}--\ref{assumption:lindeberg-type-condition-p-o(n)} and \ref{assumption:Positive-variance-limit-alpha-positive}, we have that
    \begin{align*}
        \sup _{(x_1,x_2)^\top \in \mathbb{R}^2}|\mathbb{P}(Q_1\leq x_1; Q_2\leq x_2)-\mathbb{P}(G_1 \leq x_1; G_2\leq x_2)| = \delta_n \left(1+\frac{1}{\min\{\sigma_{\hd,q},\sigma_{\hd,l}\}}\right),
    \end{align*}
    for a deterministic parameter $\delta_n$ of order $o(1)$.
\end{proposition}
\begin{proof}[Proof of Proposition~\ref{prop:negligible-Kolmogorov-distance}]
For ease of presentation, we denote 
\begin{align*}
&  g_i := (r_1r_0)^{1/2}\left(\frac{e_i(1)}{r_1}+ \frac{e_i(0)}{r_0}+ \frac{s_i(1)}{r_1}+ \frac{s_i(0)}{r_0}\right),\\ 
&  y_i := (r_1r_0)\left(-\frac{Y_i(1)-\bar{Y}(1)}{r_1^2}+\frac{Y_i(0)-\bar{Y}(0)}{r_0^2}\right).  
\end{align*}
Therefore, we can rewrite that
\[
f_1(i) = n^{-1/2} g_i,\quad f_2(i,j) = (2\sqrt{n})^{-1}H_{ij}(y_i+y_j) .
\]
  By Theorem~2.1 of \cite{koike2022high}, we have
  \begin{align*}
  &\sup _{(x_1,x_2)^\top \in \mathbb{R}^2}|\mathbb{P}(Q_1\leq x_1; Q_2\leq x_2)-\mathbb{P}(G_1 \leq x_1; G_2\leq x_2)| \\
  &= C\left(\delta_0({\bs{Q}})^{\frac{1}{3}}+\delta_1(\bs{Q})^{\frac{1}{3}}+\max _{1 \leq k \leq 2}  \left(\mathcal{M}\left(f_k\right)\right)^{1/2}\right)\left(1+\frac{1}{\min\{\sigma_{\hd,q},\sigma_{\hd,l}\}}\right),
  \end{align*}
  where $C$ is a universal constant that does not depend on $n$ and
  \begin{align*}
    \delta_0(\bs{Q}) :=& \|\cov(\bs{Q})-\cov(\bs{G})\|_{\infty},\\
    \delta_1(\bs{Q}):= & \bigg(\left|\kappa_4\left(Q_1\right)\right|+ \sum_{i} \operatorname{Inf}_i\left(f_1\right)^2\bigg)^{1/2}+\bigg(\left|\kappa_4\left(Q_2\right)\right|+ \sum_{i} \operatorname{Inf}_i\left(f_2\right)^2\bigg)^{1/2}\\
    &   + \left\|f_1\right\|_{\ell_2}\bigg(\left|\kappa_4\left(Q_2\right)\right|+ \sum_{i} \operatorname{Inf}_i\left(f_2\right)^2\bigg)^{1 / 4}.
    \end{align*}
    Now, we set 
    \[
   \delta_n =  C\left(\delta_0({\bs{Q}})^{\frac{1}{3}}+\delta_1(\bs{Q})^{\frac{1}{3}}+\max _{1 \leq k \leq 2}  \left(\mathcal{M}\left(f_k\right)\right)^{1/2}\right).
    \]
    To conclude the proof, it suffices to show that $\delta_0(\bs{Q})$, $\delta_1(\bs{Q})$, $\max _{1 \leq k \leq 2}\mathcal{M}\left(f_k\right)$ are all of order $o(1)$.
    
  {\rev By Proposition \ref{prop:var-of-l-q-term} and \ref{prop:order-of-sigma-l-q}}, we have that
   \begin{align*}
        &\cov(Q_1,Q_2)=\cov(G_1,G_2)=0,\\
        &\var(Q_1)-\sigma_{\hd,l}^2 =\frac{n-1}{n}\sigma_{\hd,l}^2-\sigma_{\hd,l}^2 = -\frac{1}{n}\sigma_{\hd,l}^2 = o(1),\\
        &\var(Q_2)-\sigma_{\hd,q}^2 =\frac{n-1}{n}\sigma_{\hd,q}^2-\sigma_{\hd,q}^2 = -\frac{1}{n}\sigma_{\hd,q}^2 = o(1).
   \end{align*}
  These estimates give that $\delta_0(\bs{Q})= o(1)$.  
   We then consider $\max_{k \in \{1, 2\}}\mathcal{M}\left(f_k\right)$. We have that
   \[
   \max_i |s_i(z)|< 2\max_i |H_{ii}(Y_i(z)-\bar{Y}(z))| < 2\max_i |Y_i(z)-\bar{Y}(z)| = o(n^{1/2}),
   \]
   which, combined with Assumption~\ref{assumption:positive-limit-of-assigned-proportion} and \ref{assumption:lindeberg-type-condition-p-o(n)}, yields that $\max_i g_i^2 = o(n)$ and $\max_i y_i^2 = o(n)$.
   As a consequence, we obtain that
    \begin{align*}
      \mathcal{M}(f_1) &= \max_{i \in [n]}\operatorname{Inf}_i(f_1)=\max_{i \in [n]} f_1(i)^2 = \max_i g_i^2 /n = o(1),\\
        \mathcal{M}(f_2) &= \max_{i \in [n]}\operatorname{Inf}_i(f_2) = \max_i\sum_{j} f_2(i,j)^2(1-\delta_{ij}) = \max_i\sum_j H_{ij}^2(y_i+y_j)^2(1-\delta_{ij})/(4n)\\
        &\leq \max_i y_i^2\sum_j H_{ij}^2(1-\delta_{ij})/n = O\Big(\max_i y_i^2 /n\Big) = o(1).
    \end{align*}
    
    Finally, we estimate $\delta_1(\bs{Q})$. We first focus on terms relating to $Q_1$ and $f_1$. We see that 
    \begin{align*}
      &\sum_i\operatorname{Inf}_i(f_1)^2 = \sum_i f_1(i)^4 = \sum_i g_i^4/n^2 = \max_i g_i^2 / n \cdot \sum_{i} g_i^2 / n = o\Big(\sum_{i} g_i^2 / n\Big).
  \end{align*}
  Then, using 
  \[
  \frac{1}{n} \sum_{i} g_i^2 = O\bigg(\max_z \frac{1}{n} \sum_{i} (e_i(z)^2 + s_i^2(z)) \bigg) = O\bigg(\max_z \frac{1}{n} \sum_{i} (Y_i - \bar{Y}(z))^2\bigg) = O(1),
  \]
  we get that $\sum_i\operatorname{Inf}_i(f_1)^2 = o(1)$ and 
  \begin{align}\label{eq:fll2}
      &\|f_1\|_{\ell_2}=\bigg\{\sum_{i} f_1\left(i\right)^2\bigg\}^{1/2} = \bigg(\sum_i g_i^2/n\bigg)^{1/2} = O(1).
  \end{align}
    On the other hand, for $\kappa_4\left(Q_1\right)$,  we have that 
    \begin{align*}
    &\E Q_1^4 = \sum_i f_1(i)^4 \E W_i^4 + 3\sum_{i\ne j} f_1(i)^2 f_1(j)^2 \E \left(W_i^2 W_j^2\right),\\
        & 3\left(\E Q_1^2\right)^2 = 3\bigg(\sum_i f_1(i)^2\E W_i^2\bigg)^2 = 3 \sum_i f_1(i)^4 \left(\E W_i^2\right)^2 + 3 \sum_{i\ne j} f_1(i)^2 f_1(j)^2  \E W_i^2\cdot  \E W_j^2 ,
    \end{align*}
  which yield that 
  \begin{align*}
       \kappa_4\left(Q_1\right) &= \sum_i f_1(i)^4 \left[\E W_i^4-3\left(\E W_i^2\right)^2\right] =\kappa_4\left(W_1\right) \sum_i f_1(i)^4 
       = \kappa_4\left(W_1\right) \sum_i g_i^4/n^2\\
       &\le \kappa_4\left(W_1\right)\left(\max_i g_i^2\right) \sum_i g_i^2/n^2 = o(1).
  \end{align*}

  Next, we focus on terms relating to $Q_2$ and $f_2$. For $\operatorname{Inf}_i(f_2)^2$, using that $(y_i + y_j)^2 \leq 2(y_i^2 + y_j^2)$, we get 
  \begin{align*}
    \sum_i \operatorname{Inf}_i(f_2)^2  = O\bigg(\sum_i\bigg\{\sum_j H_{ij}^2y_i^2(1-\delta_{ij})/n\bigg\}^2+\sum_i\bigg\{\sum_j H_{ij}^2y_j^2(1-\delta_{ij})/n\bigg\}^2 \bigg).
  \end{align*}
  Expanding the above two terms, we get
  \begin{align*}
    \sum_i\bigg\{\sum_j H_{ij}^2 y_j^2(1-\delta_{ij})/n\bigg\}^2 &= \sumtwo \h{1}{2}^4 \yi{1}^4/n^2+\sumi{1}{3} \h{1}{2}^2\h{2}{3}^2\yi{1}^2\yi{3}^2/n^2 =: M_{11}+ M_{12},\\
    \sum_i\bigg\{\sum_j H_{ij}^2 y_i^2(1-\delta_{ij})/n\bigg\}^2 &= \sumtwo \h{1}{2}^4 \yi{2}^4/n^2+\sumi{1}{3} \h{1}{2}^2\h{2}{3}^2\yi{2}^4/n^2 =: M_{13}+M_{14}.
  \end{align*}
First, we use Lemma~\ref{lem:bound-of-the-trace} and Lemma~\ref{lem:norm-of-two-matrix} to get that 
\begin{align*}
      M_{11} =M_{13} &= \sumtwo \h{1}{2}^4 \yi{1}^4/n^2 =\tr\left(\operatorname{diag}(\bs{y})^2 \diag^-(\bs{Q})\diag^-(\bs{Q})\operatorname{diag}(\bs{y})^2\right)/n^2  \\
      &\leq\sum_{i}y_i^4/n^2 <\left(\max_i y_i^2\right)\sum_{i}y_i^2/n^2 =o(1).
\end{align*}
For $M_{12}$, by repeatedly applying $\sum_{j: j \neq i}H_{ij}^2 \le H_{ii} \le 1$, we get
\begin{align*}
    M_{12} = & \sumi{1}{3} \h{1}{2}^2\h{2}{3}^2\yi{1}^2\yi{3}^2/n^2 \leq \left(\max_i y_i^2 \right)\sumi{1}{3} \h{1}{2}^2\h{2}{3}^2\yi{1}^2/n^2 \\
    &\leq \left(\max_i y_i^2 \right)\sumtwo \h{1}{2}^2\yi{1}^2/n^2 
\leq \left(\max_i y_i^2 \right)\sum_i y_i^2/n^2 =o(1),
\end{align*}
With a similar argument, we can bound $M_{14}$ as
\begin{align*}
      M_{14}=  \sumi{1}{3} \h{1}{2}^2\h{2}{3}^2\yi{2}^4/n^2 &\leq  \sum_i y_i^4/n^2 = o(1).
\end{align*}
Putting together, we see that
\[
\sum_i \operatorname{Inf}_i(f_2)^2 = o(1).
\]

We next show that $\kappa_4\left(Q_2\right) = o(1)$. For ease of presentation, we abbreviate $f_2(i,j)$ as $f_{ij}$. Since $f_{ij}=f_{ji}$, we see from some basic combinatorics that
\begin{align*}
     \E Q_2^4 &= \E\bigg(\sumtwo \fii{1}{2} W_{i_1} W_{i_2}\bigg)^4 \\
    &= C_1\sumtwo \fii{1}{2}^4 \E W_{i_1}^4 W_{i_2}^4 +  C_2\sumi{1}{3}\fii{1}{2}^2\fii{2}{3}\fii{3}{1} \E W_{i_1}^3 W_{i_2}^3 W_{i_3}^2 \\
   &\quad + C_3\sumi{1}{3}\fii{1}{2}^2\fii{2}{3}^2 \E W_{i_1}^2 W_{i_2}^4 W_{i_3}^2 +  C_4\sumi{1}{4} \fii{1}{2}\fii{2}{3}\fii{3}{4}\fii{4}{1}  \E W_{i_1}^2 W_{i_2}^2 W_{i_3}^2 W_{i_4}^2 \\
  & \quad +  12 \sumi{1}{4} \fii{1}{2}^2 \fii{3}{4}^2  \E W_{i_1}^2 W_{i_2}^2 W_{i_3}^2 W_{i_4}^2,
\end{align*}
where $C_1, \ldots, C_4$ are universal constant that do not depend on $n$. 
On the other hand, we can calculate that
\begin{align*}
    &3\left(\E Q_2^2\right)^2 = 3\bigg\{\E \bigg(\sumtwo \fii{1}{2} W_{i_1}W_{i_2} \bigg)^2\bigg\}^2 = 3\bigg(2\sumtwo \fii{1}{2}^2 \E W_{i_1}^2W_{i_2}^2\bigg)^2\\
    =&  C_6\sumtwo \fii{1}{2}^4 \left(\E W_{i_1}^2 W_{i_2}^2\right)^2 + C_7\sumi{1}{3}\fii{1}{2}^2\fii{2}{3}^2 \left(\E W_{i_1}^2 W_{i_2}^2\right)^2 + 12\sumi{1}{4} \fii{1}{2}^2 \fii{3}{4}^2  \left(\E W_{i_1}^2 W_{i_2}^2\right)^2,
\end{align*}
where $C_6$ and $C_7$ are universal constants that do not depend on $n$. Using the cancellation of the term $\sumi{1}{4} \fii{1}{2}^2 \fii{3}{4}^2$, we obtain that 
\begin{align*}
  \kappa_4\left(Q_2\right)=O\left(|M_{21}|+|M_{22}|+|M_{23}|+|M_{24}|\right),
\end{align*}
where
\begin{align*}
  & M_{21} = \sumtwo \fii{1}{2}^4,\quad M_{22} = \sumi{1}{3}\fii{1}{2}^2\fii{2}{3}\fii{3}{1},\\ & M_{23} = \sumi{1}{3}\fii{1}{2}^2\fii{2}{3}^2,\quad M_{24} = \sumi{1}{4} \fii{1}{2}\fii{2}{3}\fii{3}{4}\fii{4}{1}.
\end{align*}
We handle these terms one by one. 

The term $M_{21}$ can be written as a summation of terms of the form
\begin{align*}
  \sumtwo \h{1}{2}^4 \yi{1}^{m_1} \yi{2}^{m_2}/n^2, \quad (m_1,m_2)\in \mathbb N^2, \ m_1+m_2=4.
\end{align*}
(We adopt the convention that $0\in \mathbb{N}$.) 
By Lemmas~\ref{lem:bound-of-the-trace} and \ref{lem:norm-of-two-matrix}, for any $(m_1,m_2)\in \mathbb N^2$ with $m_1 + m_2 = 4$,
\begin{align}
\label{eq:eq1}
  \sumtwo \h{1}{2}^4 \yi{1}^{m_1} \yi{2}^{m_2}/n^2 &= \operatorname{tr}\left(\operatorname{diag}(\bs{y})^{m_1}\diag^-(\bs{Q})\operatorname{diag}(\bs{y})^{m_2}\diag^-(\bs{Q})\right)/n^2\\
  &\leq \sum_{i} y_i^4/n^2 = o(1),\nonumber
\end{align}
which implies that $|M_{21}|=o(1)$.

By Cauchy-Schwarz inequality, we have that $|M_{22}|\leq M_{23}$. The term $M_{23}$ can be written as a summation of terms of the form
\begin{align*}
  \sumi{1}{3} \h{1}{2}^2\h{2}{3}^2 \yi{1}^{m_1} \yi{2}^{m_2}\yi{3}^{m_3}/n^2, \quad (m_1,m_2,m_3)\in \mathbb N^3, \ m_1+m_2+m_3=4 .
\end{align*}
We can find $(m_1^{(1)},m_2^{(1)},m_3^{(1)})\in \mathbb N^3$ and $(m_1^{(2)},m_2^{(2)},m_3^{(2)})\in \mathbb N^3$ such that $m_i=m_i^{(1)}+m_i^{(2)}$, $i=1,2,3,$ and 
\[
m_1^{(1)}+m_2^{(1)}+m_3^{(1)} = 2,\quad m_1^{(2)}+m_2^{(2)}+m_3^{(2)} = 2.
\]
Then, applying the Cauchy-Schwartz inequality, we can obtain that
\begin{align*}
 &\bigg|\sumi{1}{3} \h{1}{2}^2\h{2}{3}^2 \yi{1}^{m_1} \yi{2}^{m_2}\yi{3}^{m_3}\bigg| \leq \\
 &\bigg(\sumi{1}{3} \h{1}{2}^2\h{2}{3}^2 \yi{1}^{2m_1^{(1)}} \yi{2}^{2m_2^{(1)}}\yi{3}^{2m_3^{(1)}}\bigg)^{1/2}\bigg(\sumi{1}{3} \h{1}{2}^2\h{2}{3}^2 \yi{1}^{2m_1^{(2)}} \yi{2}^{2m_2^{(2)}}\yi{3}^{2m_3^{(2)}}\bigg)^{1/2}. 
\end{align*}
Mimicking the above proof for $M_{12}$, by repeatedly applying $\sum_{j: j \neq i}H_{ij}^2 \le H_{ii} \le 1$ and the condition $\max_i y_i^2 = o(n)$, we get that
\begin{align*}
   \frac1{n^2}\sumi{1}{3} \h{1}{2}^2\h{2}{3}^2 \yi{1}^{2m_1} \yi{2}^{2m_2}\yi{3}^{2m_3} = o(1),\ \ \forall (m_1,m_2,m_3)\in \mathbb{N}^3, m_1+m_2+m_3=2  .
\end{align*}
As a consequence, it implies that  
\begin{align}
\label{eq:eq2}
    \frac1{n^2} \sumi{1}{3} \h{1}{2}^2\h{2}{3}^2 \yi{1}^{m_1} \yi{2}^{m_2}\yi{3}^{m_3} = o(1),\  \forall (m_1,m_2,m_3)\in \mathbb{N}^3, m_1+m_2+m_3=4, 
\end{align}
so we have $|M_{23}|=o(1)$.

Finally, $M_{24}$ can be written as a summation of terms of the form
\begin{align}
  \sumi{1}{4} \h{1}{2}\h{2}{3}\h{3}{4}\h{4}{1} \yi{1}^{m_1} \yi{2}^{m_2}\yi{3}^{m_3}\yi{4}^{m_4}/n^2, 
\end{align}
for $(m_1,m_2,m_3,m_4)\in \mathbb{N}^4$ with $m_1+m_2+m_3+m_4=4$. 
To bound this term, we estimate an intermediate quantity $$\sum_{i_1\ne i_2, i_2\ne i_3, i_3\ne i_4, i_4\ne i_1} \h{1}{2}\h{2}{3}\h{3}{4}\h{4}{1} \yi{1}^{m_1} \yi{2}^{m_2}\yi{3}^{m_3}\yi{4}^{m_4}/n^2,$$ and define 
\[
\Delta:= \bigg(\sum_{i_1\ne i_2, i_2\ne i_3, i_3\ne i_4, i_4\ne i_1}-\sumi{1}{4}\bigg) \h{1}{2}\h{2}{3}\h{3}{4}\h{4}{1} \yi{1}^{m_1} \yi{2}^{m_2}\yi{3}^{m_3}\yi{4}^{m_4}/n^2.
\]
We observe that
\begin{align*}
    &\operatorname{tr}\left[\prod_{v=1}^4 \left(\operatorname{diag}(\bs{y})^{m_v}\diag^-(\bs{H})\right)\right]/n^2 \\
    &= \sum_{i_1\ne i_2, i_2\ne i_3, i_3\ne i_4, i_4\ne i_1} \h{1}{2}\h{2}{3}\h{3}{4}\h{4}{1} \yi{1}^{m_1} \yi{2}^{m_2}\yi{3}^{m_3}\yi{4}^{m_4}/n^2,
\end{align*}
and $\Delta$ can be written as a summation of terms of the forms
\[
 \sumtwo \h{1}{2}^4 \yi{1}^{m_1^\prime} \yi{2}^{m_2^\prime}/n^2, \quad (m_1^\prime,m_2^\prime)\in \mathbb N^2,\ m_1^\prime+m_2^\prime=4 ,
\]
and
\[
  \sumi{1}{3} \h{1}{2}^2\h{2}{3}^2 \yi{1}^{m_1^\prime} \yi{2}^{m_2^\prime}\yi{3}^{m_3^\prime}/n^2, \quad (m_1^\prime,m_2^\prime,m_3^\prime)\in \mathbb N^3,\ m_1^\prime+m_2^\prime+m_3^\prime=4 .
\]
By Lemmas~\ref{lem:bound-of-the-trace} and \ref{lem:norm-of-two-matrix}, we have
\begin{align*}
  \operatorname{tr}\left[\prod_{v=1}^4 \left(\operatorname{diag}(\bs{y})^{m_v}\diag^-(\bs{H})\right)\right]/n^2 = O\bigg(\sum_{i} y_i^4/n^2\bigg) = o(1).
\end{align*}
On the other hand, by \eqref{eq:eq1} and \eqref{eq:eq2}, we have $\Delta = o(1)$. Putting together, we have that for all $ (m_1,m_2,m_3,m_4)\in \mathbb{N}^4$ with $m_1+m_2+m_3+m_4=4$,
\begin{align}
\label{eq:eq3}
     \sumi{1}{4} \h{1}{2}\h{2}{3}\h{3}{4}\h{4}{1} \yi{1}^{m_1} \yi{2}^{m_2}\yi{3}^{m_3}\yi{4}^{m_4}/n^2=o(1),
\end{align}
which yields that $|M_{24}|= o(1)$. 
Combining all the above estimates, we conclude that $\kappa_4\left(Q_2\right) = o(1).$

In light of \eqref{eq:fll2} and our bounds on $\kappa_4(Q_k)$ and $\sum_i \operatorname{Inf}_i(f_k)^2$, $k \in \{1, 2\}$, there is $\delta_1(\bs{Q}) = o(1)$, which concludes the proof.
\end{proof}

The following proposition shows the marginal convergence of $Q_1$ and $Q_2$.
\begin{proposition}
    \label{prop:marginal-weak-converge}
    Assume Assumptions~\ref{assumption:positive-limit-of-assigned-proportion}--\ref{assumption:lindeberg-type-condition-p-o(n)} holds. We have that
    \begin{itemize}
        \item[(i)] if $\liminf \sigma_{\hd,l}^2>0$, then $Q_1/\sigma_{\hd,l}\stackrel{\textnormal{d}}{\to} \mathcal{N}(0,1)$;
        \item[(ii)] if $\liminf \sigma_{\hd,q}^2>0$, then $Q_2/\sigma_{\hd,q}\stackrel{\textnormal{d}}{\to} \mathcal{N}(0,1)$.
    \end{itemize}
\end{proposition}
\begin{proof}[Proof of Proposition~\ref{prop:marginal-weak-converge}]
    By Theorem~1 of \cite{de1990central}, we have 
    $$Q_k/\var(Q_k)^{1/2} \stackrel{\operatorname{d}}{
    \to} \mathcal{N}(0,1),\quad k\in\{1,2\}, $$ 
    provided that the following two conditions hold: (i) $\E Q_k^4/(\E Q_k^2)^2 \to 3 $; (ii) $\mathcal{M}(f_k)/\var(Q_k)  \rightarrow 0.$ Moreover, by Proposition~\ref{prop:var-of-l-q-term}, we have 
    \[
    \var(Q_1)  = \sigma_{\hd,l}^2+o(1),\quad \var(Q_2)  = \sigma_{\hd,q}^2 +o(1).
    \]

    From the proof of Proposition~\ref{prop:negligible-Kolmogorov-distance}, we have seen that under Assumptions~\ref{assumption:positive-limit-of-assigned-proportion}--\ref{assumption:lindeberg-type-condition-p-o(n)}, $\kappa_4(Q_k)\rightarrow 0$ and $\mathcal{M}(f_k)\rightarrow 0$ for $k\in\{1,2\}$.
    If $\liminf \sigma_{\hd,l}^2 >0$ and $k=1$, $\kappa_4(Q_1)\rightarrow 0$ and $\mathcal{M}(f_1)\rightarrow 0$ imply conditions (i) and (ii); if $\liminf \sigma_{\hd,q}^2 >0$ and $k=2$, then $\kappa_4(Q_2)\rightarrow 0$ and $\mathcal{M}(f_2)\rightarrow 0$ imply conditions (i) and (ii). Thus, we conclude the proof.
\end{proof}
\begin{proof}[Proof of Theorem~\ref{theorem:CLT-alpha-greater-than-0}]
    Without loss of generality, we assume $\liminf\sigma^2_{\hd,l} > 0$. We split the entire sequence into two subsequences. The first subsequence is such that all $\sigma_{\hd, q}$'s are larger than $\delta_n^{1/6}$, the second is such that all $\sigma_{\hd, q}$'s are smaller than $\delta_n^{1/6}$. 
    
    For the first subsequence, we have
     \[
     1+\frac{1}{\min\{\sigma_{\hd,q},\sigma_{\hd,l}\}} = O\left(1+\sigma_{\hd,q}^{-1}+\sigma_{\hd,l}^{-1}\right) =  O(1+1+\delta_n^{-1/6}) = O(\delta_n^{-1/6}).
     \]
     which yields that
         \begin{align*}
        \sup _{(x_1,x_2)^\top \in \mathbb{R}^2}|\mathbb{P}(Q_1\leq x_1; Q_2\leq x_2)-\mathbb{P}(G_1 \leq x_1; G_2\leq x_2)| \le O(\delta_n \delta_n^{-1/6}) = O(\delta_n^{5/6}).
    \end{align*}
We first show that given any $x\in \mathbb R$,
        \begin{align*}
        \mathbb{P}\left(\frac{Q_1+Q_2}{\sigma_{\hd}}\leq x\right) \le \mathbb{P}\left(\frac{G_1+G_2}{\sigma_{\hd}}\leq x\right) + o(1).
    \end{align*}
Let $\lceil\cdot\rceil$ and $\lfloor\cdot\rfloor$ be the ceiling and floor functions, respectively.  We decompose the left-hand side as
    \begin{align*}
    & \mathbb{P}\left(\frac{Q_1+Q_2}{\sigma_{\hd}}\leq x\right) \\
    &\le \sum_{t = \lfloor - \delta_n^{-2/3}\rfloor}^{\lceil\delta_n^{-2/3}\rceil} \mathbb{P}\left(\frac{Q_1+Q_2}{\sigma_{\hd}}\leq x, (t - 1) \cdot \delta_n^{1/2} \le Q_1 \le t \cdot \delta_n^{1/2}\right) + \Prob\left(|Q_1| \ge \delta_n^{-1/6}\right) \\
    & \le \sum_{t = \lfloor - \delta_n^{-2/3}\rfloor}^{\lceil\delta_n^{-2/3}\rceil} \mathbb{P}\left(Q_2\leq \sigma_{\hd} x - (t - 1) \cdot \delta_n^{1/2}, (t - 1) \cdot \delta_n^{1/2} \le Q_1 \le t \cdot \delta_n^{1/2}\right)   + \Prob\left(|Q_1| \ge \delta_n^{-1/6}\right) \\
    & \le \sum_{t = \lfloor - \delta_n^{-2/3}\rfloor}^{\lceil\delta_n^{-2/3}\rceil} \mathbb{P}\left(G_2\leq \sigma_{\hd} x - (t - 1) \cdot \delta_n^{1/2}, (t - 1) \cdot \delta_n^{1/2} \le G_1 \le t \cdot \delta_n^{1/2}\right) \\
    &\quad + \Prob\left(|Q_1| \ge \delta_n^{-1/6}\right) + O\left(\delta_n^{-2/3}\delta_n^{5/6}\right) \\
    & \le \sum_{t = \lfloor - \delta_n^{-2/3}\rfloor}^{\lceil\delta_n^{-2/3}\rceil} \mathbb{P}\left(\frac{G_1 + G_2}{\sigma_{\hd} } \leq x, (t - 1) \cdot \delta_n^{1/2} \le G_1 \le t \cdot \delta_n^{1/2}\right) \\
    & \quad + \sum_{t = \lfloor - \delta_n^{-2/3}\rfloor}^{\lceil\delta_n^{-2/3}\rceil} \mathbb{P}\left(\sigma_{\hd} x - t \cdot \delta_n^{1/2} \le G_2 \le \sigma_{\hd} x - (t - 1) \cdot \delta_n^{1/2}, (t - 1) \cdot \delta_n^{1/2} \le G_1 \le t \cdot \delta_n^{1/2}\right) \\
    & \quad + \Prob\left(|Q_1| \ge \delta_n^{-1/6}\right) + O\left(\delta_n^{1/6}\right).
    \end{align*}
   The second term on the right-hand side is of order 
    \[
    \delta_n^{5/6} \cdot \delta_n^{-2/3} = \delta_n^{1/6} = o(1),
    \]
    since, by $\liminf \sigma_{\hd,l} >0 $ and $\sigma_{\hd,q} > \delta_n^{1/6}$, we have that 
    \begin{align*}
        &\mathbb{P}\left(\sigma_{\hd} x - t \cdot \delta_n^{1/2} \le G_2 \le \sigma_{\hd} x - (t - 1) \cdot \delta_n^{1/2}, (t - 1) \cdot \delta_n^{1/2} \le G_1 \le t \cdot \delta_n^{1/2}\right) \\
        &= 
O\left(\frac{\delta_n^{1/2}}{\sigma_{\hd,q}}\frac{\delta_n^{1/2}}{\sigma_{\hd,l}} \right) = O\left(\delta_n^{5/6}\right).
    \end{align*}
    By ~\Cref{prop:var-of-l-q-term,prop:order-of-sigma-l-q}, we have 
    \[
    \var(Q_1)/\delta_{n}^{-1/3} = \frac{n-1}{n}\sigma_{\hd,l}^2\delta_n^{1/3} =O(\delta_n^{1/3}) ,
    \]
    which, by Chebyshev's inequality, implies that $\Prob(|Q_1| \ge \delta_n^{-1/6})$ and $\Prob(|G_1| \ge \delta_n^{-1/6})$ are both negligible. As a consequence, we have
    \begin{align*}
    \mathbb{P}\left(\frac{Q_1+Q_2}{\sigma_{\hd}}\leq x\right) & \le \sum_{t = \lfloor - \delta_n^{-2/3}\rfloor}^{\lceil\delta_n^{-2/3}\rceil} \mathbb{P}\left(\frac{G_1 + G_2}{\sigma_{\hd} } \leq x, (t - 1) \cdot \delta_n^{1/2} \le G_1 \le t \cdot \delta_n^{1/2}\right) + o(1) \\
    & \le \Prob\left(\frac{G_1 + G_2}{\sigma_{\hd} } \leq x\right) + \Prob(|G_1| \ge \delta_n^{-1/6}) + o(1)\\
    & \le \Prob\left(\frac{G_1 + G_2}{\sigma_{\hd} } \leq x\right) + o(1).
    \end{align*}
    For the lower bound, we apply similar arguments as above to get that 
    \begin{align*}
    & \mathbb{P}\left(\frac{Q_1+Q_2}{\sigma_{\hd}}\leq x\right) \ge \sum_{t = \lfloor - \delta_n^{-2/3}\rfloor}^{\lceil\delta_n^{-2/3}\rceil} \mathbb{P}\left(\frac{Q_1+Q_2}{\sigma_{\hd}}\leq x, (t - 1) \cdot \delta_n^{1/2} \le Q_1 \le t \cdot \delta_n^{1/2}\right) \\
    & \ge \sum_{t = \lfloor - \delta_n^{-2/3}\rfloor}^{\lceil\delta_n^{-2/3}\rceil} \mathbb{P}\left(Q_2\leq \sigma_{\hd} x -t \cdot \delta_n^{1/2}, (t - 1) \cdot \delta_n^{1/2} \le Q_1 \le t \cdot \delta_n^{1/2}\right)\\
    & \ge \sum_{t = \lfloor - \delta_n^{-2/3}\rfloor}^{\lceil\delta_n^{-2/3}\rceil} \mathbb{P}\left(G_2\leq \sigma_{\hd} x - t \cdot \delta_n^{1/2}, (t - 1) \cdot \delta_n^{1/2} \le G_1 \le t \cdot \delta_n^{1/2}\right) - O(\delta_n^{1/6}) \\
    & \ge \sum_{t = \lfloor - \delta_n^{-2/3}\rfloor}^{\lceil\delta_n^{-2/3}\rceil} \mathbb{P}\left(\frac{G_1 + G_2}{\sigma_{\hd} } \leq x, (t - 1) \cdot \delta_n^{1/2} \le G_1 \le t \cdot \delta_n^{1/2}\right) \\
    & \qquad - \sum_{t = \lfloor - \delta_n^{-2/3}\rfloor}^{\lceil\delta_n^{-2/3}\rceil} \mathbb{P}\left(\sigma_{\hd} x - t \cdot \delta_n^{1/2} \le G_2 \le \sigma_{\hd} x - (t - 1) \cdot \delta_n^{1/2}, (t - 1) \cdot \delta_n^{1/2} \le G_1 \le t \cdot \delta_n^{1/2}\right) \\
    & \qquad - O(\delta_n^{1/6}) \\
    & \ge \sum_{t = \lfloor - \delta_n^{-2/3}\rfloor}^{\lceil\delta_n^{-2/3}\rceil} \mathbb{P}\left(\frac{G_1 + G_2}{\sigma_{\hd} } \leq x, (t - 1) \cdot \delta_n^{1/2} \le G_1 \le t \cdot \delta_n^{1/2}\right) + o(1)\\
    & \ge \Prob\left(\frac{G_1 + G_2}{\sigma_{\hd} } \leq x\right) - \Prob(|G_1| \ge \delta_n^{-1/6}) + o(1) \ge \Prob\left(\frac{G_1 + G_2}{\sigma_{\hd} } \leq x\right) + o(1).
    \end{align*}
    Putting together the upper and lower bounds, we get that the first subsequence satisfies
    \[
    \sup_{x\in\mathbb{R}}\left| \mathbb{P}\left(\frac{Q_1+Q_2}{\sigma_{\hd}}\leq x\right)-\Prob\left(\frac{G_1 + G_2}{\sigma_{\hd} } \leq x\right)\right| = o(1).
    \]

    We now consider the second subsequence where $\sigma_{\hd, q}$'s are all smaller than $\delta_n^{1/6}$ which implies that $\sigma_\hd/\sigma_{\hd, l} = 1 + o(1)$ in this sequence. As a consequence, we have $Q_2 / \sigma_\hd = o_\Prob(1)$, which means that 
    \[
    \frac{Q_1 + Q_2}{\sigma_\hd} = \frac{Q_1}{\sigma_\hd} + o_\Prob(1).
    \]
By Proposition~\ref{prop:marginal-weak-converge},  we have $Q_1/\sigma_{\hd, l}  \stackrel{\operatorname{d}}{\to} \mathcal{N}(0,1)$. 
Together with Slutsky's theorem, it implies that $(Q_1 + Q_2)/{\sigma_\hd}\stackrel{\operatorname{d}}{\to}\mathcal{N}(0,1)$.

In sum, we have that for each $x\in \mathbb R$, 
\[
\left| \mathbb{P}\left(\frac{Q_1+Q_2}{\sigma_{\hd}}\leq x\right)-\Prob\left(\frac{G_1 + G_2}{\sigma_{\hd} } \leq x\right)\right| = o(1)
\]
for both subsequences, showing that this estimate indeed holds for the whole sequence. This gives that
\[
\frac{Q_1 + Q_2}{\sigma_\hd} \stackrel{\operatorname{d}}{\to} \mathcal{N}(0,1),
\]
and the conclusion then follows from Proposition~\ref{prop:decompose-taudb-with-Ti}.
\end{proof}
\section{Inference}
\label{sec:inference-supp}
In this section, we study the validity of the proposed inference procedure. It includes the proofs for \Cref{thm:inference} and \Cref{cor:inference}. The comment of \eqref{eq:rewrite-sigma-hd} follows from the following proposition.
\begin{proposition}
\label{prop:rewrite-sigma-hd-l}
We have 
\begin{align*}
       \sigma_{\hd,l}^2 = (r_1r_0) S_{\bs{B}, \;r_1^{-1} Y(1) + r_0^{-1} Y(0)}^2.
\end{align*}
\end{proposition}
\begin{proof}[Proof of \Cref{prop:rewrite-sigma-hd-l}]
    Using \Cref{lem:variance-crt-2-forms} with $a_i = e_i(1)+s_i(1)$ and $b_i = e_i(0)+s_i(0)$, we get
    \begin{align*}
        \sigma_{\hd,l}^2 = (r_1r_0) S^2_{r_1^{-1}e(1)+ r_1^{-1}s(1) + r_0^{-1} e(0)+ r_0^{-1}s(0)}.
    \end{align*}
    Denote $\bs{P}= \bs{I}-\frac{1}{n}\bs{1}\bs{1}^\top$.
    Observe that
    \begin{align*}
        (e_1(z),\ldots,e_n(z))^\top &= (\bs{I}-\bs{H})(Y_1(z)-\bar{Y}(z),\ldots,Y_n(z)-\bar{Y}(z))^\top\\
        &=(\bs{P}-\bs{H})(Y_1(z)-\bar{Y}(z),\ldots,Y_n(z)-\bar{Y}(z))^\top,
    \end{align*}
    and
    \begin{align*}
          (s_1(z),\ldots,s_n(z))^\top &= \bs{P}\diag(\bs{H})(Y_1(z)-\bar{Y}(z),\ldots,Y_n(z)-\bar{Y}(z))^\top.
    \end{align*}
   Then, applying Lemma~\ref{lem:S2a-S2AAb} with $a_i = r_1^{-1}e_i(1)+ r_1^{-1}s_i(1) + r_0^{-1} e_i(0)+ r_0^{-1}s_i(0)$, $b_i = r_1^{-1}Y_i(1)+r_0^{-1}Y_i(0)$ and  $\bs{M} = (\bs{P}-\bs{H})+\bs{P}\diag(\bs{H})$, and noticing that $\sum_i a_i = 0$, we obtain that 
   \begin{align*}
       S^2_{r_1^{-1}e(1)+ r_1^{-1}s(1) + r_0^{-1} e(0)+ r_0^{-1}s(0)} = S^2_{\bs{M}^\top \bs{M}, r_1^{-1}Y(1)+r_0^{-1}Y(0)}.
   \end{align*}
    The conclusion then follows by the definition of $\bs{B}$ in \eqref{def:BB}.
\end{proof}
 \Cref{thm:inference} and \Cref{cor:inference} follow from the following Lemmas \ref{lem:order-of-cov-zizj-zkzl}--\ref{lem:sample-covariance-consistency}.
\begin{lemma}
    \label{lem:order-of-cov-zizj-zkzl}
    Under Assumption~\ref{assumption:positive-limit-of-assigned-proportion}, we have
    \[
    \cov(Z_iZ_j,Z_kZ_l) = O(n^{-1}).
    \]
\end{lemma}
\begin{proof}[Proof of Lemma~\ref{lem:order-of-cov-zizj-zkzl}]
Observe that
    \begin{align*}
        \cov(Z_iZ_j,Z_kZ_l) = \E Z_iZ_jZ_kZ_l-\E(Z_iZ_j)\E (Z_kZ_l). 
    \end{align*} 
First, we have
 \begin{align*}
    &\E Z_iZ_jZ_kZ_l =  \Pr(Z_i=1,Z_j=1,Z_k=1, Z_l=1) \\
    &= \Pr( Z_l=1) \Pr(Z_k=1 | Z_l=1) \Pr(Z_j=1|Z_k=1, Z_l=1) \Pr(Z_i=1|Z_j=1,Z_k=1, Z_l=1) \\
    &=\frac{n_1}{n}\frac{n_1-1}{n-1}\frac{n_1-2}{n-2}\frac{n_1-3}{n-3}.
 \end{align*}   
 On the other hand, we have
 \[
 \E(Z_iZ_j)\E (Z_kZ_l) = \left(\frac{n_1}{n}\frac{n_1-1}{n-1}\right)^2.
 \]
In sum, under Assumption~\ref{assumption:positive-limit-of-assigned-proportion}, we have that
 \begin{align*}
       \cov(Z_iZ_j,Z_kZ_l) &= \frac{n_1}{n}\frac{n_1-1}{n-1}\frac{n_1-2}{n-2}\frac{n_1-3}{n-3} - \left(\frac{n_1}{n}\frac{n_1-1}{n-1}\right)^2\\
 &= \frac{n_1^4}{n(n-1)(n-2)(n-3)}-\frac{n_1^4}{n^2(n-1)^2} + O(n^{-1})\\
 &= n_1^4\left(\frac{1}{n(n-1)(n-2)(n-3)}-\frac{1}{n^2(n-1)^2}\right) + O(n^{-1})\\
 &= n_1^4 O(n^{-5}) + O(n^{-1}) = O(n^{-1}) .
 \end{align*}
\end{proof}
Let $g_{(1)}^2\geq \ldots \geq g_{(n)}^2$ and  $y_{(1)}^2\geq \ldots \geq y_{(n)}^2$ be the ordered sequence of $\{g_i\}_{i=1}^n$ and $\{y_i\}_{i=1}^n$, respectively.
\begin{lemma}
\label{lem:quadra-term-1}
Assume Assumption~\ref{assumption:positive-limit-of-assigned-proportion} holds and $\sum_i y_i^2 = O(n)$, $\sum_i g_i^2 = O(n)$, $\max_i g_i^2 = o(n)$, $\max_i y_i^2 = o(n)$. For any symmetric matrix $\bs{D}$ with diagonal entries being $0$ and $\|\bs{D}\|_2 < C$,  we have
\begin{align}
\label{eq:sample-analogy-consistency-zizj}
    \frac{1}{n}\sum_{[i,j]} D_{ij} y_i g_j Z_i Z_j = \frac{1}{n}\sum_{[i,j]} D_{ij} y_i g_j r_1^2 + \op(1), 
\end{align}
and
\begin{align}
\label{eq:sample-analogy-consistency-1-zi1-zj}
   \frac{1}{n}\sum_{[i,j]} D_{ij} y_i g_j (1-Z_i) (1-Z_j) = \frac{1}{n}\sum_{[i,j]} D_{ij} y_i g_j r_0^2 + \op(1).
\end{align}
\end{lemma}
\begin{proof}[Proof of Lemma ~\ref{lem:quadra-term-1}]
We only prove \eqref{eq:sample-analogy-consistency-zizj}, and \eqref{eq:sample-analogy-consistency-1-zi1-zj} follows immediately by replacing $Z_i$ with $1-Z_i$.

For $i\ne j$, using $\E Z_i Z_j = r_1 \frac{n_1 - 1}{n - 1}$, we get 
\[
\E \frac{1}{n}\sum_{[i,j]} D_{ij} y_i g_j Z_i Z_j = \frac{1}{n}\sum_{[i,j]} D_{ij} y_i g_j r_1 \frac{n_1-1}{n-1} = \left(1+O(n^{-1})\right)\frac{1}{n}\sum_{[i,j]} D_{ij} y_i g_j r_1^2.
\]
On the other hand, we have 
\begin{equation}\label{eq:gyd}
\bigg|\frac{1}{n}\sum_{[i,j]} D_{ij} y_i g_j\bigg| \leq \bigg(\sum_i y_i^2/n\bigg)^{1/2}\bigg(\sum_i g_i^2/n\bigg)^{1/2} \|\bs{D}\|_2 = O(1),
\end{equation}
which implies that
\[
\E \frac{1}{n}\sum_{[i,j]} D_{ij} y_i g_j Z_i Z_j = \frac{1}{n}\sum_{[i,j]} D_{ij} y_i g_j r_1^2 + o(1).
\]
Thus, to conclude the proof using Chebyshev's inequality, it suffices to show that
\begin{align}
\label{eq:bounded-variance-1}
    \var\bigg(\sum_{[i,j]} y_i g_jD_{ij}Z_iZ_j\bigg) = o(n^2).
\end{align}
Through direct calculation, we get that 
      \begin{align*}
            &\var\bigg(\sum_{[i,j]} y_i g_jD_{ij}Z_iZ_j\bigg)=\var(Z_1 Z_2) \sumtwo \left(C_1 \yi{1}^2 \gi{2}^2 \D{1}{2}^2+ C_2\yi{1}\yi{2} \gi{1}\gi{2} \D{1}{2}^2\right)  \\
            &+ \cov(Z_1 Z_2,Z_1 Z_3) \sumi{1}{3}\left(C_3\D{1}{2}\D{2}{3}\gi{2}^2\yi{1}\yi{3}+C_4\D{1}{2}\D{2}{3}\gi{2}\yi{2}\yi{1}\gi{3}+C_5\D{1}{2}\D{2}{3}\yi{2}^2\gi{1}\gi{3}\right)\\
            & +\cov(Z_1 Z_2,Z_3 Z_4) \sumi{1}{4} C_6\D{1}{2}\D{3}{4}\yi{1}\gi{2}\yi{3}\gi{4}\\
            &=: \var(Z_1 Z_2) \left(C_1 M_{1}+ C_2 M_{2}\right)+\cov(Z_1 Z_2,Z_1 Z_3) \left(C_3 M_{3} + C_4 M_{4} + C_5 M_{5}\right) \\
            &\qquad+ \cov(Z_1 Z_2,Z_3 Z_4)  C_6 M_{6},
    \end{align*}
    where, $C_i$, $i=1,\ldots, 6$ are universal constants that do not depend on $n$. 
    By Assumption~\ref{assumption:positive-limit-of-assigned-proportion} and Lemma~\ref{lem:order-of-cov-zizj-zkzl}, we have
    \[
    \var(Z_1 Z_2)=O(1),\quad \cov(Z_1 Z_2,Z_1 Z_3) = O(1),\quad  \cov(Z_1 Z_2,Z_3 Z_4) = O(n^{-1}).
    \]
  It remains to estimate the order of $M_i$, $i=1,\ldots,6$.   

By Lemma~\ref{lem:bound-of-the-trace}, we have
    \[
    M_{1} = \tr\left(\diag(\bs{y})^2 \bs{D} \diag(\bs{g})^2 \bs{D}\right)\leq C^2 \sum_i y_{(i)}^2 g_{(i)}^2  \leq C^2 \left(\max_i g_i^2\right) \sum_i y_{i}^2  = o(n^2).
    \]
Applying the Cauchy-Schwarz inequality, we also get $ | M_{2} | \leq M_{1} = o(n^2).$ 

For $M_{3}$, let $y_{D,i}$ and $g_{D,i}$ be the $i$-th element of $\bs{D}(y_1,\ldots,y_n)^\top$  and $\bs{D}(g_1,\ldots,g_n)^\top$. We see that $\sum_i y_{D,i}^2 = O(n)$ and  $\sum_i g_{D,i}^2 = O(n)$ since $\|D\|_2 < C$. 
   By repeatedly applying $    \sum_{j\in [n]\backslash i}D_{ij} y_{j} = y_{D,i}$ and $ \sum_{j\in [n]\backslash i}D_{ij} g_{j} = g_{D,i}$,
    we obtain that 
    \begin{align*}
        M_{3}&=\sumi{1}{3}\D{1}{2}\D{2}{3}\gi{2}^2\yi{1}\yi{3} = \sumtwo\D{1}{2}\gi{2}^2\yi{1}\yiD{2}- \sumtwo\D{1}{2}\D{2}{1}\gi{2}^2\yi{1}\yi{1}\\  &=\sum_{i_1}\gi{1}^2\yiD{1}^2- \sumtwo\D{1}{2}^2\gi{2}^2\yi{1}^2=:M_{31}-M_{32}.
    \end{align*}
For $M_{31}$, it holds that
    \[
    M_{31}\leq \left(\max_i g_i^2\right)  \sum_i y_{D,i}^2  = o(n^2).
    \]
For $M_{32}$, we see that $M_{32} = M_{1}=o(n^2)$. Hence, there is $M_3 = o(n^2)$.

For $M_4$ and $M_5$, we decompose them as 
\begin{align*}
  M_{4}&=\sumtwo\D{1}{2}\gi{2}\giD{2}\yi{1}\yi{2} - \sumtwo\D{1}{2}\D{2}{1}\gi{2}\gi{1}\yi{1}\yi{2} \\
  &=\sum_{i_1}\gi{1}\giD{1}\yiD{1}\yi{1} - \sumtwo\D{1}{2}^2\gi{2}\gi{1}\yi{1}\yi{2},\\
  M_{5}&=\sumtwo\D{1}{2}\yi{2}^2\gi{1}\giD{2} - \sumtwo \D{1}{2}\D{2}{1}\yi{2}^2\gi{1}\gi{1} \\
  &=\sum_{i_1}\yi{1}^2\giD{1}^2 -  \sumtwo \D{1}{2}^2\yi{2}^2\gi{1}^2.
\end{align*}
Using similar arguments as in the analysis of $M_3$, we can show that $M_5 = o(n^2)$. 
For $M_4$, by Cauchy-Schwarz inequality, we have
\begin{align*}
    &\bigg(\sum_{i_1}\gi{1}\giD{1}\yiD{1}\yi{1}\bigg)^2 \leq \bigg(\sum_{i_1}\gi{1}^2\yiD{1}^2\bigg)\bigg(\sum_{i_1}\giD{1}^2\yi{1}^2\bigg),\\
    &\bigg(\sumtwo\D{1}{2}^2\gi{2}\gi{1}\yi{1}\yi{2}\bigg)^2 \leq  \bigg(\sumtwo \D{1}{2}^2\yi{2}^2\gi{1}^2\bigg)^2,
\end{align*}
which, combined with the arguments in the analysis of $M_3$, yields that $M_4 = o(n^2)$.

Finally, for  $M_6$, we see that
\begin{align*}
    M_6 - \bigg(\sumtwo \D{1}{2}\yi{1}\gi{2}\bigg)^2 = O\left( |M_1| + |M_2| + |M_3| + |M_4| + |M_5|\right) = o(n^2),
\end{align*}
and we have shown in \eqref{eq:gyd} that
\[
\frac{1}{n}\sum_{[i,j]} D_{ij} y_i g_j=  O(1).
\]
Thus, we have $M_6 = O(n^2)$.

Putting together the above estimates,  we obtain that 
\begin{align*}
    \var\bigg(\sum_{[i,j]} y_i g_jD_{ij}Z_iZ_j\bigg) = O(1)o(n^2)+O(1)o(n^2)+O(n^{-1})O(n^2) = o(n^2),
\end{align*}
which concludes the proof.
\end{proof}

\begin{lemma}
\label{lem:linear-term-1}
Assume Assumption~\ref{assumption:positive-limit-of-assigned-proportion} holds and $\sum_i y_i^2 = O(n)$, $\sum_i g_i^2 = O(n)$, $\max_i g_i^2 = o(n)$. For any sequence $\{a_i\}_{i=1}^n$ with $\max_i |a_i| <C$, we have that
\[
\frac{1}{n}\sum_{i} a_i y_i g_i Z_i = \frac{1}{n}\sum_{i} a_i y_i g_i r_1 + \op(1).
\]
\end{lemma}
\begin{proof}[Proof of Lemma~\ref{lem:linear-term-1}]
       It suffices to show that
\begin{align}
\label{eq:bounded-variance-2}
    \var\bigg(\sum_{i} a_i y_i g_i Z_i\bigg) = o(n^2).
\end{align}
Direct calculations give that  
\begin{align*}
    \var\bigg(\sum_{i} a_i y_i g_i Z_i\bigg) &= \var(Z_1) \sum_i a_i^2 y_i^2 g_i^2 +\cov(Z_1,Z_2)\sum_{[i,j]} a_i a_j y_i y_j g_i g_j\\
    &=: \var(Z_1) M_1 + \cov(Z_1,Z_2) M_2.
\end{align*}
Under Assumption~\ref{assumption:positive-limit-of-assigned-proportion}, we have that
\[
\var(Z_1) = O(1),\quad  \cov(Z_1,Z_2) = O(n^{-1}).
\]
It remains to estimate the order of $M_1$ and $M_2$. 

For $M_1$, by $\max_i |a_i| < C$, $\max_i g_i^2 = o(n)$ and $\sum_i y_i^2 = O(n)$, we have
\[
\sum_i a_i^2 y_i^2 g_i^2 \leq C^2 \left(\max_i g_i^2\right) \sum_i y_i^2 = o(n^2).
\]
For $M_2$, we have that
\[
M_2 = \sum_{[i,j]} a_i a_j y_i y_j g_i g_j \le \bigg(\sum_i a_i y_i g_i\bigg)^2.
\]
By Cauchy-Schwarz inequality, there is
\[
\bigg(\sum_i a_i y_i g_i\bigg)^2 \leq \max_i a_i^2 \bigg(\sum_i y_i^2\bigg)\bigg(\sum_i g_i^2\bigg) =  O(n^2),
\]
which implies that $M_2 = O(n^2)$. Thus, we have 
\[
\var\bigg(\sum_{i} a_i y_i g_i Z_i\bigg) = O(1)o(n^2)+o(n^{-1})O(n^2) = o(n^2),
\]
which concludes the proof.
\end{proof}


\begin{lemma}
\label{lem:sample-covariance-consistency}
Assume Assumptions~\ref{assumption:positive-limit-of-assigned-proportion}-\ref{assumption:lindeberg-type-condition-p-o(n)} hold. For any symmetric matrix $\bs{D}$ with $\|\diag^-(\bs{D})\|_2 < C$ and $\|\diag(\bs{D})\|_2 < C$,   we have that for $z\in\{0,1\}$,
\[
s_{\diag^- (\bs{D}),Y(z)}^2= S_{\diag^-(\bs{D}),Y(z)}^2 + \op(1), \quad s_{\diag (\bs{D}),Y(z)}^2= S_{\diag(\bs{D}),Y(z)}^2 + \op(1),
\]
\[
s_{\diag^- (\bs{D}),Y(1),Y(0)}= S_{\diag^-(\bs{D}),Y(1),Y(0)} + \op(1).
\]
\end{lemma}
\begin{proof}[Proof of Lemma~\ref{lem:sample-covariance-consistency}]
    We can write that
    \begin{align*}
        s_{\diag^- (\bs{D}),Y(z)}^2 &= \frac{1}{n r_z^2}\sum_{i\ne j:Z_i=z,Z_j=z} D_{ij}(Y_i(z)-\bar{Y}_z)(Y_j(z)-\bar{Y}_z) \\
        &= M_1+M_2+M_3,
    \end{align*}
    where
    \begin{align*}
        M_1 &= \frac{1}{n r_z^2}\sum_{i\ne j:Z_i=z,Z_j=z} D_{ij}(Y_i(z)-\bar{Y}(z))(Y_j(z)-\bar{Y}(z)),\\
        M_2 &= 2(\bar{Y}(z)-\bar{Y}_z) \frac{1}{n r_z^2}\sum_{i\ne j:Z_i=z,Z_j=z} D_{ij}(Y_i(z)-\bar{Y}(z)),\\
        M_3 &= 2(\bar{Y}(z)-\bar{Y}_z)^2 \frac{1}{n r_z^2}\sum_{i\ne j:Z_i=z,Z_j=z} D_{ij}.
    \end{align*}
    Applying Lemma~\ref{lem:quadra-term-1} with $f_i=g_i = Y_i(z)-\bar{Y}(z)$, we get
    \[
    M_1 = \frac{1}{n}\sum_{i\ne j} D_{ij}(Y_i(z)-\bar{Y}(z))(Y_j(z)-\bar{Y}(z))+\op(1).
    \]
    Applying Lemma~\ref{lem:order-sample-mean} and Lemma~\ref{lem:quadra-term-1} with $f_i = Y_i(z)-\bar{Y}(z)$ and $g_i= 1$, we get
     \[
     M_2 = 2(\bar{Y}(z)-\bar{Y}_z) \Op(1) = \op(1).
     \]
     Applying Lemma~\ref{lem:order-sample-mean} and Lemma~\ref{lem:quadra-term-1} with $f_i=g_i = 1$,  we get
     \[
     M_2 = (\bar{Y}(z)-\bar{Y}_z)^2 \Op(1) = \op(1).
     \]     
     These results together imply that
     \begin{align*}
               s_{\diag^- (\bs{D}),Y(z)}^2& = \frac{1}{n}\sum_{i\ne j} D_{ij}(Y_i(z)-\bar{Y}(z))(Y_j(z)-\bar{Y}(z)) + \op(1) \\
               &=\left(1+O(n^{-1})\right) S_{\diag^- (\bs{D}),Y(z)}^2 + \op(1) =  S_{\diag^- (\bs{D}),Y(z)}^2 + \op(1).
     \end{align*}
     
     For $s_{\diag^- (\bs{D}),Y(1),Y(0)}$, we have
     \[
     s_{\diag^- (\bs{D}),Y(1),Y(0)} = \frac{1}{n r_1r_0}\sum_{i\ne j} D_{ij}(Y_i(1)-\bar{Y}_1)(Y_j(0)-\bar{Y}_0)Z_i (1-Z_j) = M_4+M_5,
     \]
     where
     \begin{align*}
              M_4 &=  -\frac{1}{n r_1r_0}\sum_{i\ne j} D_{ij}(Y_i(1)-\bar{Y}_1)(Y_j(0)-\bar{Y}_0)Z_i Z_j,\\
              M_5 &=  \frac{1}{n r_1r_0}\sum_{i\ne j} D_{ij}(Y_i(1)-\bar{Y}_1)(Y_j(0)-\bar{Y}_0)Z_i. 
     \end{align*}   
     Similarly, applying Lemma~\ref{lem:quadra-term-1}, we get that
     \[
     M_4 = -\frac{r_1}{n r_0}\sum_{i\ne j} D_{ij}(Y_i(1)-\bar{Y}_1)(Y_j(0)-\bar{Y}_0)+\op(1).
     \]
     The term $M_5$ is decomposed as $M_5 = M_{51}+M_{52}+M_{53}+M_{54}$, where 
\begin{align*}
    M_{51} &=  \frac{1}{n r_1r_0}\sum_{i\ne j} D_{ij}(Y_i(1)-\bar{Y}(1))(Y_j(0)-\bar{Y}(0))Z_i,\\
     M_{52} &=  (\bar{Y}(1)-\bar{Y}_1)\frac{1}{n r_1r_0}\sum_{i\ne j} D_{ij}(Y_j(0)-\bar{Y}(0))Z_i,\\
       M_{53} &=  (\bar{Y}(0)-\bar{Y}_0)\frac{1}{n r_1r_0}\sum_{i\ne j} D_{ij}(Y_i(1)-\bar{Y}(1))Z_i,\\
    M_{54} &=  (\bar{Y}(0)-\bar{Y}_0)(\bar{Y}(1)-\bar{Y}_1)\frac{1}{n r_1r_0}\sum_{i\ne j} D_{ij}Z_i.   
\end{align*}
 Applying Lemma~\ref{lem:linear-term-1} with   $y_i = \sum_{j\in [n]\backslash i} D_{ij}( Y_j(0)-\bar{Y}(0))$, $g_i = Y_i(1)-\bar{Y}(1)$, and $a_i=1$, we get
 \begin{align*}
     M_{51} &= \frac{1}{n r_1r_0}\sum_{i\ne j} D_{ij}(Y_i(1)-\bar{Y}(1))(Y_j(0)-\bar{Y}(0))Z_i \\
     &= \frac{1}{n r_0}\sum_{i\ne j} D_{ij}(Y_i(1)-\bar{Y}(1))(Y_j(0)-\bar{Y}(0))+\op(1).
 \end{align*}
 Applying Lemma~\ref{lem:linear-term-1} with   $y_i = \sum_{j\in [n]\backslash i} D_{ij}( Y_j(0)-\bar{Y}(0))$, $g_i = 1$, and $a_i=1$, we get
 \[
 M_{52} = (\bar{Y}(1)-\bar{Y}_1)\Op(1) = \op(1).
 \]
  Applying Lemma~\ref{lem:linear-term-1} with   $y_i = \sum_{j\in [n]\backslash i} D_{ij}$, $g_i = ( Y_i(1)-\bar{Y}(1))$, and $a_i=1$, we get
 \[
 M_{53} = (\bar{Y}(0)-\bar{Y}_0)\Op(1) = \op(1).
 \]
 Applying Lemma~\ref{lem:linear-term-1} with   $y_i = \sum_{j\in [n]\backslash i} D_{ij}$, $g_i = 1$, and $a_i=1$, we get
 \[
 M_{53} = (\bar{Y}(0)-\bar{Y}_0)(\bar{Y}(1)-\bar{Y}_1)\Op(1) = \op(1).
 \]
 These results together imply that
\begin{align*}
 s_{\diag^- (\bs{D}),Y(1),Y(0)} &= \frac{1}{n}\sum_{i\ne j} D_{ij}(Y_i(1)-\bar{Y}_1)(Y_j(0)-\bar{Y}_0) + \op(1) \\
 &=  S_{\diag^- (\bs{D}),Y(1),Y(0)} + \op(1).
 \end{align*}
 
 Finally, for $s^2_{\diag(\bs{D}), Y(z)}$, we have
\begin{align*}
    s^2_{\diag(\bs{D}), Y(z)} = \frac{1}{n_z}\sum_{i:Z_i=z} D_{ii} (Y_i(z)-\bar{Y}_z)^2 = M_{6}+M_{7}+M_{8},
\end{align*}
where
\begin{align*} 
& M_6 =  \frac{1}{n_z}\sum_{i:Z_i=z} D_{ii} (Y_i(z)-\bar{Y}(z))^2,\\
&M_7 =2(\bar{Y}(z)-\bar{Y}_z) \frac{1}{n_z}\sum_{i:Z_i=z} D_{ii} (Y_i(z)-\bar{Y}(z)),
\\
&M_8   =(\bar{Y}(z)-\bar{Y}_z)^2 \frac{1}{n_z}\sum_{i:Z_i=z} D_{ii}.
\end{align*} 
 Applying Lemma~\ref{lem:linear-term-1} with  $y_i = g_i = Y_i(z)-\bar{Y}(z)$ and $a_i=D_{ii}$, we get  
 \[
 M_{6} = \frac{1}{n}\sum_{i:Z_i=z} D_{ii} (Y_i(z)-\bar{Y}(z))^2+ \op(1).
 \]
Applying Lemma~\ref{lem:order-sample-mean} and Lemma~\ref{lem:linear-term-1} with $y_i=1$, $ g_i = Y_i(z)-\bar{Y}(z)$, and $a_i=D_{ii}$,  we get  
 \[
 M_{7} = 2(\bar{Y}(z)-\bar{Y}_z) \Op(1) = \op(1).
 \]
 Applying Lemma~\ref{lem:order-sample-mean} and Lemma~\ref{lem:linear-term-1} with $y_i=1$, $ g_i =1$, and $a_i=D_{ii}$,  we get  
 \[
 M_{8} = (\bar{Y}(z)-\bar{Y}_z)^2 \Op(1) = \op(1).
 \]
 These results together imply that
\[
 s^2_{\diag(\bs{D}), Y(z)} = \frac{1}{n}\sum_{i} D_{ii} (Y_i(z)-\bar{Y}_z)^2+\op(1) =  S^2_{\diag(\bs{D}), Y(z)} +\op(1).
\]
To sum up, we have concluded the proof.
\end{proof}
Now, we are ready to prove \Cref{thm:inference}. The proof also includes the technical details of the comment of \eqref{eq:sigmadecomp}.
\begin{proof}[Proof of Theorem~\ref{thm:inference}] 
Recall that we denote $\bs{P}= \bs{I} - \frac{1}{n}\bs{1}\bs{1}^\top$. By Lemma~\ref{lem:norm-of-two-matrix} and the fact $0\le H_{ii}\le 1$, there is
\[
\|\diag(\bs{\star})\|_2 = O(1), \quad \|\diag^-(\bs{\star})\|_2 = O(1), \quad \star \in \{\bs{H},\bs{Q}, \bs{P}\}.
\]
We then expand $\bs{B}$ as
\begin{align*}
    \bs{B} &= \left(\bs{P}  - \bs{H} + \bs{P} \diag(\bs{H})\right)^\top \left(\bs{P}  - \bs{H} + \bs{P} \diag(\bs{H})\right)\\
    &= \bs{P}-\bs{H}+ \diag(\bs{H})\bs{P} \diag(\bs{H}) + (\bs{P}-\bs{H}) \diag(\bs{H}) +  \diag(\bs{H}) (\bs{P}-\bs{H}).
\end{align*}
Therefore, we have
\begin{align*}
    \diag(\bs{B})&=\diag(\bs{P})-\diag(\bs{H})+ \diag(\bs{H})\diag(\bs{P}) \diag(\bs{H})   \\
    &+ (\diag(\bs{P})-\diag(\bs{H})) \diag(\bs{H}) +\diag(\bs{H}) (\diag(\bs{P})-\diag(\bs{H})),\\
        \diag^-(\bs{B})&=\diag^-(\bs{P})-\diag^-(\bs{H})+ \diag(\bs{H})\diag^-(\bs{P}) \diag(\bs{H})  \\
    &+ (\diag^-(\bs{P})-\diag^-(\bs{H})) \diag(\bs{H}) + \diag(\bs{H}) (\diag^-(\bs{P})-\diag^-(\bs{H})).
\end{align*}
By sub-additivity and sub-multiplicativity of the $l_2$-norm, we have
\[
\|\diag(\bs{B})\|_2  = O(1), \quad \|\diag^-(\bs{B})\|_2  = O(1).
\]
Therefore, for $\star \in \{\bs{H},\bs{Q},\bs{B},\bs{P}\}$, under Assumptions~\ref{assumption:positive-limit-of-assigned-proportion}-\ref{assumption:lindeberg-type-condition-p-o(n)}, we can derive that
\begin{align}
&s^2_{\diag(\star),Y(z)} = S^2_{\diag(\star),Y(z)}+\op(1)\quad s^2_{\diag^-(\star),Y(z)} = S^2_{\diag^-(\star),Y(z)}+\op(1),\nonumber \\
\label{eq:consistency-of-sample-correlation}
&    s_{\diag^-(\star),Y(1),Y(0)} = S_{\diag^-(\star),Y(1),Y(0)} + \op(1).
\end{align}
by using Lemma~\ref{lem:sample-covariance-consistency}.
Thus, we have proved the consistency of those empirical estimators of covariances.

Next, we prove that
\begin{align}
        \cI_{3} = &\sum_{z \in \{0, 1\}}  \left(S^2_{\diag(\bs{B}), Y(z)} - S^2_{\diag(\bs{Q}), Y(z)}- S^2_{\diag^-(\bs{H}),Y(z)}\right)  \nonumber \\ 
    & + 2S_{\diag^-(\bs{H}),Y(1),Y(0)} - S^2_{\diag(\bs{H}),Y(1)-Y(0)} - S^2_{e(1)-e(0)} + O\left(n^{-1}\right).\label{eq:cI3}
\end{align}
Some direct calculations give that
\[
B_{ii} = 1-\frac{1}{n}+ \left(1-\frac{2}{n}\right)H_{ii}- \left(1+\frac{1}{n}\right)H_{ii}^2,
\]
which implies that $B_{ii}-Q_{ii} = 1+ O(n^{-1})$.

Applying the equation 
\[
2S_{\diag(\bs{D}), Y(1),Y(0)} = \sum_{z\in\{0,1\}} S_{\diag(\bs{D}), Y(z)}^2-S_{\diag(\bs{D}), Y(1)-Y(0)}^2,
\]
with $\bs{D}\in \{\bs{B}, \bs{Q}, \bs{H}\}$ and the equation 
\[
S_{Y(1)-Y(0)}^2 = S_{\bs{H},Y(1)-Y(0)}^2+S_{e(1)-e(0)}^2,
\]
we obtain that
\begin{align*}
   \cI_3 &= 2S_{\diag(\bs{B}), Y(1),Y(0)}-2S_{\diag(\bs{Q}), Y(1),Y(0)}  \\
&=\sum_{z\in\{0,1\}} \left(S_{\diag(\bs{B}), Y(z)}^2-S_{\diag(\bs{Q}), Y(z)}^2\right)-S_{\diag(\bs{B}), Y(1)-Y(0)}^2+S_{\diag(\bs{Q}), Y(1)-Y(0)}^2\\
&=\sum_{z\in\{0,1\}} \left(S_{\diag(\bs{B}), Y(z)}^2-S_{\diag(\bs{Q}), Y(z)}^2\right)-S_{Y(1)-Y(0)}^2+O\left(n^{-1}\right) \\
&=\sum_{z\in\{0,1\}} \left(S_{\diag(\bs{B}), Y(z)}^2-S_{\diag(\bs{Q}), Y(z)}^2\right)-S_{\bs{H},Y(1)-Y(0)}^2-S_{e(1)-e(0)}^2+O\left(n^{-1}\right)\\
 &= \sum_{z \in \{0, 1\}} \left(S^2_{\diag\{\bs{B}\}, Y(z)} - S^2_{\diag\{\bs{Q}\}, Y(z)}- S^2_{\diag^-\{\bs{H}\},Y(z)}\right)   \\ 
    &\qquad+ 2S_{\diag^-\{\bs{H}\},Y(1),Y(0)} - S^2_{\diag\{\bs{H}\},Y(1)-Y(0)} - S^2_{e(1)-e(0)}+O\left(n^{-1}\right).
\end{align*}
We replace all terms in the formula of $\sigma_\hd^2$ with their empirical estimators, except for the term
\[
-S^2_{\diag\{\bs{H}\},Y(1)-Y(0)} - S^2_{e(1)-e(0)},
\]
which constitutes the bias of $\hat{\sigma}^2_\hd$.

Using \eqref{eq:consistency-of-sample-correlation}, we get that under Assumptions \ref{assumption:positive-limit-of-assigned-proportion}-\ref{assumption:lindeberg-type-condition-p-o(n)},
\begin{align}
    \label{eq:inference-p/n-alpha}
    \hat{\sigma}^2_\hd = {\sigma}^2_\hd +S^2_{\diag\{\bs{H}\},Y(1)-Y(0)} + S^2_{e(1)-e(0)} + \op(1).
\end{align}

{\revone Since $\liminf_{n\rightarrow \infty} \sigma_{\hd}^2>0$ under \Cref{assumption:Positive-variance-limit-alpha-positive}, we can choose $
	\epsilon_n := -\frac{\hat{\sigma}^2_{\hd}-({\sigma}_{\hd}^2 +  S_{e(1)-e(0)}^2 + S^2_{\diag\{\bs{H}\},Y(1)-Y(0)})}{\sigma^2_{\hd}} = \op(1)$,
	such that
\begin{align*}
    & \frac{\hat{\sigma}^2_{\hd}}{\sigma^2_{\hd}} = 1 + \frac{  S_{e(1)-e(0)}^2 + S^2_{\diag\{\bs{H}\},Y(1)-Y(0)}}{\sigma^2_{\hd}} + \frac{\hat{\sigma}^2_{\hd}-({\sigma}_{\hd}^2 +  S_{e(1)-e(0)}^2 + S^2_{\diag\{\bs{H}\},Y(1)-Y(0)})}{\sigma^2_{\hd}} \\
    \geq & 1 +  \frac{\hat{\sigma}^2_{\hd}-({\sigma}_{\hd}^2 +  S_{e(1)-e(0)}^2 + S^2_{\diag\{\bs{H}\},Y(1)-Y(0)})}{\sigma^2_{\hd}}  = 1 - \epsilon_n.
\end{align*}}

It remains to show that if we are under Assumptions \ref{assumption:positive-limit-of-assigned-proportion}-\ref{assumption:lindeberg-type-condition-p-o(n)} and, in addition, Assumption~\ref{assumption:p-o(n)-case}, then
\begin{align}
\label{eq:inference-p-o(n)}
    {\sigma}^2_\hd +S^2_{\diag\{\bs{H}\},Y(1)-Y(0)} + S^2_{e(1)-e(0)} = {\sigma}^2_\adj + S^2_{e(1)-e(0)} + o(1).
\end{align}
{\revone Note that the above equality, together with \Cref{assumption:Positive-variance-limit-p-o(n)} and~\eqref{eq:inference-p/n-alpha}, immediately yields that $\hat{\sigma}^2_{\hd}/\sigma^2_{\adj} \geq 1 -\epsilon_n$, where $\epsilon_n := -\frac{\hat{\sigma}^2_{\hd}-({\sigma}_{\adj}^2 +  S_{e(1)-e(0)}^2)}{\sigma^2_{\adj}} = \op(1)$.}

We now prove~\eqref{eq:inference-p-o(n)}. Comparing the left-hand and right-hand sides of \eqref{eq:inference-p-o(n)} with the formulas of $\hat{\sigma}^2_{\hd}$ and $\sigma^2_{\adj}$, we find that it suffices to prove that  
\[
S^2_{\diag\{\bs{H}\},Y(1)-Y(0)} =o(1),\quad S^2_{s(z)} = o(1), \quad S_{\bs{Q}, \;r_1^{-2} Y(1) - r_0^{-2} Y(0)}^2 = o(1).
\]
Under Assumption~\ref{assumption:p-o(n)-case}, using $\sum_i H_{ii} = p$, we obtain that 
\begin{align*}
    S^2_{\diag\{\bs{H}\},Y(1)-Y(0)}  &\leq \frac{2}{n-1}\sum_i H_{ii} (Y_i(1)-\bar{Y}(1))^2+\frac{2}{n-1}\sum_i H_{ii} (Y_i(0)-\bar{Y}(0))^2\\
    & \leq  \frac{2}{n-1}\sum_{i=1}^p (Y_{(i)}(1)-\bar{Y}(1))^2+\frac{2}{n-1}\sum_{i=1}^p (Y_{(i)}(0)-\bar{Y}(0))^2 =o(1).
\end{align*}
On the other hand, using $\sum_i H_{ii}^2 \le \sum_i H_{ii} = p$, we obtain that
\begin{align*}
    S^2_{s(z)}  \leq  \frac{1}{n-1} \sum_i H_{ii}^2 (Y_i(z)-\bar{Y}(z))^2 \leq \frac{1}{n-1}\sum_{i=1}^p (Y_{(i)}(z)-\bar{Y}(z))^2 = o(1).
\end{align*}
Finally $S_{\bs{Q}, \;r_1^{-2} Y(1) - r_0^{-2} Y(0)}^2 = o(1)$ follows from the above analysis and the proof of Theoerm~\ref{theorem:CLT-p=o(n)}. 

Putting together, the conclusion then follows.
\end{proof}
\begin{proof}[Proof of Corollary~\ref{cor:inference}]
  In the proof of Theorem~\ref{thm:inference}, we have derived that (recall \eqref{eq:cI3})
  \[
  \cI_3 = \sum_{z\in\{0,1\}} \left(S_{\diag(\bs{B}), Y(z)}^2-S_{\diag(\bs{Q}), Y(z)}^2\right)-S_{Y(1)-Y(0)}^2+O\left(n^{-1}\right).
  \]
 Then, the conclusion follows by replacing all terms with their empirical estimators except for $S_{Y(1)-Y(0)}^2$ and by using a proof similar to that of Theorem~\ref{thm:inference}.
\end{proof}

{\rev
\section{Comparisons with $\hat{\tau}_{\CMA}$ and $\hat{\tau}_{\CMO}$}\label{sec:exact-unbiased}

\subsection{Comparison with $\hat{\tau}_{\CMA}$}

For ease of presentation, with a slight abuse of notation throughout this section we write $\bs{X}_i := \bs{X}_i - \bar{\bs{X}}$. From \citet[Theorem~4.2]{chang2021exact}, we have that $\hat{\tau}_{\CMA}$ can be expressed as
\begin{equation}\label{eq:cmadecomp}
	\begin{aligned}
		\hat{\tau}_{\CMA} = &\hat{\tau}_{\unadj} - \frac{(n_1 - 1) n}{n_1 (n - 1)} \cdot \frac{1}{n_1} \sum_{i=1}^n Z_i \hat{\bs\beta}_1^\top \bs{X}_i + \frac{(n_0 - 1) n}{n_0 (n - 1)} \cdot \frac{1}{n_0} \sum_{i=1}^n (1 - Z_i) \hat{\bs\beta}_0^\top \bs{X}_i +\\
		& \frac{n_0}{n_1(n_1-1)} \sum_{i:Z_i=1} H_{ii}(Y_i(1)-\bar{Y}_1) - \frac{n_1}{n_0(n_0-1)} \sum_{i:Z_i=0} H_{ii}(Y_i(0)-\bar{Y}_0) - \Delta_{\CMA,1} + \Delta_{\CMA,0},
	\end{aligned}
\end{equation}
where by writing $\bar{\bs{X}}_z := \frac{1}{n_z} \sum_{i: Z_i = z} \bs{X}_i$, we may express $\Delta_{\CMA,z}$ for $z \in \{0, 1\}$ as
\begin{align*}
 &\Delta_{\CMA,z}= \frac{C_{\CMA,z}}{n_z} \cdot \frac{n}{n - 1} \sum_{i:Z_i=z} (\bs{X}_i-\bar{\bs{X}}_z)^\top \bs{S}_{\bs{X}}^{-2} (\bs{X}_i-\bar{\bs{X}}_z)(Y_i(z)-\bar{Y}_z); \\
&C_{\CMA,z}=\frac n{n_z^3}(\frac{n_z}n-\frac{3n_z(n_z-1)}{n(n-1)}+\frac{2n_z(n_z-1)(n_z-2)}{n(n-1)(n-2)})\frac{n(n-1)(n-2)}{(n_z-1)(n_z-2)n_z}\frac{n_z^3}{n^3}.
\end{align*}

\begin{proposition}
\label{prop:equivalence-of-CMA}
    Under Assumptions~\ref{assumption:positive-limit-of-assigned-proportion}--\ref{assumption:2-moment-of-finite-population-for-y} and the first half of \Cref{assumption:lindeberg-type-condition-p-o(n)}, we have $|\htaudb-\hat{\tau}_{\CMA}| = \Op(n^{-1})$.
\end{proposition}
\begin{proof}
   Recall the definition of $\htaudb$, we have the first five terms of $\hat{\tau}_{\CMA}$ match $\htaudb$, except that there are small differences in their scalings. Hence, there exists a large constant $C$, such that
\begin{align}
    |\htaudb - \hat{\tau}_{\CMA}|& \leq \frac{C}{n}|\hat{\bs{\beta}}_1^\top \bar{\bs{X}}_1| + \frac{C}{n} |\hat{\bs{\beta}}_0^\top \bar{\bs{X}}_0| + \frac{C}{n^2} \Big|\sum_{i:Z_i=1} H_{ii}(Y_i(1)-\bar{Y}_1)\Big| \nonumber\\
    & + \frac{C}{n^2} \Big|\sum_{i:Z_i=0} H_{ii}(Y_i(0)-\bar{Y}_0)\Big| +  |\Delta_{\CMA,1}| + |\Delta_{\CMA,0}|. \nonumber
\end{align}
To conclude the proof, it suffices to show that for $z\in \{0,1\}$
\begin{align}
\label{eq:equivalence-of-estimators-1}
    &\frac{1}{n}\sum_{i:Z_i=z} H_{ii}(Y_i(z)-\bar{Y}_z) = \Op(1),\quad \bar{\bs{X}}_z^\top\hat{\bs{\beta}}_z = \Op(1);\\
    \label{eq:equivalence-of-estimators-2}
    &|\Delta_{\CMA,z}| = \Op(n^{-1}).
\end{align}
  We now prove the first result of \eqref{eq:equivalence-of-estimators-1}. Below we just prove the case $z = 1$; and $z = 0$ follows from an analogous argument. We have
    \begin{align*}
        \frac{1}{n}\sum_{i:Z_i=1} H_{ii}(Y_i(1)-\bar{Y}_1) = \frac{1}{n}\sum_{i:Z_i=1} H_{ii}(Y_i(1)-\bar{Y}(1)) + \Big(\frac{1}{n}\sum_{i:Z_i=1} H_{ii}\Big)(\bar{Y}(1)-\bar{Y}_1).
    \end{align*}
Apparently \Cref{lem:order-sample-mean} yields $\bar{Y}(1) - \bar{Y}_1 = \Op(1 / \sqrt{n})$; applying \Cref{lem:linear-term-1} with $a_i=H_{ii}$, $y_i = 1$, $g_i=1$ and $a_i=H_{ii}$, $y_i = (Y_i(1)-\bar{Y}(1))$, $g_i=1$ respectively, we have
    \begin{align*}
            &\frac{1}{n}\sum_{i:Z_i=1} H_{ii} = \frac{r_1}{n}\sum_{i=1}^n H_{ii} + \op(1) = \frac{r_1 p }{n} + \op(1) = \Op(1);\\
        &\frac{1}{n}\sum_{i:Z_i=1} H_{ii}(Y_i(1)-\bar{Y}(1)) = \frac{r_1}{n}\sum_{i=1}^n H_{ii}(Y_i(1)-\bar{Y}(1)) + \op(1) = \Op(1);
    \end{align*}
    where the last equality follows from
    \[
    \frac{1}{n}\sum_{i=1}^n H_{ii}(Y_i(1)-\bar{Y}(1)) \leq  \sqrt{\Big\{\frac{1}{n}\sum_{i=1}^n H_{ii}^2\Big\} \Big\{\frac{1}{n}\sum_{i=1}^n (Y_i(1)-\bar{Y}(1))^2 \Big\}} = O(1).
    \]
    Putting together proves the first result of~\eqref{eq:equivalence-of-estimators-1}.

On the other hand, we have show  in the proof of \Cref{proposition:asymptotic-equivalance-htauadj} that
\begin{align*}
    \bar{\bs{X}}_1^\top\hat{\bs{\beta}}_1 = \frac{n-1}{(n_1-1)n_1}\bs{Z}^\top \bs{A}(1) \bs{Z} +\frac{n-1}{(n_1-1)n_1}\bs{Z}^\top \bs{H} \bs{Z} (\bar{Y}(1)-\bar{Y}_1),
\end{align*}
and we have shown in the proof of \Cref{lem:ZHZ-order-A(z)-order}, that
    \begin{align*}
      \bs{Z}^\top \bs{A}(z) \bs{Z} &= \sum_{[i,j]} r_1^2 H_{ij} (Y_j(z)-\bar{Y}(z))+  \sum_{i} r_1 H_{ii} (Y_i(z)-\bar{Y}(z)) + \op(n)\\
      &= \sum_{i} r_1r_0 H_{ii} (Y_i(z)-\bar{Y}(z)) + \op(n) =  \Op(n);\\
    \bs{Z}^\top \bs{H} \bs{Z} &= \sum_{[i,j]} r_1^2 H_{ij} +  \sum_{i} r_1 H_{ii}  + \op(n) = \sum_{i} r_1r_0 H_{ii} + \op(n) = \Op(n).
    \end{align*}
It follows that $\bar{\bs{X}}_1^\top\hat{\bs{\beta}}_1 = \Op(n^{-1}\cdot n) + \Op(n^{-1}\cdot n\cdot n^{-1/2}) = \Op(1)$ and similarly $\bar{\bs{X}}_0^\top\hat{\bs{\beta}}_0 = \Op(1)$.

Finally, we show the magnitude of $\Delta_{\CMA,1}$. The proof for $\Delta_{\CMA,0}$ is similar. We first decompose $\Delta_{\CMA,1}$ as
\[
\Delta_{\CMA,1} = \frac{n^2 C_{\CMA,1}}{n_1}\Big\{ M_1 + (\bar{Y}(1)-\bar{Y}_1)M_2\Big\},
\]
where 
\begin{align*}
	M_1 := & \frac{1}{n (n - 1)}\sum_{i:Z_i=1} (\bs{X}_i-\bar{\bs{X}}_1)^\top \bs{S}_{\bs{X}}^{-2} (\bs{X}_i-\bar{\bs{X}}_1)(Y_i(1)-\bar{Y}(1)); \\
	M_2 := & \frac{1}{n (n - 1)}\sum_{i:Z_i=1} (\bs{X}_i-\bar{\bs{X}}_1)^\top \bs{S}_{\bs{X}}^{-2} (\bs{X}_i-\bar{\bs{X}}_1).
\end{align*}
Our only remaining job then is to show that both $M_1$ and $M_2$ are of order $\Op(1)$. We first focus on $M_1$. 
Using that $(n - 1)^{-1} \bs{X}_i^\top \bs{S}_{\bs{X}}^{-2} \bs{X}_j = H_{ij}$ and $(n - 1)^{-1}\bs{X}_i^\top \bs{S}_{\bs{X}}^{-2} \bs{X}_j (Y_j(1)-\bar{Y}(1)) = A_{ij}(1)$, we can rewrite $M_1$ as
\begin{align*}
    M_1 = \frac{1}{n}\sum_{i:Z_i=1} H_{ii}(Y_i(1)-\bar{Y}(1)) -\frac{2\bs{Z}^\top\bs{A}(1)\bs{Z}}{n_1 n} +  \Big(\frac{\bs{Z}^\top\bs{H}\bs{Z}}{n_1^2}\Big)\frac{1}{n}\sum_{i:Z_i=1}(Y_i(1)-\bar{Y}(1)). 
\end{align*}
    It then follows from the same analysis as above that
\[
M_1 = \Op(1) + \Op(n^{-1}) + \Op(n^{-1})\Op(n^{-1/2}) = \Op(1).
\]

Similarly, we can prove $M_2 = \Op(1)$; putting together we prove the desired result.
\end{proof}


\subsection{Comparison with $\hat{\tau}_{\CMO}$}

For ease of presentation, with a slight abuse of notation throughout this section we write $\bs{X}_i := \bs{X}_i - \bar\bX$. we now define $\bs{W}_i := (1, \bs{X}_i^\top)^\top$ and $\bs{S}_{\bs{W}}^2 := \frac{1}{n - 1} \sum_{i=1}^n \bs{W}_i \bs{W}_i^\top$. From \citet[Section~3.3]{chiang2023regression}, we have that $\hat{\tau}_{\CMO}$ can be expressed as
\begin{align*}
	\hat{\tau}_{\CMO} =& \hat{\tau}_{\unadj} - \sum_{i=1}^n \frac{Z_i-r_1}{(n - 1)r_1}\bs{W}_i^\top\bs{S}_{\bs{W}}^{-2}\Big(\frac{1}{n_1}\sum_{j:Z_j=1}\bs{W}_jY_j\Big) -  \sum_{i=1}^n \frac{Z_i-r_1}{(n - 1)r_0}\bs{W}_i^\top\bs{S}_{\bs{W}}^{-2}\Big(\frac{1}{n_0}\sum_{j:Z_j=0}\bs{W}_jY_j \Big) \\
	& +
	\frac{n_0}{n_1^2(n - 1)} \sum_{i:Z_i=1} \bs{W}_i^\top \bs{S}^{-2}_{\bs{W}} \bs{W}_i Y_i  - \frac{n_1}{n_0^2(n - 1)} \sum_{i:Z_i=0} \bs{W}_i^\top \bs{S}^{-2}_{\bs{W}} \bs{W}_i Y_i \\
	& -  \frac{n_0}{n_1^2(n-1)^2}\sum_{[i,j]} \bs{W}_i^\top \bs{S}_{\bs{W}}^{-2}Z_j \bs{W}_j Y_j +  \frac{n_1}{n_0^2(n-1)^2}\sum_{[i,j]} \bs{W}_i^\top \bs{S}_{\bs{W}}^{-2}(1-Z_j) \bs{W}_j Y_j.
\end{align*}

\begin{proposition}
	\label{prop:equivalence-of-CMO}
	Under \Cref{assumption:positive-limit-of-assigned-proportion}--\ref{assumption:lindeberg-type-condition-p-o(n)} and $\max_{z\in\{0,1\}} |\bar{Y}(z)|= O(1)$, we have
	\begin{align*}
		\hat{\tau}_{\CMO}-\htaudb = \Delta_{\CMO} + \op(n^{-1/2})
	\end{align*}
	where  
	\[
	\Delta_{\CMO} = \frac{1}{n }\sum_{i=1}^n (Z_i-r_1)H_{ii} \Big(\frac{\bar{Y}(1)}{r_1}+\frac{\bar{Y}(0)}{r_0}\Big)  -\frac{1}{n}\sum_{[i,j]} (Z_i-r_1)(Z_j-r_1) H_{ij}\Big(\frac{\bar{Y}(1)}{r_1^2}-\frac{\bar{Y}(0)}{r_0^2}\Big).
	\]
	Moreover, we have $\Delta_{\CMO} = \op(n^{-1/2})$ when $p = o(n)$; and $\Delta_{\CMO} = \Op(n^{-1/2})$ when $p \asymp n$.
\end{proposition}

\begin{proof}
	Using that
	$$(n-1)\bs{S}_{\bs{W}}^{2} = \begin{pmatrix}
		n & \\
		& (n-1)\bs{S}_{\bs{X}}^2
	\end{pmatrix},$$ 
	we have, for any $1\leq i,j\leq n$, 
	\begin{equation}\label{eq:bsw}
		(n - 1)^{-1}\bs{W}_{i}^\top \bs{S}_{\bs{W}}^{-2} \bs{W}_{j} = H_{ij} + n^{-1}.
	\end{equation}
	With above, we can express the second and third terms in the decomposition of $\hat{\tau}_{\CMO}$ as
	\begin{align*}
		& -\sum_{i=1}^n\sum_{j:Z_j=1} \frac{Z_i-r_1}{r_1}(H_{ij}+n^{-1})\frac{1}{n_1}Y_j -  \sum_{i=1}^n \sum_{j:Z_j=0} \frac{Z_i-r_1}{r_0}(H_{ij}+n^{-1})\frac{1}{n_0}Y_j\\
		=& -\sum_{i=1}^n\sum_{j:Z_j=1} \frac{Z_i-r_1}{r_1}H_{ij}\frac{1}{n_1}Y_j -  \sum_{i=1}^n \sum_{j:Z_j=0} \frac{Z_i-r_1}{r_0}H_{ij}\frac{1}{n_0}Y_j.
	\end{align*}

Now using that $\sum_{i=1}^n H_{ij} = 0$, we can further rewrite the above expression as
	\begin{align*}
		& -\sum_{i=1}^n\sum_{j:Z_j=1} \frac{Z_i}{r_1}H_{ij}\frac{1}{n_1}Y_j +  \sum_{i=1}^n \sum_{j:Z_j=0} \frac{1 - Z_i}{r_0}H_{ij}\frac{1}{n_0}Y_j\nonumber\\
		& = - \underset{=: M_1}{\underbrace{\sum_{i=1}^n\sum_{j=1}^n \frac{Z_iZ_j}{r_1}H_{ij}\frac{1}{n_1}(Y_j - \bar{Y}_1)}} + \underset{=: M_0}{\underbrace{\sum_{i=1}^n \sum_{j=1}^n \frac{(1 - Z_i) (1 - Z_j)}{r_0}H_{ij}\frac{1}{n_0}(Y_j - \bar{Y}_0)}} + \Delta_{\CMO,1},
	\end{align*}
where
\[
\Delta_{\CMO,1} = - \sum_{i=1}^n\sum_{j=1}^n \frac{Z_iZ_j}{r_1 n_1 }H_{ij}\bar{Y}_1 + \sum_{i=1}^n \sum_{j=1}^n \frac{(1-Z_i)(1-Z_j)}{r_0n_0}H_{ij}\bar{Y}_0.
\]

We now focus on $M_1$, expanding $H_{ij}$, we have
\begin{equation*}
		M_1 = \frac{1}{r_1 (n - 1)}\sum_{i=1}^n Z_i \bs{X}_i^\top \bs{S}_{\bs{X}}^{-2} \left(\frac{1}{n_1}\sum_{j=1}^n Z_j \bs{X}_j(Y_j-\bar{Y}_1)\right) = \frac{n}{n - 1}\bar{\bs{X}}_1^\top \bs{S}_{\bs{X}}^{-2} \left(\frac{1}{n_1}\sum_{j=1}^n Z_j \bs{X}_j(Y_j-\bar{Y}_1)\right),
	\end{equation*}
which matches the second term in the decomposition of $\hat{\tau}_{\CMA}$ in~\eqref{eq:cmadecomp} (this follows directly from the definition of $\hat{\bs\beta}_1$ in~\eqref{eq:newbeta}). Analogously, we have that $M_0$ matches the third term in the decomposition.

Using again~\eqref{eq:bsw} and $\sum_{i=1}^n H_{ij} = 0$, we can express the last four terms of $\hat{\tau}_{\CMO}$ as
	\begin{align}
		&\frac{n_0}{n_1^2} \sum_{i:Z_i=1} (H_{ii}+n^{-1}) Y_i  - \frac{n_1}{n_0^2} \sum_{i:Z_i=0} (H_{ii}+n^{-1}) Y_i - \nonumber \\
		& \frac{n_0}{n_1^2(n-1)}\sum_{[i,j]} (H_{ij}+n^{-1})Z_j  Y_j + \frac{n_1}{n_0^2(n-1)}\sum_{[i,j]} (H_{ij}+n^{-1})(1-Z_j)  Y_j \nonumber\\
		=& \frac{n_0}{n_1^2} \sum_{i:Z_i=1} (H_{ii}+n^{-1}) Y_i  - \frac{n_1}{n_0^2} \sum_{i:Z_i=0} (H_{ii}+n^{-1}) Y_i - \nonumber\\
		&\frac{n_0}{n_1^2(n-1)}\sum_{j:Z_j=1} \big(-H_{jj}+\frac{n-1}{n}\big)  Y_j + \frac{n_1}{n_0^2(n-1)}\sum_{j:Z_j=0} \big(-H_{jj}+\frac{n-1}{n}\big)  Y_j \nonumber\\
		=& \frac{n_0}{n_1^2} \sum_{i:Z_i=1} H_{ii} Y_i  - \frac{n_1}{n_0^2} \sum_{i:Z_i=0} H_{ii} Y_i +\frac{n_0}{n_1^2(n-1)}\sum_{i:Z_i=1} H_{ii}  Y_i - \frac{n_1}{n_0^2(n-1)}\sum_{i:Z_i=0} H_{ii}  Y_i \nonumber\\
		= & \frac{n}{n-1}\Big( \frac{n_0}{n_1^2} \sum_{i:Z_i=1} H_{ii} Y_i  - \frac{n_1}{n_0^2} \sum_{i:Z_i=0} H_{ii} Y_i\Big) \nonumber\\
		= & \underset{=: M_2}{\underbrace{\frac{n}{n-1}\Big\{ \frac{n_0}{n_1^2} \sum_{i:Z_i=1} H_{ii} (Y_i-\bar{Y}_1)  - \frac{n_1}{n_0^2} \sum_{i:Z_i=0} H_{ii} (Y_i-\bar{Y}_0)\Big\}}}+ \Delta_{\CMO,2}, \nonumber 
	\end{align}
	where
	\[
	\Delta_{\CMO,2} = \frac{n}{n-1}\Big( \frac{n_0}{n_1^2} \sum_{i:Z_i=1} H_{ii} \bar{Y}_1  - \frac{n_1}{n_0^2} \sum_{i:Z_i=0} H_{ii} \bar{Y}_0\Big).
	\]
	
	Notice that the term $M_2$ defined above match the fourth and fifth terms of $\hat{\tau}_{\CMA}$, except that there are small differences in their scalings. As a consequence, there exists a constant $C$, such that
	\begin{align*}
		|\hat{\tau}_{\CMO}-\hat{\tau}_{\CMA}-(\Delta_{\CMO,1} + \Delta_{\CMO,2})| \leq &   \frac{C}{n^2} \Big|\sum_{i:Z_i=1} H_{ii}(Y_i(1)-\bar{Y}_1)\Big| + \frac{C}{n^2} \Big|\sum_{i:Z_i=0} H_{ii}(Y_i(0)-\bar{Y}_0)\Big| + \\
		& |\Delta_{\CMA,1}| + |\Delta_{\CMA,0}|. \nonumber
	\end{align*}
	
	In light of the results in~\Cref{prop:equivalence-of-CMA}, to get the first result, we need to show that
	\begin{align}\label{eq:comdecom}
		\Delta_{\CMO,1} + \Delta_{\CMO,2} = & \frac{1}{n }\sum_{i=1}^n (Z_i-r_1)H_{ii} \Big(\frac{\bar{Y}(1)}{r_1}+\frac{\bar{Y}(0)}{r_0}\Big) \\
		&  - \frac{1}{n}\sum_{[i,j]} (Z_i-r_1)(Z_j-r_1) H_{ij}\Big(\frac{\bar{Y}(1)}{r_1^2}-\frac{\bar{Y}(0)}{r_0^2}\Big) + \op(n^{-1/2}) \nonumber.
	\end{align}
	To achieve this goal, we first decompose $	\Delta_{\CMO,1} + \Delta_{\CMO,2} = M_3 + M_4$,
	where
	\begin{align*}
		M_3 &:=  \Big\{  \frac{nn_0}{(n-1)n_1^2} \sum_{i=1}^n Z_i H_{ii}  - \sum_{i=1}^n\sum_{j=1}^n \frac{Z_iZ_j}{n_1 r_1}H_{ij}\Big\}\bar{Y}_1,\\
		M_4 &:=  \Big\{  - \frac{nn_1}{(n-1)n_0^2} \sum_{i=1}^n (1-Z_i) H_{ii}   + \sum_{i=1}^n \sum_{j=1}^n \frac{(1-Z_i)(1-Z_j)}{r_0n_0}H_{ij}\Big\}\bar{Y}_0.
	\end{align*}
For $M_3$, first we have
\begin{equation}
	\frac{n}{n-1} \frac{n_0}{n_1^2} \sum_{i=1}^n Z_i H_{ii} = \frac{nn_0}{(n-1)n_1^2} \sum_{i=1}^n Z_i \Big(H_{ii}-\frac{p}{n}\Big) +  \frac{p n_0}{(n-1) n_1}.\label{eq:decompose-of-M3-term1-CMO}
\end{equation}
	Using \eqref{eq:decompose-Z-H-Z}, we further have
	\begin{align}
		&\sum_{i=1}^n\sum_{j=1}^n \frac{Z_iZ_j}{n_1 r_1}H_{ij} = \frac{\bs{Z}^\top \tilde{\bs{H}}\bs{Z}}{nr_1^2} +  \frac{p n_0}{(n-1)n_1} \label{eq:decompose-of-M3-term2-CMO},
	\end{align}
	where we may rewrite $\bs{Z}^\top \tilde{\bs{H}}\bs{Z}$ as
	\begin{align*}
		\bs{Z}^\top \tilde{\bs{H}}\bs{Z} =& (\bs{Z}-r_1\bs{1})^\top \tilde{\bs{H}}(\bs{Z}-r_1\bs{1}) =  \sum_{i=1}^n (Z_i-r_1)^2 \Big(H_{ii}-\frac{p}{n}\Big) +\\
		&\sum_{[i,j]} (Z_i-r_1)(Z_j-r_1) (H_{ij}+\frac{p}{n(n-1)}).
	\end{align*}
Using $(Z_i-r_1)^2 = (1-2r_1)(Z_i - r_1) + r_1 r_0$, we get 
\begin{align*}
	\sum_{i=1}^n (Z_i-r_1)^2 \Big(H_{ii}-\frac{p}{n}\Big) & = \sum_{i=1}^n (1-2r_1)(Z_i - r_1) \Big(H_{ii}-\frac{p}{n}\Big) + \sum_{i=1}^n r_1 r_0 \Big(H_{ii}-\frac{p}{n}\Big) \\
	& = \sum_{i=1}^n (1-2r_1)(Z_i - r_1) \Big(H_{ii}-\frac{p}{n}\Big).
\end{align*}
	In light of the above, and using moreover $\sum_{[i,j]} (Z_i-r_1)(Z_j-r_1) = -\sum_{i=1}^n (Z_i-r_1)^2 = O(n)$, we have
	\begin{equation}\label{eq:decompose-Z-tldH-Z}
			\bs{Z}^\top \tilde{\bs{H}}\bs{Z} = \sum_{i=1}^n (1-2r_1)(Z_i-r_1)\Big(H_{ii}-\frac{p}{n}\Big) + \sum_{[i,j]} (Z_i-r_1)(Z_j-r_1) H_{ij} + O(1).
	\end{equation}
	
	Substituting \eqref{eq:decompose-Z-tldH-Z} into \eqref{eq:decompose-of-M3-term2-CMO} and using \eqref{eq:decompose-of-M3-term1-CMO}--\eqref{eq:decompose-of-M3-term2-CMO}, we have
	\begin{equation}\label{eq:decompose-of-M3-term3-CMO}
		\begin{aligned}
			&   \frac{nn_0}{(n-1)n_1^2} \sum_{i=1}^n Z_i H_{ii}  - \sum_{i=1}^n\sum_{j=1}^n \frac{Z_iZ_j}{n_1 r_1}H_{ij} = \frac{nn_0}{(n-1)n_1^2} \sum_{i=1}^n Z_i \Big(H_{ii}-\frac{p}{n}\Big) - \frac{\bs{Z}^\top \tilde{\bs{H}}\bs{Z}}{nr_1^2} \\
			& = (\frac{1}{n r_1}+O(n^{-2}))\sum_{i=1}^n (Z_i-r_1)\Big(H_{ii}-\frac{p}{n}\Big) - \frac{1}{nr_1^2}\sum_{[i,j]} (Z_i-r_1)(Z_j-r_1) H_{ij} + \Op(n^{-1})\\
			&= \frac{1}{n r_1}\sum_{i=1}^n (Z_i-r_1)H_{ii} - \frac{1}{nr_1^2}\sum_{[i,j]} (Z_i-r_1)(Z_j-r_1) H_{ij} + \Op(n^{-1}).
		\end{aligned}
	\end{equation}
	
	On the other hand, using \Cref{lem:ZHZ-order-A(z)-order}, and again~\eqref{eq:decompose-of-M3-term1-CMO}--\eqref{eq:decompose-of-M3-term2-CMO}, we have
	\begin{equation}\label{eq:decompose-of-M3-term4-CMO}
		\frac{nn_0}{(n-1)n_1^2} \sum_{i=1}^n Z_i H_{ii}  - \sum_{i=1}^n\sum_{j=1}^n \frac{Z_iZ_j}{n_1 r_1}H_{ij} = \Op(n^{-1/2})+\op(1) = \op(1).
	\end{equation}
	
	In light of~\eqref{eq:decompose-of-M3-term3-CMO},~\eqref{eq:decompose-of-M3-term4-CMO}, $\bar{Y}_1-\bar{Y}(1) = \Op(n^{-1/2})$, and the assumption that $\max_{z \in \{0,1\}} |\bar{Y}(z)| = O(1)$, we have that
	\begin{align*}
		M_3 &=  \Big(\frac{1}{n r_1}\sum_{i=1}^n (Z_i-r_1)H_{ii} - \frac{1}{nr_1^2}\sum_{[i,j]} (Z_i-r_1)(Z_j-r_1) H_{ij} + \Op(n^{-1})\Big)( \bar{Y}(1)+\Op(n^{-1/2})) \\
		&= \frac{1}{n r_1}\sum_{i=1}^n (Z_i-r_1)H_{ii} \bar{Y}(1) - \frac{1}{nr_1^2}\sum_{[i,j]} (Z_i-r_1)(Z_j-r_1) H_{ij}\bar{Y}(1) + \op(n^{-1/2}).
	\end{align*}
	Similarly, we have
	\begin{align*}
		M_4 = \frac{1}{n r_0}\sum_{i=1}^n (Z_i-r_1)H_{ii} \bar{Y}(0) + \frac{1}{nr_0^2}\sum_{[i,j]} (Z_i-r_1)(Z_j-r_1) H_{ij}\bar{Y}(0) + \op(n^{-1/2}).
	\end{align*}
	\eqref{eq:comdecom} follows immediately.
	
	It remains to show the order of $\Delta_{\CMO}$. Let $(T_i)_{i=1}^n$ be the vector of independent Bernoulli random variables obtained by Hajek's coupling. Mimicking the proof of \Cref{prop:decompose-taudb-with-Ti}, we have
	\begin{align*}
		&\frac{1}{n }\sum_{i=1}^n (Z_i-r_1)H_{ii} = \frac{1}{n}\sum_{i=1}^n (T_i-r_1)(H_{ii}-p / n) + \op(n^{-1/2}),\\
		&\frac{1}{n}\sum_{[i,j]} (Z_i-r_1)(Z_j-r_1) H_{ij} =  \frac{1}{n}\sum_{[i,j]} (T_i-r_1)(T_j-r_1) H_{ij}+ \op(n^{-1/2}),
	\end{align*}
	where
	\begin{align*}
		&\var\Big(\frac{1}{n}\sum_{[i,j]} (T_i-r_1)(T_j-r_1) H_{ij}\Big) = \frac{1}{n^2}\sum_{[i,j]} 2(r_1r_0)^2 H_{ij}^2 = \frac{2(r_1r_0)^2}{n^2}\sum_{i} (H_{ii}-H_{ii}^2)  = O(pn^{-2}),\\
		&\var\Big(\frac{1}{n}\sum_{i=1}^n (T_i-r_1)(H_{ii}-p / n)\Big) = \frac{r_1r_0}{n^2}\sum_{i}(H_{ii}-p / n)^2 = O(pn^{-2}).
	\end{align*}

Using above, and the assumption that $\max_{z \in \{0,1\}} |\bar{Y}(z)| = O(1)$, we obtain the desired order of $\Delta_{\CMO}$ under both the $p = o(n)$ and $p \asymp n$ regime.
\end{proof}}

\section{Justification of assumptions}
\label{sec:justification-of-assumption}
In this section, we prove Propositions \ref{proposition:justify-sum-of-max-p-o(n)}--\ref{proposition:compare-the-variance}, which provide some justifications of our assumptions.
For the proof of \Cref{proposition:justify-sum-of-max-p-o(n)}, we will use the classical Bernstein inequality.
\begin{lemma}[Bernstein inequality]
\label{lem:Bernstein's inequality}
Let  $X_1, \ldots, X_n$ be independent centered random variables. Suppose that ${\displaystyle |X_{i}|\leq M}$ almost surely for all $i$. Then, for all $t>0$, we have that
    \[
    \mathbb{P}\left(\sum_{i=1}^n X_i \geq t\right) \leq \exp \left(-\frac{\frac{1}{2} t^2}{\sum_{i=1}^n \mathbb{E}\left[X_i^2\right]+\frac{1}{3} M t}\right).
    \]
\end{lemma}

\begin{proof}[Proof of Proposition~\ref{proposition:justify-sum-of-max-p-o(n)}]
Fix $z\in \{0,1\}$. For ease of presentation, we denote  $Y_i(z)-\E Y_i(z)$  by $U_i$ and $Y_{(i)}(z)-\E Y_i(z)$ by $U_{(i)}$. By definition,
     $(U_{(1)}-\bar{U})^2\geq (U_{(2)}-\bar{U})^2\geq \ldots \geq (U_{(n)}-\bar{U})^2$. We further define $U_{<1>}^2\geq\ldots\geq U_{<n>}^2$ as the ordered sequence of $\{U_i^2\}_{i=1}^n$.
Then, we have that 
     \[
     \sum_{i=1}^p (Y_{(i)}(z)-\bar{Y}(z))^2 = \sum_{i=1}^p (U_{(i)}-\bar{U})^2 \leq 2p\bar{U}^2 + 2\sum_{i=1}^p U_{(i)}^2 \leq 2p\bar{U}^2 + 2\sum_{i=1}^p U_{<i>}^2.
     \]
     Since $\E \bar{U} = 0$ and $\var(\bar{U}) = \var(U_1)/n = O(n^{-1})$, by Chebyshev's inequality, we have that
     \[
     \mathbb{P}(p\bar{U}^2\geq c_{n1}) \leq \frac{\var(\bar{U})p}{c_{n1}} = \var(U_1) \frac{p}{n c_{n1}}.
     \]
     Thus, choosing $c_{n1} = (p/n)^{1/2} = o(1)$, we get that with probability $1-(p/n)^{1/2} = 1-o(1)$,
     \[
     p\bar{U}^2 < c_{n1}.
     \]
     It remains to show that there exists $c_{n2}\rightarrow 0$ such that
    \[
    \mathbb{P}\left(\sum_{i=1}^p U_{<i>}^2\geq c_{n2}\right) = o(1) .
    \]
    Note $\sum_{i=1}^p U_{<i>}^2$ is increasing in $p$, so in the following proof, we assume that $p\to \infty$ without loss of generality. 
    
    
    Now, we consider the following two cases for the distribution of $U_1^2$:
    \begin{itemize}
        \item[(1)] $U_1^2$ is bounded almost surely, i.e., there exists an $M>0$ such that 
        \[
        \mathbb{P}(U_1^2 \geq M) =0.
        \]
   \item[(2)] $U_1^2$ is unbounded, i.e., for any $M>0$, we have
        \[
        \mathbb{P}(U_1^2 \geq M) >0.
        \]
    \end{itemize}
    In case (1), we have that almost surely,
    \[
    \frac{1}{n}\sum_{i=1}^p U_{<i>}^2<pM/n=o(1),
    \]
    in which case we can choose $c_{n2} = pM/n$.

   On the other hand, suppose case (2) holds. Then, we define the upper quantiles of $U_1$ as
    \[
    Q_{a}: = \sup \{M\in \mathbb{R}\mathrel{|}\mathbb{P}(U_1^2 \geq M)\geq a\},\quad a > 0.
    \]
    By definition, $\mathbb{P}(U_1^2 \geq Q_{a}) \geq a$
    and $Q_{a} \rightarrow \infty$ as $a\rightarrow 0$.
    For any $c_{n2}>0$ and $\alpha=p/n$, we have 
    \begin{align*}
            &~\mathbb{P}\left( \sum_{i=1}^p U_{<i>}^2/n  \geq c_{n2}\right) \\
            \leq&~ \mathbb{P}\left(\sum_i I(U_i^2\geq Q_{2\alpha}) < p\right)+ \mathbb{P}\left(  \sum_{i=1}^p U_{<i>}^2/n  \geq c_{n2},\sum_i I(U_i^2\geq Q_{2\alpha}) \geq p\right) \\
            \leq&~ \mathbb{P}\left(\sum_i I(U_i^2\geq Q_{2\alpha}) < p\right) + \mathbb{P}\left( \sum_{i=1}^n U_{i}^2 I(U_i^2\geq Q_{2\alpha}) /n \geq c_{n2} \right)\\      
            =:&~ \mathbb{P}(\mathcal{E}_1)+\mathbb{P}(\mathcal{E}_2) .
    \end{align*}
    We next deal with the events $\mathcal{E}_1$ and $\mathcal{E}_2$, respectively.
     
 For $\mathcal{E}_1$, let $e := \mathbb{P}(U_i^2 \geq Q_{2\alpha})\geq 2\alpha$. Then, we apply Bernstein's inequality with  $X_i = e- I(U_i^2 \geq Q_{2\alpha})$, $t = ne/2$, $\E X_i^2 <e$, and $ |X_i| < 2$ 
 to get that 
 $$\sum_i I(U_i^2 \geq Q_{2\alpha})\geq ne/2 \geq n\alpha =p $$ 
 holds with probability at least
     \begin{align*}
             1-\exp\left(-\frac{3}{32}ne\right)\ge 1-\exp\left(-\frac{3p}{16}\right) = 1-o(1).
     \end{align*}
This implies that $\mathbb{P}(\mathcal{E}_1) = o(1)$.

For $\mathcal{E}_2$, using Markov's inequality, we get that
\begin{align*}
        \mathbb{P}\left(\sum_{i=1}^n U_{i}^2 I(U_i^2\geq Q_{2\alpha}) /n \geq c_{n2}\right)\leq \frac{1}{c_{n2}}\E U_{i}^2 I(U_i^2\geq Q_{2\alpha}) .
\end{align*}
  Since $Q_{2\alpha}\rightarrow \infty$ as $\alpha \rightarrow 0$ and $\E U_{i}^2<\infty$, we have $ \E U_{i}^2 I(U_i^2\geq Q_{2\alpha}) \to 0$. Thus, we can choose $c_{n2} = \left[\E U_{i}^2 I(U_i^2\geq Q_{2\alpha})\right]^{1/2} = o(1)$ such that $\mathbb{P}(\mathcal{E}_2) = c_{n2}\to 0$.

  In sum, we have proved that with probability $1-o(1)$,  
  \[
   \sum_{i=1}^p (Y_{(i)}(z)-\bar{Y}(z))^2  < 2c_{n1}+2c_{n2} = o(1).
  \]
  Hence, the conclusion follows.
\end{proof}

\begin{proof}[Proof for Corollary~\ref{corollary:upper-lower-bound-for-sigma-alpha>0}]
For simplicity of notations,  we denote by $\tilde{Y}_i(z) := Y_i(z)-\bar{Y}(z)$.
    By definition and \Cref{prop:rewrite-sigma-hd-l}, we have
    \begin{align}
        \sigma^2_{\hd,l} &= \frac{1}{(n-1)(r_1r_0)}\sum_i \left(r_0s_i(1)+r_0e_i(1)+r_1s_i(0)+r_1 e_i(0)\right)^2,\label{sigmahdl}\\
       \sigma^2_{\hd,q} & = \frac{(r_1r_0)^2}{n-1} \sum_{[i,j]}H_{ij}^2\left( \frac{\tilde{Y}_i(1)}{r_1^2}-\frac{\tilde{Y}_i(0)}{r_0^2}\right)\left(\frac{\tilde{Y}_j(1)}{r_1^2}-\frac{\tilde{Y}_j(0)}{r_0^2}\right)\nonumber\\
       &+\frac{(r_1r_0)^2}{n-1}\sum_i (H_{ii}-H_{ii}^2)\left( \frac{\tilde{Y}_i(1)}{r_1^2}-\frac{\tilde{Y}_i(0)}{r_0^2}\right)^2  .\label{sigmahdq}
    \end{align}
Note $\sigma^2_{\hd,l} \geq 0 $ and $\sigma^2_{\hd,q} \geq 0 $ by~\Cref{prop:var-of-l-q-term}. We first prove that under Assumption~\ref{assumption:maximum-leverage-score-close-to-alpha},
    \begin{align}
    \label{eq:eq4}
            \sum_i \left(s_i(z)-\alpha\tilde{Y}_i(z)\right)^2 = o(n).
    \end{align}
    Applying the inequality
    $
    \sum_i (a_i-\bar{a})^2 \leq \sum_i a_i^2 
    $
    with $a_i = (H_{ii}-\alpha)\tilde{Y}_i(z)$,
    we get that 
\begin{align*}
    \sum_i \left(s_i(z)-\alpha\tilde{Y}_i(z)\right)^2 = \sum_i (a_i-\bar{a})^2 \leq  \sum_i a_i^2 = \sum_i (H_{ii}-\alpha)^2\tilde{Y}_i(z)^2,
\end{align*}
where the right-hand side is bounded by
\[
\sum_i (H_{ii}-\alpha)^2\tilde{Y}_i(z)^2 \le  \max_i |H_{ii}-\alpha|^2\cdot \sum_i \tilde{Y}_i(z)^2 = o(n)
\]
when $\max_i |H_{ii}-\alpha| = o(1)$ and $\sum_i \tilde{Y}_i(z)^2 = O(n)$, or bounded by 
\begin{align}
    \sum_i (H_{ii}-\alpha)^2\tilde{Y}_i(z)^2 &\le \bigg(\sum_i|\tilde{Y}_i(z)|^{2+\eta}\bigg)^{\frac{2}{2+\eta}}\bigg(\sum_i|H_{ii}-\alpha|^{\frac{2(2+\eta)}{\eta}} \bigg)^{\frac{\eta}{2+\eta}} \nonumber\\
    &< \bigg(\sum_i|\tilde{Y}_i(z)|^{2+\eta}\bigg)^{\frac{2}{2+\eta}}\bigg(\sum_i|H_{ii}-\alpha|^2 \bigg)^{\frac{\eta}{2+\eta}}=o(n)\label{ineq_case2}
\end{align}
when $\sum_i|H_{ii}-\alpha|^2  = o(n)$ and $\sum_i|\tilde{Y}_i(z)|^{2+\eta} = O(n)$. In the derivation of \eqref{ineq_case2}, the first inequality uses H{\"o}lder's inequality, while the second inequality is due to $\max_i |H_{ii}-\alpha| < 1$. 
In either case, we have proved \eqref{eq:eq4}, which implies that replacing $s_i(z)$ with $\alpha\tilde{Y}_i(z)$ in the formula of $\sigma^2_{\hd,l}$ leads to a negligible difference, i.e.,
\begin{align*}
    \sigma^2_{\hd,l} = \frac{1}{(n-1)(r_1r_0)}\sum_i \left\{r_0\alpha\tilde{Y}_i(1)+r_0e_i(1)+r_1\alpha \tilde{Y}_i(0)+r_1 e_i(0)\right\}^2 + o(1).
\end{align*}

Recall the definition of $d_i(z)$ at the first page of this supplement, we have $d_i(z) = \tilde{Y}_i(z)-e_i(z)$. 
Notice that $\sum_i e_i(z)d_i(z') = 0$ for $z,z'\in \{0,1\}$ and  
\begin{align*}
    &R^2 \sum_i \left(r_0\tilde{Y}_i(1)+r_1 \tilde{Y}_i(0)\right)^2 = \sum_i \left(r_0d_i(1)+r_1 d_i(0)\right)^2,\\
     &(1-R^2) \sum_i \left(r_0\tilde{Y}_i(1)+r_1 \tilde{Y}_i(0)\right)^2 = \sum_i \left(r_0e_i(1)+r_1 e_i(0)\right)^2,\\
     &\sigma^2_\cre = \frac{1}{(n-1)(r_1r_0)}\sum_i \left(r_0\tilde{Y}_i(1)+r_1 \tilde{Y}_i(0)\right)^2.
\end{align*}
With these identities, we derive that
\begin{align*}
    \sigma^2_{\hd,l} 
        &=\frac{1}{(n-1)(r_1r_0)}\sum_i \left(r_0(1+\alpha)e_i(1)+r_1(1+\alpha)e_i(0)+r_0\alpha d_i(1)+r_1\alpha d_i(0)\right)^2+o(1)\\
    &=\frac{1}{(n-1)(r_1r_0)}\left[(1+\alpha)^2\sum_i \left(r_0e_i(1)+r_1 e_i(0)\right)^2 + \alpha^2 \sum_i \left(r_0  d_i(1)+r_1 d_i(0)\right)^2 \right]+o(1)\\
    &=\left[(1+\alpha)^2-(1+2\alpha)R^2\right] \sigma^2_\cre +o(1).
\end{align*}
Therefore, we have  
\begin{align}
\label{eq:eq6}
    \sigma^2_\hd \geq \sigma_{\hd,l}^2 =  \left[(1+\alpha)^2-(1+2\alpha)R^2\right] \sigma^2_\cre +o(1),
\end{align}
which gives the lower bound on $\sigma^2_\hd$.

For the upper bounds on $\sigma^2_\hd$ and $\sigma^2_{\hd,q}$, applying the Cauchy-Schwarz inequality and the identity $\sum_j H_{ij}^2=({\mathbf H^2})_{ii}=H_{ii}$, we get that  
\begin{align*}
    \sum_{[i,j]}H_{ij}^2\left( \frac{\tilde{Y}_i(1)}{r_1^2}-\frac{\tilde{Y}_i(0)}{r_0^2}\right)\left(\frac{\tilde{Y}_j(1)}{r_1^2}-\frac{\tilde{Y}_j(0)}{r_0^2}\right) &\leq \sum_{[i,j]}H_{ij}^2\left( \frac{\tilde{Y}_i(1)}{r_1^2}-\frac{\tilde{Y}_i(0)}{r_0^2}\right)^2\\
    &= \sum_i (H_{ii}-H_{ii}^2)\left( \frac{\tilde{Y}_i(1)}{r_1^2}-\frac{\tilde{Y}_i(0)}{r_0^2}\right)^2.
\end{align*}
Plugging it into \eqref{sigmahdq} yields that
\begin{equation}\label{sigmahdq2}
    \sigma^2_{\hd,q}  \leq  \frac{2(r_1r_0)^2}{n-1}\sum_i (H_{ii}-H_{ii}^2)\left( \frac{\tilde{Y}_i(1)}{r_1^2}-\frac{\tilde{Y}_i(0)}{r_0^2}\right)^2. 
\end{equation}
Now, the upper bounds on $\sigma^2_\hd$ and $\sigma^2_{\hd,q}$ follows immediately from the estimates
\begin{align*}
   \sum_i (H_{ii}-\alpha)\left( \frac{\tilde{Y}_i(1)}{r_1^2}-\frac{\tilde{Y}_i(0)}{r_0^2}\right)^2 = o(n),\quad \sum_i (H_{ii}^2-\alpha^2)\left( \frac{\tilde{Y}_i(1)}{r_1^2}-\frac{\tilde{Y}_i(0)}{r_0^2}\right)^2 = o(n).
\end{align*}
It suffices to prove that under Assumption~\ref{assumption:maximum-leverage-score-close-to-alpha},
\begin{align}
\label{eq:eq5}
    M_1:=\sum_i |H_{ii}-\alpha|\tilde{Y}_i(z)^2 = o(n),\quad M_2:=\sum_i |H_{ii}^2-\alpha^2| \tilde{Y}_i(z)^2 = o(n).
\end{align}
for $z\in\{0,1\}$.

When $\max_i|H_{ii}-\alpha| = o(1)$ and $\sum_i \tilde{Y}_i(z)^2 = O(n)$, $M_1$ is bounded as 
\[
\max_i|H_{ii}-\alpha| \cdot \sum_i \tilde{Y}_i(z)^2=o(n).
\]
When $\sum_i|H_{ii}-\alpha|^2  = o(n)$ and $\sum_i|\tilde{Y}_i(z)|^{2+\eta} = O(n)$, $M_1$ is bounded as 
\begin{align*}
   &~ \bigg(\sum_i|\tilde{Y}_i(z)|^{2+\eta^\prime}\bigg)^{\frac{2}{2+\eta^\prime}}\bigg(\sum_i|H_{ii}-\alpha|^{\frac{2+\eta^\prime}{\eta^\prime}} \bigg)^{\frac{\eta^\prime}{2+\eta^\prime}} \\
   \le &~\bigg(\sum_i|\tilde{Y}_i(z)|^{2+\eta^\prime}\bigg)^{\frac{2}{2+\eta^\prime}}\bigg(\sum_i|H_{ii}-\alpha|^2 \bigg)^{\frac{\eta^\prime}{2+\eta^\prime}} \\
   \le &~n^{\frac{2}{2+\eta^\prime}}\left(\sum_i|\tilde{Y}_i(z)|^{2+\eta}/n\right)^{\frac{2}{2+\eta}} \bigg(\sum_i|H_{ii}-\alpha|^2 \bigg)^{\frac{\eta^\prime}{2+\eta^\prime}} =o(n),
\end{align*}
by using H{\"o}lder's inequality in the first two steps, where $\eta^\prime \in (0,\eta)$ is chosen to be a small constant such that $(2+\eta^\prime)/\eta^\prime >2$. 
To sum up, under Assumption~\ref{assumption:maximum-leverage-score-close-to-alpha}, we have $M_1 = o(n)$. 
The bound on $M_2$ then follows easily:
\[
M_2 \leq \max_i|H_{ii}+\alpha| \cdot \sum_i |H_{ii}-\alpha|\tilde{Y}_i(z)^2 \le 2M_1 = o(n).
\]

Finally, using \eqref{sigmahdq2} and \eqref{eq:eq5}, we obtain
\begin{align*}
    \sigma^2_{\hd,q}  &\leq  \frac{2(r_1r_0)^2}{n-1}\sum_i (\alpha-\alpha^2)\left( \frac{\tilde{Y}_i(1)}{r_1^2}-\frac{\tilde{Y}_i(0)}{r_0^2}\right)^2+o(1) \\
    &=  2(r_1r_0)^2 \alpha(1-\alpha) S_{r_1^{-2} Y(1) - r_0^{-2} Y(0)}^2+o(1).
\end{align*}
Together with \eqref{eq:eq6}, it concludes the proof.
\end{proof}

For the proof of  Proposition~\ref{proposition:compare-the-variance}, we need to use the following lemma, which is an i.i.d. version of Theorem~$2$ in \cite{whittle1960bounds}.

\begin{lemma}
    \label{lem:bound-for-the-moment}
    Let $\xi = (\xi_1,\ldots,\xi_n)$ be a random vector with centered i.i.d. entries. Let $\bs{A}$ be an arbitrary deterministic matrix. For any $s\geq2$, there exists a constant $C(s)$ depending on $s$ such that
    \begin{align*}
           \E\left|\bs{\xi}^\top \bs{A} \bs{\xi}-\E(\bs{\xi}^\top \bs{A} \bs{\xi})\right|^{s} \leq C(s) \left(\E |\bs{\xi}_1|^{2s}\right)^{1/2} \Big(\sum_{i,j} |A_{ij}|^2\Big)^{s/2}.
    \end{align*}   
\end{lemma}
The next lemma follows from a simple calculation.
\begin{lemma}
\label{lem:mean-of-quadra}
Let $\xi = (\xi_1,\ldots,\xi_n)$ be a random vector with centered i.i.d. entries. Let $\bs{A}$ be an arbitrary deterministic matrix. Then, we have
    \[
    \E(\bs{\xi}^\top \bs{A} \bs{\xi}) = \tr(A)\E\xi_1^2.
    \]
\end{lemma}
\begin{proof}[Proof of \Cref{lem:mean-of-quadra}]
By the mean zero and i.i.d.~conditions for the entries of $\xi$, we have
    \begin{align*}
        \E(\bs{\xi}^\top \bs{A} \bs{\xi}) =  \sum_{[i,j]} \E(\xi_i \xi_j A_{ij})+\sum_{i} \E(\xi_i^2 A_{ii}) = 0+ (\E\xi_1^2) \sum_{i} \E A_{ii} = \tr(A)\E\xi_1^2.
    \end{align*}
    This concludes the proof.
\end{proof}

\begin{proof}[Proof of Proposition~\ref{proposition:compare-the-variance}]
     We observe that
     \begin{equation}\label{eq:Stau-S}
         S^2_{\tau}-S^2_{e(1)-e(0)}-S^2_{\diag\{\bs{H}\},Y(1)-Y(0)} = S^2_{\bs{H},\tau} - S^2_{\diag\{\bs{H}\},Y(1)-Y(0)}.  
     \end{equation}
Through a direct calculation, we can write $S^2_{\bs{H},\tau}$ as 
\begin{align*}
     S^2_{\bs{H},\tau} &=  \frac{1}{n-1}\Big\{ \left(\bs{\varepsilon}(1)-\bs{\varepsilon}(0)\right)^\top\bs{H}\left(\bs{\varepsilon}(1)-\bs{\varepsilon}(0)\right)    +\left(\bs{\beta}_1 - \bs{\beta}_0\right)^\top \bs{X}^\top  \bs{P}\bs{X} \left(\bs{\beta}_1 - \bs{\beta}_0\right) \\
    &\qquad\qquad + 2\left(\bs{\varepsilon}(1)-\bs{\varepsilon}(0)\right)^\top \bs{P}\bs{X} \left(\bs{\beta}_1 - \bs{\beta}_0\right)\\
   &=: M_1+M_2+2M_3,
\end{align*}
where we denote $\bs{\varepsilon}(z)= (\varepsilon_1(z),\ldots, \varepsilon_n(z))^\top$, $\bs{P}= \bs{I} - \frac{1}{n}\bs{1}\bs{1}^\top$, and $\bs{X} = (\bs{X}_1,\ldots, \bs{X}_n)^\top$.
We next estimate the terms $M_i$, $i=1,2,3$, one by one. 

For $M_1$, applying Lemmas~\ref{lem:bound-for-the-moment} and \ref{lem:mean-of-quadra} with $s=2$, $\bs{A} = \bs{H}/(n-1)$, and $\xi_i = \varepsilon_i(1)-\varepsilon_i(0)$, and using the independence between $\bs{H}$ and $\bs{\varepsilon}(z)$, we obtain that 
     \begin{align*}
         \E (M_1|\bs{H})=\E  M_1  &= \frac{\tr(\bs{H})}{n-1}  \var(\varepsilon_1(1)-\varepsilon_1(0)) = \frac{p}{n-1}\left(\sigma_{\varepsilon(1)}^2+\sigma_{\varepsilon(0)}^2\right) \\
         &= \alpha \left(\sigma_{\varepsilon(1)}^2+\sigma_{\varepsilon(0)}^2\right) + O(n^{-1}),
     \end{align*}
     with $\sigma_{\varepsilon(z)}^2$, $z\in \{0,1\}$, denoting the variance of $\varepsilon_1(z)$, 
     and that
\begin{align*}
    \E\big(|M_1-\E(M_1|\bs{H})|^2~\big |\bs{H}\big)&\leq C(2) \left(\E\xi_1^4 \right)^{1/2} \frac{1}{(n-1)^2}\sum_{i,j} H_{ij}^2\\
    &= C(2) \left(\E\xi_1^4 \right)^{1/2}\frac{p}{(n-1)^2} = O(n^{-1}).
\end{align*}
Thus, by choosing $c_{n1} = \left[C(2) (\E\xi_1^4 )^{1/2}{p}(n-1)^{-2}\right]^{1/3}=o(1)$, we have
\begin{align*}
\mathbb{P}\big(|M_1-\E M_1|\geq c_{n1}\big |\bs{H}\big) &= \mathbb{P}\big(|M_1-\E(M_1|\bs{H})|\geq c_{n1}\big |\bs{H}\big) \\
&\leq  c_{n1}^{-2}\E\big[|M_1-\E(M_1|\bs{H})|^2\big | \bs{H} \big] \leq c_{n1}.
\end{align*}
Then, using the law of total expectation, we obtain that
\[
\mathbb{P}\big(|M_1-\E M_1|\geq c_{n1}\big) \leq c_{n1} = o(1).
\]
 
 For $M_2$, 
 applying Lemma~\ref{lem:bound-for-the-moment} with $s=2$, $\bs{A} = \bs{P}/(n-1)$, and $\xi_i = \bs{X}_i^\top (\bs{\beta}_1 - \bs{\beta}_0)$, we obtain that 
     \[
     \E M_2 = \frac{\tr(\bs{P})}{n-1} \E \big|\bs{X}_1^\top (\bs{\beta}_1 -  \bs{\beta}_0)\big|^2= \|\mathbf O^\top (\bs{\beta}_1-\bs{\beta}_0)\|_2^2. 
     \]
     Notice that due to the condition $\E|\bs{X}_1^\top\bs{\beta}_z|^{4}  < C$, $z\in \{0,1\}$, we have 
     \begin{equation}\label{eq:l22}
         \|\mathbf O^\top \bs{\beta}_z\|_2^2 = \E|\bs{X}_1^\top\bs{\beta}_z|^{2} \le \left(\E|\bs{X}_1^\top\bs{\beta}_z|^{4}\right)^{1/2}\le C^{1/2} .
     \end{equation}
    Next, applying Lemma~\ref{lem:mean-of-quadra}, we obtain that 
     \begin{align*}
         \E|M_2-\E M_2|^2 &\leq C(2)\left(\E\xi_1^4\right)^{1/2}\frac{1}{(n-1)^2}\sum_{i,j} P_{ij}^2  \\
         &= C(2)\left(\E\xi_1^4\right)^{1/2} (n-1)^{-1} = O(n^{-1}).
     \end{align*}
  Hence, by choosing $c_{n2} = \left[C(2)(\E\xi_1^4)^{1/2} (n-1)^{-1}\right]^{1/3}=o(1)$, we have
\[
\mathbb{P}\big(|M_2-\E M_2|\geq c_{n2}\big) \leq c_{n2} = o(1).
\]

For $M_3$, we observe that $\E M_3 = 0$ due to the independence between $\bs{X}$ and $\varepsilon(z)$. Denoting $\bs{\xi} = (\xi_1,\ldots,\xi_n)$ with $\xi_i = \bs{X}_i^\top (\bs{\beta}_1 -\bs{\beta}_0)$, 
 we obtain that
\begin{align*}
    \E M_3^2 &= \E \left[\frac{1}{n-1} (\varepsilon(1)-\varepsilon(0))\bs{P}\bs{\xi}\right]^2=(n-1)^{-2} \left(\sigma_{\varepsilon(1)}^2+\sigma_{\varepsilon(0)}^2\right) \E (\bs{\xi}^\top \bs{P}\bs{\xi})\\
    &=(n-1)^{-1} \left(\sigma_{\varepsilon(1)}^2+\sigma_{\varepsilon(0)}^2\right) \E M_2 = O(n^{-1}).
\end{align*}
where we used $\bs{P}^2 = \bs{P}$ and \eqref{eq:l22}. 
Then, we choose $c_{n3} = (\E M_3^2)^{1/3} = o(1)$ such that
\[
\mathbb{P}(|M_3|\geq c_{n3}) < c_{n3}^{-2}\E M_3^2=c_{n3} = o(1).
\]
     
     To sum up, we have shown that with probability $1-o(1)$,  
     \begin{equation}
     \label{eq:eq13}
     \left|S^2_{\bs{H},\tau}-\E S^2_{\bs{H},\tau}\right| \le c_{n1}+c_{n2}+2c_{n3} = o(1),
     \end{equation}
     where
     \begin{align*}
              \E S^2_{\bs{H},\tau} =  \alpha \left(\sigma_{\varepsilon(1)}^2+\sigma_{\varepsilon(0)}^2\right) + \|\mathbf O^\top (\bs{\beta}_1-\bs{\beta}_0)\|_2^2+O(n^{-1}). 
     \end{align*}

Next, we handle $S^2_{\diag\{\bs{H}\},Y(1)-Y(0)}$. It is easy to see that
    \begin{align}\label{eq:cn40}
           |S^2_{\diag\{\bs{H}\},Y(1)-Y(0)}- \alpha S^2_{Y(1)-Y(0)}| < \max_i |H_{ii}-\alpha| \cdot S^2_{Y(1)-Y(0)}.
    \end{align}
 By Proposition~\ref{proposition:4+eta-th-moment}, we have that with probability $1-o(1)$, 
    \begin{align}\label{eq:cn4}
    \max_i |H_{ii}-\alpha| < c_{n4}:=n^{-\delta},
    \end{align}
     for some constant $ \delta \in (0, \frac{\eta}{8 + 2 \eta})$. 
For $S^2_{Y(1)-Y(0)}$, applying Lemmas~\ref{lem:bound-for-the-moment} and \ref{lem:mean-of-quadra} with $s=2$, $\bs{A}=\bs{P}/(n-1)$ and $\xi_i = (Y_i(1)-Y_i(0))-(\mu_1-\mu_0)$, we obtain that
         \[
    \E S^2_{Y(1)-Y(0)}  = \frac{\tr(\bs{P})}{n-1} \mathbb E\xi_1^2 = \sigma_{\varepsilon(1)}^2+\sigma_{\varepsilon(0)}^2 + \|\mathbf O^\top (\bs{\beta}_1-\bs{\beta}_0)\|_2^2.
     \]
     and
     \begin{align*}
          \E|S^2_{Y(1)-Y(0)}-\E S^2_{Y(1)-Y(0)}|^2 &\leq C(2)\left(\E \xi_1^4\right)^{1/2}\frac{1}{(n-1)^2}\sum_i \sum_j P_{ij}^2\\
           &= C(2)\left(\E \xi_1^4\right)^{1/2} (n-1)^{-2} =  O(n^{-1}).
     \end{align*}
   Thus, by choosing $c_{n5}=\big[\E|S^2_{Y(1)-Y(0)}-\E S^2_{Y(1)-Y(0)}|^2\big]^{1/3} = o(1)$, we have that 
   \begin{equation}
   \label{eq:eq14}
            \mathbb P \left(|S^2_{Y(1)-Y(0)}-\E S^2_{Y(1)-Y(0)}| \ge c_{n5}\right) \le c_{n5}=o(1).
   \end{equation}
    Plugging \eqref{eq:cn4} and \eqref{eq:eq14} into \eqref{eq:cn40}, we obtain that with probability $1-o(1)$, 
    \begin{equation}
        \label{eq:eq15}
        \begin{split}
        &\left|S^2_{\diag\{\bs{H}\},Y(1)-Y(0)}- \alpha \left(\sigma_{\varepsilon(1)}^2+\sigma_{\varepsilon(0)}^2 + \|\mathbf O^\top (\bs{\beta}_1-\bs{\beta}_0)\|_2^2\right)\right| \\
          < &~c_{n4}\left(\sigma_{\varepsilon(1)}^2+\sigma_{\varepsilon(0)}^2 + \|\mathbf O^\top (\bs{\beta}_1-\bs{\beta}_0)\|_2^2+c_{n5}\right) + \alpha c_{n5}= o(1).
        \end{split}
    \end{equation}
   
    Finally, combining \eqref{eq:Stau-S}, \eqref{eq:eq13} and \eqref{eq:eq15}, we obtain that with probability $1-o(1)$,
    \[
    S^2_{\tau}-S^2_{e(1)-e(0)}-S^2_{\diag\{\bs{H}\},Y(1)-Y(0)} > (1-\alpha)\|\mathbf O^\top (\bs{\beta}_1-\bs{\beta}_0)\|_2^2  + o(1) \geq o(1).
    \]   
    The conclusion then follows.
\end{proof}

Next, we give the proof of \Cref{proposition:4+eta-th-moment}.
\begin{proof}[Proof of \Cref{proposition:4+eta-th-moment}]

For simplicity of notations, we denote  
\begin{equation}\label{eq:mWP}
    \mW:= n^{-1/2}\begin{pmatrix}\bXX_1,\cdots, \bXX_n\end{pmatrix}^\top,\quad  \bP:=\bs{I} - \frac{1}{n}\bs{1}\bs{1}^\top.
\end{equation} 
Then, we can write the matrix $\bs{H}$ as
\begin{equation}\label{eq:mH} 
\bs{H} = \bP\mW  \big(\mW^\top \bP\mW \big)^{-1} \mW^\top\bP.
\end{equation}

Now, we introduce a truncated matrix $\wt \bXX=(\wt \bXX_1, \ldots, \wt \bXX_n)^\top$ as 
\begin{equation}\label{truncateZ} 
\wt V_{ij}:= \mathbf 1\left( |V_{ij}|\le \varphi_n \log n\right)\cdot V_{ij},\quad \text{with}\quad \varphi_n:=n^{\frac{2}{4+\eta}},
\end{equation}
and denote $\wt\mW:= n^{-1/2}\begin{pmatrix}\wt\bXX_1,\cdots, \wt\bXX_n\end{pmatrix}^\top$. 
Combining the moment bound $\max_j\E |V_{ij}|^{4+\eta} <C$ with Markov's inequality, we obtain from a simple union bound that
\begin{equation}\label{XneX222}
\mathbb \mathbb{P}(\wt \bXX= \bXX) = 1- \Prob\left(\max_{i,j}|V_{ij}| > \varphi_n \log n \right)=1-\OO \left( (\log n)^{-(4+\eta)}\right).
\end{equation}
By definition, we have 
\begin{equation} \label{EwtZ}
\begin{split}
 \E  \wt  V_{ij} &= - \mathbb E \left[ \mathbf 1\left( |V_{ij}|> \varphi_n \log n \right)V_{ij}\right] ,\\  
\E  |\wt  V_{ij}|^2 &= 1 - \mathbb E \left[ \mathbf 1\left( |V_{ij}|> \varphi_n \log n \right)|V_{ij}|^2\right] .
\end{split}
\end{equation}
Using the tail probability expectation formula, we can check that
\begin{align*}
&  \mathbb E \left| \mathbf 1\left( |V_{ij}|> \varphi_n\log n \right)V_{ij}\right| = \int_0^\infty \Prob\left( \left| \mathbf 1\left(  |V_{ij}|> \varphi_n\log n \right)V_{ij}\right| > s\right)\dd s \\
& = \int_0^{\varphi_n\log n}\Prob\left( |V_{ij}|> \varphi_n\log n \right)\dd s +\int_{\varphi_n\log n}^\infty \Prob\left(|V_{ij}| > s\right)\dd s  \\
& \lesssim \int_0^{\varphi_n\log n}\left(\varphi_n\log n \right)^{-(4+\eta)}\dd s +\int_{\varphi_n\log n}^\infty s^{-(4+\eta)}\dd s \lesssim \left(\varphi_n\log n\right)^{-(3+\eta)}.
\end{align*}
Here, for simplicity of notations, given two quantities $a_n,b_n$ depending on $n$, we have used $a_n\lesssim b_n$ to mean that $|a_n|=\OO(|b_n|)$. Similarly, we have
\begin{align*}
&  \mathbb E \left| \mathbf 1\left( |V_{ij}|> \varphi_n \log n \right)V_{ij}\right|^2  =  2\int_0^\infty s \Prob\left( \left| \mathbf 1\left( |V_{ij}|>\varphi_n\log n \right)V_{ij}\right| > s\right)\dd s  \\
&=  2\int_0^{\varphi_n\log n} s \Prob\left( |V_{ij}|> \varphi_n\log n \right)\dd s +2\int_{\varphi_n\log n}^\infty s\Prob\left(|V_{ij}| > s\right)\dd s  \\
 & \lesssim \int_0^{\varphi_n\log n}s\left(\varphi_n\log n \right)^{-(4+\eta)}\dd s +\int_{\varphi_n\log n}^\infty s^{-(3+\eta)}\dd s \lesssim \left(\varphi_n\log n\right)^{-(2+\eta)}.
\end{align*}
From the above two estimates, we can derive that
\begin{align}\label{meanshif}
&|\mathbb E  \wt V_{ij}| \le n^{-3/2}, \quad  \mathbb E |\wt V_{ij}|^2 =1+ \OO(n^{-1}),\\
&\E  \|\wt \bXX-\bXX\|_F^2 = \sum_{i,j} \mathbb E \left| \mathbf 1\left( |V_{ij}|> \varphi_n \log n \right)V_{ij}\right|^2 \lesssim n^{\frac{4}{4+\eta}}(\log n)^{-(2+\eta)}. 
\end{align}
As a consequence, we get that 
\begin{equation}\label{eq:op_norm}
\|\mathbb E\wt \mW\|_F \le n^{-1/2}\Big(\sum_{i,j}|\mathbb E  \wt V_{ij}|^2\Big)^{1/2} \le n^{-1},  \quad \mathbb P \left( \|\wt \mW - \mW\|_F \ge n^{-\frac{\eta}{8+2\eta}}\right) =o(1). 
\end{equation}

Let $\bD$ be a $p\times p$ diagonal matrix with entries $D_{ii}=\var (\wt  V_{1i})$, $i\in [p]$. By \eqref{meanshif}, we have 
\begin{equation}\label{eq_Dii1}
   { \max_{i\in [n]}} |D_{ii}-1|=\OO(n^{-1}). 
\end{equation}
Now, we introduce the matrices $\bW:=(\wt\mW-\E\wt\mW)\bD^{-1/2}$ and 
\begin{equation}\label{eq:mH1}
\bs{H}' = \bP\bW \big(\bW^\top \bP \bW\big)^{-1} \bW^\top\bP.
\end{equation}
By definition and \eqref{eq_Dii1}, the entries of $\bW$ are independent random variables satisfying 
\begin{equation}\label{eq:centerW}
    \E \cal W_{ij}=0,\quad \E |\cal W_{ij}|^2=n^{-1},\quad |\cal W_{ij}|\le \frac{2\varphi_n\log n}{n^{1/2}},\quad i \in [n],\  j \in [p].
\end{equation}
Moreover, from \eqref{eq:op_norm}, we see that   
\begin{equation}\label{eq:op_norm2}
\mathbb P \left( \| \bW\bD^{1/2} - \mW\|_F \ge 2n^{-\frac{\eta}{8+2\eta}}\right) =o(1). 
\end{equation}
On the other hand, it is well-known that the empirical spectral distribution of $\mW \bP\mW^\top$ satisfies the famous Marchenko-Pastur (MP) law \citep{MP}, and their eigenvalues are all inside the support of the MP law, $[(1-\sqrt{\alpha})^2, (1+\sqrt{\alpha})^2],$ with high probability \citep{No_outside}. In particular, the following estimate is a direct consequence of the results in \cite{No_outside}: for any small constant $0<c<(1-\sqrt{\alpha})^2$, 
\begin{equation}\label{op rough2}  
\begin{split}
&\mathbb P \left(   (1-\sqrt{\alpha})^2 - c \le \lambda_{\min} (\mW^\top \bP \mW) \le \lambda_{\max} (\mW^\top \bP \mW) \le (1+\sqrt{\alpha})^2 + c\right) \\
&=1-o(1),
\end{split}
\end{equation}
where $\lambda_{\min}$ and $\lambda_{\max}$ denote the minimum and maximum eigenvalues, respectively. With \eqref{eq_Dii1}, \eqref{eq:op_norm2} and \eqref{op rough2}, we obtain the following two estimates: there exists a constant $C_1>0$ (depending on $\limsup \alpha$) such that 
\begin{equation}\label{op rough3}  
\mathbb P \left(   C_1^{-1} \le \lambda_{\min} (\bW^\top \bP \bW) \le \lambda_{\max} (\bW^\top \bP \bW) \le C_1\right) =1-o(1),
\end{equation}
and
\begin{equation}\label{eq:op_norm3}  
\mathbb P \left(  \|\bs{H}'-\bs{H}\|_2\ge C_1 n^{-\frac{\eta}{8+2\eta}}\right)\to 0.
\end{equation}
Since $|H_{ii}'-H_{ii}|\le \|\bs{H}'-\bs{H}\|_2$, to conclude the proof, it suffices to show that
\begin{equation}\label{eq:limH'}
    \mathbb{P}\left(\max_{i\in [n]}|H_{ii}'-\alpha|> n^{-\delta}\right) \to 0
\end{equation}
for any constant $0<\delta < \frac{\eta}{8+2\eta}$.

Let $\e_n=n^{-1/2}$. By \eqref{op rough3}, there exists a constant $C_2>0$ such that
\begin{equation}\label{eq:H'-H}
\mathbb P \left( \|\bs{H}' -\bs{H}_{\e}\|_2\ge C_2\e_n\right)={ o(1)},
\end{equation}
where $\bs{H}_{\e}$ is defined as\footnote{{\revone From here and below, we use ``$A^{-1}$'' and ``$\frac{1}{A}$'' interchangeably to denote the inverse of a matrix $A$. The latter notation is a common convention in random matrix theory, as it helps streamline the presentation of lengthy formulas involving matrix inverses (see e.g.~\citet{erdHos2017dynamical}).}}
\begin{equation}\label{eq:H'-H2}
\bs{H}_{\e}:=\bP\bW\frac{1}{(\bW^\top\bP)(\bP\bW) - \ii \e_n\bs{I}}\bW^\top\bP.
\end{equation}
Observe the following matrix identity
\begin{equation}\label{eq:H'-H3}
\bs{H}_{\e}= 1+\frac{\ii \e_n}{\bP\bW \bW^\top\bP - \ii \e_n \bs{I}}.
\end{equation}
Now, to conclude \eqref{eq:limH'}, it suffices to prove that 
\begin{equation}\label{eq:limvWv}
    \Prob\left[\left|\left( \frac{\ii \e_n}{\bP\bW \bW^\top \bP - \ii \e_n \bs{I}}\right)_{ii} + 1-\alpha\right| \ge n^{-\delta}\right] \le n^{-C}
\end{equation}
for any constant $0<\delta < \frac{\eta}{8+2\eta}$ and large constant $C>1$. Then, taking a simple union bound, we obtain that 
\begin{equation}\label{eq:limvWv2}
    \Prob\left[\max_{i\in [n]}\left|\left( \frac{\ii \e_n}{\bP\bW \bW^\top \bP - \ii \e_n \bs{I}}\right)_{ii} + 1-\alpha\right| \ge n^{-\delta}\right] \le n^{-(C-1)},
\end{equation}
which concludes \eqref{eq:limH'}.

For the proof of \eqref{eq:limvWv}, we will adopt Theorem 11.2 of \cite{Anisotropic}. More precisely, under the conditions on $\bW$ in \eqref{eq:centerW}, the following estimate was proved in Theorem 11.2 of \cite{Anisotropic}: for any deterministic unit vectors $\mathbf u,\mathbf v\in \mathbb R^n$, 
\begin{equation}\label{eq:locallaw}
  \Prob \left( \left| \mathbf u^\top \left(\frac{\ii \e_n}{\bP\bW \bW^\top \bP - \ii \e_n\bs{I}}\right)\mathbf v - \mathbf u^\top \left( {m(\ii \e_n)\bP - \frac{1}{n}\bs{1}_n\bs{1}_n^\top}\right)\mathbf v\right| \ge \frac{\varphi_n}{n^{\frac 1 2-c}}\right) \le n^{-C}  
\end{equation}
holds for any small constant $c>0$ and large constant $C>1$.
Here, $m(z)$ is the unique analytic function in a neighborhood around the origin that satisfies the equation
$$ -m(z) + \frac{p}{n}\frac{m(z)}{m(z)+z}=1. $$
In particular, from this equation, we can solve that that 
\begin{equation}\label{eq:solvem}
   | m(\ii \e_n)+(1-\alpha)|\le C_3\e_n
   \end{equation}
for a constant $C_3>0$. Plugging \eqref{eq:solvem} into \eqref{eq:locallaw} and taking $\mathbf u=\mathbf v=\bs{e}_{i,n}$, the $i$-th canonical basis unit vector, we obtain the estimate \eqref{eq:limvWv} since $\left(n^{-1}\bs{1}_n\bs{1}_n^\top\right)_{ii}=n^{-1}$. This concludes the proof.
\end{proof}

\begin{remark}
    An estimate of the form \eqref{eq:locallaw} is often called a \emph{local law} of $(\bP\bW \bW^\top\bP  - z\bs{I})^{-1}$, the Green's function of $\bP\bW \bW^\top\bP$. Such local laws of sample covariance matrices were also established in many other papers under different settings, see e.g., \cite{BMP_AOP,isotropic,BPZ_AOS,DY,VESD} (we remark that this list is far from being comprehensive). The setting in Theorem 11.2 of \cite{Anisotropic} is closest to our current one, but there is a minor difference that $|\cal W_{ij}|$ is of order $O(n^{-1/2+\e})$ in \cite{Anisotropic}. However, using the argument in \cite{DY}, it is rather straightforward to extend Theorem 11.2 of \cite{Anisotropic} to our setting with $|\cal W_{ij}|\le \frac{\varphi_n\log n}{n^{1/2}}$ in \eqref{eq:centerW} and conclude \eqref{eq:locallaw}. We omit the details here.
    
\end{remark}


{\rev 
Finally, we give the proof of \Cref{prop:justify-assumption-3}.

\begin{proof}[Proof of \Cref{prop:justify-assumption-3}]
For ease of notation, we write $Y_i(z)$ as $Y_i$ and $e_i(z)$ as $e_i$. By definition, we can write that
\[
e_i = Y_i-\bar{Y}-\bs{e}_{i,n}^\top\bs{H}\bs{Y}.
\]
Let $R_i = Y_i - \E Y_i$ and $\bs{R} = (R_1,\ldots,R_n)$. We can rewrite the above expression as
\[
e_i = R_i-\bar{R}-f_i,\]
where $f_i$ is defined as 
\[
f_i:= \bs{e}_{i,n}^\top\bs{H}\bs{R}=\bs{e}_{i,n}^\top\bP\mW\Big(\mW^\top\bP\mW\Big)^{-1}\mW^\top\bP\bs{R}.
\]
Here, we adopt the notations in \eqref{eq:mWP} and \eqref{eq:mH}. Since $\E |R_i|^{2+\eta} < C$, by Markov's inequality, we have that
\[
\Prob(\max_{i \in [n]} |R_i| > n^{\frac{2}{4+\eta}}) \leq n\Prob( |R_1| > n^{\frac{2}{4+\eta}}) \leq n\left(\E|R_1|^{2+\eta}\right) n^{-\frac{4+2\eta}{4+\eta}} = O(n^{-\frac{\eta}{4+\eta}}) . 
\]
In other words, with probability $1-o(1)$, 
\begin{align}
\label{eq:order-of-max-Ri-1}
    \max_{i \in [n]} |R_i| \leq n^{\frac{2}{4+\eta}}.
\end{align}
On the other hand, by Chebyshev's inequality, we have that 
\begin{align}
\label{eq:order-of-bar-R-1}
    |\bar{R}| \le n^{-(1/2-c)}
\end{align}
with probability $1-o(n^{-2c})$ for any constant $c>0$. It remains to control the size of $f_i$'s. 
This is done for the parts (i) and (ii) of \Cref{prop:justify-assumption-3} separately. 

\medskip
\noindent   \textbf{Proof of \Cref{prop:justify-assumption-3} (i)}:
     First, we notice that
     \begin{align*}
         |f_i| \leq (\sum_{i=1}^n H_{ij}^2)^{1/2}\|\bs{R}\|_2 = H_{ii}^{1/2}\|\bs{R}\|_2.
     \end{align*}
By the law of large numbers, we have that with probability $1-o(1)$,  
\begin{align}
\label{eq:order-of-R-2-norm}
   \|\bs{R}\|_2^2 \le \frac{3n}{2}\E R_1^2.
\end{align}
On the other hand, by \Cref{proposition:4+eta-th-moment}, we have that with probability $1-o(1)$,
\begin{align}
\label{eq:order-of-max-Hii}
    \max_{i\in [n]} |H_{ii}| \leq \frac{p}{n} + n^{-\delta}.
\end{align}
Combining equations \eqref{eq:order-of-R-2-norm} and \eqref{eq:order-of-max-Hii}, we obtain that with probability $1-o(1)$, 
\begin{align}
\label{eq:order-of-max-fi}
    \max_{i\in [n]} \Big|\frac{f_i}{\sqrt{n}}\Big| \leq C \Big\{\Big(\frac{p}{n}\Big)^{1/2} + n^{-\delta/2}\Big\}
\end{align}
\Cref{prop:justify-assumption-3} (i) then follows from \eqref{eq:order-of-max-Ri-1}, \eqref{eq:order-of-bar-R-1}, and \eqref{eq:order-of-max-fi}.

\medskip
\noindent   \textbf{Proof of \Cref{prop:justify-assumption-3} (ii)}:
Recall the truncated matrices $\wt \bXX$ and $\wt\mW$ defined around \eqref{truncateZ} and the matrices $\bW$ and $\bs{H}'$ defined around \eqref{eq:mH1}. 
Let 
$f_i' := \bs{e}_{i,n}^\top\bs{H}'\bs{R}$. 
We have
\[
\max_{i \in [n]} |f_i| \leq \max_{i \in [n]} |{f}_i'| + \max_{i \in [n]} |f_i-{f}_i'|.
\]
By \eqref{eq:op_norm3} and \eqref{eq:order-of-R-2-norm}, we have that with probability $1-o(1)$, 
\begin{align}
\label{eq:order-of-fi-fi-prime}
    \max_{i \in [n]} |f_i-{f}_i'| \le \|\bs{R}\|_2 \|\bs{H}'-\bs{H}\|_2 \le C n^{1/2} n^{-\frac{\eta}{8+2\eta}} = Cn^{\frac{2}{4+\eta}}
\end{align}
for a large constant $C>0$. It remains to bound $\max_{i \in [n]} |{f}_i'|$. 

We denote $\bW=(\bW_1,\ldots, \bW_n)^\top.$ Recall that the entries of $\bW$ are independent random variables satisfying \eqref{eq:centerW}. Moreover, by  \eqref{eq_Dii1}, we have that
\begin{align}
    \E|\cW_{ij}|^{4+\eta} &= D_{jj}^{-\frac{4+\eta}{2}}n^{-\frac{4+\eta}{2}}\E |\tilde{V}_{ij}-\E\tilde{V}_{ij}|^{4+\eta} \nonumber\\
    &\leq 2n^{-\frac{4+\eta}{2}}\cdot 2^{3+\eta} \Big(\E |\tilde{V}_{ij}|^{4+\eta}+|\E\tilde{V}_{ij}|^{4+\eta}\Big) \nonumber\\
    & \leq  2^{4+\eta} n^{-\frac{4+\eta}{2}} \Big(\E |V_{ij}|^{4+\eta}+n^{-\frac{3(4+\eta)}{2}}\Big) \le C n^{-\frac{4+\eta}{2}}\label{eq:momentUU}
\end{align}
for some large constant $C>0$. Now, we express $f_i'$ as
\begin{align*}
    f_i' =&  (\bW_i-\bar{\bW})^\top \big(\bW^\top\bP\bW\big)^{-1} \sum_j (\bW_j-\bar{\bW}) R_j\\
        =& (\bW_i-\bar{\bW})^\top \big(\bW^\top\bP\bW\big)^{-1}  \Big(\sum_j \bW_jR_j-n\bar{\bW}\bar{R}\Big).
\end{align*}
It follows that
\begin{align*}
   \max_{i\in [n]} |f_i'| \leq \Delta_1 + \Delta_2  + \max_{i \in [n]} |{g}_i|,
\end{align*}
where $\Delta_1$, $\Delta_2$, and $g_i$ are defined as 
\begin{align*}
    &\Delta_1 := \max_{i \in [n]}\Big|f_i'-\bW_i^\top\big(\bW^\top\bP\bW\big)^{-1} \sum_j \bW_jR_j\Big|,\\
    &\Delta_2 := \max_{i \in [n]}\Big|\bW_i^\top\big\{\big(\bW^\top\bP\bW\big)^{-1}- \big(\bW^\top\bW\big)^{-1}\big\} \sum_j \bW_jR_j\Big|, \\
    & {g}_i := \bW_i^\top\big(\bW^\top\bW\big)^{-1} \sum_j \bW_jR_j.
\end{align*}
To conclude the proof, it suffices to prove the following estimates: with probability $1-o(1)$,
\begin{align}
\label{eq:order-of-Delta-1-Delta-2-max-gi}
    \Delta_1  \leq C n^{1/4},\quad \Delta_2 \leq C n^{1/4},
   \quad \max_{i \in [n]} |{g}_i| \leq C n^{1/4}.
\end{align}
\Cref{prop:justify-assumption-3} (ii) then follows immediately from \eqref{eq:order-of-max-Ri-1}, \eqref{eq:order-of-bar-R-1}, \eqref{eq:order-of-fi-fi-prime} and \eqref{eq:order-of-Delta-1-Delta-2-max-gi}.

To show \eqref{eq:order-of-Delta-1-Delta-2-max-gi}, we will use the following two lemmas, whose proofs are postponed until we complete the proof of \Cref{prop:justify-assumption-3} (ii).
\begin{lemma}
\label{lem:orders-of-some-quantities}
    Under the assumptions of \Cref{prop:justify-assumption-3} (ii), we have that with probability $1- O(n^{-\eta/4})$,
\begin{equation}\label{eq:Wil2}
\max_{i \in [n]}\|\bW_i\|_2^2 \leq  2,
\end{equation}
and that with probability $1-O(n^{-(1+\eta/4)})$, 
\begin{equation}\label{eq:Rl2}
\|\bs{R}\|_2 \leq Cn^{1/2},\quad |\bar{R}| \leq n^{-1/4},
\end{equation}
for some large constant $C>0$. 
\end{lemma}

\begin{lemma}
\label{lem:bound-for-eigenvalues-of-some-random-matrices}
Under the assumptions of \Cref{prop:justify-assumption-3} (ii), there exists a large constant $C_1>1$ such that for any constant $C>1$,
\begin{align}
  &C_1^{-1}\leq\lambda_{\min} (\bW^\top\bP\bW) \leq \lambda_{\max} (\bW^\top\bP\bW) \leq C_1, \label{op WPW}\\ 
  &C_1^{-1}\leq\lambda_{\min} (\bW^\top\bW) \leq \lambda_{\max} (\bW^\top\bW) \leq C_1,\label{op WW}
\end{align}
with probability $1-O(n^{-C})$. (Note the probability bound here is stronger than that in \eqref{op rough3}.)
\end{lemma}

By the above two lemmas, there exists a constant $C>1$ such that with probability $1-O(n^{-(1+\eta/4)})$,
\begin{align}
\label{eq:order-of-bar-U-22}
    &\|\bar{\bW}\|_2^2 = n^{-2}\bs{1}^\top\bW\bW^\top\bs{1} \le n^{-1} \lambda_{\max}(\bW\bW^\top) \le   Cn^{-1},\\
    &\Big\|\sum_j \bW_jR_j\Big\|_2^2 = \bs{R}^\top \bW\bW^\top \bs{R} \le  \lambda_{\max}(\bW\bW^\top) \|\bs{R}\|_2^2 \le Cn\label{eq:order-of-sum-U-R}.
\end{align}
Now, combining Lemmas \ref{lem:orders-of-some-quantities} and \ref{lem:bound-for-eigenvalues-of-some-random-matrices} with equations \eqref{eq:order-of-bar-U-22} and \eqref{eq:order-of-sum-U-R}, we can bound that with probability $1-o(1)$, 
\begin{align}
    \Delta_1 & \le  \max_{i \in [n]}\|\bW_i\|_2\big\|\big(\bW^\top\bP\bW\big)^{-1}\big\|_2  n \|\bar{\bW}\|_2|\bar{R}| + \|\bar{\bW}\|_2\big\|\big(\bW^\top\bP\bW\big)^{-1}\big\|_2\big\|\sum_j \bW_jR_j\big\|_2\nonumber\\
    &\quad +  \|\bar{\bW}\|_2\big\|\big(\bW^\top\bP\bW\big)^{-1}\big\|_2 n\|\bar{\bW}\|_2|\bar{R}|\nonumber\\
    & = O(n \cdot n^{-1/2}\cdot n^{-1/4}+ n^{-1/2}\cdot n^{1/2}+n^{-1/2}\cdot n \cdot n^{-1/2}\cdot n^{-1/4}) = O(n^{1/4}) .\nonumber
\end{align}
This concludes the first estimate in \eqref{eq:order-of-Delta-1-Delta-2-max-gi}.

For the term $\Delta_2$, noticing $\bW^\top\bP\bW = \bW^\top\bW-n\bar{\bW}\bar{\bW}^\top$ and using the Sherman–Morrison formula, we obtain that
\begin{align*}
    &\bW_i^\top\Big\{\big(\bW^\top\bP\bW\big)^{-1}- \big(\bW^\top\bW\big)^{-1}\Big\} \sum_j \bW_jR_j \\
    =& \frac{n\bW_i^\top (\bW^\top\bW )^{-1}\bar{\bW} \bar{\bW}^\top (\bW^\top\bW )^{-1}\sum_j \bW_jR_j}{1-n\bar{\bW}^\top (\bW^\top\bW )^{-1}\bar{\bW}}.
\end{align*}
Thus, we can bound $\Delta_2$ by 
\begin{align}\label{eq:boundDelta2}
    \Delta_2 \leq \frac{n\max_{i \in [n]}|\bW_i^\top(\bW^\top\bW)^{-1}\bar{\bW}||\bar{\bW}^\top(\bW^\top\bW)^{-1}\sum_j \bW_jR_j|}{|1-n\bar{\bW}^\top(\bW^\top\bW)^{-1}\bar{\bW}|}.
\end{align}
To control the three factors on the RHS, we need the following two lemmas, whose proofs are postponed until we complete the proof of \Cref{prop:justify-assumption-3} (ii). 

\begin{lemma}
\label{lem:order-of-1-n-barU-UU-barU}
Under the assumptions of \Cref{prop:justify-assumption-3} (ii), there exists a constant $c>0$ such that for any constant $C>1$,
 \begin{align}
 \big|1-n\bar{\bW}^\top\big(\bW^\top\bW\big)^{-1}\bar{\bW}\big|\ge c \label{eq:H71}
\end{align}
with probability $1-O(n^{-C})$.
\end{lemma}

\begin{lemma}
\label{lem:order-of-maximum-of-independent-product}
Under the assumptions of \Cref{prop:justify-assumption-3} (ii), with probability $1-o(1)$, 
 \begin{align}
    & \max_{i \in [n]}|\bW_i^\top(\bW_{(-i)}^\top\bW_{(-i)})^{-1}\sum_{j\ne i} \bW_jR_j|\le n^{1/4},\label{eq:H81}\\
    & \max_{i \in [n]}|\bW_i^\top(\bW_{(-i)}^\top\bW_{(-i)})^{-1}\sum_{j\ne i} \bW_j|\le n^{1/4} ,\label{eq:H82}
\end{align}
where $\bW_{(-i)}$ is the $(n-1)\times p$ matrix obtained by removing the $i$-th row from $\bW$. 
\end{lemma}

We now bound \eqref{eq:boundDelta2}. First, we have that
\begin{align}
n\max_{i \in [n]}\big|\bW_i^\top\big(\bW^\top\bW\big)^{-1}\bar{\bW}\big| & \leq \max_{i \in [n]}\big|\bW_i^\top\big(\bW^\top\bW\big)^{-1}\bW_i\big| \nonumber\\
&+ \max_{i \in [n]}\big|\bW_i^\top\big(\bW^\top\bW\big)^{-1}\sum_{j\ne i}\bW_j\big|.\label{eq:WWWW}
\end{align}
By Lemmas \ref{lem:orders-of-some-quantities} and \ref{lem:bound-for-eigenvalues-of-some-random-matrices}, the first term on the RHS satisfies that with probability $1-o(1)$,
\begin{align}
\label{eq:order-of-max-Ui-U-U-Ui}
\max_{i \in [n]}\big|\bW_i^\top\big(\bW^\top\bW\big)^{-1}\bW_i\big|  \leq \max_{i\in [n]}\|\bW_i\|_2^2 \big\|\big(\bW^\top\bW\big)^{-1}\big\|_2 \le 2 C_1  .
\end{align}
For the second term on the RHS of \eqref{eq:WWWW}, we use the Sherman–Morrison formula to write it as
\begin{align*}
    \bW_i^\top(\bW^\top\bW)^{-1}\sum_{j\ne i} \bW_j &=  \frac{\bW_i^\top(\bW_{(-i)}^\top\bW_{(-i)})^{-1}\sum_{j\ne i} \bW_j}{1+\bW_i^\top(\bW_{(-i)}^\top\bW_{(-i)})^{-1}\bW_i^\top}.
\end{align*}
Then, using \Cref{lem:order-of-maximum-of-independent-product}, we conclude from this equation that with probability $1-o(1)$,
\begin{align}\label{eq:bddW}
\max_{i \in [n]}\big|\bW_i^\top\big(\bW^\top\bW\big)^{-1}\sum_{j\ne i}\bW_j\big| \leq \max_{i \in [n]}\big|\bW_i^\top(\bW_{(-i)}^\top\bW_{(-i)})^{-1}\sum_{j\ne i} \bW_j\big| \leq n^{1/4}.
\end{align}
Plugging \eqref{eq:order-of-max-Ui-U-U-Ui} and \eqref{eq:bddW} into \eqref{eq:WWWW}, we obtain that with probability $1-o(1)$,
\begin{align}
    \label{eq:order-of-quantities-Delta2-1}
n\max_{i \in [n]}\big|\bW_i^\top\big(\bW^\top\bW\big)^{-1}\bar{\bW}\big| = O(n^{1/4}). 
\end{align}
Next, by \Cref{lem:bound-for-eigenvalues-of-some-random-matrices} and equations \eqref{eq:order-of-bar-U-22} and \eqref{eq:order-of-sum-U-R}, we have that with probability $1-o(1)$,
\begin{align}
 \label{eq:order-of-quantities-Delta2-2}   \big|\bar{\bW}^\top(\bW^\top\bW)^{-1}\sum_j \bW_jR_j\big| \leq \|\bar{\bW}\|_2 \big\|\big(\bW^\top\bW\big)^{-1}\big\|_{2}\big\|\sum_j \bW_jR_j\big\|_2 =O(1).
\end{align}
Applying \eqref{eq:order-of-quantities-Delta2-1}, \eqref{eq:order-of-quantities-Delta2-2}, and \eqref{eq:H71} to \eqref{eq:boundDelta2}, we conclude the second estimate in \eqref{eq:order-of-Delta-1-Delta-2-max-gi}.

Finally, we bound $\max_{i\in[n]}|g_i|$:
\begin{equation}\label{eq:bddgi}
\max_{i \in [n]} |g_i| \leq \max_{i \in [n]}|\bW_i^\top(\bW^\top\bW)^{-1}\sum_{j\ne i} \bW_jR_j| + \max_{i \in [n]}|\bW_i^\top(\bW^\top\bW)^{-1} \bW_iR_i|.
\end{equation}
For the first term on the RHS, a similar argument as in \eqref{eq:bddW} based on the Sherman–Morrison formula and \Cref{lem:order-of-maximum-of-independent-product} yields that 
\begin{align}\label{eq:bddWR}
\max_{i \in [n]}\big|\bW_i^\top\big(\bW^\top\bW\big)^{-1}\sum_{j\ne i}\bW_jR_j\big| \leq n^{1/4}
\end{align}
with probability $1-o(1)$. For the second term on the RHS of \eqref{eq:bddgi}, using \eqref{eq:order-of-max-Ui-U-U-Ui}, we get that
\begin{align}\label{eq:bddWR2}
\max_{i \in [n]}|\bW_i^\top(\bW^\top\bW)^{-1} \bW_iR_i| \leq 2C_2 \max_{i \in [n]} |R_i|
\end{align}
with probability $1-o(1)$. Since $\E|R_i|^{4+\eta} =O(1)$ by assumption, using the Markov's inequality and a union bound argument, we get that with probability $1-O(n^{-\eta/4})$,
\begin{align}
\label{eq:order-of-max-Ri-2}
    \max_{i \in [n]} |R_i| < n^{1/4}.
\end{align}
Plugging \eqref{eq:order-of-max-Ri-2} into \eqref{eq:bddWR2}, we get that with probability $1-o(1)$,
\[
\max_{i \in [n]}|\bW_i^\top(\bW^\top\bW)^{-1} \bW_iR_i| =O(n^{1/4}).
\]
Together with \eqref{eq:bddWR}, it concludes the last estimate in \eqref{eq:order-of-Delta-1-Delta-2-max-gi}.
\end{proof}

For the proof of \Cref{lem:orders-of-some-quantities}, we will use the following lemma, which is stated as \cite[Theorem~$2$]{whittle1960bounds}.
\begin{lemma}
    \label{lem:bound-for-the-moment-of-independent-sum}
    Let $\bs{\xi} = (\xi_1,\ldots,\xi_q)$ be a random vector with independent entries. Let  $\bs{\nu} = (\nu_1,\ldots,\nu_q)$ be an arbitrary deterministic vector.  If $\E |\xi_j|^s$, $j\in [q]$, exists for some $s\geq 2$, there exists a constant $C(s)$ depending only on $s$ such that
    \begin{align*}
           \E|\bs{\nu}^\top\bs{\xi}|^{s} &\leq C(s) \Big\{\sum_{j=1}^q \nu_j^2 (\E |\xi_j|^s)^{2/s}\Big\}^{s/2}.
    \end{align*}
\end{lemma}

\begin{proof}[Proof of \Cref{lem:orders-of-some-quantities}]
We first write that 
\begin{align*}
    \max_{i \in [n]} \|\bW_i\|_2^2 \le \frac{p}{n} + \max_{i \in [n]}\Big| \sum_{j\in [p]} (\cW_{ij}^2-n^{-1})\Big|. 
\end{align*}
Using \eqref{eq:momentUU}, we obtain that 
\[
\max_{j \in [p]}\E \left[|\cW_{ij}^2-n^{-1}|^{(4+\eta)/2} \right]\leq \max_{j \in [p]} 2^{(4+\eta)/2-1}(\E |\cW_{ij}|^{4+\eta}+n^{-(4+\eta)/2}) =O(n^{-(4+\eta)/2}).
\]
Now, applying \Cref{lem:bound-for-the-moment-of-independent-sum} with $\bs{\xi} = (\cW_{i1}^2-n^{-1},\ldots,\cW_{ip}^2-n^{-1})$ and $s=(4+\eta)/2$, we obtain that 
\begin{align*}
    \E \Big| \sum_{j\in [p]} (\cW_{ij}^2-n^{-1})\Big|^{(4+\eta)/2} \lesssim  p^{1+\eta/4} \max_{j \in [p]} \E |\xi_j|^{(4+\eta)/2} \lesssim n^{-(1+\eta/4)}.
\end{align*}
By Markov's inequality and a union bound, we have that with probability $1-O(n^{-\eta/4})$,
\[
\max_{i\in [n]}\Big| \sum_{j\in [p]} (\cW_{ij}^2-n^{-1})\Big| \le 1,
\]
which yields that
\[
\max_{i \in [n]}\|\bW_i\|_2^2 \le  1+\frac{p}{n} \le 2. 
\]
This concludes the estimate \eqref{eq:Wil2}. 
With a similar argument based on \Cref{lem:bound-for-the-moment-of-independent-sum}, we get that 
     \begin{align*}
  \E \Big|\sum_i R_i^2 - n\E R_1^2\Big|^{(4+\eta)/2} &\lesssim n^{1+\eta/4}\max_{i\in[n]}\E\left|R_i^2-\E R_i^2\right|^{(4+\eta)/2}  \lesssim n^{1+\eta/4}.
    \end{align*}
Applying Markov's inequality, we conclude that 
   \[
   \Prob\Big(\Big|\sum_i R_i^2  - n\E R_1^2\Big|\ge n\Big)=O(n^{-(1+\eta/4)}).
   \]
 This gives the first estimate in \eqref{eq:Rl2}.   
 For the second estimate, using \Cref{lem:bound-for-the-moment-of-independent-sum} again with $\bs{\xi} = (R_1,\ldots,R_n)$ and $s=4+\eta $, we obtain that 
   \[
   \E|\bar{R}|^{4+\eta} \lesssim n^{-(4+\eta)/2} \max_{i\in[n]}\E |R_i|^{4+\eta} \lesssim n^{-(4+\eta)/2}.
   \]
   Applying Markov's inequality, we conclude the second estimate in \eqref{eq:Rl2}.      
\end{proof}

\begin{proof}[Proof of \Cref{lem:bound-for-eigenvalues-of-some-random-matrices}] 
When $\eta$ is arbitrarily large, the estimates \eqref{op WPW} and \eqref{op WW} are immediate consequences of Theorem 11.3 and Theorem 3.12 of \cite{Anisotropic}, respectively. 
The estimate \eqref{op WW} for the general case with an arbitrary (small) constant $\eta>0$ follows from Lemma 3.11 of \cite{DY}.
The upper bound in \eqref{op WPW} is a trivial consequence of that in \eqref{op WW}. For the lower bound, using the Sherman–Morrison formula
$$  \big(\bW^\top\bP\bW\big)^{-1} = \big(\bW^\top\bW\big)^{-1} + \frac{n(\bW^\top\bW )^{-1}\bar{\bW} \bar{\bW}^\top (\bW^\top\bW )^{-1}}{1-n\bar{\bW}^\top (\bW^\top\bW )^{-1}\bar{\bW}},$$
we obtain that 
\begin{align*}
   & \big\|\big(\bW^\top\bP\bW\big)^{-1}\big\|_2\le \big\|\big(\bW^\top\bW\big)^{-1}\big\|_2 \left( 1+\frac{n\|\bar\bW\|_2^2\|(\bW^\top\bW)^{-1}\|_2}{|1-n\bar{\bW}^\top\big(\bW^\top\bW\big)^{-1}\bar{\bW}|}\right).
\end{align*} 
Then, with \Cref{lem:order-of-1-n-barU-UU-barU} and equations \eqref{op WW} and \eqref{eq:order-of-bar-U-22}, we conclude that 
$$\lambda_{\min} (\bW^\top\bP\bW) = \big\|\big(\bW^\top\bP\bW\big)^{-1}\big\|_2 =O(1)$$
with probability $1-O(n^{-C})$ for any constant $C>1$. This gives the lower bound in \eqref{op WPW}.
\end{proof}
\begin{proof}[Proof of \Cref{lem:order-of-1-n-barU-UU-barU}]
Let $\e_n=n^{-1/2}$ and define 
$$ \bcH:=\bW (\bW^\top \bW)^{-1}\bW^\top,\quad \bcH_{\e}:=\bW\frac{1}{\bW^\top \bW - \ii \e_n\bs{I}}\bW^\top.$$
By \eqref{op WW}, there exists a constant $C_2>0$ such that
\begin{equation}\label{eq:H'-H3.5}
\mathbb P \left( \|\bcH -\bcH_{\e}\|_2\ge C_2\e_n\right)=O(n^{-C}) 
\end{equation}
for any constant $C>1$. Similar to \eqref{eq:H'-H3}, we have the following matrix identity
\begin{equation}\label{eq:H'-H4}
\bcH_{\e}= 1+\frac{\ii \e_n}{\bW \bW^\top - \ii \e_n \bs{I}}.
\end{equation}
A similar local law as in \eqref{eq:locallaw} has been established in Theorem 3.6 of \cite{Anisotropic}:
for any deterministic unit vectors $\mathbf u,\mathbf v\in \mathbb R^n$, 
\begin{equation}\label{eq:locallaw2}
  \Prob \left( \left| \mathbf u^\top \left(\frac{\ii \e_n}{\bW \bW^\top - \ii \e_n\bs{I}}\right)\mathbf v - m(\ii \e_n)  \mathbf u^\top \mathbf v\right| \ge \frac{\varphi_n}{n^{\frac 1 2-c}}\right) \le n^{-C}  
\end{equation}
holds for any small constant $c>0$ and large constant $C>1$. Taking $\mathbf u=\mathbf v= \mathbf 1_n/\sqrt{n}$ in the above estimate and using \eqref{eq:solvem} and \eqref{eq:H'-H4}, we obtain that 
$$\left| n^{-1}\mathbf 1_n^\top \bcH_{\e}\mathbf 1_n -\alpha\right| \le \frac{\varphi_n}{n^{\frac 1 2-c}} + C_3\e_n$$
with probability $1-O(n^{-C})$. Together with \eqref{eq:H'-H3.5}, it implies that 
$$ \left|n\bar{\bW}^\top\big(\bW^\top\bW\big)^{-1}\bar{\bW} -\alpha\right|=\left|n^{-1}\mathbf 1_n^\top \bcH\mathbf 1_n -\alpha\right|\le \frac{\varphi_n}{n^{\frac 1 2-c}} + (C_2+C_3)\e_n$$
with probability $1-O(n^{-C})$. This concludes the proof since $\alpha\le 1-c$ for a constant $c>0$ by the given assumption. 
\end{proof}

\begin{proof}[Proof of \Cref{lem:order-of-maximum-of-independent-product}]
By \Cref{lem:bound-for-eigenvalues-of-some-random-matrices} (note $\bW_{(-i)}$ satisfies the same assumptions as $\bW$ with $n$ replaced by $n-1$) and equation \eqref{eq:order-of-sum-U-R}, there exists a constant $C>0$ such that
\[
\Big\|(\bW_{(-i)}^\top\bW_{(-i)})^{-1}\sum_{j\ne i} \bW_jR_j\Big\|_2\leq \big\|\big(\bW_{(-i)}^\top\bW_{(-i)}\big)^{-1}\big\|_2\Big\|\sum_{j\ne i} \bW_jR_j\Big\|_2 \le Cn^{1/2}
\]
with probability $1-O(n^{-(1+\eta/4)})$. Therefore, we have that
\begin{align}\label{eq:bddPn1/4}
    &\Prob\Big(\big|\bW_i^\top \bs{\nu}_i\big|>n^{1/4}\Big) \le \Prob\Big(\big|\bW_i^\top\bs{\nu}_i\big|>n^{1/4} \Big| \|\bs{\nu}_i\|_2 \leq C n^{1/2} \Big)+ O(n^{-(1+\eta/4)}),
\end{align}
where we denote $\bs{\nu}_i: = (\bW_{(-i)}^\top\bW_{(-i)})^{-1}\sum_{j\ne i} \bW_jR_j$. Note that $\bs{\nu}_i$ is independent of $\bW_i$. Hence, applying Lemma~\ref{lem:bound-for-the-moment-of-independent-sum} with $\bs{\xi}=(\cW_{i1},\ldots, \cW_{ip})$ and $s=4+\eta$, we obtain that  
 \[
 \E\Big(\big|\bW_i^\top\bs{\nu}_i\big|^{4+\eta} \Big|\bs{\nu}_i\Big) \lesssim  \|\bs{\nu}_i\|_2^{4+\eta} \cdot \max_{j \in [p]}\E|\cW_{ij}|^{4+\eta} .
 \]
By \eqref{eq:momentUU}, we see that conditioning on $\|\bs{\nu}_i\|_2 \leq Cn^{1/2}$,
$$\E\Big(\big|\bW_i^\top\bs{\nu}_i\big|^{4+\eta} \Big|\|\bs{\nu}_i\|_2 \leq C n^{1/2} \Big)\lesssim 1, $$
which, combined with Markov's inequality, implies that 
$$\Prob\Big(\big|\bW_i^\top\bs{\nu}_i\big|>n^{1/4} \Big| \|\bs{\nu}_i\|_2 \leq C n^{1/2} \Big)=O(n^{-(1+\eta/4)}).$$
Plugging it into \eqref{eq:bddPn1/4} and applying a union bound, we obtain that 
$$\Prob\Big(\max_{i\in [n]}\big|\bW_i^\top \bs{\nu}_i\big|>n^{1/4}\Big) =O(n^{-\eta/4}) .$$
This concludes \eqref{eq:H81}. The estimate \eqref{eq:H82} can be proved in the same way.
\end{proof}}

\section{Additional numerical experiments}\label{sec:addsimu}

In this section, we conduct additional simulation analysis to examine the finite sample performance of the proposed estimator and inference procedure. In the main text, we consider the setup that $\bs{\mathcal{X}}$ and $\bs{\check{\varepsilon}}(z)$ (which is used to generate the independent $t$ residual) have i.i.d. entries from $t$ distribution with $3$ degrees of freedom. Here, we consider $2$ more setups:
\begin{itemize}
    \item $\bs{\mathcal{X}}$ and $\bs{\check{\varepsilon}}(z)$ have i.i.d. entries from Cauchy distribution
    \item $\bs{\mathcal{X}}$ have i.i.d. entries from Cauchy distribution and $\bs{\check{\varepsilon}}(z)$ have i.i.d. entries from $t$ distributions with degrees of freedom $3$. We also modify the model to
    \begin{align*}
    &Y_i(1) = \mu_1 + \sca(\tra(\bs{X}_i^\top\bs{\beta}_1)) + \varepsilon_i(1)/\sqrt{\gamma}, \\
    &Y_i(0) = \mu_0+ \sca(\tra(\bs{X}_i^\top\bs{\beta}_0)) + \varepsilon_i(0)/\sqrt{\gamma},
\end{align*}
where for a finite population $\{a_i\}_{i=1}^n$: 
\[
\tra(a_i) = b_{(\pi(i))},
\]
 $b_{(1)}\leq b_{(2)}\leq\ldots \leq b_{(n)}$ is the ordered sequence of $\{b_i\}_{i=1}^n$ with $b_i$ generated from $t$ distribution with degrees of freedom $3$, and $\pi(i)$ is the rank of $a_i$.
\end{itemize}
For both setups, we consider the same factorial experiments regarding $\gamma$, $\delta$, $\alpha$, and the generating models of $\varepsilon_i(z)$. Note that the first setup represents the most challenging case in which \Cref{assumption:2-moment-of-finite-population-for-y}--\ref{assumption:lindeberg-type-condition-p-o(n)} fail. In the second setup, albeit with extremely heavy-tail covariates, $Y_i(z)$ have bounded $3$rd moment, the \Cref{assumption:2-moment-of-finite-population-for-y}--\ref{assumption:lindeberg-type-condition-p-o(n)} hold.  \Cref{fig:bias-cauchy}--\ref{fig:ci_cauchy} show the results for the first setup. \Cref{fig:bias-cauchy-transy}--\ref{fig:ci_cauchy_transy} show the results for the second setup.

For the first setup, $(\htaudb,\hat{\sigma}_{\hd,\textrm{cb}}^2)$ outperforms its competitors in terms of relative RMSE, relative bias, more reliable inference and shorter confidence intervals in all cases, except under the independent $t$ residual with $\gamma=3$ and $\alpha \leq 0.1$. The performance of $\hat{\tau}_\lin$ can be more catastrophic than in the main text when $\alpha$ is large. For example, the relative RMSE can be as large as $40$. Interestingly, although our asymptotic theory does not apply to these extreme regimes, in most of the cases the relative RMSE and relative confidence interval length produced by our debiased estimator is not too far away from $1$. In other words, $(\htaudb,\hat{\sigma}_{\hd,\textrm{cb}}^2)$ does not give significant harm compared to 
without covariate adjustment in these extreme setups. This demonstrates the robustness of our method when faced with extreme cases.

For the second setup, $\htaudb$ outperforms other competitors for smaller relative bias and relative RMSE under the worst-case residual. Although our theory only guarantees that our method has a better estimation efficiency and a shorter confidence interval length than the unadjusted method under a high signal-to-noise ratio and light-tailed covariates, it is interesting that we can observe improved efficiency even with heavy-tailed covariates.

We notice that for both setups, when $\gamma=3$, sometimes, $\hat{\tau}_{\adj}$ 
slightly outperforms $\htaudb$ in terms of relative RMSE but with larger bias. Since under the worst-case residual, $\hat{\tau}_{\adj}$ has very large relative RMSE, we still recommend using $\htaudb$ for heavy-tail covariates. 
\begin{figure}[H]
  \centering
\begin{subfigure}[b]{\linewidth}
  \includegraphics[width=\linewidth]{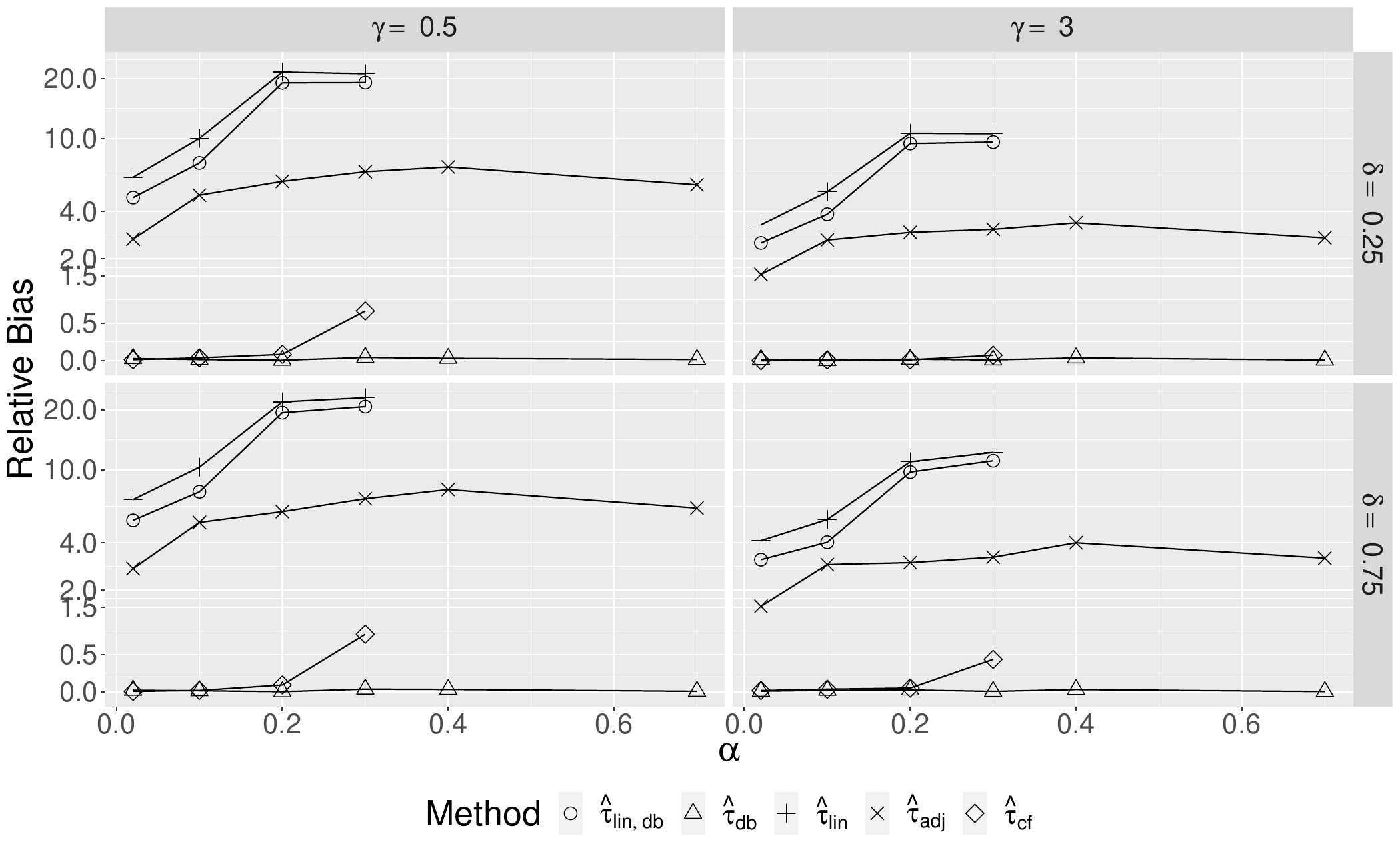} 
  \caption{Worst case residual}     \label{fig:bias_worst_cauchy}
\end{subfigure}
 \begin{subfigure}[b]{\linewidth}
  \centering
  \includegraphics[width=\linewidth]{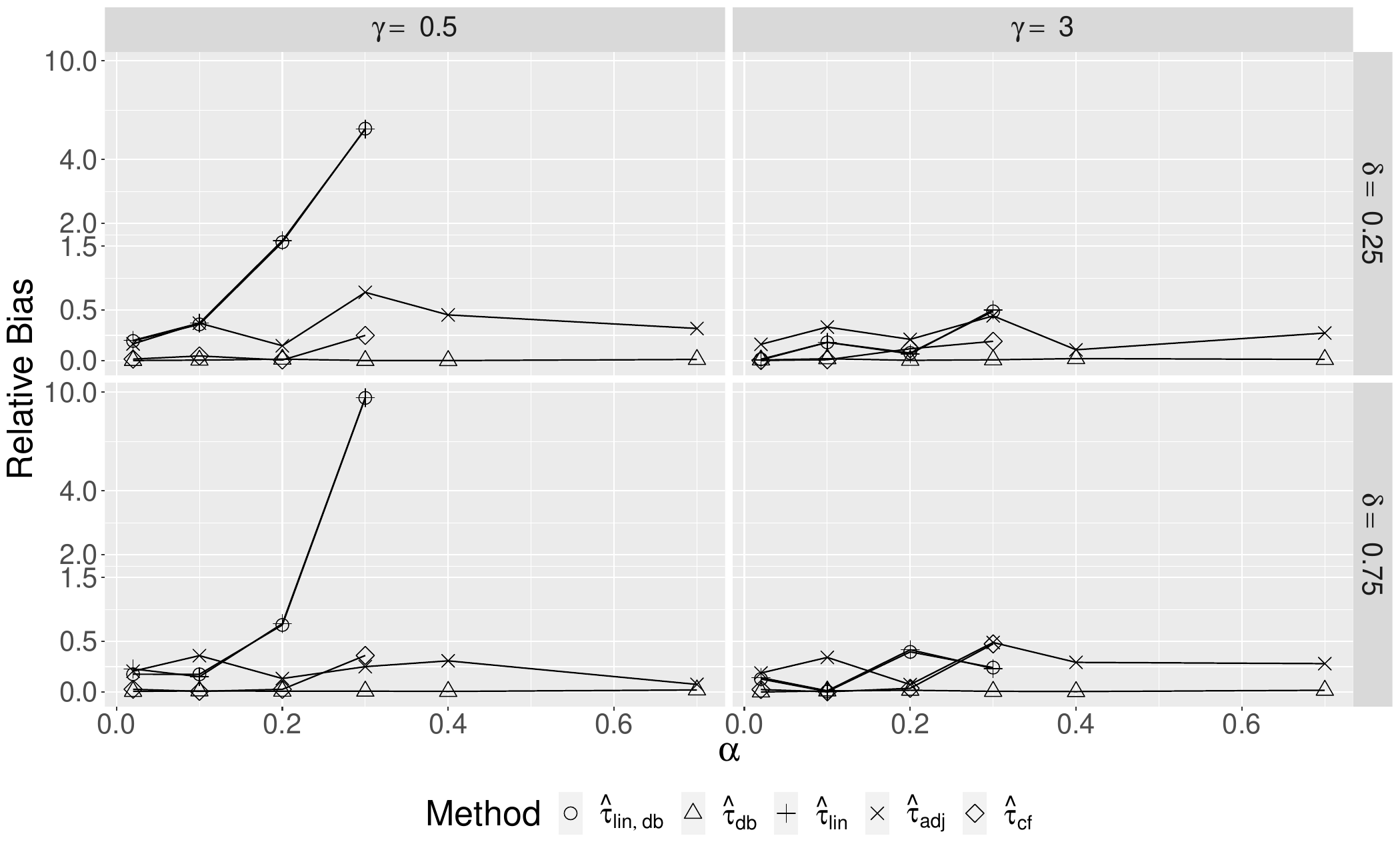}
 \caption{Independent $t$ residual} \label{fig:bias_t_cauchy}
\end{subfigure}

\caption{\emph{The first set up.} Relative bias for different choices of $\gamma$, $\delta$ and $\alpha$ under the worst-case residual and independent $t$ residual. For both figures, we use a transformation of $\log_{10}(1+x)$ for the y-axis to adapt the curve display.} \label{fig:bias-cauchy}

\end{figure}

\begin{figure}[H]
  \centering
\begin{subfigure}[b]{\linewidth}
  \includegraphics[width=\linewidth]{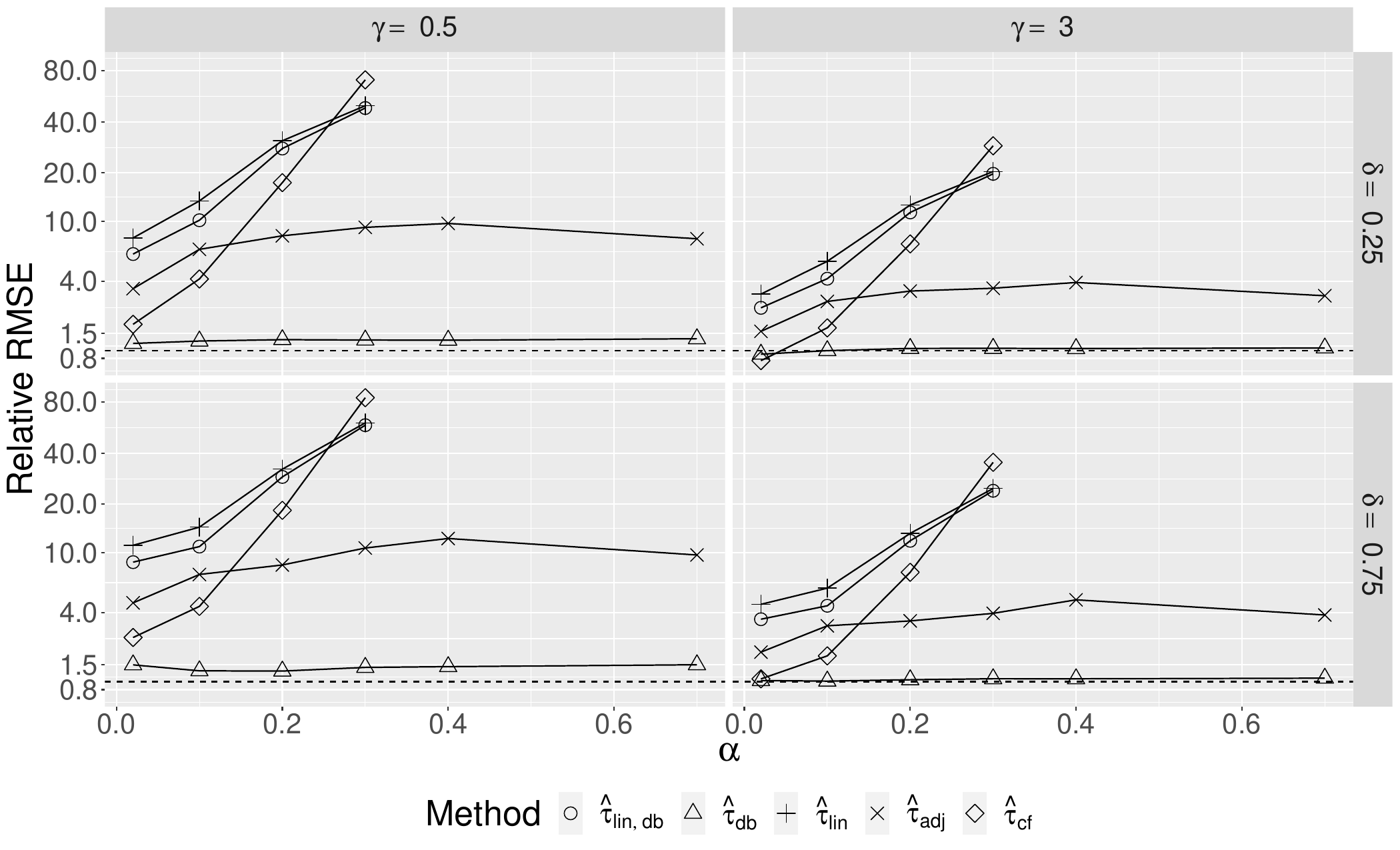}
  \caption{Worst case residual}     \label{fig:rmse_worst_cauchy}
\end{subfigure}
 \begin{subfigure}[b]{\linewidth}
  \centering
  \includegraphics[width=\linewidth]{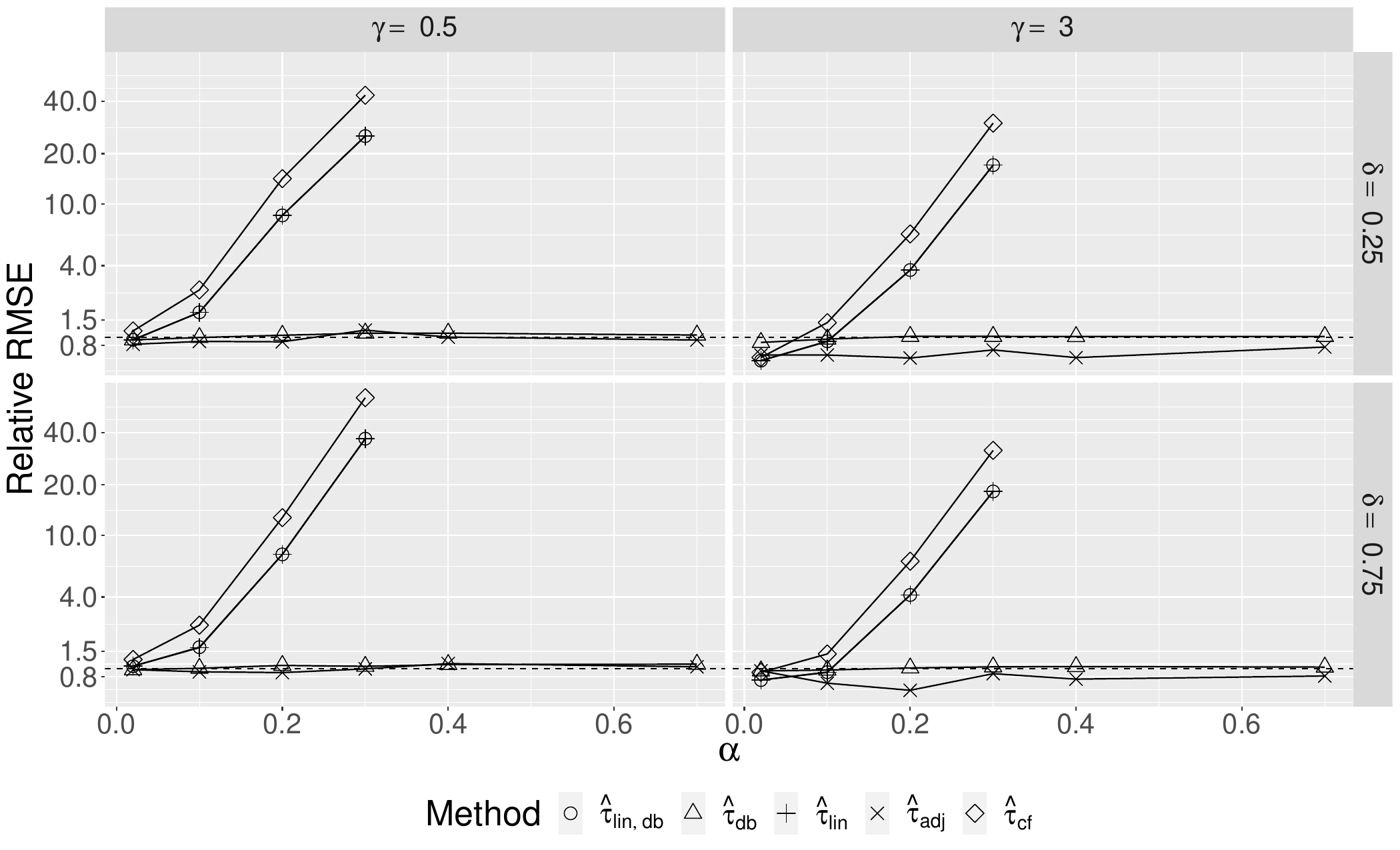}
 \caption{Independent $t$ residual} \label{fig:rmse_t_cauchy}
\end{subfigure}

\caption{\emph{The first set up.} Relative RMSEs for different choices of $\gamma$, $\delta$ and $\alpha$ under the worst-case residual and independent $t$ residual.  The dashed lines signify $1$. For both figures, we use a transformation of $\log_{10}(1+x)$ for the y-axis to adapt the curve display.} \label{fig:rmse_cauchy}
\end{figure}

\begin{figure}[H]
  \centering
\begin{subfigure}[b]{\linewidth}
  \includegraphics[width=\linewidth]{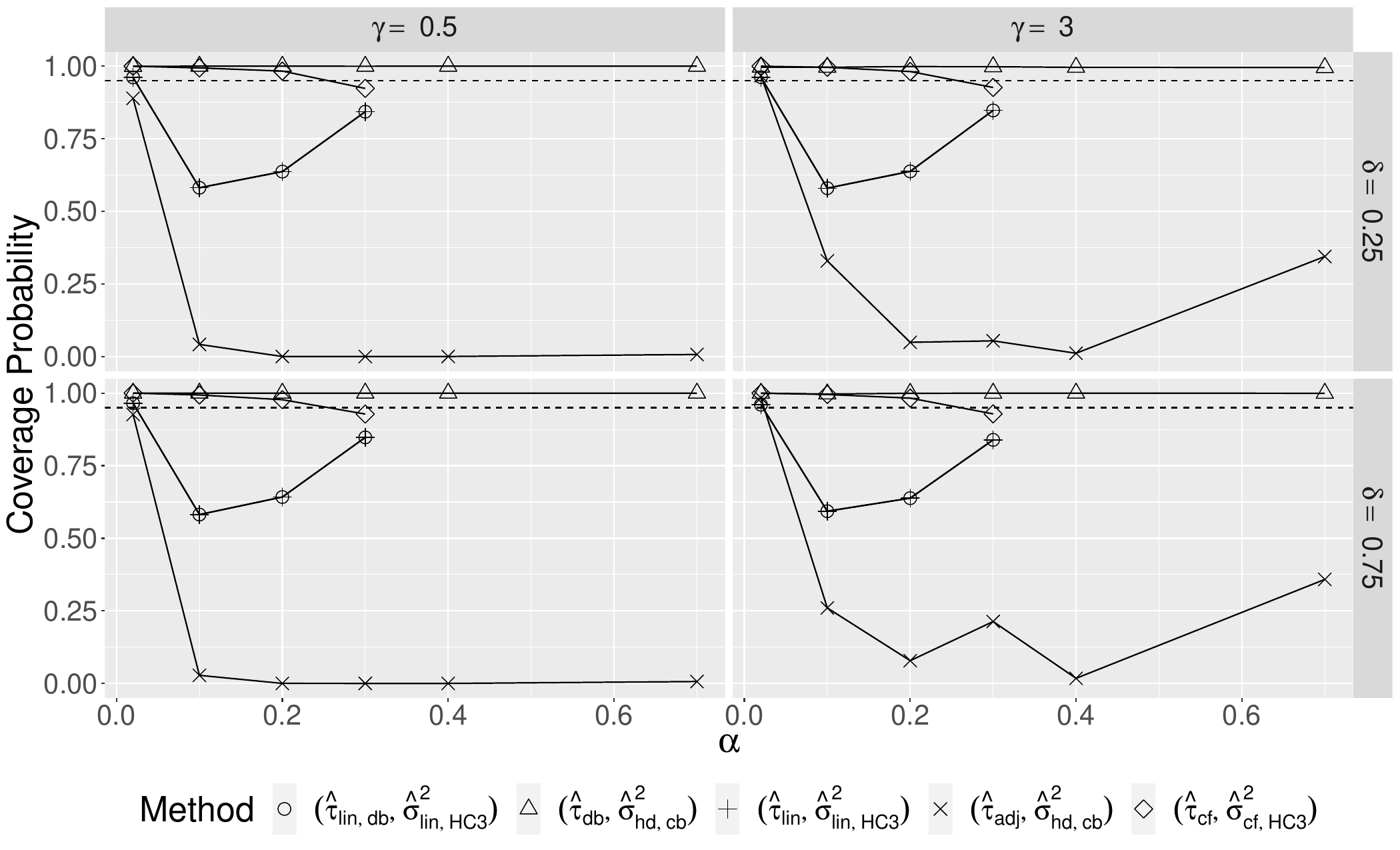}
  \caption{Worst case residual}     \label{fig:cp_worst_cauchy}
\end{subfigure}
 \begin{subfigure}[b]{\linewidth}
  \centering
  \includegraphics[width=\linewidth]{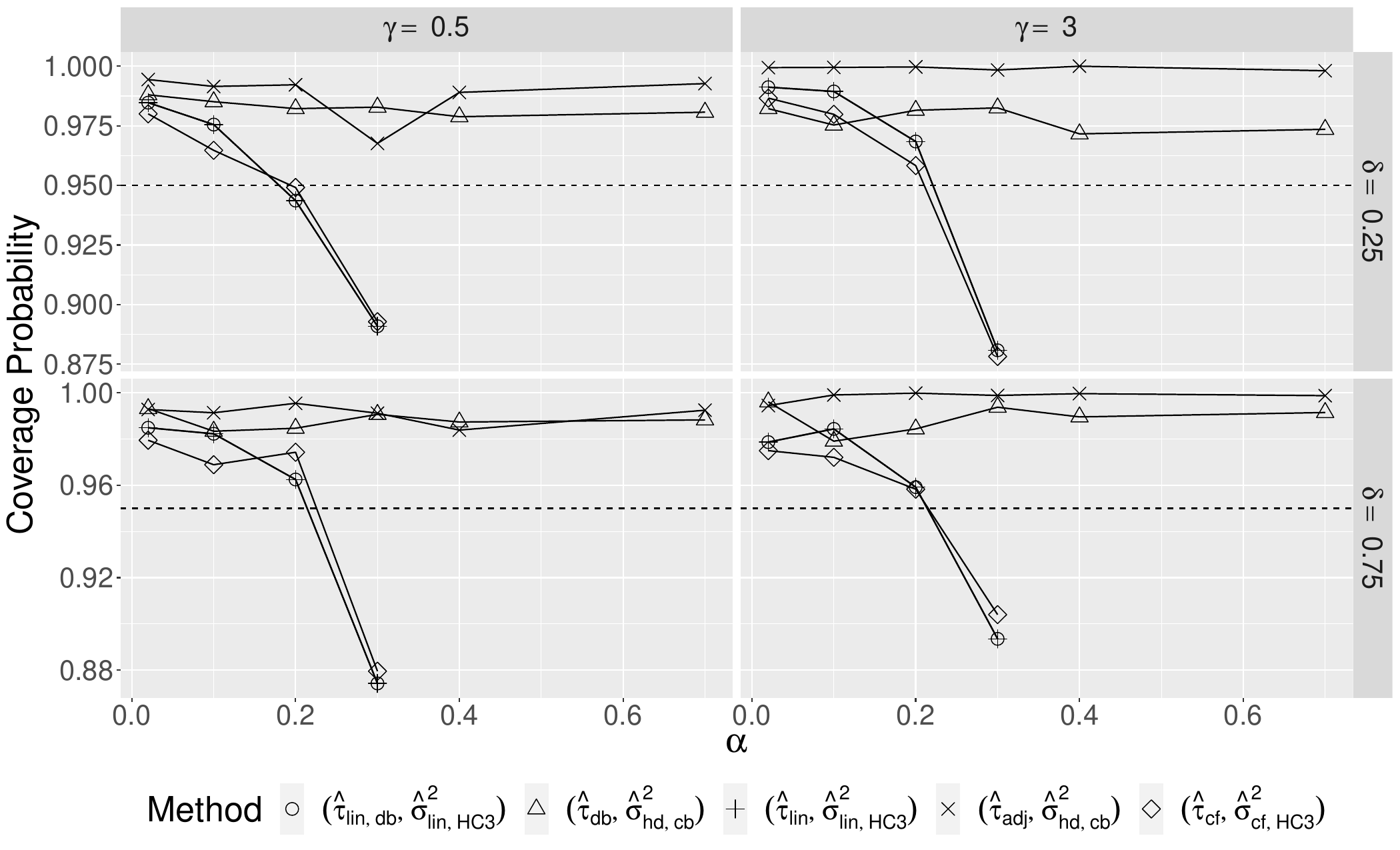}
 \caption{Independent $t$ residual} \label{fig:cp_t_cauchy}
\end{subfigure}
\caption{\emph{The first set up.} Coverage probabilities for different choices of $\gamma$, $\delta$ and $\alpha$ under the worst-case residual and independent $t$ residual.  The dashed lines signify $0.95$.} \label{fig:cp_cauchy}
\end{figure}

\begin{figure}[H]
  \centering
\begin{subfigure}[b]{\linewidth}
  \includegraphics[width=\linewidth]{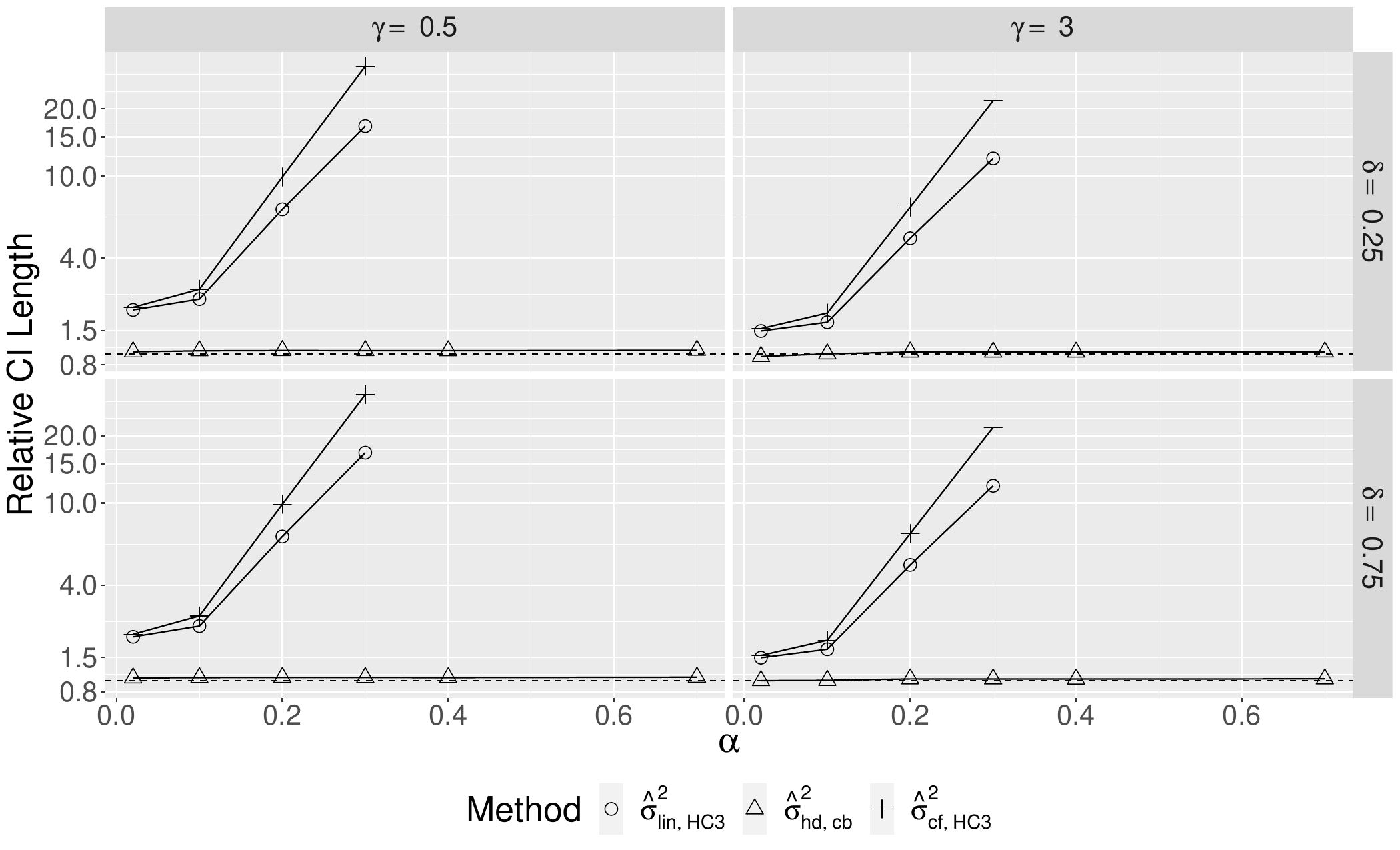}
  \caption{Worst case residual}     \label{fig:ci_worst_cauchy}
\end{subfigure}
 \begin{subfigure}[b]{\linewidth}
  \centering
  \includegraphics[width=\linewidth]{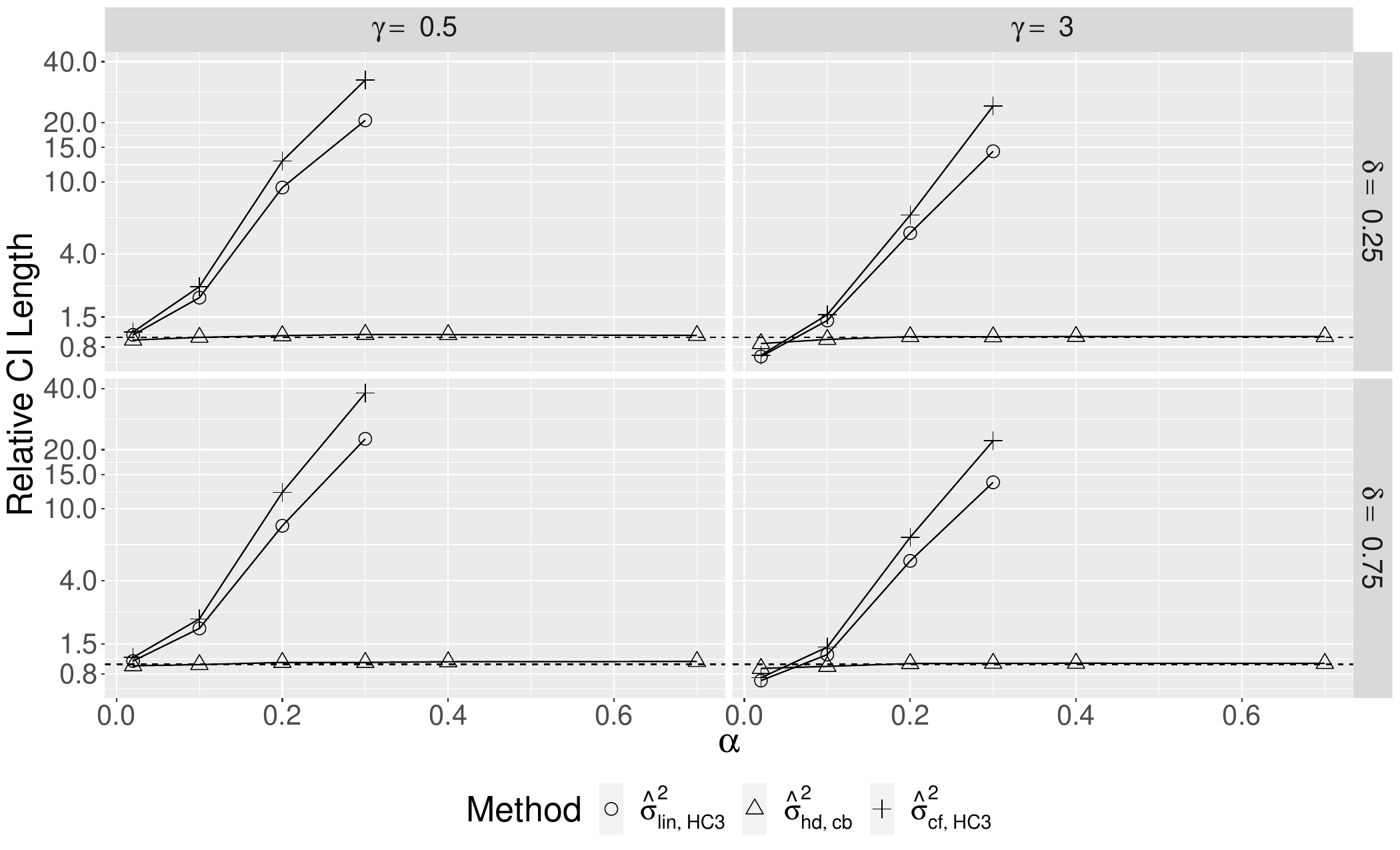}
 \caption{Independent $t$ residual} \label{fig:ci_t_cauchy}
\end{subfigure}
\caption{\emph{The first set up.} Relative confidence interval length for different choices of $\gamma$, $\delta$ and $\alpha$ under the worst-case residual and independent $t$ residual.  The dashed lines signify $1$. For both figures, we use a transformation of $\log_{10}(1+x)$ for the y-axis to adapt the curve display.} \label{fig:ci_cauchy}
\end{figure}

\begin{figure}[H]
  \centering
\begin{subfigure}[b]{\linewidth}
  \includegraphics[width=\linewidth]{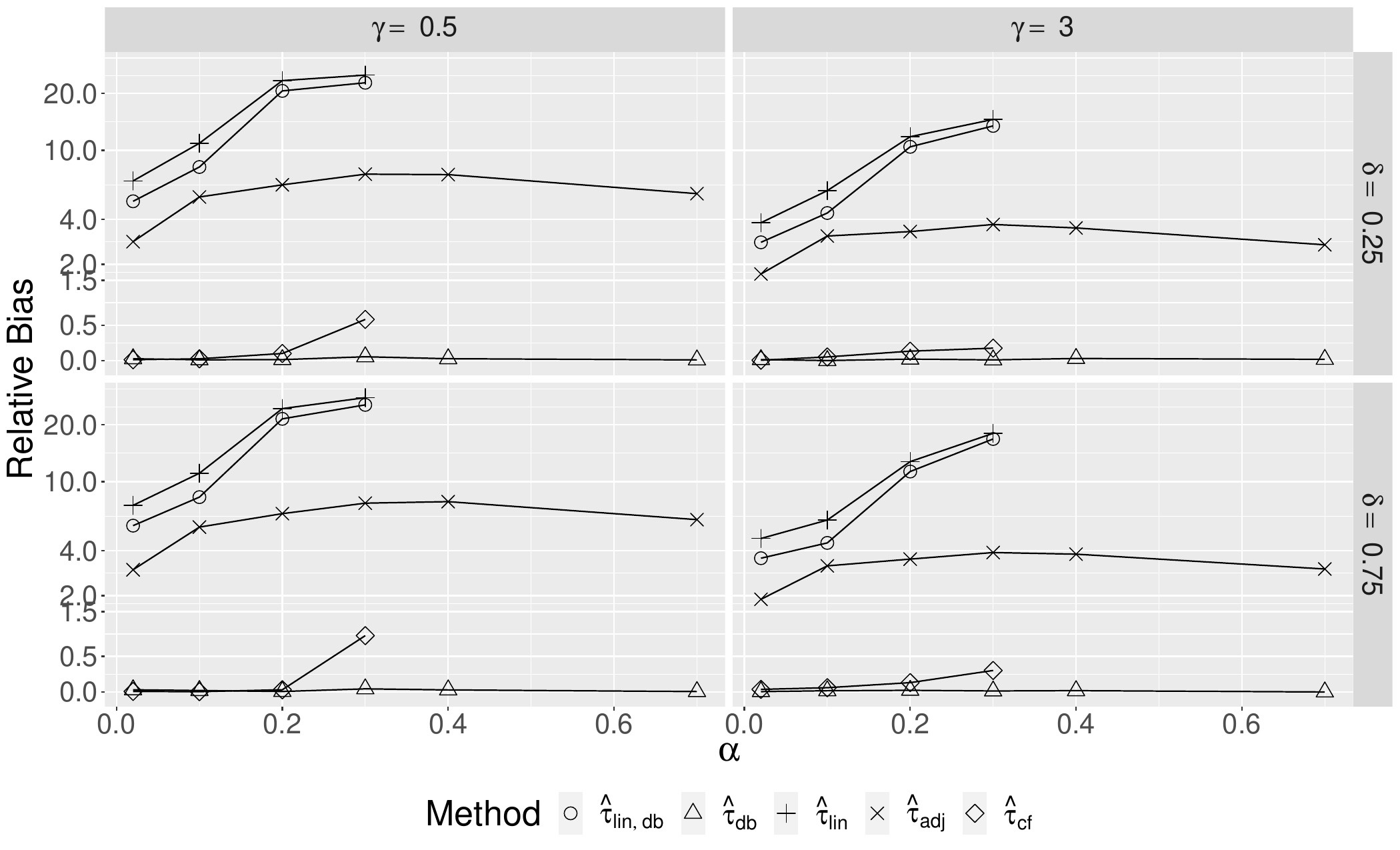}
  \caption{Worst case residual}     \label{fig:bias_worst_cauchy_transy}
\end{subfigure}
 \begin{subfigure}[b]{\linewidth}
  \centering
  \includegraphics[width=\linewidth]{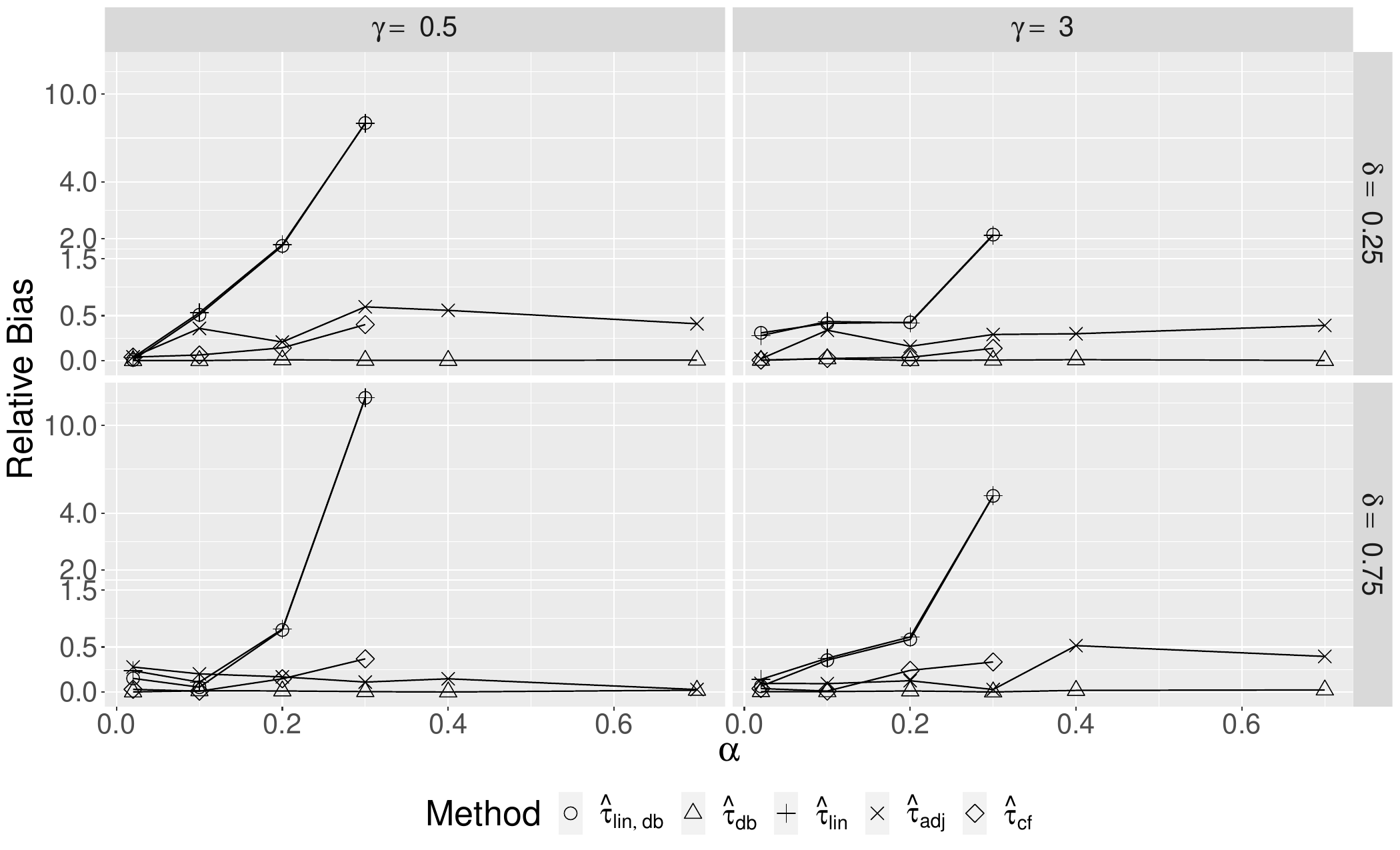}
 \caption{Independent $t$ residual} \label{fig:bias_t_cauchy_transy}
\end{subfigure}
\caption{\emph{The second set up.} Relative bias for different choices of $\gamma$, $\delta$ and $\alpha$ under the worst-case residual and independent $t$ residual.  For both figures, we use a transformation of $\log_{10}(1+x)$ for the y-axis to adapt the curve display.} \label{fig:bias-cauchy-transy}
\end{figure}

\begin{figure}[H]
  \centering
\begin{subfigure}[b]{\linewidth}
  \includegraphics[width=\linewidth]{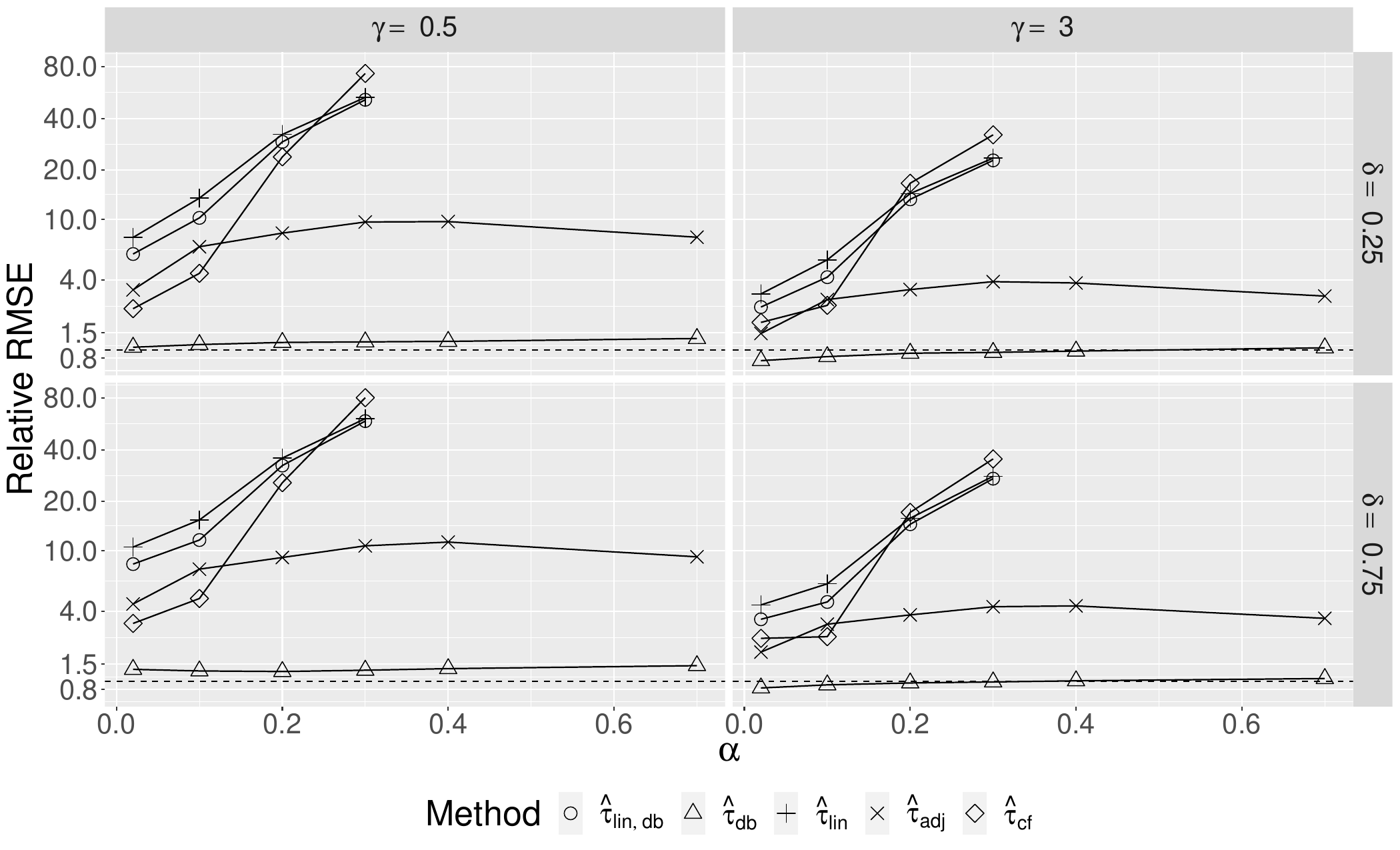}
  \caption{Worst case residual}     \label{fig:rmse_worst_cauchy_transy}
\end{subfigure}
 \begin{subfigure}[b]{\linewidth}
  \centering
  \includegraphics[width=\linewidth]{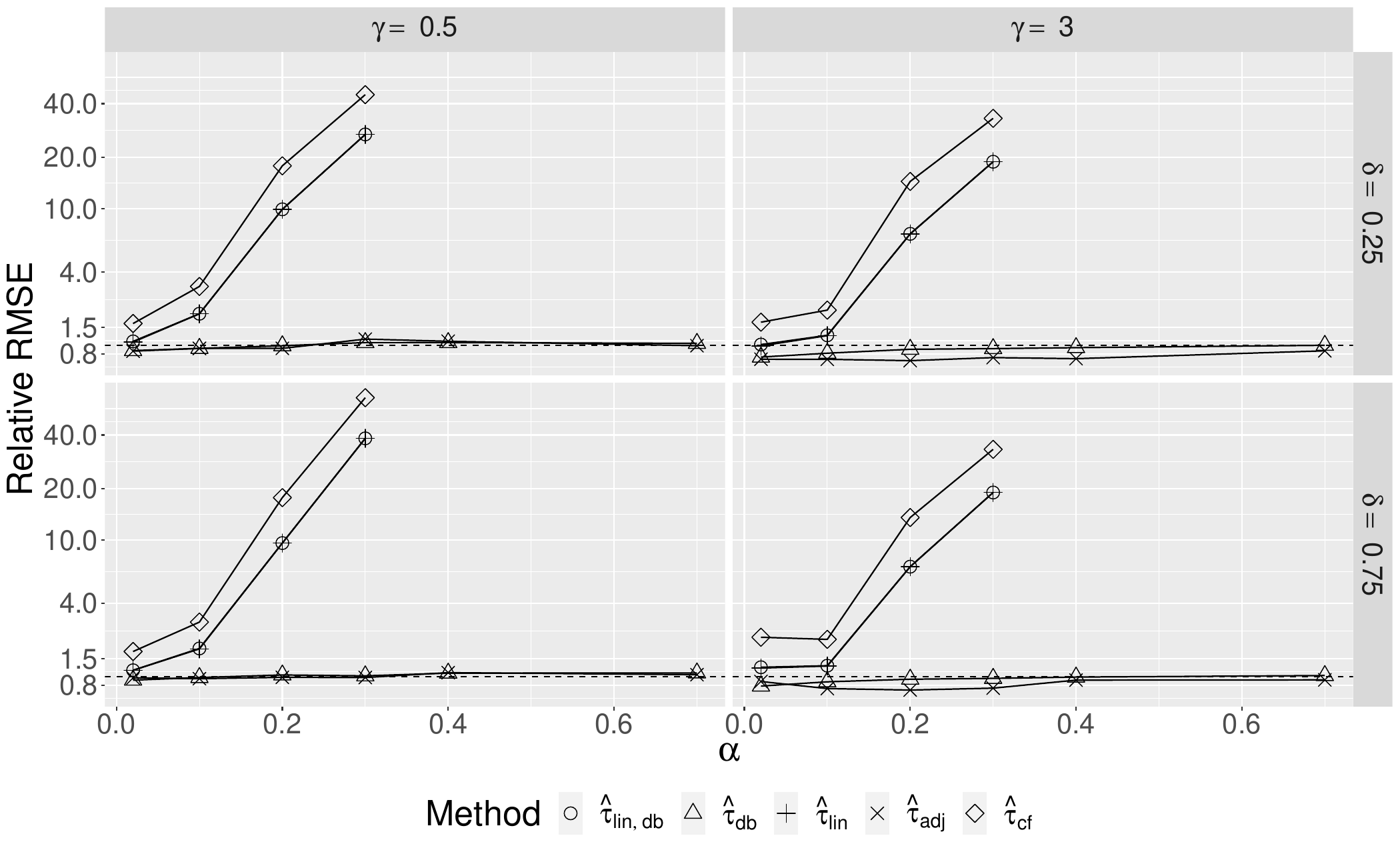}
 \caption{Independent $t$ residual} \label{fig:rmse_t_cauchy_transy}
\end{subfigure}
\caption{\emph{The second set up.} Relative RMSEs for different choices of $\gamma$, $\delta$ and $\alpha$ under the worst-case residual and independent $t$ residual.  The dashed lines signify $1$. For both figures, we use a transformation of $\log_{10}(1+x)$ for the y-axis to adapt the curve display.} \label{fig:rmse_cauchy_transy}
\end{figure}

\begin{figure}[H]
  \centering
\begin{subfigure}[b]{\linewidth}
  \includegraphics[width=\linewidth]{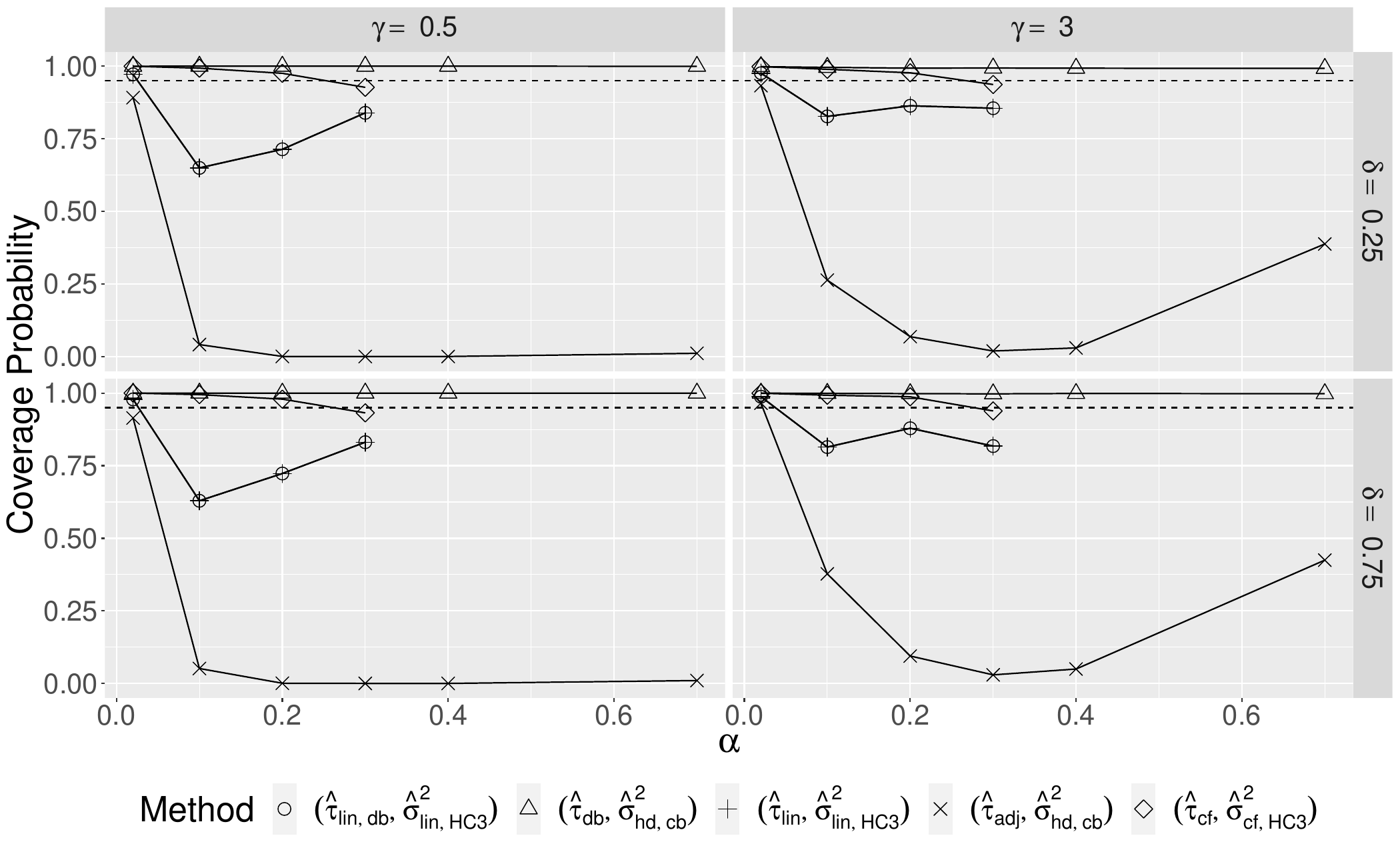}
  \caption{Worst case residual}     \label{fig:cp_worst_cauchy_transy}
\end{subfigure}
 \begin{subfigure}[b]{\linewidth}
  \centering
  \includegraphics[width=\linewidth]{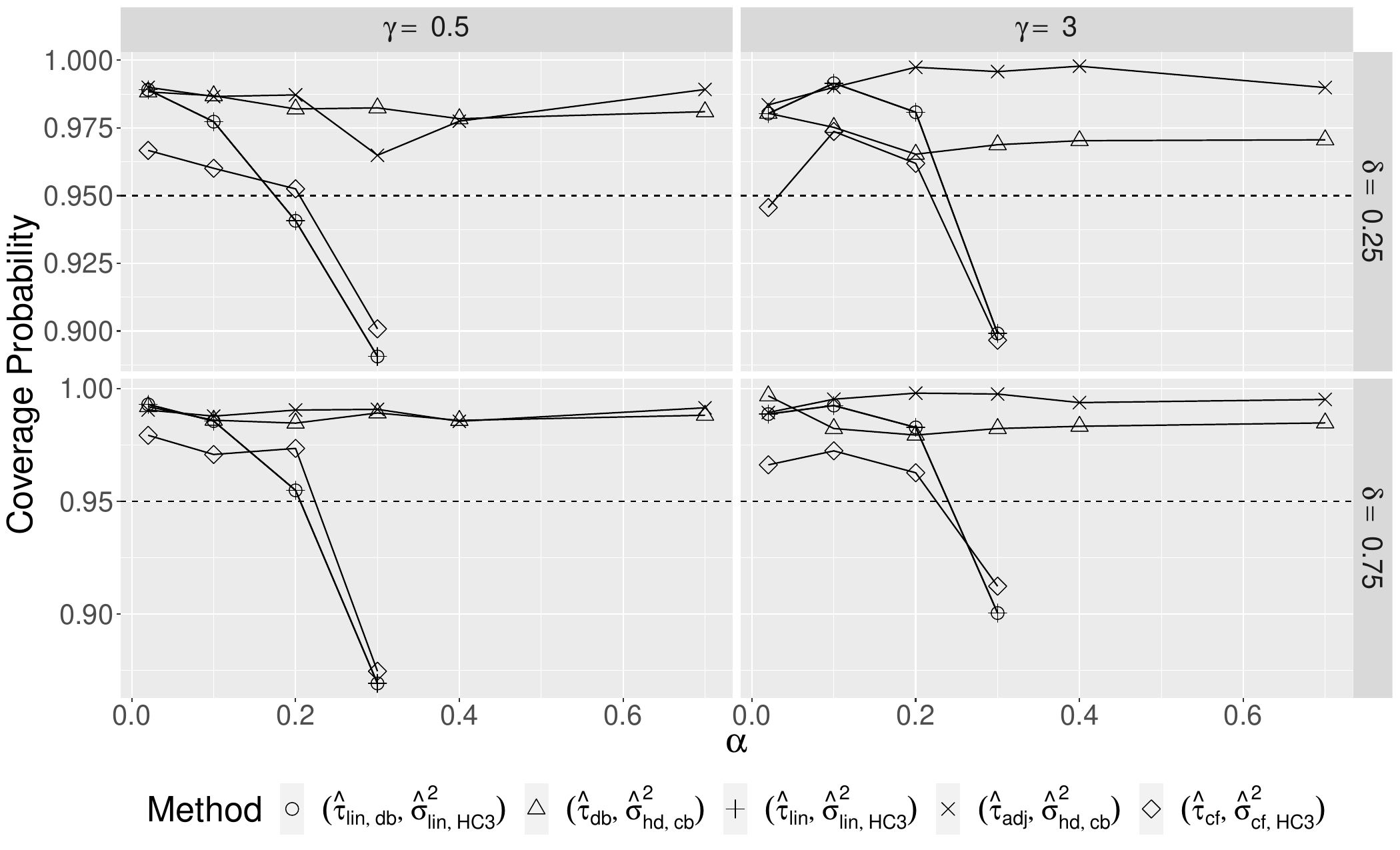}
 \caption{Independent $t$ residual} \label{fig:cp_t_cauchy_transy}
\end{subfigure}
\caption{\emph{The second set up.} Coverage probabilities for different choices of $\gamma$, $\delta$ and $\alpha$ under the worst-case residual and independent $t$ residual.  The dashed lines signify $0.95$.} \label{fig:cp_cauchy_transy}
\end{figure}

\begin{figure}[H]
  \centering
\begin{subfigure}[b]{\linewidth}
  \includegraphics[width=\linewidth]{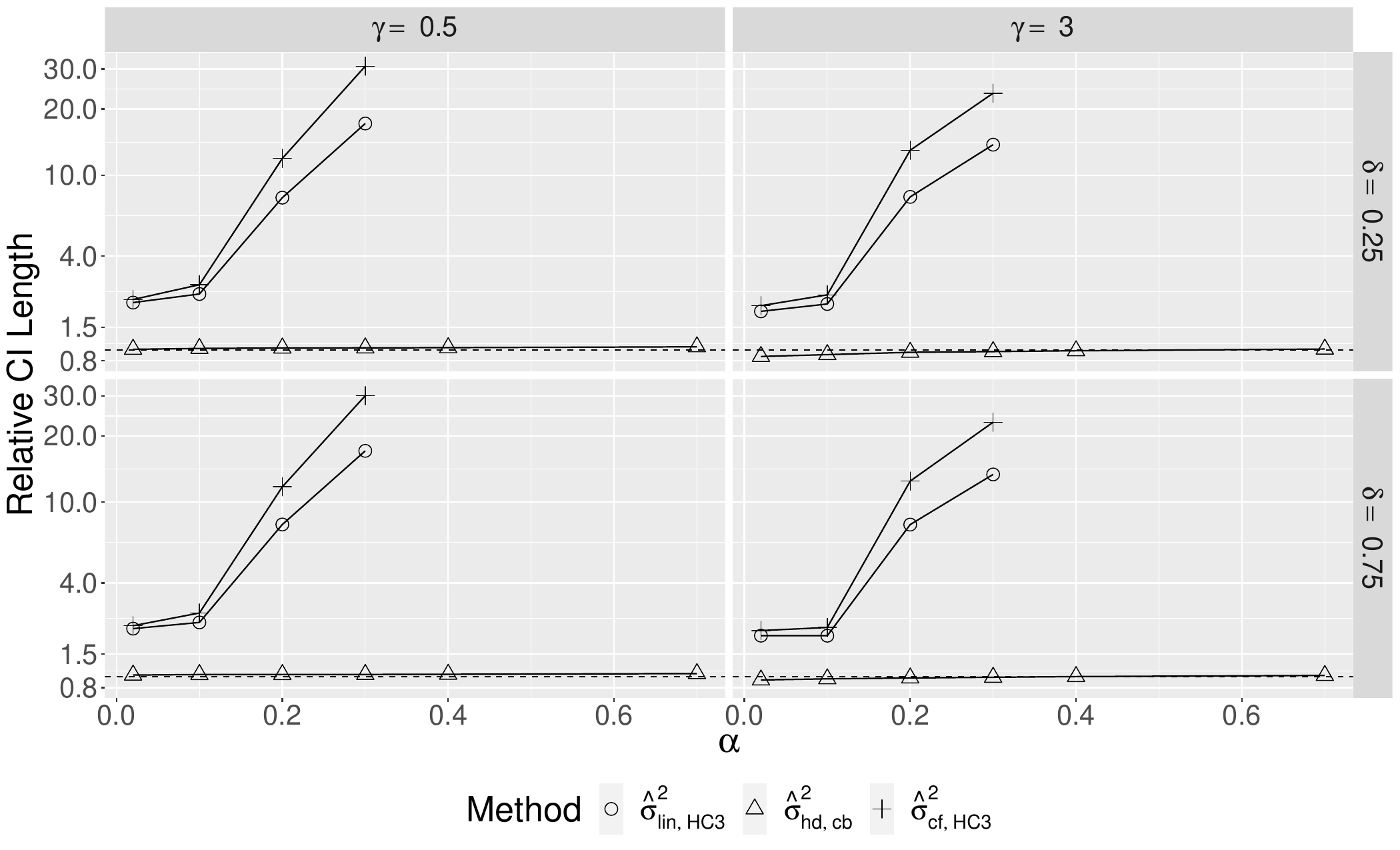}
  \caption{Worst case residual}     \label{fig:ci_worst_cauchy_transy}
\end{subfigure}
 \begin{subfigure}[b]{\linewidth}
  \centering
  \includegraphics[width=\linewidth]{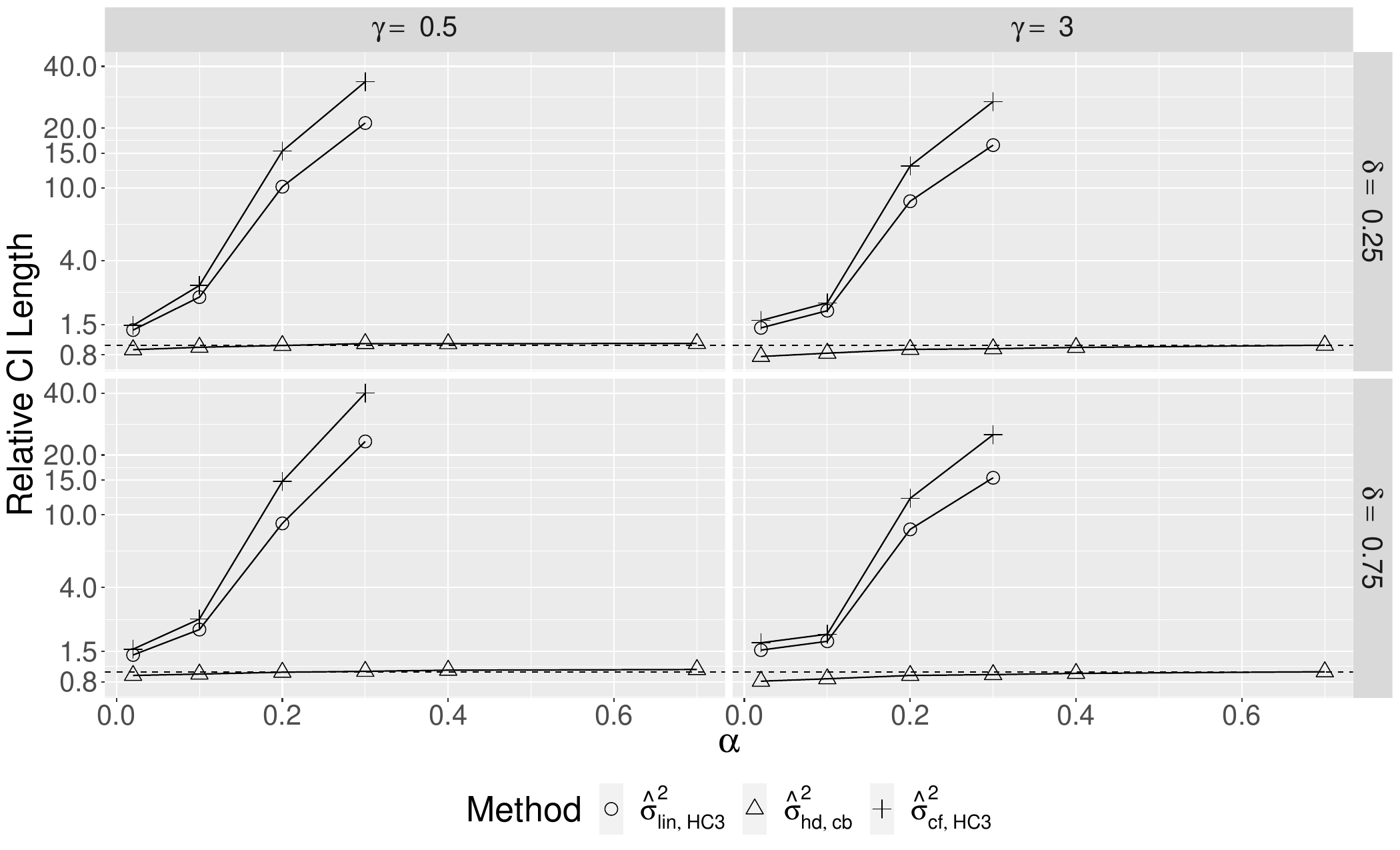}
 \caption{Independent $t$ residual} \label{fig:ci_t_cauchy_transy}
\end{subfigure}
\caption{\emph{The second set up.} Relative confidence interval length for different choices of $\gamma$, $\delta$ and $\alpha$ under the worst-case residual and independent $t$ residual.  The dashed lines signify $1$. For both figures, we use a transformation of $\log_{10}(1+x)$ for the y-axis to adapt the curve display.} \label{fig:ci_cauchy_transy}
\end{figure}

\newpage

	\printglossary[title={Notation table}]

\end{document}